\documentclass[12pt]{article}

\usepackage[margin=0.75in]{geometry}
\usepackage{color}
\usepackage{graphicx}
\usepackage{epsfig}
\graphicspath{{./eps/}}
\usepackage{epstopdf}
\usepackage{subfig}

\usepackage{amsmath}
\usepackage{amssymb}
\usepackage{amsthm}

\usepackage{rotating}
\usepackage{verbatim}
\usepackage{hyperref}
\usepackage{bm}
\usepackage[normalem]{ulem}
\usepackage[authoryear]{natbib}
\usepackage{bbm}
\usepackage{array}
\usepackage{enumerate}
\usepackage{float}
\usepackage{authblk}% delete this for AOS template
\usepackage{algorithm,algorithmic}

\usepackage[flushleft]{threeparttable}
\usepackage{tablefootnote}

\usepackage[perpage]{footmisc}

\newtheorem{theorem}{Theorem}[section]
\newtheorem{lemma}[theorem]{Lemma}

\newtheorem{proposition}[theorem]{Proposition}

\theoremstyle{definition}
\newtheorem{definition}{Definition}[section]

\usepackage{thmtools}
\declaretheorem[style=definition]{remark}
\declaretheorem[style=definition]{example}%
\renewcommand{\P}{\mathbb{P}}
\newcommand{\pr}[1]{\mbox{$\mathbb{P}\left(#1\right)$}}
\newcommand{\prs}[2]{\mbox{$\mathbb{P}_{#1}\left(#2\right)$}}
\newcommand{\E}[1]{\mbox{$\mathbb{E}\left[#1\right]$}}
\newcommand{\Pc}[2]{\mbox{$\mathbb{P}\left(\left.#1 \ \right| #2\right)$}}
\newcommand{\Pcr}[2]{\mbox{$\mathbb{P}\left(#1 \left|\ #2\right.\right)$}}
\newcommand{\Pcrs}[3]{\mbox{$\mathbb{P}_{#1}\left(#2 \left|\ #3\right.\right)$}}
\newcommand{\Ec}[2]{\mbox{$\mathbb{E}\left[\left.#1 \ \right| #2\right]$}}
\newcommand{\iid}{\stackrel{\emph{i.i.d.}}{\sim}}
\def\independenT#1#2{\mathrel{\rlap{$#1#2$}\mkern2mu{#1#2}}}
\newcommand{\indp}{\protect\mathpalette{\protect\independenT}{\perp}}
\def\eqd{\,{\buildrel \mathcal{D} \over =}\,}

\newcommand{\mrc}[1]{\mbox{$\varrho\left(#1\right)$}}

\makeatletter
\newcommand*{\indep}{%
  \mathbin{%
    \mathpalette{\@indep}{}%
  }%
}
\newcommand*{\nindp}{%
  \mathbin{%                   % The final symbol is a binary math operator
    \mathpalette{\@indep}{\not}% \mathpalette helps for the adaptation
  }%
}
\newcommand*{\@indep}[2]{% create independent symbol
  \sbox0{$#1\perp\m@th$}%        box 0 contains \perp symbol
  \sbox2{$#1=$}%                 box 2 for the height of =
  \sbox4{$#1\vcenter{}$}%        box 4 for the height of the math axis
  \rlap{\copy0}%                 first \perp
  \dimen@=\dimexpr\ht2-\ht4-.2pt\relax
  \kern\dimen@
  {#2}%
  \kern\dimen@
  \copy0 %                       second \perp
} 
\makeatother

\newcommand{\bs}{\boldsymbol}

\newcommand{\ovec}{\ensuremath{\operatorname{vec}}}

\newcommand{\diag}[1]{\mbox{$\mathrm{diag}\left\{#1\right\}$}}
\newcommand{\one}[1]{\mbox{$\mathbf{1}_{\left\{#1\right\}}$}}

\def\R{\mathbb{R}}
\def\Z{\mathbb{Z}}

\newcommand{\dhmargin}[2]{{\color{red}#1}\marginpar{\color{red}\raggedright\tiny[DH]:#2}}
\newcommand{\rev}[1]{{\color{black}#1}}
\newcommand{\pkg}[1]{{\normalfont\fontseries{b}\selectfont #1}}

\newcommand{\N}{\mathcal{N}}
\newcommand{\bX}{\bs{X}}
\newcommand{\Xk}{\tilde{X}}
\newcommand{\bZ}{\bs{Z}}
\newcommand{\bz}{\bs{z}}
\newcommand{\bS}{\bs{\Sigma}}
\newcommand{\bO}{\bs{\Omega}}
\newcommand{\hbS}{\hat{\bs{\Sigma}}}
\newcommand{\bXk}{\tilde{\bs{X}}}

\newcommand{\bsfx}{\bs{\mathsf{x}}}
\newcommand{\bsfX}{\bs{\mathsf{X}}}
\newcommand{\bsfXk}{\tilde{\bs{\mathsf{X}}}}
\newcommand{\bsfXsw}{[\bsfXk_1,\,\bsfX_{\text{-}1}]}
\newcommand{\bsfXksw}{[\bsfX_1,\,\bsfXk_{\text{-}1}]}

\newcommand{\bsfV}{\bs{\mathsf{V}}}
\newcommand{\bsfW}{\bs{\mathsf{W}}}
\newcommand{\bsfWk}{\bs{\tilde{\mathsf{W}}}}

\newcommand{\by}{\bs{y}}
\newcommand{\bx}{\bs{x}}
\newcommand{\St}{ {K }}%For dicrete states 1
\newcommand{\st}{ k }%2
\newcommand{\bst}{ \bs{k} }%3
\newcommand{\bSt}{ \bs{K} }%3

\newcommand{\bb}{\bs{\beta}}

\newcommand{\bmu}{\bs{\mu}}
\newcommand{\hbmu}{\hat{\bs{\mu}}}

\newcommand{\bW}{\bs{W}}
\newcommand{\bw}{\bs{w}}
\newcommand{\bQ}{\bs{Q}}
\newcommand{\bL}{\bs{L}}

\newcommand{\bU}{\bs{U}}
\newcommand{\bV}{\bs{V}}

\newcommand{\bC}{\bs{C}}
\newcommand{\bG}{\bs{G}}
\newcommand{\bss}{\bs{s}} 

\newcommand{\Fxy}{F_{Y,X}}
\newcommand{\Fyx}{F_{Y|X}}
\newcommand{\Fx}{F_X}
\newcommand{\noj}{{\text{-}j}}
\newcommand{\swap}[1]{{\text{swap}(#1)}}
\newcommand{\tp}{^{\top}}
\newcommand{\oKL}{\widehat{\text{KL}}}
\newcommand{\nullset}{\mathcal{H}_{0}}
\newcommand{\Xnoj}{X_{\noj}}

\newcommand{\floor}[1]{\lfloor#1\rfloor}%floor function

\newcommand{\nex}{three } % number of examples

\def\ct{ \bs{H} }  % locally define contingency table
\def\smallcell{ \bs{\mathsf{H}} } % the cell of ct
\def\icc{k} % index for connected components
\def\ncc{\ell}% number of connected components
\def \ncg {m}   % number of folds for splitting
\def\setE{\mathcal{M}_{}}  % the support of [x,xk] 

\graphicspath{{./fig/}}

\numberwithin{equation}{section}

\title{Relaxing the Assumptions of Knockoffs by Conditioning}
\author[]{Dongming Huang}
\author[]{Lucas Janson}
\date{}
\affil[]{Department of Statistics, Harvard University}

\begin{document}
\maketitle

\begin{abstract}
The recent paper \citet{EC-ea:2018} introduced model-X knockoffs, a method for variable selection that provably and non-asymptotically controls the false discovery rate with no restrictions or assumptions on the dimensionality of the data or the conditional distribution of the response given the covariates. The one requirement for the procedure is that the covariate samples are drawn independently and identically from a precisely-known (but arbitrary) distribution. The present paper shows that the exact same guarantees can be made \emph{without} knowing the covariate distribution fully, but instead knowing it only up to a parametric model
with as many as $\Omega(n^{*}p)$ 
%\footnote{The notation $a_n=\Omega(b_n)$ for two sequences $\{a_n\}_{n=1}^\infty$ and $\{b_n\}_{n=1}^\infty$ means that there is a positive constant $c$ and integer $n_0$ such that $a_n\ge c b_n$ for all $n\ge n_0$.} 
parameters, where $p$ is the dimension and $n^{*}$ is the number of covariate samples (which may exceed the usual sample size $n$ of labeled samples when unlabeled samples are also available). 
%including when the number of parameters far exceeds the dimension.
 % This requirement is notably different from the majority of work on inference for conditional models, which normally assumes the \emph{response} samples are drawn independently from a known family of conditional models given the covariates (such as a generalized linear model), parameterized by roughly $p$ (the number of covariates) parameters, while assuming little-to-nothing about the covariate distribution. Either type of assumption astronomically reduces an exponentially large (in $p$) function space in order to make inference mathematically feasible, and in practice acts instead as an approximation, with the accuracy of the approximation depending on the strength and type of domain knowledge a user has about his or her data. However, the latter approach would seem to have a slight advantage since, even in theory, it does not require \emph{exact} knowledge of the response's conditional distribution while the former approach does require the covariate distribution to be known exactly. This paper removes this disadvantage by showing that model-X knockoffs can be carried out with the same guarantees while only requiring the covariate distribution be known up to a parameterized family of distributions. 
The key is to treat the covariates as if they are drawn conditionally on their observed value for a sufficient statistic of the model. 
Although this idea is simple, even in Gaussian models conditioning on a sufficient statistic leads to a distribution supported on a set of zero Lebesgue measure, requiring techniques from topological measure theory to establish valid algorithms. We demonstrate how to do this for three models of interest, with simulations showing the new approach remains powerful under the weaker assumptions. 

\end{abstract}

%\tableofcontents

\section{Introduction}
\subsection{Problem statement}
In this paper we consider random variables $(Y,X_1,\dots,X_p)$ where $Y$ is a response or outcome variable, each $X_j$ is a potential explanatory variable (also known as a covariate or feature) and $p$ is the dimensionality, or number of covariates. For instance, $Y$ could be the binary indicator of whether a patient has a disease or not, and $X_j$ could be the number of minor alleles at a specific location (indexed by $j$) on the genome, also known as a single nucleotide polymorphism (SNP). A common question of interest is which of the $X_j$ are important for determining $Y$, with importance defined in terms of conditional independence. That is, $X_j$ is considered \emph{un}important (or \emph{null}) if 
\[Y \indp X_j \mid \Xnoj,\]
where $\Xnoj = \{X_1,\dots,X_p\}\setminus \{X_j\}$; stated another way, $X_j$ is unimportant exactly when $Y$'s conditional distribution does not depend on $X_j$. %Denote by $\nullset$ the set of all $j$ such that $Y\indp X_j \,|\, X_\noj$. 
Denote by $\nullset$ the set of all $j$ such that $X_j$ is unimportant.
As discussed in \citet{EC-ea:2018}, under very mild conditions the complement of  the set of unimportant variables, i.e., the \emph{important} (or \emph{non-null}) variables, constitutes the Markov blanket $S$ of $Y$, namely, the unique smallest set $S$ such that $Y\indp X_{S}\mid X_{\text{-}S}$. Note that when $Y\,|\, X_1,\dots,X_p$ follows a generalized linear model (GLM) with no redundant covariates, the set of important variables exactly equals the set of variables with nonzero coefficients, as usual \citep{EC-ea:2018}.

In our search for the Markov blanket we usually cannot possibly hope for perfect recovery, so we instead attempt to maximize the number of important variables discovered while probabilistically controlling the number of false discoveries. In this paper, as with most others in the knockoffs literature,\footnote{\citet{LJ-WS:2016} show how the last step of knockoffs can easily be modified to control other error rates such as the $k$-familywise error rate.} we consider the false discovery rate (FDR) \citep{YB-YH:1995}, defined for a (random) selected subset of variables $\hat{S}$ as 
\[\text{FDR} := \E{\frac{|\hat{S}\,\cap \,\nullset|}{|\hat{S}|}},\]
i.e., the expected fraction of discoveries that are not in the Markov blanket (false discoveries), where we use the convention that $0/0 = 0$. Controlling the FDR at, say, $10\%$ is powerful as compared to controlling more classical error rates like the familywise error rate, while still being interpretable, allowing a statistician to report a conclusion such as ``here is a set of covariates $\hat{S}$, 90\% of which I expect to be important.'' 

\subsection{Our contribution}\label{sec:contribution}
In our discussion of approaches to this problem, we will draw on a fundamental decomposition of the joint distribution $\Fxy$ of $(Y,X_1,\dots,X_p)$ into the product of the conditional distribution $\Fyx$ of $Y\,|\, X_1,\dots,X_p$ and the joint distribution $\Fx$ of $X_1,\dots,X_p$. The canonical approach to inference, which we refer to as the `fixed-X' approach, assumes $\Fyx$ is a member of a parametric family of conditional distributions (e.g., a GLM), while placing weak or no assumptions on $\Fx$. In fact, the fixed-X approach usually treats the observed values of $X_{i,1},\dots,X_{i,p}$ for $i=1,\dots,n$ as fixed; that is, it performs inference \emph{conditionally} on the observed values of $X_1,\dots,X_p$ in the data, which also allows the covariate rows to be drawn from different distributions or even be deterministic (fixed). The approach proposed in \citet{EC-ea:2018}, referred to therein as the `model-X' approach, assumes the observations $(Y_i,X_{i,1},\dots,X_{i,p})\iid\Fxy$ and places no restrictions on $\Fx$ but assumes it is known exactly, while assuming nothing about $\Fyx$. So, to summarize slightly imprecisely, the canonical, fixed-X approach to inference places all assumptions on $\Fyx$ and none on $\Fx$, while the model-X approach does the opposite by placing all assumptions on $\Fx$ and none on $\Fyx$. 

Note that both $\Fyx$ and $\Fx$ are exponentially complex in $p$: in the simple case where each element of $(Y,X_1,\dots,X_p)$ is categorical with $k$ categories, i.e., $(Y,X_1,\dots,X_p)\in\{1,\dots,k\}^{p+1}$, it is easily seen that a fully general model for $\Fyx$ has $(k-1)k^p$ free parameters while $\Fx$ has only slightly fewer with $k^p-1$. So both fixed-X and model-X approaches astronomically reduce an exponentially large (in $p$) space of distributions in order to make inference feasible, highlighting the importance of robustness, assumption-checking, and domain knowledge for justifying the resulting inference;
%(\emph{approximations} is probably more accurate than \emph{assumptions}, since one never really believes such assumptions to hold exactly in practice)
see \citet[Chapter 1]{LJ:2017} for a detailed discussion of the role of fixed-X and model-X\footnote{Therein referred to as `model-based' and `model-free', respectively.} assumptions in high-dimensional inference. With that said, one apparent advantage of the fixed-X approach is that it does not require \emph{exact} knowledge of $\Fyx$, while the model-X approach of \cite{EC-ea:2018} does require $\Fx$ be known exactly. 

The present paper removes this apparent advantage by showing that model-X knockoffs can still provide powerful and exact, finite-sample inference even when the covariate distribution is only known up to a parameterized family of distributions (also known as a model), as opposed to known exactly. In fact, \rev{in Section~\ref{sec:ex}} we will show examples in which the number of parameters we allow for $\Fx$'s model is $\Omega(n^{*}p)$, where $n^{*}$ is the total number of samples of $X$ (including unlabeled samples), which is always at least as large as the number of labeled samples $n$, and can be much larger in some applications. This is much greater than the number of parameters allowed in the model for $\Fyx$ in fixed-X inference (see Section~\ref{sec:relwork}). Table~\ref{table:modeldim} provides a summarized comparison of the model flexibility allowed in the fixed-X and model-X approaches.
\begin{table}[H]
\begin{tabular}{rcccc}
                                                  & Model for $\Fyx$                                        & Model for $\Fx$                                                                     &  &  \\ \cline{2-3}
\multicolumn{1}{r|}{Fixed-X}                      & \multicolumn{1}{c|}{$o(n)$ parameters\tablefootnote{In the exceptional case of Gaussian linear regression, $n$ parameters are allowed.}\textsuperscript{,}\tablefootnote{Except for Gaussian linear regression, fixed-X inferential guarantees are only asymptotic.}} & \multicolumn{1}{c|}{arbitrary}          &  &  \\ \cline{2-3}
\multicolumn{1}{r|}{Model-X \citep{EC-ea:2018}}             & \multicolumn{1}{c|}{arbitrary}                   & \multicolumn{1}{c|}{$0$ parameters}      &  &  \\ \cline{2-3}
\multicolumn{1}{r|}{Model-X (this paper)} & \multicolumn{1}{c|}{arbitrary}                   & \multicolumn{1}{c|}{$\Omega( n^{*} p)$ parameters} &  &  \\ \cline{2-3}

\end{tabular}
\caption{
\rev{Maximum complexity of models allowed by existing methods (see Section~\ref{sec:relwork}) and our proposal (see the list in Section~\ref{sec:ck} and also Section~\ref{sec:unlabel} for the explanation for $\Omega( n^{*} p)$) for controlled variable selection.}
%Maximum complexity of models allowed by existing methods (see Section~\ref{sec:relwork}) for controlled variable selection.
% in the fixed-X approach, the model-X approach of \citet{EC-ea:2018}, and the conditional model-X approach of the present paper. 
Note that without assuming a model, $\Fyx$ and $\Fx$ are of similar complexity (exponentially large in $p$).}
\label{table:modeldim}
\end{table}
Of course the above discussion and table refer only to the \emph{mathematical} complexity of models allowed by the fixed-X and model-X approaches. An analyst's decision between them should depend on how well domain knowledge and/or auxiliary data support their (very different) assumptions. But in light of Table~\ref{table:modeldim}, it seems the conditional model-X approach is easiest to justify unless substantially more is known about $\Fyx$ than $\Fx$.

%, (2) while the model-X approach assumes a certain level of knowledge of $\Fx$, in principle (ignoring computational considerations) inference can be carried out for any $\Fx$; in contrast, the fixed-X approach both assumes knowledge about $\Fyx$ \emph{and} is only known to be viable for certain classes of models, (3) in high dimensions ($p>n$) even a $p$-dimensional model such as a GLM for $\Fyx$ is not actually sufficiently constrained to perform even asymptotic fixed-X inference, so that extra, often quite stringent, structural assumptions must be imposed to effectively reduce the dimension further, such as that a $p$-dimensional parameter has no more than $o(\sqrt{n}/\log(p))$ nonzero entries, (4) on the contrary, the model-X approach actually still works even with a larger-than-$p$-dimensional family, which we will see examples of in, e.g., Section~\ref{sec:eg-ggm}.

\subsection{Related work}\label{sec:relwork}
By far the most common fixed-X approaches to inference rely on GLMs with $p$ parameters, reducing model complexity from exponential to linear in $p$. 
When $p$ is smaller than the number of observations $n$, inference for GLMs other than Gaussian linear models relies on large-sample approximation by assuming at least $p/n \rightarrow 0$ [Huber1973, Portnoy1985]. 
\rev{Note that the commonly studied problem of inference for a single parameter can generally be translated to FDR control using the Benjamini--Hochberg \citep{YB-YH:1995} or Benjamini--Yekutieli \citep{YB-DY:2001} procedures (see, e.g., \citet{AJ-HJ:2018}), so that it makes sense to compare such inference with our paper that is focused on multiple testing.} 
In high dimensions, i.e., when $p>n$, %, i.e., when there are more covariates $p$ than observations $n$, 
even reducing the complexity of $\Fyx$ to $p$ parameters with a GLM is insufficient for fixed-X inference, as GLMs become unidentifiable in this regime due to the design matrix columns being linearly dependent. Early solutions for fixed-X inference in high-dimensional GLMs relied on $\beta$-min conditions that lower-bound the magnitude of nonzero coefficients to obtain asymptotically-valid p-values for individual variables (see, e.g., \citet{AC-SL:2013}). More recent work removes the $\beta$-min condition in favor of strong sparsity assumptions on the coefficient vector, usually $o(\sqrt{n}/\log(p))$ nonzeros, with notable examples including the debiased Lasso (see, e.g., \citet{CZ-SZ:2014,AJ-AM:2014,vdG-ea:2014}) and the extended score statistic (see, e.g., \citet{AB-VC-CH:2014,AB-VC-KK:2015,VC-CH-MS:2015,YN-HL:2017}), both of which provide asymptotically-valid p-values for GLMs with some additional assumptions on the `compatibility' of the design matrix. In recent work that seems to straddle the fixed-X and model-X paradigms, \citet{YZ-JB:2018} and \citet{YZ-JB:2018b} %\ljmargin{}{Also have a look at this recent paper \url{https://projecteuclid.org/euclid.ejs/1538791404} by the same authors and see if it merits citation as well} 
compute asymptotically-valid p-values for the Gaussian linear model without any extra restrictions like sparsity or $\beta$-min on $\Fyx$, but with added assumptions on $\Fx$ about the sparsity of conditional linear dependence among covariates.  %assumption: either knowing population covariance or putting a sparse linear model for the synthesized feature and the stabilized feature
%\dhmargin{Later \citet{YZ-JB:2018b}}{it adds less ingredient related to us} generalizes this work by allowing heteroscedasticity in the noise and either the regression coeffficients or the X's dependence to be non-sparse.  %The heteroscedasticity and the non-sparsity are what they meant by 'model misspecification'. 

Another branch of recent research called post-selection inference can be viewed as a different approach to high-dimensional inference: it aims to test random hypotheses selected by a high-dimensional regression and provide valid p-values by conditioning on the selection event (see, e.g., \citet{WF-DS-JT:2014,JL-DS-YS-JT:2016} for foundational contributions, and \citet[Appendix A]{EC-ea:2018} for more about the difference between post-selection inference and our approach).

%\ljmargin{Review existing model-X work for high-dimensional inference: MX knockoffs, CRT, robustness paper, some follow-up papers on other knockoff constructions like the HMM work and 'knockoffs for the mass', three new papers on approximate knockoffs, which purport to construct valid knockoffs without any assumptions on $\Fx$: \url{http://statweb.stanford.edu/~candes/papers/DeepKnockoffs.pdf} and its refs 16 and 20.}{}

The method of knockoffs was first introduced by \citet{RB-EC:2015} for low-dimensional homoscedastic linear regression with fixed design. %\dhmargin{}{I think we should not completely ignore RB-EC:2016} 
The model-X knockoffs framework proposed by \citet{EC-ea:2018} read this idea from a different perspective, providing valid finite-sample inference with no assumptions on $\Fyx$ but assuming full knowledge of $\Fx$. Exact knockoff generation methods have been found for $\Fx$ following a multivariate Gaussian \citep{EC-ea:2018}, a Markov chain or hidden Markov models \citep{MS-CS-EC:2017}, \rev{a graphical model \citep{SB-EC-LJ-WW:2019}, and certain latent variable models \citep{JG-AG-JZ:2018}}. In the case that $\Fx$ is only known approximately, the robustness of model-X knockoffs is studied by \citet{BR-CE-SR:2018}. When $\Fx$ is completely unknown some recent works have proposed methods to generate approximate knockoffs \citep{JJ-JY-MS:2019,YR-MS-EC:2018,YL-CZ:2018} which have shown promising empirical results, particularly in low-dimensional problems, but come with no theoretical guarantees. \rev{In contrast, the current paper proposes to construct valid knockoffs that provide exact finite sample error control.} 
%\ljmargin{}{I think the GAN author has been deanonymized now}

This paper is based on the idea of performing inference conditional on a sufficient statistic for $\Fx$'s model so as to make that inference parameter-free. In low-dimensional inference, likely the simplest example of such an idea is a permutation test for independence, which can be thought of as a randomization test performed conditional on the order statistics of an observed i.i.d. vector of scalar $X$ with unknown distribution (the order statistics are sufficient for the family of all one-dimensional distributions). Although permutation tests can only test marginal independence, not conditional independence as addressed in the present paper, \citet{RP:1984} constructs a conditional permutation test that does test conditional independence assuming a logistic regression model for $X_j\mid X_{\noj}$, and allows the parameters of the logistic regression model to be unknown by conditioning on that model's sufficient statistic. However that sufficient statistic is composed of inner products between the vector of observed $X_j$'s and each of the vectors of observed values of the other covariates $X_{\noj}$, precluding inference except in the case of covariates with a very small set of discrete values, and almost entirely precluding inference in a high-dimensional setting.\footnote{See the paragraph preceding \citet[Theorem 1]{RP:1984} for a description of the test's limitations.} A different conditional permutation test was recently proposed by \cite{BT-WY-BR-SR:2018} to test conditional independence in the model-X framework, but while their conditioning improves robustness, they still require the same assumptions as the original conditional randomization test \citep{EC-ea:2018}, namely, that $X_j\mid X_{\noj}$ is known exactly. To our knowledge, the present paper is the first to use the idea of conditioning on sufficient statistics for high-dimensional inference, enabling powerful and exact FDR-controlled variable selection under arguably weaker assumptions than any existing work.

% say something about high-dimensionality being the case of interest up higher

%and the model-X approach reduces $\Fx$ to a single point by taking it to be known exactly.
%and fixed-x knockoff as well?

\subsection{Outline}
The rest of the paper is structured as follows: Section~\ref{sec:main} describes the main result and the proposed method of conditional knockoffs to generalize model-X knockoffs to the case when $\Fx$ is known only up to a distributional family, as opposed to exactly. 
%Section~\ref{sec:ex} applies the proposed method to \nex different models for $\Fx$, and provides simulations showing that the proposed method often loses almost no power in exchange for its increased generality over model-X knockoffs when $\Fx$ is known exactly. 
Section~\ref{sec:ex} applies conditional knockoffs to \nex different models for $\Fx$, and 
provides explicit algorithms for constructing valid knockoffs. Simulations are also presented, showing that conditional knockoffs often loses almost no power in exchange for its increased generality over model-X knockoffs with exactly-known $\Fx$. 
%Section~\ref{sec:sims} provides simulations showing that the proposed method normally loses almost no power in exchange for its increased generality, as compared to model-X knockoffs when $\Fx$ is known exactly. 
Finally, Section~\ref{sec:disc} provides some synthesis of the ideas in this paper and directions for future work.

\section{Main Idea and General Principles}\label{sec:main}
%\section{Main Idea and General Principles}\label{sec:main}
Before going into more detail, we introduce some notation. Suppose we are given i.i.d. row vectors $(Y_i,X_{i,1},\dots,X_{i,p})\in\R^{p+1}$ for $i=1,\dots,n$. We then stack these vectors into a design matrix $ \bX\in\R^{n\times p}$ whose $i$th row is denoted by $\bx_i^{\top} = (X_{i,1},\dots,X_{i,p})\in\R^p$, and a column vector $\by\in\R^n$ %%\ljmargin{}{is there a reason you switched to capital $\bs{Y}$, in contrast to the MX knockoffs paper notation which uses $\bs{y}$?} 
whose $i$th entry is $Y_i$. We are about to define model-X knockoffs $(\Xk_{i,1},\dots,\Xk_{i,p})$, and $\bXk\in\R^{n\times p}$ will analogously denote these row vectors stacked to form a knockoff design matrix. 
A square bracket around matrices, such as $[\bX,\bXk]$, denotes the horizontal concatenation of these matrices.
%Round parentheses around random elements, such as $(\by, \bX)$, denotes the joint random element.
We use $[p]$ for $\{1,2,\dots,p\}$, and $i:j$ for $\{i,i+1,\dots, j\}$ for any $i\leq j$; 
for a set $A\subseteq [p]$, let $\bX_A$ denote the matrix with columns given by the columns of $\bX$ whose indices are in $A$, and for singleton sets we streamline notation by writing $\bX_j$ instead of $\bX_{\{j\}}$. For sets $A_1, \dots, A_m$, denote by $\prod_{j=1}^{m}A_j$ their Cartesian product. \rev{For two disjoint sets $A$ and $B$, we denote their union by $A\uplus B$.} We will denote by $\mathbb{N}$ the set of strictly positive integers. 

%[unused] Denote by $\bX_j(I)$ the subvector of $\bX_j$ with indices $I\subseteq [n]$.

\subsection {Model-X Knockoffs}\label{sec:mxk}
We begin with a short review of model-X knockoffs \citep{EC-ea:2018}. The authors define model-X knockoffs for a random vector $X\in\R^p$ of covariates as being a random vector $\Xk\in\R^p$ such that for any set $A\subseteq [p]$
\begin{equation}\label{eq:mxk1}
\Xk\indp Y\,|\,X, \mbox{ and } (X,\Xk)_{\text{swap}(A)} \eqd (X,\Xk),
\end{equation}
where the swap($A$) subscript on a $2p$-dimensional vector (or matrix with $2p$ columns) denotes that vector (matrix) with the $j$th and $(j+p)$th entries (columns) swapped, for all $j\in A$. To use knockoffs for variable selection, suppose some statistics $Z_j$ and $\tilde{Z}_{j}$ are used to measure the importance of $X_j$ and $\tilde{X}_{j}$, respectively, in the conditional distribution $Y\mid X_1,\dots,X_p,\Xk_1,\dots,\Xk_p$, with 
\[
(Z_1,\dots, Z_p, \tilde{Z}_{1},\dots, \tilde{Z}_{p})=z([\bX,\bXk],\by), 
\]
for some function $z$ such that swapping $\bX_j$ and $\tilde{\bX}_{j}$ swaps the components $Z_j$ and $\tilde{Z}_{j}$, i.e., for any $A\subseteq [p]$, 
\[
z([\bX,\bXk]_{\text{swap}(A)},\by)=z([\bX,\bXk],\by)_{\text{swap}(A)}. 
\]
For example, $z([\bX,\bXk],\by)$ could perform a cross-validated Lasso regression of $\by$ on $[\bX,\bXk]$ and return the absolute values of the $2p$-dimensional fitted coefficient vector. More generally the $Z_j$ can be almost any measure of variable importance one can think of, including measures derived from arbitrarily-complex machine learning methods or from Bayesian inference, and this flexibility allows model-X knockoffs to be powerful even when $\Fyx$ is quite complex.

The pairs $(Z_j,\tilde{Z}_j)$ of variable importance measures are then plugged into scalar-valued antisymmetric functions $f_j$ to produce $W_j=f_j(Z_j, \tilde{Z}_{j})$, which measures the \emph{relative} importance of $X_j$ to $\Xk_j$. Viewed as a function of all the data, $W_j = w_j([\bX,\bXk],\by)$ can be shown to satisfy the \emph{flip-sign} property, which dictates that for any $A\subseteq [p]$,
\[
w_j([\bX,\bXk]_{\text{swap}(A)},\by)=\left\{
\begin{array}{c l}	
   w_j([\bX,\bXk],\by) , &\text{ if }\; j \notin A,\\
   -w_j([\bX,\bXk],\by), &\text{ if }\; j \in A.
\end{array}\right.
\]
Taking $Z_j$ and $\tilde{Z}_{j}$ as the absolute values of Lasso coefficients as in the above example, one might choose $W_j=Z_j-\tilde{Z}_{j}$, referred to in \citet{EC-ea:2018} as the \emph{Lasso coefficient-difference} (LCD) statistic. 
Finally, given a target FDR level $q$, the knockoff filter selects the variables ${\hat{S}=\{j\; : \: W_j\geq T\}}$ where $T$ is either the \emph{knockoff threshold} $T_0$ or the \emph{knockoff+ threshold} $T_+$:
\[
T_{0} = \min\left\{  t>0 : \frac{ \#\{j\; : \: W_j\leq -t\}}{\#\{j\; : \: W_j\geq t\}} \leq q\right\}, \;\;\;\quad T_{+} = \min\left\{  t>0 : \frac{ 1+\#\{j\; : \: W_j\leq -t\}}{\#\{j\; : \: W_j\geq t\}} \leq q\right\}.
\]
\citet[Theorem 3.4]{EC-ea:2018} prove that $\hat{S}$ with $T_+$ exactly (non-asymptotically) controls the FDR at level $q$, and that $\hat{S}$ with $T_0$ exactly controls a modified FDR, $\E{\frac{|\hat{S}\,\cap \, \nullset |}{|\hat{S}|+1/q}}$, %\ljmargin{}{we should use $S$ instead of $\nullset$ to match the equation on page 2} 
at level $q$. The key to the proof of exact control is the aforementioned flip-sign property of the $W_j$, and that property follows from the following crucial property of model-X knockoffs: for any subset $A\subseteq  \nullset$, 
\[
( [\bX,\bXk]_{\text{swap}(A)}, \by) \eqd ( [\bX,\bXk], \by),
\]
which is proved in \citet[Lemma 3.2]{EC-ea:2018} to hold for knockoffs satisfying Equation~\eqref{eq:mxk1}.

The proofs of exact control required just one assumption, that one could construct knockoffs satisfying Equation~\eqref{eq:mxk1}. To satisfy that assumption, \citet{EC-ea:2018} assumes throughout that $\Fx$ is known exactly. We will relax this assumption, but first slightly generalize the definition of valid knockoffs:
\begin{definition}[Model-X knockoff matrix]\label{def:mxk}
The random matrix $\bXk\in\R^{n\times p}$ is a \emph{model-X knockoff matrix} for  the random matrix $\bX\in\R^{n\times p}$ if
for any subset $A\subseteq [p]$,
\begin{equation}\label{eq:mxk}
\bXk\indp\by\,|\,\bX,\;\text{ and }[\bX,\,\bXk]_{\text{swap}(A)} \eqd [\bX,\,\bXk],
\end{equation}
\end{definition}
Note that Equation~\eqref{eq:mxk} is more general than Equation~\eqref{eq:mxk1}, and indeed \eqref{eq:mxk1} implies \eqref{eq:mxk} as long as the rows of $[\bX,\,\bXk]$ are independent. However, the proof of \citet{EC-ea:2018}'s crucial Lemma 3.2 and, ultimately, FDR control in the form of their Theorem 3.4 used only Equation~\eqref{eq:mxk}. Therefore Definition~\ref{def:mxk} is the `correct' definition, since the ability to generate knockoffs satisfying Definition~\ref{def:mxk} is all that is needed for the theoretical guarantees of knockoffs in \citet{EC-ea:2018} to hold, and it is well-defined for any matrix $\bX$, even when the rows are not independent. We will use this general definition because although we also assume samples are drawn i.i.d. from a distribution, those samples will no longer be independent when we condition on a sufficient statistic for the model for $\Fx$. Hereafter, \textit{model-X knockoffs} and \textit{knockoffs} will always refer to model-X knockoff matrices as defined by Definition~\ref{def:mxk} unless otherwise specified.

%TODO: steamline the equations into 1 line. 
\rev{
For completeness, we restate the FDR control theorem in \citet{EC-ea:2018}. 
\begin{theorem}
	Suppose $\bXk$ is a knockoff matrix for $\bX$ and the statistics $W_j$'s satisfy the flip-sign property. For any $q\in [0,1]$, if $\hat{S}$ is selected by the knockoff method with threshold $T$ being either $T_+$ or $T_0$, then 
	\[
	\E{ \frac{|\hat{S} \, \cap \, \nullset | }{\max\left( |\hat{S}| , 1 \right)}} \leq q, \mbox{ for $T_+$  } ;~~	\E{\frac{|\hat{S}\,\cap \, \nullset |}{|\hat{S}|+1/q}} \leq q,
\mbox{ for $T_0  $ . }
	\]
\end{theorem}
}
\rev{
It is worth mentioning that if $\bXk_{j}$ is identical to $\bX_{j}$, then $W_j=0$ and $j$ cannot be selected by the knockoff filter. Formally, we call such a column in the knockoff matrix \textit{trivial}. 
}

\subsection{Conditional Knockoffs }\label{sec:ck}
The main idea of this paper is that if $\Fx$ is known only up to a parametric model, and that parametric model has sufficient statistic (for $n$ i.i.d. observations drawn from $\Fx$) given by $T(\bX)$, then by definition of sufficiency the distribution of $\bX\,|\,T(\bX)$ does not depend on the model parameters and is thus known exactly a priori. To leverage this for knockoffs, consider the following definition.
\begin{definition}[Conditional model-X knockoff matrix]\label{def:cond-ko}
%[old definition]For the random matrix $\bX\in\R^{n\times p}$ with distribution following the model $G_{\bs{\Theta}}=\{g_{\bs{\theta}}:\bs{\theta}\in\bs{\Theta}\subseteq \R^m\}$, the random matrix $\bXk$ is a \emph{conditional model-X knockoff matrix} for $\bX$ if $T(\bX)$ is a \ljmargin{sufficient}{I still want to discuss whether this should be in the definition} statistic for $G_{\bs{\Theta}}$ such that for any subset $A\subseteq [p] $,
The random matrix $\bXk\in\R^{n\times p}$ is a \emph{conditional model-X knockoff matrix} for  the random matrix $\bX\in\R^{n\times p}$ if there is a statistic $T(\bX)$ such that for any subset $A\subseteq [p] $,
\begin{equation}\label{eq:cmxk}
\bXk\indp\by\,|\,\bX,\;\text{ and }\left.[\bX,\,\bXk]_{\text{swap}(A)} \eqd [\bX,\,\bXk]\;\right|\, T(\bX),
\end{equation}
\end{definition}

By the law of total probability, \eqref{eq:cmxk} implies \eqref{eq:mxk}, thus conditional model-X knockoffs are also model-X knockoffs:
\begin{proposition}\label{prop:cko}
	If $\bXk$ is a conditional model-X knockoff matrix for $\bX$, then it is also a model-X knockoff matrix. 
\end{proposition}

% inference conditional on T
%The only difference that comes from using conditional model-X knockoffs for variable selection is that all stochasticity in the procedure is taken to be conditional on $T(\bX)$, and conditional on $T(\bX)$, the conditional knockoffs are simply valid model-X knockoffs as defined in \citet{EC-ea:2018}. Therefore, as compared to \citet{EC-ea:2018}, the knockoffs are required to have the same properties (but now conditional on $T(\bX)$), the knockoff statistics can be constructed in the same way, and the choice of selection threshold proceeds in the same way as well. The exact finite-sample inference results now hold not just conditionally on $\by$, but also on $T(\bX)$. 

Proposition~\ref{prop:cko} says that all the guarantees of model-X knockoffs \rev{(i.e., Theorem 2.1)}, such as exact FDR control and the flexibility in measuring variable importance, immediately hold more generally when $\bXk$ is a \emph{conditional} model-X knockoff matrix. Definition~\ref{def:cond-ko} is especially useful when the distribution of $\bX$ is known to be in a model $G_{\bs{\Theta}}=\{g_{\bs{\theta}}:\bs{\theta}\in\bs{\Theta}\}$ with parameter space $\bs{\Theta}$, and $T(\bX)$ is a sufficient statistic for $G_{\bs{\Theta}}$, because then the distribution of $\bX\mid T(\bX)$ is known exactly even though the unconditional distribution of $\bX$ is not. Exact knowledge of the distribution of $\bX\mid T(\bX)$ in principle allows us to construct knockoffs, similar to how exact knowledge of the unconditional distribution of $\bX$ has enabled all previous knockoff construction algorithms. 
As a simple example, when $G_{\Theta}$ is the set of all $p$-dimensional distributions with mutually-independent entries, the set of order statistics for each column of $\bX$ constitutes a sufficient statistic $T(\bX)$, and a conditional knockoff matrix $\bXk$ can be generated by randomly and independently permuting each column of $\bX$. 
Unfortunately for more interesting models that allow for dependence among the covariates, even for canonical $G_{\bs{\Theta}}$ like multivariate Gaussian, the distribution of $\bX\mid T(\bX)$ is often much more complex than those for which knockoff constructions already exist. Using novel methodological and theoretical tools, in Section~\ref{sec:ex} we provide efficient and exact algorithms for constructing nontrivial conditional knockoffs when $\Fx$ comes from each of the following \nex models:
\begin{enumerate}
	\item {\bf Low-dimensional Gaussian}: $$\Fx \in\left\{\N(\bmu, \bS): \bmu\in \R^p,\;\bS\in\R^{p\times p},\;\bS\succ\bs{0}\right\},$$ when $n > 2p$. In this case, the number of model parameters is $p+\frac{p(p+1)}{2} = \Omega(p^2)$, and also $\Omega(n p) $ in the most challenging case when $p=\Omega(n)$. %, whose order can be as large as $n^2$. 
	\item {\bf Gaussian graphical model}: $$\Fx\in\left\{\N(\bmu,\bS): \bmu\in \R^p,\;\bS\in\R^{p\times p},\; \bS \succ \bs{0},\;\left(\bS^{-1}\right)_{j,k} = 0\text{ for all } (j,k)\notin E\right\}$$ for some known sparsity pattern 
%\dhmargin{}{should be `edge set'? Or we are highlighting $E$ being in a certain class so we still prefer `sparsity pattern'?} 
$E$. For example, $\bS^{-1}$ could be banded with bandwidth as large as $n/8-1$,\footnote{Here we assume $n/8\leq p$.} allowing a number of parameters as large as % \dhmargin{the notion of "as large as" may be misleading or wrong} 
$p + \left(\frac{np}{8} - \frac{n(n-8)}{128}\right) = \Omega(np )$. Note that $p$ is not explicitly constrained, so this model allows both low- and high-dimensional data sets. 
	\item {\bf Discrete graphical model}: $$\Fx \in\left\{\text{distribution on } \prod_{j=1}^p [K_j]: X_j\indp X_{[p]\setminus N_E(j) } \mid X_{ N_E(j)\setminus\{j\} } \text{ for all } (j,k)\notin E\right\}$$ %\pi\in [0,1]^{\prod_{j=1}^p K_j},\;\sum_{k=1}^{\prod_{j=1}^p K_j} \pi_k = 1,\;
for some known positive integers $K_1,\dots,K_p$ and known sparsity pattern $E$, where $N_E(j)$ is the closed neighborhood of $j$. For example, $X$ could be a $K$-state (non-stationary) Markov chain whose $K-1+(p-1)K(K-1)$ parameters are the probability mass function of $X_1$ and the transition matrices $\Pc{X_j}{X_{j-1}}$ for each $j\in\{2,\dots,p\}$, where $K$ can be as large as $\sqrt{\frac{n-2}{2}}$, allowing a number of parameters as large as $\sqrt{\frac{n-2}{2}} -1 + (p-1)\left(\sqrt{\frac{n-2}{2}}\right)\left(\sqrt{\frac{n-2}{2}}-1\right) = \Omega(np)$. Again, $p$ is not explicitly constrained, so this model allows both low- and high-dimensional data sets. 
%}{I reworded substantially---please check}% [the reason is, each fold has at least 1+K^2, thus for any j, at least one of the partitions by its neighbors' values has at least 2 samples, and thus can be different with positive probability and can potentially output a nontrivial knockoff]
% and empirically the conditional knockoff still has power in detecting important variables.
\end{enumerate}

%is to compare the observed test statistic with the value that would have been obtained under the null hypothesis with a different treatment vector that have the same inner product as. 

%connect to permutation test as conditioning on order statistic.
%connect to Rosenbaum (1984) as conditioning on sufficient statistic of logistic regression for propensity scoring.

%[??]define conditional knockoffs for a given model. Conditional knockoffs can be thought of as model-X knockoffs where everything is conditioned on $T(X)$, or by simply decomposing $\pr{X,\Xk,T(X)} = \Pc{X,\Xk}{T(X)}\pr{T(X)}$ we can see that conditional knockoffs are perfectly valid regular knockoffs as well!

\begin{remark}
It is worth mentioning that conditioning may shrink the set of nonnull hypotheses. For instance, if $\nullset=\emptyset$ and $T(\bX)$ is chosen to be $\bX$, then all variables are automatically null conditional on $T(\bX)$, and thus conditional knockoffs cannot select any nonnull variables. For a detailed discussion, see Appendix~\ref{app:cond-hypo}. %[compiling of the other file is needed]
%For most problems of interest, we provide the following sufficient condition for $\nullset \supset \mathcal{H}_{0,T}$.
\end{remark}

\rev{
\begin{remark}
Any algorithm that generates conditional knockoffs given one sufficient statistic $T(\bX)$ (i.e., satisfying Equation~(2.3) for $T(\bX)$) by definition is also a valid algorithm for generating conditional knockoffs given any sufficient statistic $S(\bX)$ that is a function of $T(\bX)$. This means that any valid conditional knockoff algorithm satisfies Equation~(2.3) for the minimal sufficient statistic, since by definition a minimal sufficient statistic is a function of any other sufficient statistic. So we could say that the minimal sufficient statistic is in some sense the optimal one to condition on, in that the choice to condition on the minimal sufficient statistic allows for the most general set of conditional knockoff algorithms of any sufficient statistic one could choose to condition on for a given model.
% Denote by $\mathcal{K}_{T}$ the class of conditional knockoff matrices such that Equation~\eqref{eq:cmxk} holds for a sufficient statistic $T(\bX)$. If $S(\bX)$ is a minimal sufficient statistic, $\mathcal{K}_{S}$ is the largest class. Indeed, for any sufficient statistic $T(\bX)$, by the definition of minimal sufficient statistics, there exists a function $h$ such that $S(\bX)=h(T(\bX))$. Therefore if Equation~\eqref{eq:cmxk} holds for  $T(\bX)$ it also holds for $S(\bX)$, and thus $\mathcal{K}_{T}\subset \mathcal{K}_{S}$. Also note that the smallest class is $\mathcal{K}_{\bX} = \{\bX\}$, namely, the class of trivial knockoff, because $\mathcal{K}_{\bX}\subset \mathcal{K}_{T}$ for any $T(\bX)$. 
\end{remark}
}

\subsection{Integrating Unlabeled Data}\label{sec:unlabel}
In addition to the $n$ labeled pairs  $\left\{(Y_i, \bx_i)\right\}_{i=1}^{n}$, we might also have unlabeled %\ljmargin{\emph{unlabeled}}{We should decide between `unlabeled' and `unsupervised' and use the same terminology throughout} 
%\ljmargin{}{I'm not crazy about the subscript/superscript notations here}
data $\{\bx_i^{(u)}\}_{i=1}^{n^{(u)}}$, i.e., covariate samples without corresponding responses/labels. This extra data can be integrated seamlessly into the construction of conditional knockoffs: stack the labeled covariate matrix $\bX$ on top of the unlabeled covariate matrix $\bX^{(u)}$ to get $\bX^{*}\in \R^{n^{*}\times p}$, where $n^{*}=n+n^{(u)}$, then construct conditional knockoffs $\bXk^{*}$ for $\bX^{*}$, and finally take $\bXk$ to be the first $n$ rows of $\bXk^{*}$.
 
\begin{proposition}\label{prop:unsuper}
Suppose the rows of $\bX^{*}$ are i.i.d. covariate vectors and $\bX$ is the matrix composed of the first $n$ rows of $\bX^{*}$. Let $\by$ be the response vector for $\bX$. 
If for some statistic $T(\bX^{*})$ and any set $A\subseteq [p]$,
\[
\bXk^{*} \indp\by\,|\,\bX^{*},\text{ and }[\bX^{*},\,\bXk^{*}]_{\text{\emph{swap}}(A)} \eqd [\bX^{*},\,\bXk^{*}] \,\Big{|}\, T(\bX^{*}), 
\]
then if $\bXk$ is the matrix composed of the first $n$ rows of $\bXk^{*}$, then $\bXk$ is a model-X knockoff matrix for $\bX$.
%it holds that 
%\begin{enumerate}[(a).]
%	\item $\bXk \indp\by\,|\,\bX$,
%	\item $[\bX,\,\bXk]_{\text{\emph{swap}}(A)} \eqd [\bX,\,\bXk] \Big{|} T(\bX^{*})$.
%\end{enumerate}
\end{proposition}
Note that by taking $T(\bX^{*})$ to be constant, the same result holds unconditionally: if ${\bXk^{*} \indp\by\,|\,\bX^{*}}$ and 
$
[\bX^{*},\,\bXk^{*}]_{\text{swap}(A)} \eqd [\bX^{*},\,\bXk^{*}] 
$
for any $A\subseteq [p]$, then $\bXk$ is a valid knockoff matrix for $\bX$. Thus constructing knockoffs for $\bX^{*}$, conditional or otherwise, produces valid knockoffs for $\bX$ automatically. Of course, if $\Fx$ is known and the rows of $\bX^*$ are i.i.d., it is natural to construct each row of $\bXk^{*}$ independently, in which case the presence of $\bX^{(u)}$ changes nothing about the construction of the relevant knockoffs $\bXk$. But as seen in Section~\ref{sec:ck}, when $\Fx$ is not known exactly the flexibility with which we can model it depends on the sample size, with the number of parameters allowed to be as large as $\Omega(np)$ in all the models in this paper. What Proposition~\ref{prop:unsuper} shows is that $n$ can be replaced with $n^{*}$, which can dramatically increase the modeling flexibility allowed by conditional knockoffs, especially in high dimensions. For example, our conditional knockoffs construction in Section \ref{sec:ldg} for arbitrary multivariate Gaussian distributions naively requires $n>2p$, but we now see it actually just requires $n^{*}>2p$, which is much easier to satisfy when $n^{(u)}$ is large, as it often is in, for instance, genomics or economics applications. Even when $n$ alone is large enough to construct nontrivial knockoffs for a desired model, constructing conditional knockoffs with unlabeled data as described in this section will tend to increase power.
 
\section{Conditional Knockoffs for Three Models of Interest}\label{sec:ex}

In this section, we provide efficient algorithms to generate exact conditional model-X knockoffs under three different models for $\Fx$, as well as numerical simulations comparing the variable selection power of the knockoffs thus constructed with those constructed by existing algorithms that require $\Fx$ be known exactly. 

%All proofs are deferred to Appendix~\ref{app:main-proof}. Any sampling described in the algorithms is conducted independently of all previous sampling in the same algorithm, unless stated otherwise. All simulations use a Gaussian linear model for the response: $Y_i\mid \bx_{i}\sim \N(\frac{1}{\sqrt{n}}\bx_i\tp \bb , 1)$ where $\bb$ has 60 non-zero entries with random signs and equal amplitudes; all the same simulations are also rerun with a nonlinear model (logistic regression) with similar results, presented in Appendix~\ref{app:Simulation}. We use the LCD knockoff statistic and the knockoff+ threshold with target FDR $q=20$\%; see Section~\ref{sec:mxk} for details. Only power curves (power $=\E{ \frac{S\cap \hat{S}}{|S|} }$) are shown because the FDR is always controlled (both theoretically and empirically). Source code for running conditional knockoffs can be found at \url{https://github.com/stathuang/cknockoff} and tutorials demonstrating the usage of the code are available at \url{http://lucasjanson.fas.harvard.edu/code/ConditionalKnockoffs}.

All proofs are deferred to Appendix~\ref{app:main-proof}. Any sampling described in the algorithms is conducted independently of all previous sampling in the same algorithm, unless stated otherwise. All simulations use a Gaussian linear model for the response: $Y_i\mid \bx_{i}\sim \N(\frac{1}{\sqrt{n}}\bx_i\tp \bb , 1)$ where $\bb$ has 60 non-zero entries with random signs and equal amplitudes. \rev{Note the sparsity and magnitude equalities are simply chosen for convenience---we present additional simulations varying these choices in Appendix~\ref{app:vary}.} 
\rev{
We remind the reader that, although we use linear regression as an illustrative example in the simulations, our methods apply to more general regressions, and} all the same simulations are also rerun with a nonlinear model (logistic regression) with similar results, presented in Appendix~\ref{app:nonliner}. We use the LCD knockoff statistic \rev{with tuning parameter chosen by 10-fold cross-validation} and the knockoff+ threshold with target FDR $q=20$\%; see Section~\ref{sec:mxk} for details. Only power curves (power $=\E{ \frac{ | S\cap \hat{S}| }{|S|} }$) are shown because the FDR is always controlled (both theoretically and empirically). \rev{The procedure we compare to, unconditional knockoffs, refers to model-X knockoffs where $\Fx$ is taken to be known exactly (knockoff statistics and thresholds are chosen identically).}

\subsection{Low-Dimensional Multivariate Gaussian Model}\label{sec:ldg}
Despite the focus in variable selection on high-dimensional problems, we start with a low-dimensional example as it represents an interesting and instructive case. Suppose that
\begin{equation}\label{model:ldg}
\bx_i\iid\N(\bmu,\bS)	
\end{equation}
for some unknown $\bmu$ and positive definite $\bS$. Let $\hbmu := \bX^{\top}\bs{1}_n/n$ denote the vector of column means of $\bX$, %(we will at times abuse notation in linear algebraic expressions by having $\hbmu$ denote the $n\times p$ matrix composed of such means: $\bs{1}_n\bs{1}_n^{\top}\bX/n$---there should be no confusion as the surrounding linear algebra will make clear what the dimension of $\hbmu$ must be), 
and let $\hbS := (\bX-\bs{1}_{n}\hbmu\tp)^{\top}(\bX-\bs{1}_{n}\hbmu\tp)/n$ be the empirical covariance matrix of $\bX$. Then $T(\bX) = (\hbmu,\hbS)$ constitutes a (minimal, complete) sufficient statistic for the model \eqref{model:ldg} for $\bX$.

\subsubsection{Generating Conditional Knockoffs}
When $n> 2p$, we can construct knockoffs for $\bX$ conditional on $\hbmu$ and $\hbS$ via Algorithm~\ref{alg:ldg}. 
\begin{algorithm}[h]
\caption{Conditional Knockoffs for Low-Dimensional Gaussian Models} \label{alg:ldg}
\begin{algorithmic}[1]
\ENSURE$\bX\in \R^{n\times p}$.
\REQUIRE $n> 2p$.
\STATE Find $\bss\in\R^p$ such that $\bs{0}_{p\times p}\prec\diag{\bss}\prec 2 \hbS$.
\STATE Compute the Cholesky decomposition of $n \left(  2\diag{\bss} - \diag{\bss}\hbS^{-1}\diag{\bss}  \right) $ as $\bL^{\top}\bL$.
\STATE Generate $\bW$ a $n\times p$ matrix whose entries are i.i.d. $\N(0,1)$ and independent of $\bX$ and compute the Gram--Schmidt orthonormalization $[\!\!\underbrace{\bs{Q}}_{n\times(p+1)}\!\!,\underbrace{\bU}_{n\times p}\, ]$ of the columns of $\left[ \bs{1}_n,\, \bX,\, \bW \right]$. \label{line:ldg-W}
\STATE Set 
\begin{equation}\label{eq:ldg_ko_short}
\bXk = \bs{1}_{n}\hbmu\tp + (\bX-\bs{1}_{n}\hbmu\tp)(\bs{I}_p - \hbS^{-1}\diag{\bss}) + \bU\bL.
\end{equation}\vspace{-.5cm}
\RETURN $\bXk$.
\end{algorithmic}
\end{algorithm}
%\ljmargin{}{I changed notation slightly, update anything in the proofs accordingly}

%\begin{theorem}\label{thm:ldg_ko}
%Let $n\ge 2p+1$, $\bss\in\R^p$ be such that $\bs{0}_{p\times p}\prec\diag{\bss}\prec2\hbS$, $\bL^{\top}\bL = 2\diag{\bss} - \diag{\bss}\hbS^{-1}\diag{\bss}$ be a Cholesky decomposition, and $\bU = \bZ\bV$ where $\bZ$ is a $n\times (n-1-p)$ orthonormal matrix that is orthogonal to $\left[ \bs{1}_n,\, (\bX-\hbmu)\right]$ and $\bV$ is independent of $\bX$ and uniformly distributed on all $ (n-p-1)\times p$ matrices with orthonormal columns. Then
%\begin{equation}\label{eq:ldg_ko_short}
%\bXk = \hbmu + (\bX-\hbmu)(\bs{I}_p - \hbS^{-1}\diag{\bss}) + \bU\bL
%\end{equation}
%is a conditional model-X knockoff matrix for $\bX$ given $\hbmu$ and $\hbS$.
%
%\begin{proof} The proof is given in Appendix~\ref{app:ldgprf}.\end{proof}
%\end{theorem}

In Algorithm~\ref{alg:ldg}, $n>2p$ is needed because in Line~\ref{line:ldg-W} the $n\times (2p+1)$ matrix $\left[ \bs{1}_n,\, \bX,\, \bW \right]$ must have at least as many rows as columns to be a valid input to the Gram--Schmidt orthonormalization algorithm. %is required to have full column rank so that Gram--Schmidt orthonormalization can be implemented.}{is this needed?} %\ljmargin{---}{Add sentence here explicitly about why $n> 2p$ is needed}
The astute reader may notice a strong similarity between Equation~\eqref{eq:ldg_ko_short} and the fixed-X knockoff construction in \citet[Equation (1.4)]{RB-EC:2015}. %, and indeed they are nearly the same, and 
Indeed nearly the same tools can be used to find a suitable $\bss$; in Appendix~\ref{app:detail-ldg} we slightly adapt three methods from \citet{RB-EC:2015} and \citet{EC-ea:2018} for computing suitable $\bss$. The computational complexity of Algorithm~\ref{alg:ldg} depends on the method used to find $\bss$, with the fastest option requiring $O\left( n p^2 \right)$ time.

\rev{
The differences between Equation~\eqref{eq:ldg_ko_short} and the fixed-X knockoff construction are the additional accounting for the mean by adding/subtracting $\hbmu$, the lack of requiring that $\bX$ have normalized columns, the ``$\prec$" relationships (as opposed to ``$\preceq$"), and most importantly the requirement that $\bU$ be random. 
Indeed, as can be seen in the proof of Theorem~\ref{thm:ldg_ko}, the precise uniform distribution of $\bU$ is crucial. 
And it bears repeating that unlike fixed-X knockoffs, Algorithm~\ref{alg:ldg} produces valid \emph{model-X} knockoffs and hence permits importance statistics without the ``sufficiency property'' and applies to \emph{any} $\Fyx$, not just homoscedastic linear regression. 
}

%$\bU$ in \eqref{eq:ldg_ko_short} can be numerically constructed as follows. Let $\bW$ be a $n\times p$ random matrix whose entries are i.i.d. standard normal and independent of $\bX$, let the Gram-Schmidt  orthonormalization of $\left[( \bs{1}_n,\, \bX-\hbmu) , \bW \right]$  be $[ \bs{Q}_{0},\bs{Q}_{1} ]$, and set $\bU=\bs{Q}_{1}$; see Proposition~\ref{prop:ldg-roation}. The computational complexity for computing $\bs{Q}_{1}$ is $O(n p^2)$. 

\begin{theorem}\label{thm:ldg_ko}
Algorithm~\ref{alg:ldg} generates valid knockoffs for model~\eqref{model:ldg}.
\end{theorem}

The challenge in proving Theorem~\ref{thm:ldg_ko} is that the conditional distribution of $[\bX,\bXk]\mid T(\bX)$ is supported on an uncountable subset of zero Lebesgue measure, and its distribution is only defined through the distribution of $\bX\mid T(\bX)$ and the conditional distribution of $\bXk\mid\bX$. Although $\bX\mid T(\bX)$ and $\bXk\mid\bX$ are both conditionally uniform on their respective supports, and the latter's normalizing constant does not depend on $\bX$, these facts alone are not sufficient to conclude that $[\bX,\bXk]\mid T(\bX)$ is uniform on its support (see Appendix~\ref{app:counterexample} for a simple counterexample), which is what we need to prove. Although these distributions on zero-Lebesgue-measure manifolds can be characterized using geometric measure theory (as in, e.g., \citet{PD-SH-MS:2013}), we bypass this approach by directly using the concept of invariant measures from topological measure theory; see Appendix~\ref{app:ldgprf}. %The proof of Theorem~\ref{thm:ldg_ko} relies crucially on $\bU$'s distribution, but also techniques from topological measure theory since the conditional distribution of $[\bX,\bXk] $ given $\hbmu$ and $\hbS$ is supported on a Lebesgue-measure-zero \dhmargin{manifold}{Although it is true, we no longer prove/use it. Replaced by 'subset'? }.

A useful consequence of Theorem~\ref{thm:ldg_ko} is the double robustness property that if knockoffs are constructed by Algorithm~\ref{alg:ldg} and knockoff statistics are used which obey the sufficiency property of \citet{RB-EC:2015} (that is, the knockoff statistics only depend on $\by$ and $[\bX,\,\bXk]$ through $[\bs{1}_{n},\bX,\,\bXk]^{\top}\by$ and $[\bs{1}_{n},\bX,\,\bXk]^{\top}[\bs{1}_{n},\bX,\,\bXk]$), then the resulting variable selection controls the FDR exactly as long \emph{as at least one of} the following holds:
\begin{itemize}
\item $\bx_i\iid\N(\bmu,\bS)$ for some $\bmu$ and $\bS$, both unknown (\emph{regardless of $\Fyx$}), or
\item $y_i\,|\,\bx_i\iid\N(\bx_i^{\top}\bb,\sigma^2)$ for some $\bb$ and $\sigma^2$, both unknown (\emph{regardless of $\Fx$}).
\end{itemize}

In Appendix~\ref{app:detail-ldg} we extend Algorithm~\ref{alg:ldg} to the case when the mean is known (Algorithm~\ref{alg:ldg-known-mean}) or a subset of columns of $\bX$ are additionally conditioned on (Algorithm~\ref{alg:ldg-block}). Both extensions may be of independent interest, but will also be used as subroutines when generating knockoffs for Gaussian graphical models in Section~\ref{sec:eg-ggm}.

\subsubsection{Numerical Examples}\label{sec:ldg-sims}
We present two simulations comparing the power of conditional knockoffs to the analogous unconditional construction that uses the exactly-known $\Fx$. \rev{
We remind the reader that the simulation setting is at the beginning of Section~\ref{sec:ex}. 
}
The vector $\bss$ in Algorithm~\ref{alg:ldg} is computed using the SDP method of Equation~\eqref{opt:sdp}, and the analogous vector for the unconditional construction is chosen by the analogous SDP method \citep{EC-ea:2018}. Although in both examples $n^{*}>2p$, the number of unknown parameters in the Gaussian model for $\Fx$ is $p+\frac{p(p+1)}{2}>500,000$, vastly larger than any of the sample sizes.

% these two plots have different size with other plots because the legend is too compact!
\begin{figure}[h]\centering
%  \subfloat[FDR]{     \includegraphics[width=0.5\linewidth]{figures/LDG1/FDR agains n} }
%    ~ %\hfill
  \subfloat []{   \includegraphics[width=0.45\linewidth]{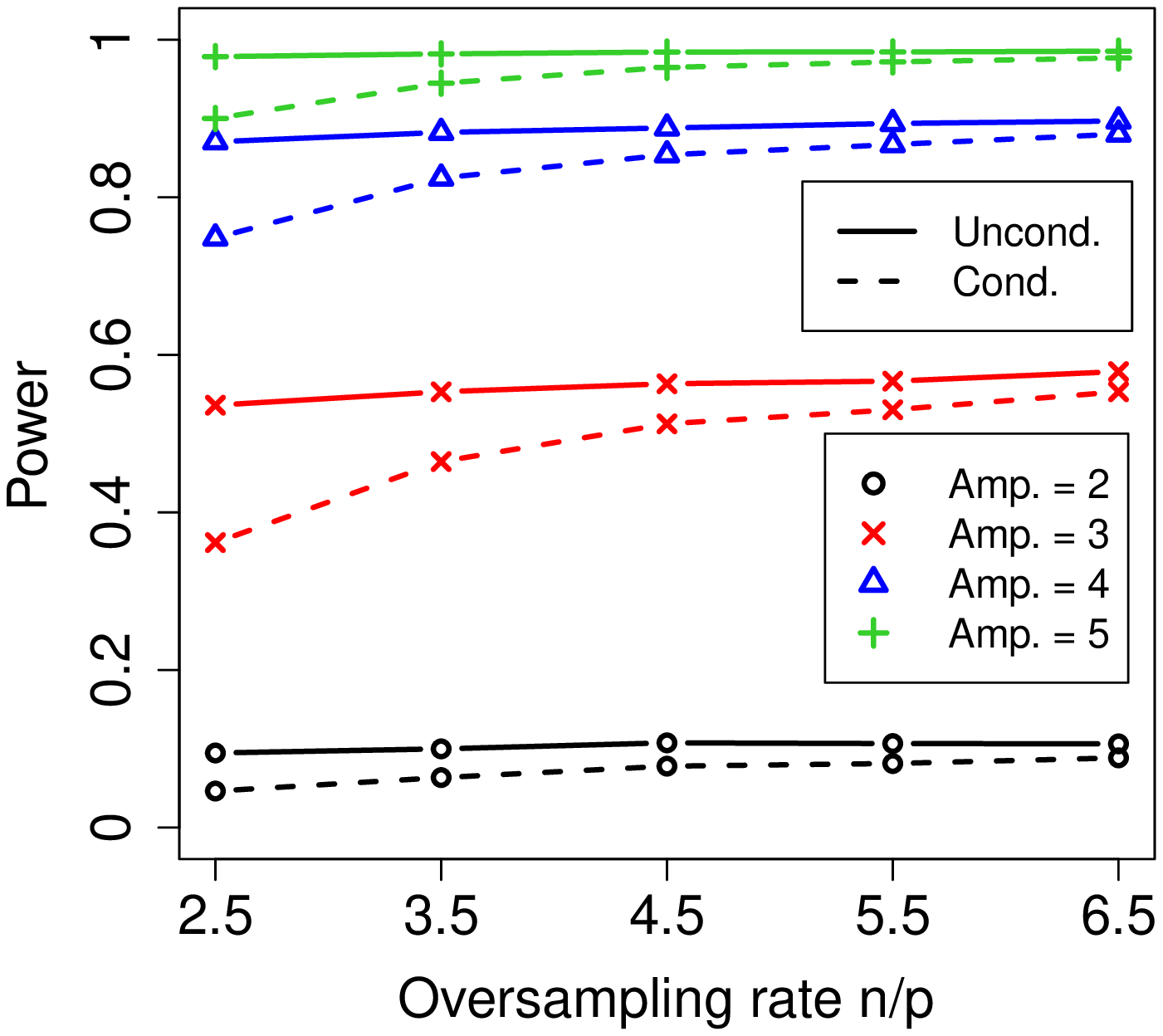}
  \label{fig:LDG1}}
%~\hfill
%  \subfloat[FDR]{     \includegraphics[width=0.5\linewidth]{figures/LDG_unlabel/FDR agains A} }
%    ~ %\hfill
  \subfloat[
  %Fixed $n=300$, different $n_u$
  ]{     \includegraphics[width=0.45\linewidth]{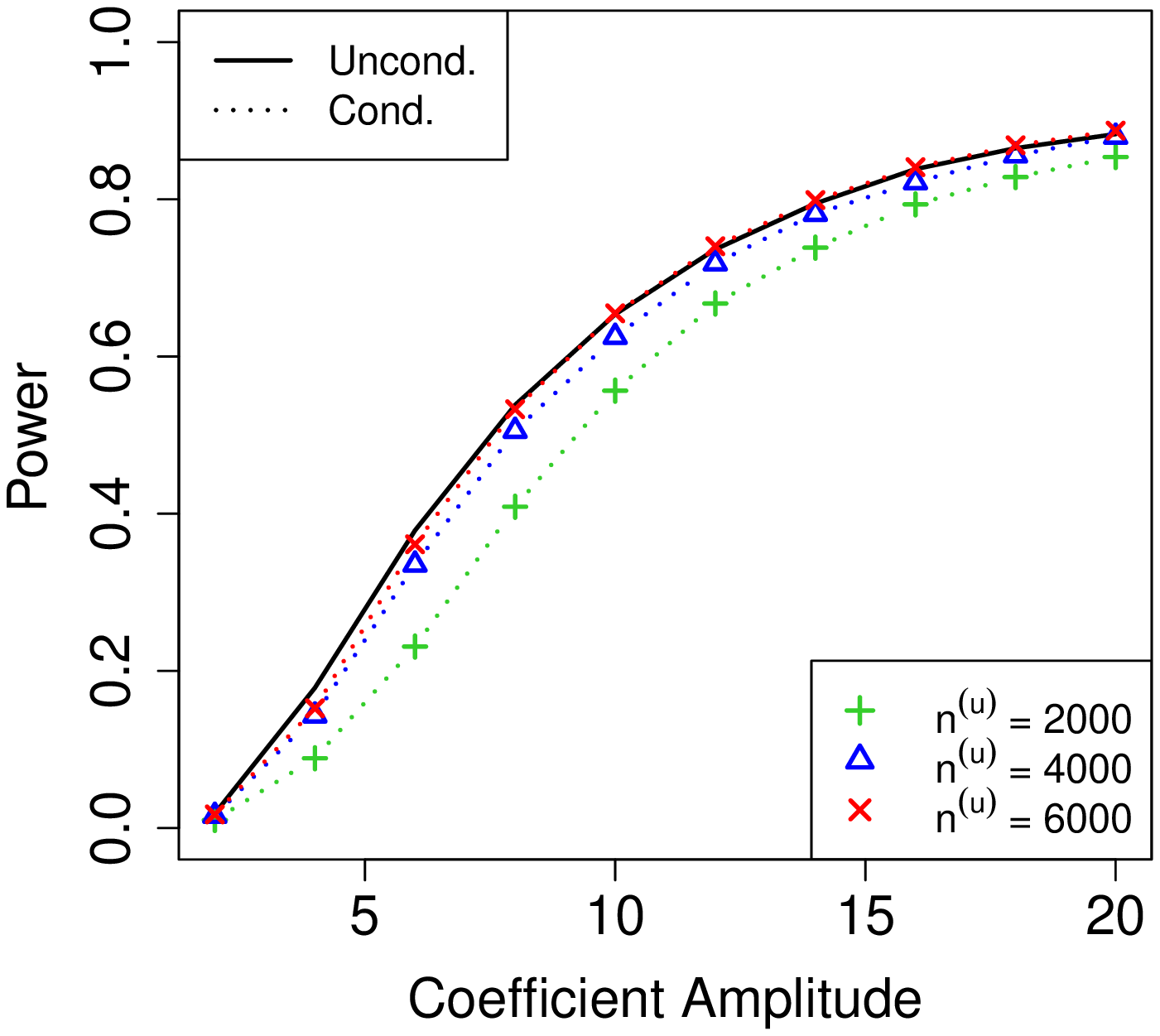}
  \label{fig:LDG2}
  }
    \caption{Power curves of conditional and unconditional knockoffs for an AR(1) model with $p=1000$ (a) as $n/p$ varies for various coefficient amplitudes and (b) as the coefficient amplitude varies for various values of $n^{(u)}$, with $n=300$ fixed. Standard errors are all below 0.008.}\label{fig:ldg}
    %The nominal FDR level is $0.2$.
 \end{figure}

Figure \ref{fig:LDG1} fixes $p=1000$ and plots the difference in power between unconditional and conditional knockoffs as $n> 2p$ increases for a few different signal amplitudes. The power of the conditional and unconditional constructions is quite close except when $n = 2.5p$ is just above its threshold of $2p$, and even then the power of the conditional construction is respectable. 

Figure \ref{fig:LDG2} shows how unlabeled samples improve the power of conditional knockoffs. The model is the same as the first example but the labeled sample size is fixed at $n=300$ and we vary the number of unlabeled samples. Again, the power of the conditional and unconditional constructions is extremely close except when $n^{*}=2.3p$ is just above its threshold, and again even in that setting the power of the conditional construction is respectable. Note that unlabeled samples here have enabled the \emph{low-dimensional} Gaussian construction to apply in a high-dimensional setting with $n<p$, since $n^{*}>2p$.

%\subsection{Autoregressive Gaussian}
%stationary time series... probably not going to get to GARCH

\subsection{Gaussian Graphical Model}%with a Sparse Precision Matrix}
\label{sec:eg-ggm}
Ignoring unlabeled data, the method of the previous subsection is constrained to low-dimensional (or perhaps more accurately, medium-dimensional, since it allows $p=\Omega(n)$) settings and cannot be immediately extended to high dimensions. 
%\ljmargin{Indeed, it is easy to see that when $p>n$, the only valid knockoff matrix conditional on $\hbmu$ and $\hbS$ is the trivial choice of $\bXk = \bX$, which is of course powerless (though technically valid).}{is this easy to see?} 
In many applications however, particularly in high dimensions, the covariates are modeled as multivariate Gaussian with \emph{sparse} precision matrix $\bS^{-1}$, and when the sparsity pattern is known a priori, we can condition on much less. For instance, time series models such as autoregressive models assume a banded precision matrix with known bandwidth, and the model used in this subsection would also allow for nonstationarity. Spatial models often assume a (known) neighborhood structure such that the only nonzero precision matrix entries are index pairs corresponding to spatial neighbors. 
%More generally, sparsity in the precision matrix, but with \emph{unknown} sparsity pattern, is a common assumption in Gaussian graphical models which are used to model many types of data in high dimensions such as gene expressions. Although the construction in this section no longer holds exactly when the sparsity pattern is unknown, approximate knockoffs could still be constructed by first using a method for estimating the sparsity pattern \citep[Chapter 13]{buhlmann2011statistics} and then treating it as known. % other reference:\cite{NM-PB:2006,MY-YL:2007,JF-TH-RT:2008,PR-MW-GR-BY:2011}) 

Precisely, suppose $\bX$'s rows $\bx_i^\top$ are i.i.d. draws from a distribution known to be in the model 
\begin{equation}\label{eq:gauss_sparse}
\left\{\N(\bmu,\bS):\bmu\in \R^{p}, \, \left(\bS^{-1}\right)_{j,k} = 0\text{ for all } j\neq k\text{ and }(j,k)\notin E, \bS \succ \bs{0}\right\}
\end{equation}
where $E\subseteq [p]\times[p]$ is some symmetric set of integer pairs (i.e., $(j,k)\in E \Rightarrow (k,j)\in E$) with no self-loops. Then the undirected graph  $G\; :=\;  ([p], E)$ defines a Gaussian graphical model with vertex set $[p]$ and edge set $E$. For any $j\in [p]$, define $I_{j}=\{k:(j,k)\in E \}$ for the vertices that are adjacent to $j$.  %, and we can further define $I_j := \{k:(j,k)\in E\}\setminus\{j\}$ as $j$'s neighborhood (excluding itself) in that graph.
We will use the terms `vertex' ($j\in [p]$) and `variable' ($X_j$) interchangeably.
$\hbmu$ and $\hbS_E$ together constitute a sufficient statistic, where $\hbS_E:=\left\{\hbS_{j,k}: j=k\text{ or }(j,k)\in E\right\}$. We will show in this section how to generate conditional knockoffs, and we will characterize the sparsity patterns $E$ for which we can generate knockoffs with $\bXk_j\neq\bX_j$ for all $j\in [p]$.

\begin{remark}\label{rem:unknown-graph}
More generally, sparsity in the precision matrix, but with \emph{unknown} sparsity pattern, is a common assumption in Gaussian graphical models which are used to model many types of data in high dimensions such as gene expressions. Although the construction in this section no longer holds exactly when the sparsity pattern is unknown, approximate knockoffs could still be constructed by first using a method for estimating the sparsity pattern \citep[Chapter 13]{buhlmann2011statistics} and then treating it as known. Note that we only require the edge set $E$ to contain all non-zero entries of $\bS^{-1}$, which is no harder than the exact identification of the non-zero entries. 
\end{remark}

\subsubsection{Generating Conditional Knockoffs by Blocking}\label{sec:ggm-alg-proof}
%First consider the ideal case when the graph $G$ separates into disjoint connected components $G_1,\dots, G_{\ncc}$. Then $X$ can be divided into independent subvectors, $X_{V_1}, \dots, X_{V_{\ncc}}$, where the $V_{\icc}$'s are the vertex sets of the $G_{\icc}$'s, 
%TODO: change the words so that $G_{\icc}$
First consider the ideal case when the graph $G$ separates into disjoint connected components whose respective vertex sets are $V_1,\dots, V_{\ncc}$. Then $X$ can be divided into independent subvectors, $X_{V_1}, \dots, X_{V_{\ncc}}$, 
and if each $|V_{\icc}|<n/2$, we can construct low-dimensional conditional knockoffs separately and independently for each $\bX_{V_{\icc}}$ as in Section~\ref{sec:ldg}. Moving to the general case when $G$ is connected, we can do something intuitively similar by conditioning on a subset of variables in addition to $\hbmu$ and $\hbS_E$. If there is a subset of vertices $B$ such that the subgraph $G_{B}$ induced by deleting $B$ separates into small disjoint connected components, then we should be able to construct conditional knockoffs as above for $\bX_{B^c}$ by conditioning on $\bX_B$. We think of the variables in $B$ as being \emph{blocked} to separate the graph into small disjoint parts, hence we refer to this $B$ as a \emph{blocking set}.

The following definition formalizes when we can apply the above procedure, and Algorithm~\ref{alg:sparse-gaussian-or} states that procedure precisely.
\begin{definition}\label{def:ggm-sep}
A graph $G$ is \emph{$n$-separated} by a set $B\subset [p]$ if the subgraph $G_{B}$ induced by deleting all vertices in $B$ has connected components whose respective vertex sets we denote by $V_1,\dots,V_{\ncc}$ such that for all $k\in [\ncc]$, 
\[
2 | V_{\icc} |+  | I_{V_{\icc}} \, \cap \, B | < n, 
\]
where $I_{V_{\icc}}\; :=\;  \bigcup\limits_{j\in V_{\icc}} I_{j}$ is the neighborhood of $V_{\icc}$ in $G$. 
\end{definition}

Note that when the $V_{\icc}$ separated $X$ into independent subvectors, we only needed $2|V_{\icc}|<n$; now that they only represent \emph{conditionally} independent subvectors, we must also account for $V_{\icc}$'s neighbors in $B$ that we condition on, resulting in the requirement that $2|V_{\icc}|+| I_{V_{\icc}} \, \cap \, B | <n$.

%\dhmargin{-}{For example, if $n'$ is used, it replaces the bundle in line:ggm-alg1the complexity below.}
%\ljmargin{-}{I don't know why there's a space before line 1 of Alg 2, but can you figure out how to get rid of it?}: spacing issue
\begin{algorithm}[h]
\caption{Conditional Knockoffs for Gaussian Graphical Models} \label{alg:sparse-gaussian-or}
%\algsetup{linenodelimiter=}
\begin{algorithmic}[1]
\ENSURE $\bX\in\R^{n\times p}$, $G =([p],E)$, $B\in [p]$.
\REQUIRE{For some $n'\leq n$, $G$ is $n'$-separated by $B$ into connected component vertex sets $V_1,\dots,V_{\ncc}$.}%$\{V_{\icc}\}_{\icc=1}^{\ncc}$.}
\FOR{$\icc = 1,\dots, \ncc$}
%[old]\STATE Generate a $n\times(n-1- |I_{V_{\icc}} \cap B|)$ orthonormal matrix $\bs{Q}$ that is orthogonal to $[\bs{1}_{n},\bX_{ I_{V_{\icc} } \cap B}]$, and compute $\bs{Q}_{\perp}$ an orthonormal basis for $\bs{Q}$'s null space.\label{line:alg-ggm-1}
%[old statement before adding the two modificatons]\STATE {\color{red} Construct low-dimensional knockoffs $\bs{J}_{(\icc)}$ for $\bs{Q}\tp \bX_{V_{\icc}}$ via Algorithm~\ref{alg:ldg-known-mean}, a modification of Algorithm~\ref{alg:ldg} for known mean presented in Appendix~\ref{app:detail-ldg}.}\label{line:ggm-alg1}%
\STATE Construct partial low-dimensional knockoffs $\bXk_{V_{\icc}}$ for $\bX_{V_{\icc}}$ conditional on $\bX_{I_{V_{\icc}}\cap B}$ via Algorithm~\ref{alg:ldg-block} (a slight modification of Algorithm~\ref{alg:ldg}).     \label{line:ggm-alg1}
%[old]\STATE Set $\bXk_{V_{\icc}}=\bs{Q}\bs{J}_{(\icc)} + \bs{Q}_{\perp}\bs{Q}_{\perp}\tp \bX_{V_{\icc}}$.\label{line:genxk}
\ENDFOR
\STATE Set $\bXk_B=\bX_B$.
\RETURN $\bXk$.
\end{algorithmic}
\end{algorithm}

Algorithm~\ref{alg:sparse-gaussian-or} constructs knockoffs for the model \eqref{eq:gauss_sparse} by first conditioning on $\bX_B$ and then running a slight modification of Algorithm~\ref{alg:ldg} (Algorithm~\ref{alg:ldg-block} in Appendix~\ref{app:detail-ldg-block}) on the variables/columns $V_{k}$ corresponding to the induced subgraphs. The computational complexity of Algorithm~\ref{alg:sparse-gaussian-or} is 
%${O\left( n \sum_{k=1}^{\ncc}\left(  \left| I_{V_{\icc}} \, \cap \, B \right|^{2} \left| V_{\icc} \right|+ (1+ 2 | V_{\icc} |+  \left| I_{V_{\icc}} \, \cap \, B \right| \right)^{2} \right)}$, 
$O\left( n \sum_{\icc=1}^{\ncc}\left(  \left| I_{V_{\icc}} \, \cap \, B \right|^{2} \left| V_{\icc} \right|+  | V_{\icc} |^{2} \right)\right)$, 
which is upper-bounded by %the simpler expression 
$O\left( \ell n n'^2 + n p \max_{\icc \in [\ncc]}  | I_{V_{\icc}} \, \cap \, B |^{2} \right) $
 (both complexities assume the most efficient construction of $\bss$ is used as a primitive in Algorithm~\ref{alg:ldg-block}). 

\begin{theorem}\label{prop:ggm-or}
Algorithm~\ref{alg:sparse-gaussian-or} generates valid knockoffs for model~\eqref{eq:gauss_sparse}.
\end{theorem}

Algorithm~\ref{alg:sparse-gaussian-or} raises two key issues: how to find a suitable blocking set $B$, and how to address the fact that $\bXk_B=\bX_B$ are trivial knockoffs, so using conditional knockoffs from Algorithm~\ref{alg:sparse-gaussian-or} will have no power to select any of the variables in $B$. 

Algorithm~\ref{alg:ggm-greedy-search-graph} provides a simple greedy way to find a suitable $B$ or, given an initial blocking set $B$, can also be used to shrink $B$ (see Proposition~\ref{prop:ggm-perm-block}). 
%
%
% The algorithm visits every vertex once in the order $\pi$ of $[p]$ and decides whether the vertex is blocked or \emph{free} (not blocked). 
%Algorithm~\ref{alg:ggm-greedy-search-graph} implicitly constructs a graph $\bar{G}$ from $G$ by tracking the neighbors of each $j$ and creating new nodes to stand in for the knockoff variables.
%The neighbors of $j$ in $\bar{G}$ among the non-knockoff are grouped into $N_j$.
%
The algorithm visits every vertex in $G$ once in the order $\pi$ and decides whether each vertex it visits is blocked or \emph{free} (not blocked). Meanwhile, it constructs a graph $\bar{G}$ from $G$, which gets expanded every time a vertex $j$ is determined to be free: all pairs of $j$'s neighbors in $\bar{G}$ get connected (if not already) and a new vertex $\tilde{j}$ that has the same neighborhood as $j$ in $\bar{G}$ is added to the graph. A vertex is blocked if, when it is visited, its degree in $\bar{G}$ is greater than $n'-3$.

\begin{algorithm}[h]
\caption{Greedy Search for a Blocking Set} \label{alg:ggm-greedy-search-graph}
\begin{algorithmic}[1]
\ENSURE  $\pi$ a permutation of $[p]$, $G=([p],E)$, $n'$.
\STATE Initialize a graph $\bar{G}=G$,  and $B=\emptyset$. 
\FOR{$t= 1,\dots,p$}
\STATE Let $j=\pi_t$, and $\bar{I}_{j}$ be the neighborhood of $j$ in the graph $\bar{G}$.
\IF{$n' \geq 3+ |\bar{I}_{j}|  $} \label{line:alg-search-graph}
%\STATE Add an edge between each pair of $j$'s neighbors  in $\bar{G}$. 
%\STATE Add a vertex $\tilde{j}$ to $\bar{G}$ and add edges between $\tilde{j}$ and each of $j$'s neighbors in $\bar{G}$. 
\STATE Add edges between all pairs of vertices in $\bar{I}_{j}$.
\STATE Add a vertex $\tilde{j}$ to $\bar{G}$ and add edges between $\tilde{j}$ and all vertices in $\bar{I}_{j}$. 
\ELSE
\STATE $B\leftarrow B\cup \{j\}$.
\ENDIF
\ENDFOR
\RETURN $B$.
\end{algorithmic}
\end{algorithm}

\begin{proposition}\label{prop:ggm-block-comp}
If $B$ is the blocking set determined by Algorithm~\ref{alg:ggm-greedy-search-graph} with input $(\pi,n')$, then $G$ is $n$-separated by $B$ for any $n\ge n'$.
\end{proposition}
%This algorithm takes an input permutation $\pi$ of $[p]$, although 

%dhmargin: No need to mention the proof, and the proof is not in that appendix, too. 
%The virtual implementation of Algorithm~\ref{alg:ggm-greedy-search-graph} is given in Appendix~\ref{app:detail-ggm}. 
\rev{Algorithm~\ref{alg:ggm-greedy-search-graph} is meant to be intuitive but a more efficient implementation} is given in Appendix~\ref{app:detail-ggm}. 
Algorithm~\ref{alg:ggm-greedy-search-graph} can also be made even greedier by choosing the next $j$ at each step as the unvisited vertex in $[p]$ with the smallest degree in $\bar{G}$ (breaking ties at random), instead of following the ordering $\pi$. The algorithm also takes an input $n'$, which one may prefer to choose smaller than $n$ for computational or statistical efficiency, as we investigate in Section~\ref{sec:ggm-sims} (smaller $n'$ will mean smaller $V_k$ to generate knockoffs for in Line~\ref{line:ggm-alg1}
%Lines~\ref{line:ggm-alg1}--\ref{line:genxk} 
of Algorithm~\ref{alg:sparse-gaussian-or}). The flexibility in both $\pi$ and $n'$ is mainly motivated by the second aforementioned issue of trivial knockoffs $\bXk_B=\bX_B$, addressed next.

An intuitive solution to prevent the trivial knockoffs $\bXk_B$ in Algorithm~\ref{alg:sparse-gaussian-or} is to split the rows of $\bX$ in half and run Algorithm~\ref{alg:sparse-gaussian-or} on each half with disjoint blocking sets $B_{1}$ and $B_{2}$ such that $G$ is $n/2$-separated by both blocking sets. Then the knockoffs for variables in $B_1$ will be trivial for half the rows of $\bXk$ and those for variables in $B_2$ will be trivial for the other half of the rows of $\bXk$, but since $B_1$ and $B_2$ are disjoint, no variables will have entirely trivial knockoffs. Even though some knockoff variables are trivial for half their rows, we find the power loss for these variables to be surprisingly small, see the simulations in Section~\ref{sec:ggm-sims}.
%Precisely, we will randomly divide the rows of the covariate matrix into two halves of roughly equal sizes, run the Algorithm~\ref{alg:ggm-greedy-search} to each half with $B_1$ and $B_2$ respectively, and concatenate the resulting knockoff matrices by rows. 

This data-splitting idea is generalized in Algorithm~\ref{alg:datasplitting} to splitting the rows of $\bX$ into $\ncg$ folds and running Algorithm~\ref{alg:sparse-gaussian-or} on each fold with a different input $B$. 
\begin{algorithm}[h]
\caption{Conditional Knockoffs for Gaussian Graphical Models with Data Splitting}\label{alg:datasplitting}
\begin{algorithmic}[1]
\ENSURE $\bX\in\R^{n\times p}$, $G=([p],E)$, $B_1,\dots, B_{\ncg}\subset[p]$, $n_1,\dots,n_\ncg\in\mathbb{N}$
\REQUIRE $\bigcup\limits_{i=1}^{ \ncg} B_{i} ^{c}= [p]$, $G$ is $n_{i}$-separated by $B_{i}$ for all $i = 1,\dots,\ncg$, and $\sum\limits_{i=1}^{\ncg}n_{i}= n$.
\STATE Partition the rows of $\bX$ into submatrices $\bX^{(1)},\dots,\bX^{(\ncg)}$ with each $\bX^{(i)}\in\R^{n_i\times p}$.
\FOR{$i = 1,\dots,\ncg$}
\STATE Run Algorithm~\ref{alg:sparse-gaussian-or} on $\bX^{(i)}$ with blocking set $B_i$ to obtain $\bXk^{(i)}$.
\ENDFOR
\RETURN $\bXk = \left[ \bXk^{(1)}; \dots; \bXk^{(\ncg)}\right]$ (the row-concatenation of the $\bXk^{(i)}$'s).
\end{algorithmic}
\end{algorithm}

In Algorithm~\ref{alg:datasplitting}, since $\bigcup\limits_{i=1}^{ \ncg} B_{i} ^{c}= [p]$, for each $j\in [p]$ there is at least one $i$ such that $j\notin B_i$, and thus $\bXk_{j}\neq\bX_j$. Before characterizing when it is possible to find such $B_i$, we formalize the requirements of Algorithm~\ref{alg:datasplitting} into a definition.
\begin{definition}
$G=([p], E)$ is \emph{$(\ncg,n)$-coverable} if there exist $B_1,\dots, B_{\ncg}$ subsets of $[p]$ and integers $n_1\dots,n_{\ncg}$ such that $\bigcup\limits_{i=1}^{ \ncg} B_{i} ^{c}= [p]$, $G$ is $n_{i}$-separated by $B_{i}$ for all $i=1,\dots,\ncg$, and $\sum\limits_{i=1}^{\ncg}n_{i}\leq n$. 
\end{definition}

The following common graph structures are $(\ncg,n)$-coverable:
\begin{itemize}
\item If the largest connected component of $G$ is not larger than $(n-1)/2$, $G$ is $(1,n)$-coverable. 
\item If $G$ is a Markov chain of order $r$ (making the model a time-inhomogeneous AR($r$) model), i.e., $E = \{(i,j): 1\le |i-j|\le r\}$, and $n\geq 2+8r$, then $G$ is $(2,n)$-coverable.
\item If $G$ is a \emph{$\ncg$-colorable} (also known as \emph{$\ncg$-partite}), i.e., the vertices can be divided into $m$ disjoint sets such that the vertices in each subset are not adjacent, and $n\geq \ncg(3+\max_{j} | I_j|)$, then $G$ is $(\ncg,n)$-coverable. For example, 
\begin{itemize}
\item A tree ($\ncg=2$) in which the maximal number of children of any vertex is no more than $(n-8)/2$, 
\item A circle with $p$ even ($\ncg=2$) and $n\geq 10$, or  with $p$ odd ($\ncg=3$) and $n\geq 15$,
\item A finite subset of the $d$-dimensional lattice $\Z^{d}$ where vertices separated by distance 1 are adjacent ($\ncg=2$) and $n\geq 6+4d$.
\end{itemize}
\end{itemize}
For simple graphs such as those listed above, finding appropriate blocking sets $B_i$ can be done by inspection; see Appendix~\ref{app:ggm-eg}. % for all the examples listed above.  
More generally, determining $(\ncg,n)$-coverability for an arbitrary graph or, given an $(\ncg,n)$-coverable graph, determining blocking sets $B_i$'s that are optimal in some sense (e.g., minimizing $\Big|\bigcup\limits_{i\leq \ncg}B_i \Big|$) are beyond the scope of this work. However, in Algorithm~\ref{alg:greedy-blocking} in Appendix~\ref{app:detail-ggm}, we provide a randomized greedy search for suitable $B_{i}$'s that be applied in practice when the graph structure is too complex to find such $B_{i}$'s by inspection.

\subsubsection{Numerical Examples}\label{sec:ggm-sims}
We present two simulations comparing the power of Algorithm~\ref{alg:datasplitting} with its unconditional counterpart, one a time-varying AR$(1)$ model and the other a time-varying AR$(10)$. Line~\ref{line:ggm-alg1} of Algorithm~\ref{alg:sparse-gaussian-or} uses Algorithm~\ref{alg:ldg} with the vector $\bss$ computed using the SDP method of Equation~\eqref{opt:sdp}, and the unconditional construction also uses the SDP method \citep{EC-ea:2018}. Algorithm~\ref{alg:datasplitting} was run with $m=2$ and $B_1$ and $B_2$ chosen by fixing $n'$ (specified in the following paragraphs) and running Algorithm~\ref{alg:ggm-greedy-search-graph} twice with two different $\pi$'s. The first run used the original variable ordering for $\pi$, and the second run used ordered $B_1$ followed by the ordered remaining variables.\footnote{This is a nonrandomized version of Algorithm~\ref{alg:greedy-blocking}, which works well for AR models because of their graph structure.}
\rev{
We remind the reader that the simulation setting is at the beginning of Section~\ref{sec:ex}. 
}

In Figure~\ref{fig:ggm1}, the $\bx_i\in \R^{2000}$ are i.i.d. AR$(1)$ with autocorrelation coefficient 0.3 (although the autocorrelation coefficient does not vary with time, this is not assumed by Algorithm~\ref{alg:datasplitting}). We chose $n' = 40$, resulting in $210$ variables that are each blocked in half the samples. The number of unknown parameters is $3p-1=5,999$ while the sample sizes simulated are much smaller, $n\le 350$, yet the power of conditional knockoffs is nearly indistinguishable from that of unconditional knockoffs which uses the exactly-known distribution of $X$.

In Figure~\ref{fig:ggm10}, the $\bx_i\in\R^{2000}$ are time-varying AR(10); specifically, $\bx_i\iid \N(\bs{0},\bS)$ where $\bS$ is the renormalization of $\bS^0$ to have 1's on the diagonal, and $\left(\bS^0\right)^{-1}_{j,k} = \one{j=k}-0.05\cdot \one{1\leq |j-k|\leq 10}.$ We chose $n'=50$, resulting in $1,660$ variables that are each blocked in half the samples. The number of unknown parameters is $2p+10p-10\times 11/2 = 23,945$ while the sample sizes are again much smaller, $n\le 500$, and the power difference between conditional and unconditional knockoffs remains very slight. 

%{\color{red} . max SE = 0.007456619 for AR(1) and 0.007558055 for AR(10). }

%{\color{red} $|B_1\cup B_2|$: for AR(1), 210. for AR(10), 1660. Roughly, for AR(r), every $n'/2$ variables need a block of $r$, so $p/(n'/2)*r*m$. }

\begin{figure}[h]
\centering
%  \subfloat[FDR]{     \includegraphics[width=0.5\linewidth]{figures/GGM_AR1/FDR agains A} }
%    ~ %\hfill
  \subfloat[
  %AR(1) with coefficient $0.3$
  ]{ \includegraphics[width=0.45\linewidth]{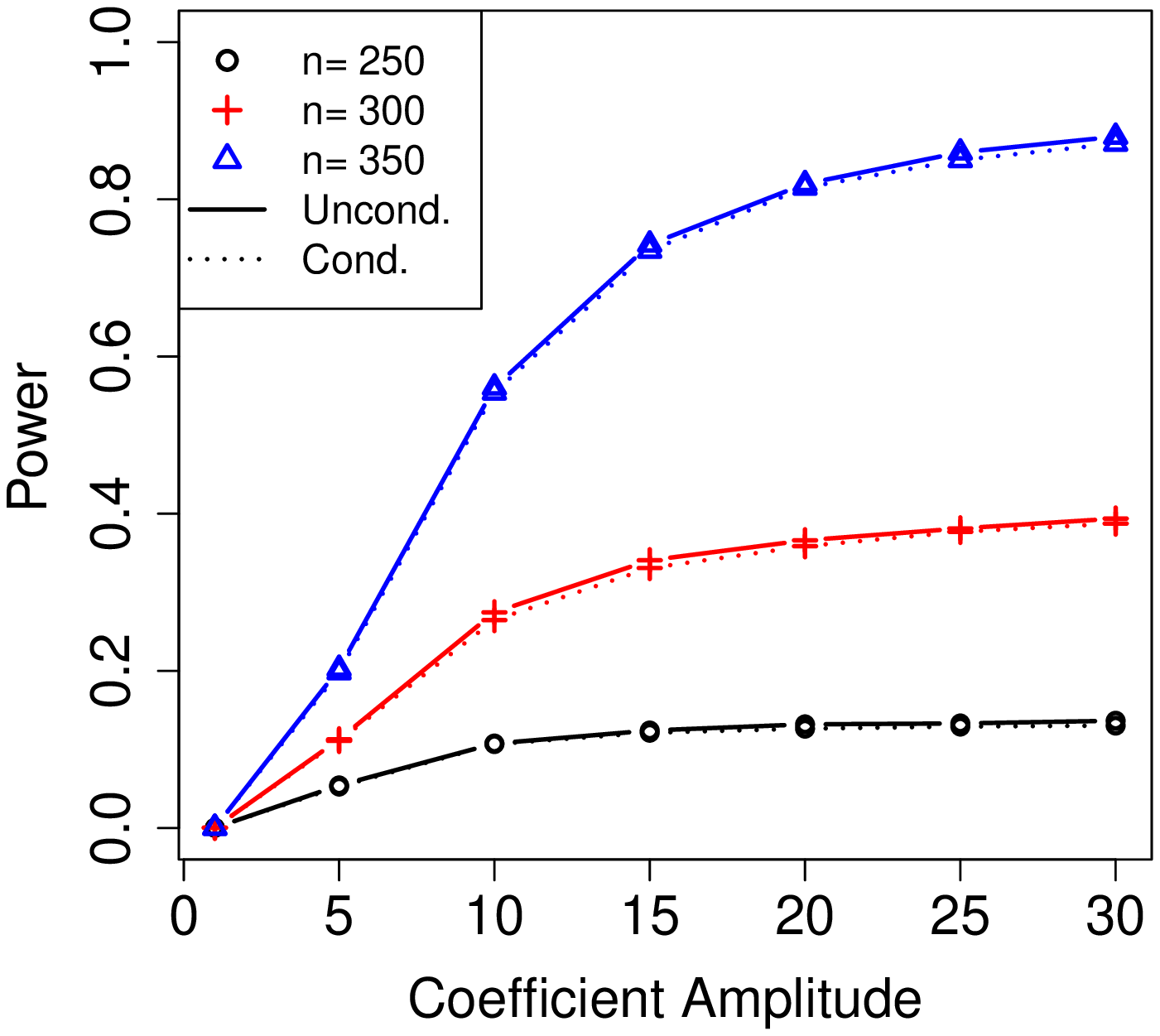}
   \label{fig:ggm1}
    }
%  \subfloat[FDR]{     \includegraphics[width=0.5\linewidth]{figures/GGM_AR10/FDR agains A} }
%    ~ %\hfill
  \subfloat[
  %AR(10) with partial-coefficient roughly $0.05$
  ]{     \includegraphics[width=0.45\linewidth]{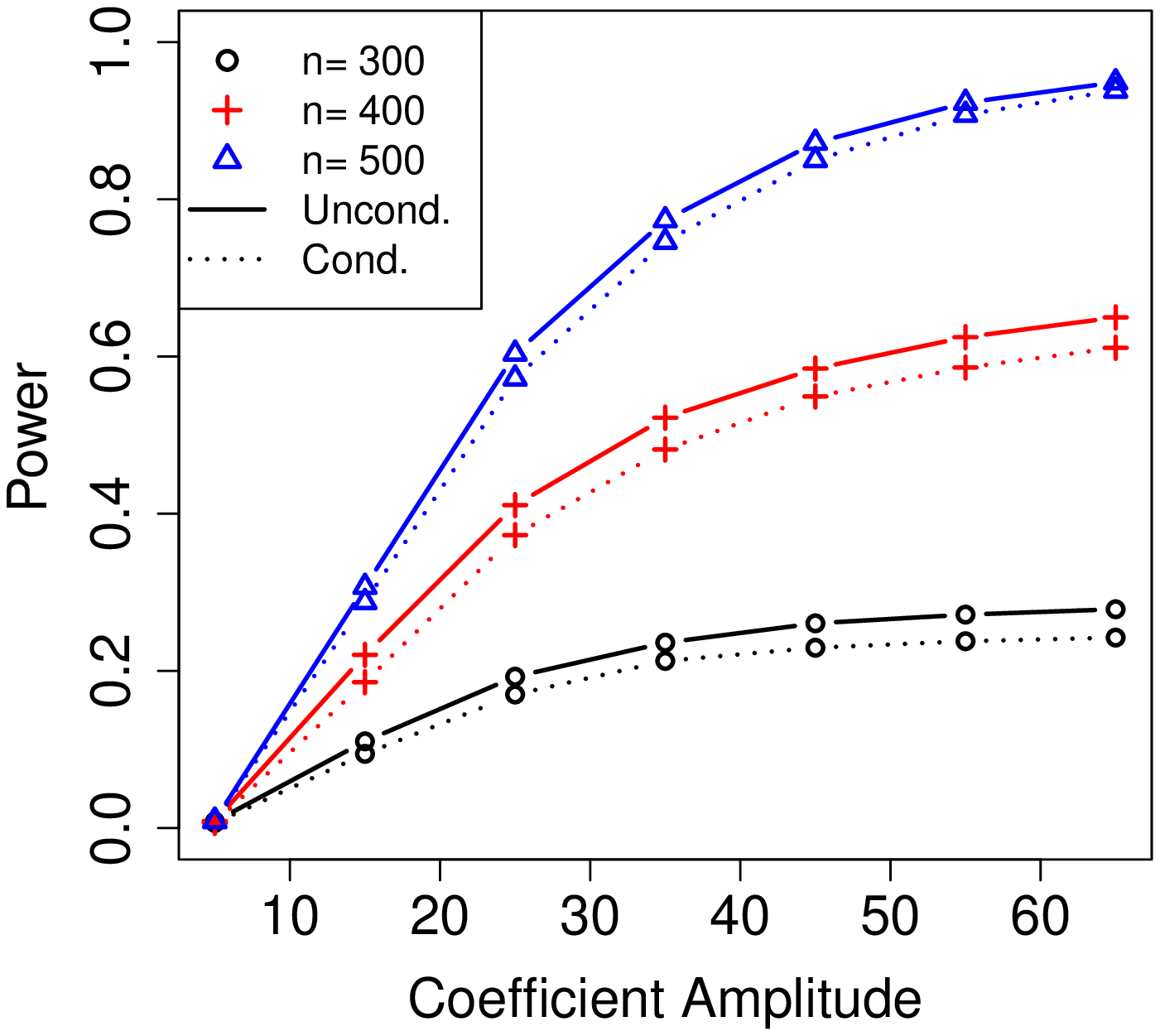} 
  \label{fig:ggm10}
  }
 \caption{Power curves of conditional and unconditional knockoffs for $p=2000$ and a range of $n$ for (a) an AR$(1)$ model and (b) an AR$(10)$ model. Standard errors are all below $0.008$.}
\label{fig:ggm} %The nominal FDR level is $0.2$.
 \end{figure}

Note that the simulation in Figure~\ref{fig:ggm1} blocked on just roughly 10\% of its variables (i.e., $|B_1\cup B_2|/p\approx 10\%$), and since the signals are uniformly distributed, one might worry that in specific applications where the blocked variables and signals happened to align, the power loss might be much worse. But Figure~\ref{fig:ggm10}'s simulation blocked on over 80\% of its variables and still suffered very little power loss compared to unconditional knockoffs, suggesting that even the blocking of signal variables has only a small effect on power thanks to the data splitting in Algorithm~\ref{alg:datasplitting}.

Finally, we examine the sensitivity of the power of conditional knockoffs to the choice of $n'$ in Algorithm~\ref{alg:ggm-greedy-search-graph} for choosing the $B_i$. In the case of AR($1$) with $n=300$ and $p=2000$, Figure~\ref{fig:ggm_nprime_dc} shows the averaged density\footnote{3200 independent simulations were averaged and the kernel density estimate used a Gaussian kernel with a bandwidth of 0.01.} of original-knockoff correlations $\tilde{\rho}_j=\bX_j\tp \bXk_j/ (\|\bX_j\|\|\bXk_j\|)$ for three different choices of $n'$, and Figure~\ref{fig:ggm_nprime_power} shows the corresponding power curves. Recall that smaller $n'$ means blocking on more variables but generating better knockoffs for the non-blocked variables in each step $i$ of Algorithm~\ref{alg:datasplitting}. Figure~\ref{fig:ggm_nprime_dc} shows quite different correlation profiles for different $n'$, with $n'=40$ seeming to provide the density with mass most concentrated to the left. Indeed Figure~\ref{fig:ggm_nprime_power} shows $n'=40$ is most powerful, but only by a small margin---the power is quite insensitive to the choice of $n'$. In applications, the choice of $n'$ may rely on an approximate version of Figure~\ref{fig:ggm_nprime_dc} obtained by simulating $\bX$ from an estimated model. %\ljmargin{-}{please fill in upper-bound on standard errors in Fig 3 caption. Also, please replace 3a's x-axis label with Original-Knockoff Correlations.}

\begin{figure}\centering
%  \subfloat[FDR]{     \includegraphics[width=0.5\linewidth]{figures/GGM_AR1/FDR agains A} }
%    ~ %\hfill

%  \subfloat[FDR]{     \includegraphics[width=0.5\linewidth]{figures/GGM_AR10/FDR agains A} }
%    ~ %\hfill
  \subfloat[
  %Empirical Diagonal Correlation 
  ]{     \includegraphics[width=0.45\linewidth]{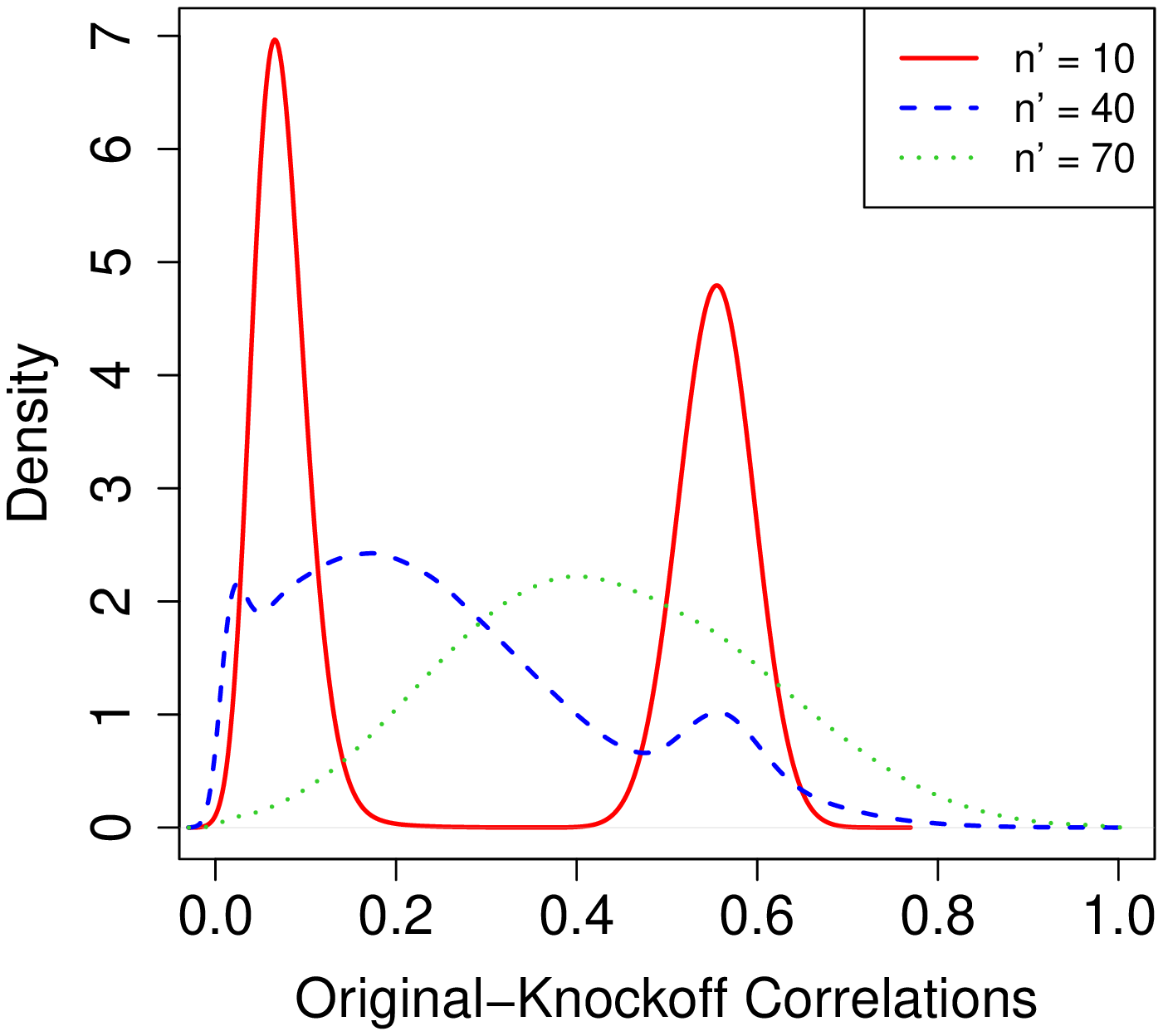} 
  \label{fig:ggm_nprime_dc}
  }
\subfloat[
%Power
]{     \includegraphics[width=0.45\linewidth]{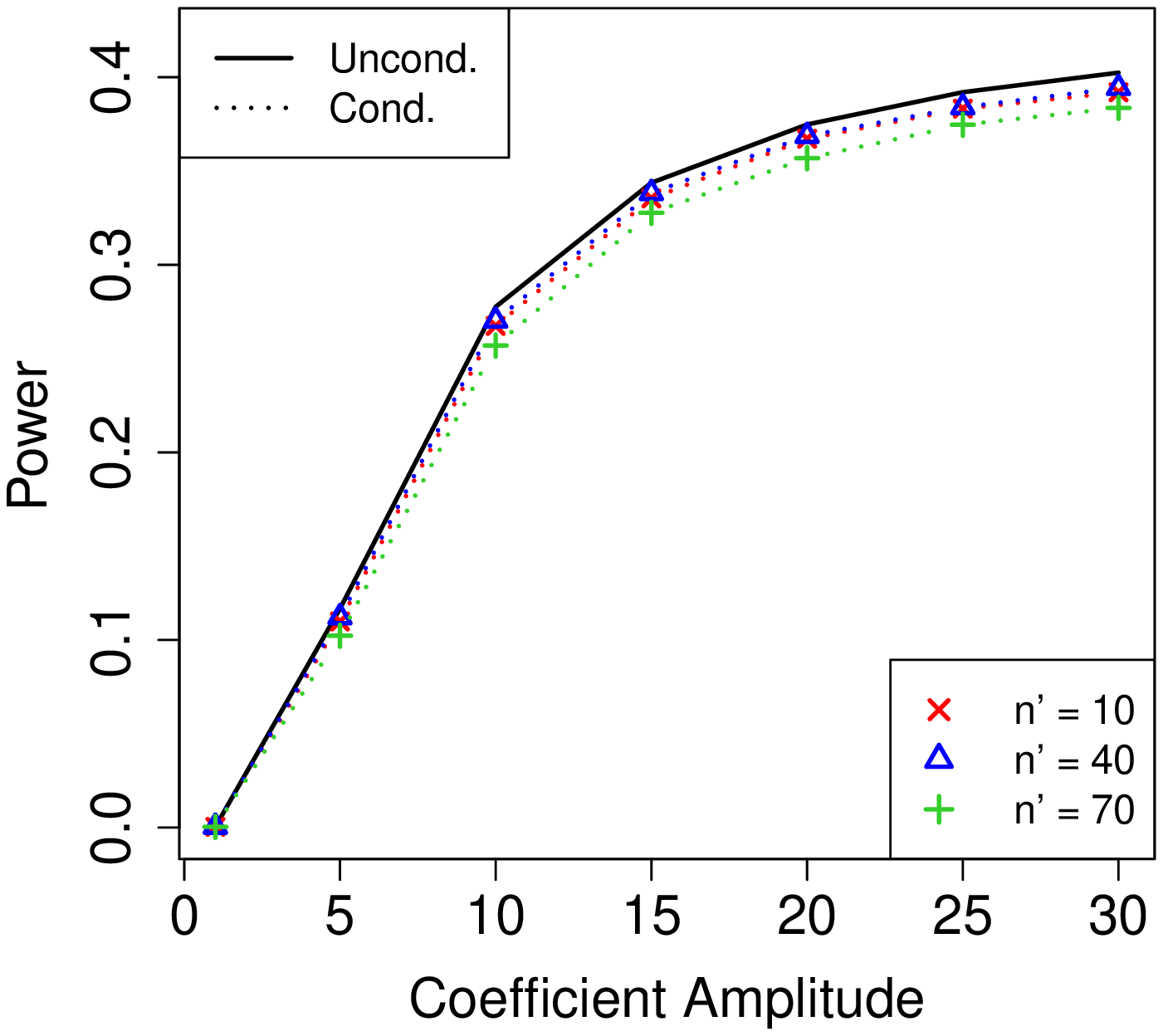}
   \label{fig:ggm_nprime_power}
    }
 \caption{Sensitivity of conditional knockoffs to the choice of $n'$ for an AR(1) model with $n=300$ and $p=2000$. (a) Histograms of the original-knockoff correlations and (b) power curves. Standard errors in (b) are all below 0.004.}\label{fig:ggm_nprime}
%The nominal FDR level is $0.2$.
 \end{figure}

\rev{
In Appendix~\ref{app:Simulation}, we provide additional experiments that compare the performance of conditional knockoffs that are generated using different sufficient statistics (Appendix~\ref{app:simu-different-ss}) and examine the scenario where a superset of the edge set $E$ is unknown and is instead estimated using the data (Appendix~\ref{app:simu-unknown-graph}). 
}

%\subsection{Gaussian Factor Model?}

\subsection{Discrete Graphical Model}\label{sec:eg-dgm}
We now turn to applying conditional knockoffs to discrete models for $X$. Such models are used, for example, for survey responses, general binary covariates, and single nucleotide polymorphisms (mutation counts at loci along the genome) in genomics. Many discrete models assume some form of local dependence, for instance in time or space. We will show how to construct conditional knockoffs when that local dependence is modeled by (undirected) graphical models (see, e.g., \citet[Chapter 2]{edwards2000introduction}), for example, Ising models, Potts models, and Markov chains. 
% \ljmarin{}{Add citation, e.g., to Introduction to Graphical Modeling by David Edwards, Chapter 2}
% In particular, in Section \ref{sec:eg-mc}, we improve the construction for Markov chain models by conditioning on fewer statistics.

%\dhmargin{-}{we no longer use the global markov property. Maybe delete it?}
A random vector $X$ is Markov with respect to a graph $G=([p],E)$ if for any two disjoint subsets $A,A' \subset [p]$ and a \emph{cut set} $B\subset [p]$ such that every path from $A$ to $A'$ passes through $B$, it holds that  $X_A \indp X_{A'} \mid X_B$. Denote by $I_j$ the vertices adjacent to $j$ in $G$ (excluding $j$ itself). $X$ being Markov  implies the \emph{local Markov property} that $X_j \indp X_{ (\{j\}\cup  I_j )^{c} } \mid X_{I_j}$. 

In this section, we assume $X$ is locally Markov with respect to a known graph $G$ and each variable $X_j$ takes $\St_{j} \ge 2$ discrete values (for simplicity label these values $[\St_j]=\{1,\dots, \St_{j} \}$). Although the algorithms in this section can be applied when $\St_{j}$ is infinite, we assume for simplicity that $\St_{j}$ is finite. Formally, we assume 
\begin{equation}\label{model:dgm}
\Fx \in\left\{\text{distribution on } \prod_{j=1}^p [ \St_j] \text{ satisfying the local Markov property w.r.t. } G \right\}.
\end{equation}
%We introduce some notation to streamline the presentation in this section. For a $n\times c$ matrix $\bZ$, and integers $k_1,\dots, k_c$, let \ljmargin{$\mrc{\bZ,\st_1,\dots, \st_c}\; :=\; \sum_{i=1}^{n} \one{Z_{i,1}=\st_1,\dots,Z_{i,c}=\st_c}$}{Can we use $C$ instead of $\Phi$? Or some other capital non-Greek letter that we haven't used yet?}, i.e., the number of rows of $\bZ$ that equal the vector $(\st_1,\dots,\st_c)$. Here, $\one{\cdot}$ is the indicator function. 

\subsubsection{Generating Conditional Knockoffs by Blocking} \label{sec:eg-dgm-1}
Our algorithm for generating conditional knockoffs for discrete graphical models uses again the ideas of blocking and data splitting in Section~\ref{sec:eg-ggm}. However, unlike Section~\ref{sec:eg-ggm} which built upon the low-dimensional construction of Section~\ref{sec:ldg}, there is no known efficient algorithm for constructing conditional knockoffs for general discrete models in low dimensions. As such, instead of blocking to isolate small graph components, we now block to isolate \emph{individual} vertices, and as such need to be more careful with data splitting to ensure the resulting knockoffs remain powerful.

Suppose $B$ is a cut set such that every path connecting \emph{any} two different vertices in $B^{c}$ passes through $B$; call such a set a \emph{global cut set} with respect to $G$. The local Markov property implies the elements of $X_{B^c}$ are conditionally independent given $X_B$:
\[
\Pcr{X_{B^c}}{X_{B}} \, =\, \prod_{j \in B^c} \Pcr{X_{j}}{X_{B}}\, =\, \prod_{j \in B^c} \Pcr{X_{j}}{X_{ I_j}},
\]
where we used the fact that for any $j\in B^c$, $I_j\subseteq B$ and $X_j\indp X_{B\setminus  I_j} \mid X_{I_j}$. For any $A\subseteq [p]$ and $\st_1,\dots, \st_p$, denote by $\bst_{A}$ the vector of $\st_{j}$'s for $j\in A$ and by $[\bSt_{A}]$ the cartesian product $\prod\limits_{j\in A} [\St_{j}]$. Then the conditional probability $\Pcr{X_{j}}{X_{ I_j}}$ can be written as
\[
\prod_{ \st_{j} \in [\St_j],\bst_{I_j}\in [\bSt_{I_j}] } \theta_j ( \st_{j}, \bst_{I_j} )^{\one{X_j= \st_{j}, X_{I_j}=\bst_{I_j}}},
\]
%, where $\bst_{I_j}$ is the vector of $\st_l$'s for all $\ell \in I_j$ and $[\bSt_{I_j}]:=\; \prod\limits_{\ell\in I_j} [\St_l]$. 
with parameters $\theta_j ( \st_{j}, \bst_{I_j} )\in [0,1]$ for all $\st_{j}$, $\bst_{I_j}$, with the convention that $0^0 :=\; 1$. 
%The joint distribution of $X$ can now be written as 
%\[
%%g_{\psi, \bs{\theta}}(X) \, =\, 
%\psi_B(X_B)\prod_{j \in B^c} \left( \prod_{ \st_{j} \in [\St_j],  \bst_{I_j} \in [\bSt_{I_j}]    } \theta_j ( \st_{j}, \bst_{I_j} )^{\one{X_j= \st_{j}, X_{I_j}=\bst_{I_j}}}  \right),
%\]
%for some probability mass function $\psi_B$. The joint distribution for $n$ i.i.d. samples from the graphical model is then
%% not used : $\st_l \in [\St_l] :  \forall \, \ell \in I_j  $
%\[
%\prod_{i=1}^{n}\psi_B(X_{i,B})\prod_{j \in B^c} \left( \prod_{ \st_{j} \in [\St_j], \bst_{I_j} \in [\bSt_{I_j}]  } \theta_j ( \st_{j}, \bst_{I_j} )^{ 
%N_j(\st_j,\bst_{I_j})
%%\mrc{[\bX_j, \bX_{I_j}], \st_{j}, \bst_{I_j}}
%}  \right),
%\]
%where $N_j(\st_j,\bst_{I_j})= \sum_{i=1}^n \one{X_{i,j}=\st_j,\bX_{i,{I_j}}=\bst_{I_j}}$. %\ljmargin{}{I removed $\Phi$ from the main paper since we don't use it much and I don't like the notation. So since you presumably use it in the proofs, you will have to define it there.}
Let $\psi_B(X_B)$ be the probability mass function for $X_B$, 
the joint distribution for $n$ i.i.d. samples from the graphical model is then
% not used : $\st_l \in [\St_l] :  \forall \, \ell \in I_j  $
\[
\prod_{i=1}^{n}\psi_B(X_{i,B})\prod_{j \in B^c} \left( \prod_{ \st_{j} \in [\St_j], \bst_{I_j} \in [\bSt_{I_j}]  } \theta_j ( \st_{j}, \bst_{I_j} )^{ 
N_j(\st_j,\bst_{I_j})
%\mrc{[\bX_j, \bX_{I_j}], \st_{j}, \bst_{I_j}}
}  \right),
\]
where $N_j(\st_j,\bst_{I_j})= \sum_{i=1}^n \one{X_{i,j}=\st_j,\bX_{i,{I_j}}=\bst_{I_j}}$. 
Let $T_{B}(\bX)$ be the statistic that includes $\bX_B$ and the counts $N_j( \st_{j}, \bst_{I_j})$ for all $j\in B^c$ and all possible $( \st_{j}, \bst_{I_j})$.  Then $T_{B}(\bX)$ is a sufficient statistic for model \eqref{model:dgm}. Conditional on $T_B(\bX)$, the random vectors $\{\bX_j,  j\in B^c\}$ are independent and each $\bX_j$ is uniformly distributed  on all $\bw\in [\St_j]^{n}$ such that $\sum_{i=1}^n \one{w_i=\st_j,\bX_{i,{I_j}}=\bst_{I_j}}=N_{j}( \st_{j},\bst_{I_j})$ for any $(\st_j, \bst_{I_j})$. Algorithm~\ref{alg:block-dg} generates knockoffs conditional on $T_B(\bX)$ by, for each $j$, uniformly permuting subsets of entries of $\bX_j$ to produce $\bXk_j$. The subsets of entries are defined by blocks of identical rows of $\bX_{I_j}$ so that $\sum_{i=1}^n \one{\bXk_{i,j}=\st_j,\bX_{i,{I_j}}=\bst_{I_j}}=N_{j}( \st_{j},\bst_{I_j})$, as required.
\begin{algorithm}[h]
\caption{Conditional Knockoffs for Discrete Graphical Models} \label{alg:block-dg}
\begin{algorithmic}[1]
  \ENSURE $\bX\in\mathbb{N}^{n\times p}$, $G =([p],E)$, $B\in [p]$.
\REQUIRE $B$ is a global cut set of $G$.
\FOR{$j$ in $[p]\setminus  B$ }
\STATE Initialize  $\bXk_j$ to $\bX_j$.
\FOR{$\bst_{I_j} \in [\bSt_{I_j}]$} \label{alg:dgm-enum}
\STATE Uniformly randomly permute the entries of $\bXk_j$ whose corresponding rows of $\bX_{I_j}$ equal $\bst_{I_j}$.
%\STATE Let $V$ be the row indices for which $\bX_{I_j}$ equals $\bst_{I_j}$
%\STATE Set the subvector $\bXk_{j}(V)$ to be a uniform random permutation of $\bX_{j}(V)$. % {\color{red} And only takes the one that is different from $\bX_{j}$?}
\ENDFOR
\ENDFOR
\STATE Set $\bXk_B=\bX_B$. 
\RETURN $\bXk = [\bXk_1,\dots,\bXk_p]$.
\end{algorithmic}
\end{algorithm}

The computational complexity of Algorithm~\ref{alg:block-dg} is $O\left(  \sum\limits_{j\in B^{c}}(n+ \min (\prod\limits_{\ell\in I_j}\St_{\ell} , n |I_j| ) ) \right)$, which is  shown in Appendix~\ref{app:detail-dgm}.  If {$n>\max_{j\in B^c} \prod\limits_{\ell\in I_j}\St_{\ell}$}, as needed to guarantee nontrivial knockoffs for all $j\in B^c$ are generated with positive probability, then the complexity can be simplified to $O\left( n (p-|B|) \right)$. In general, Algorithm~\ref{alg:block-dg}'s computational complexity is bounded by the simple expression $O(n p\bar{d})$, where $\bar{d}$ is the average degree in $B^c$.

\begin{theorem}\label{thm:dgm-block}
Algorithm~\ref{alg:block-dg} generates valid knockoffs for model~\eqref{model:dgm}.
\end{theorem}

As with Algorithm~\ref{alg:sparse-gaussian-or}, in Algorithm~\ref{alg:block-dg} variables in $B$ are blocked and their knockoffs are trivial: $\bXk_B = \bX_B$. One way to mitigate this drawback is to, after running Algorithm~\ref{alg:block-dg}, expand the graph to include the generated knockoff variables and then conduct a second knockoff generation with the expanded graph. 
\rev{We elaborate on this idea and present Algorithm~\ref{alg:dgm-expanding}, a modified version of Algorithm~\ref{alg:block-dg},  in Appendix~\ref{app:mc}. }
Another systematic way to address this issue is to take the same approach as Algorithm~\ref{alg:datasplitting} by splitting the data and running Algorithm~\ref{alg:block-dg} (or Algorithm~\ref{alg:dgm-expanding}) on each split with different $B$'s; see Algorithm~\ref{alg:datasplitting-discrete}. 
\begin{algorithm}[h]
\caption{Conditional Knockoffs for Discrete Graphical Models with Data Splitting}\label{alg:datasplitting-discrete}
\begin{algorithmic}[1]
\ENSURE $\bX\in\mathbb{N}^{n\times p}$, $G=([p],E)$, $B_1,\dots, B_{\ncg}\subset[p]$, $n_1,\dots,n_\ncg\in\mathbb{N}$.
\REQUIRE$[p] = \bigcup_{i=1}^{\ncg}B_i^c$ and each $B_i$ is a global cut set.
\STATE Partition the rows of $\bX$ into submatrices $\bX^{(1)},\dots,\bX^{(\ncg)}$ with each $\bX^{(i)}\in\mathbb{N}^{n_i\times p}$.
\FOR{$i = 1,\dots,\ncg$}
\STATE Run Algorithm~\ref{alg:block-dg} or \ref{alg:dgm-expanding} on $\bX^{(i)}$ with $B_i$ to obtain $\bXk^{(i)}$. \label{line:dgm-sub}
\ENDFOR
\RETURN $\bXk = \left[ \bXk^{(1)}; \dots; \bXk^{(\ncg)}\right]$ (row-concatenation of $\bXk^{(i)}$'s).
\end{algorithmic}
\end{algorithm}

If $n_i > \max\limits_{j\in B_i^c} \prod\limits_{\ell\in I_j} \St_{\ell}$ for all $i\leq m$ and all the model parameters $\theta_j ( \st_{j}, \bst_{I_j} )$ are positive, then Algorithm~\ref{alg:datasplitting-discrete} produces nontrivial knockoffs for all $j$ with positive probability. Note that in the continuous case, similar mild conditions guarantee that Algorithm~\ref{alg:datasplitting} produces nontrivial knockoffs for all $j$ with \emph{probability 1}. This is unachievable in general in the discrete case no matter how the sufficient statistic is chosen, as there is always a positive probability (for every $j$) that the sufficient statistic takes a value such that $\bXk_j = \bX_j$ is uniquely determined given that sufficient statistic (e.g., if $\bX_{i,j}= 1$ for all $i$).

	%It worths mentioning that the definition of $(\ncg,n)$-coverable graph in Section~\ref{sec:gbic} guarantees nontrivial knockoff with probability $1$ whereas the condition above only stochastically guarantees the existence of nontrivial knockoffs. This is indeed the cases for discrete data when $T(\bX)$ and $\bX_B$ uniquely determines $\bX_{B^c}$.

%dhmargin{}{The $m$-colorable condition is essential but the way to translate the coloring to the blocking set can vary and I don't know what is the optimal way.  }
One way to ensure $B_1,\dots,B_m$ satisfy the requirements of Algorithm~\ref{alg:datasplitting-discrete} is if assigning each $B_i^c$ a different color produces a proper coloring of $G$.\footnote{A coloring of $G$ is \emph{proper} if no adjacent vertices have the same color. } The end of Section~\ref{sec:ggm-alg-proof} listed some common graph structures with known chromatic numbers,\footnote{The chromatic number of a graph $G$ is the minimal $m$ such that $G$ is $m$-colorable.} which subsume many common models including Ising models and Potts models. Although not specified in Section~\ref{sec:ggm-alg-proof}, a Markov chain of order $\ncg-1$ is $\ncg$-colorable and a planar graph (map) is 4-colorable. Also, for any graph of maximal degree $d$, a $(d+1)$-coloring can be found in $O(dp)$ time by greedy coloring \citep[Chapter 2]{RL:2016}. In general, both finding the chromatic number and finding a corresponding coloring of a graph $G$ are NP-hard \citep{GM-JD:2002}, but there exist efficient algorithms that in practice are able to color graphs with a near-optimal number of colors (see \citet{EM-PT:2010} for a survey). % (e.g., \cite{KD-MR-SM:1998, ML-RM:2001}).

\subsubsection{Refined Constructions for Markov Chains}\label{sec:eg-mc}
For Markov chains, we develop two alternative conditional knockoff constructions that take advantage of the Markovian structure. Although we generally expect these constructions to dominate Algorithm~\ref{alg:datasplitting-discrete} when $G$ is a Markov chain, we found the difference in power to be negligible in every simulation we tried, and so we defer these algorithms to Appendix~\ref{app:mc} and only provide a brief summary here.

Suppose the components of $X$ follow a $K$-state discrete Markov chain, and let $\pi^{(1)}_{\st} = \pr{X_{1}=\st}$ and $\pi^{(j)}_{\st_{},\st'} = \Pc{X_{j}=\st'}{X_{j-1}=\st_{}}$ be the model parameters. Then the joint distribution for $n$ i.i.d. samples is,
\begin{align*}
\pr{\bX} 
&= \; \prod_{\st=1}^{ \St_{}}(\pi_{  \st}^{(1)})^{\sum_{ \st' =1}^{ \St_{}}N_{\st,\st' }^{(2)}} \prod_{j=2}^p \prod_{\st =1}^{\St_{}}\prod_{ \st' =1}^{\St_{}}(\pi_{\st ,\st' }^{(j)})^{N_{\st ,\st' }^{(j)}},
\end{align*}
where $N^{(j)}_{\st_{},\st'} = \sum_{i=1}^n\one{X_{i,j-1}=\st,X_{i,j}=\st'}$. So all the $N^{(j)}_{\st_{},\st'}$'s together form a sufficient statistic, which we denote by $T(\bX)$. As opposed to the statistics $N_j( \st_{j}, \bst_{ \{j-1,j+1\} })$'s used in Section~\ref{sec:eg-dgm-1}, $T(\bX)$ is minimal, and thus we expect that generating knockoffs conditional on it will be more powerful than knockoffs generated conditional on a non-minimal statistic. Conditional on $T(\bX)$, the columns of $\bX$ still comprise a Markov chain whose distribution can be used to generate knockoffs in two possible ways:
\begin{enumerate}
	\item \emph{The sequential conditional independent pairs (SCIP)} algorithm \citep{EC-ea:2018,MS-CS-EC:2017} has computational complexity exponential in $n$, but by splitting the samples into small folds and generating conditional knockoffs separately for each fold, $n$ is artificially reduced and the computation made tractable.
	\item \emph{Refined blocking} modifies Algorithm~\ref{alg:block-dg} by first drawing a new contingency table that is exchangeable with the the three-way contingency table for $(\bX_{j-1},\bX_{j},\bX_{j+1})$ and then sampling $\bXk_j$ given the new contingency table. 
\end{enumerate}

\begin{comment}
{\color{red}
The graph-expanding can improve the knockoff generation for Markov chain. 
Without graph-expanding, Algorithm~\ref{alg:datasplitting-discrete} will split the data into two folds, each of which is used to generate knockoff for either all odd or all even variables. With graph-expanding, we are able to `unblock' half of the blocked variables in each fold, and thus obtain strictly better knockoff. However, this introduce artificial unfairness among variables, and 
a better implementation of the graph-expanding is the following. 

We split the data into four folds, and label each variable $j$ by $j\mod 4$. For the $i$th fold data where $i\in [4]$,  Algorithm~\ref{alg:dgm-expanding} first generates knockoff for variables with labels $i \mod 4$ and $(i+2)\mod 4$ via Algorithm~\ref{alg:block-dg} by blocking the variables with labels $(i+1) \mod 4$ and $(i+3)\mod 4$, and then expands the graph to include the new knockoff vertices. Finally, Algorithm~\ref{alg:block-dg} can be applied again  to generate knockoffs for the variables with label $(i+1)\mod 4$ by blocking all other variables in the expanded graph. In this way, only one quarter of each column in the final knockoff matrix is forced to be trivial.
}
\end{comment}

\subsubsection{Numerical Examples}\label{sec:dgm-sims}

%We present two simulations, one comparing the power of Algorithm~\ref{alg:datasplitting-discrete} with its unconditional counterpart for discrete Markov chains \citep{MS-CS-EC:2017}, and the other demonstrating the versatility of our conditional approach by showing its power on an Ising model, for which no other construction exists in the literature. 
We present two simulations, comparing the power of Algorithm~\ref{alg:datasplitting-discrete} with its unconditional counterpart for discrete Markov chains \citep{MS-CS-EC:2017} and for Ising models \citep{SB-EC-LJ-WW:2019}. 
\rev{
We remind the reader that the simulation setting is at the beginning of Section~\ref{sec:ex}. 
}

In Figure~\ref{fig:MC}, the \rev{$\bx_i\in \{0,1\}^{1000}$}  are i.i.d. from an inhomogeneous binary Markov chain with $p=1000$. The initial distribution is ${\pr{ X_1=0}=\pr{ X_1=1}=.5}$, and the transition probabilities 
\[
{\pr{ X_j=0| X_{j-1}=1}=Q_{10}^{(j)}},\; \quad {\pr{ X_j=1| X_{j-1}=0}=Q_{01}^{(j)}} 
\]
are randomly generated as 
\[
Q_{10}^{(j)}=\frac{U_1^{(j)} }{0.4+U_1^{(j)} +U_2^{(j)} },\;\quad  Q_{01}^{(j)}=\frac{U_3^{(j)} }{0.4+U_3^{(j)} +U_4^{(j)} }, 
\]
where $U_i^{(j)} \iid \text{Unif}([0,1])$ but held fixed across all replications. 
%{\color{red}
%Markov chains are 2-colorable, and Algorithm~\ref{alg:datasplitting-discrete} can takes $B_1$ as the even variables and $B_2$ as the odds, with $n_1=n_2=n/2$. To gain a small improvement, we alternatively use 4-fold data splitting with graph-expanding procedure described at the end of Section~\ref{sec:eg-mc} to generate knockoff.} 
We implemented Algorithm~\ref{alg:datasplitting-discrete} with $B_1$ as the even variables and $B_2$ as the odds, with $n_1=n_2=n/2$, and used Algorithm~\ref{alg:dgm-expanding} (with $Q=2$) in Line~\ref{line:dgm-sub}. %\footnote{Alternatively, one can partition the data into four submatrices ($m=4$ in Algorithm~\ref{alg:datasplitting-discrete}) and generate knockoffs with one-step graph-expanding that unblocks different variables so as to be egalitarian.}
The number of unknown parameters in the model is $2p-1=1,999$ and all plotted power curves have $n\leq 350$. Despite the high-dimensionality, conditional knockoffs are nearly as powerful as the unconditional SCIP procedure of \citet{MS-CS-EC:2017} which requires knowing the exact distribution of $X$.
%\ljmargin{}{Please fill in standard errors in the caption. The Uncond/Cond legend in 4a weirdly has the legend lines going right up to the legend border, please fix that.}

\begin{figure}[h]\centering
%  \subfloat[FDR]{     \includegraphics[width=0.5\linewidth]{figures/MC/FDR agains A} }
%    ~ %\hfill
  \subfloat[%inhomogeneous Markov Chain with random transition matrix.%The solid line is the unconditional knockoff,  the dot line is the conditional knockoff by blocking and the dash line is conditional SCIP. 
  ]{\includegraphics[width=0.45\linewidth]{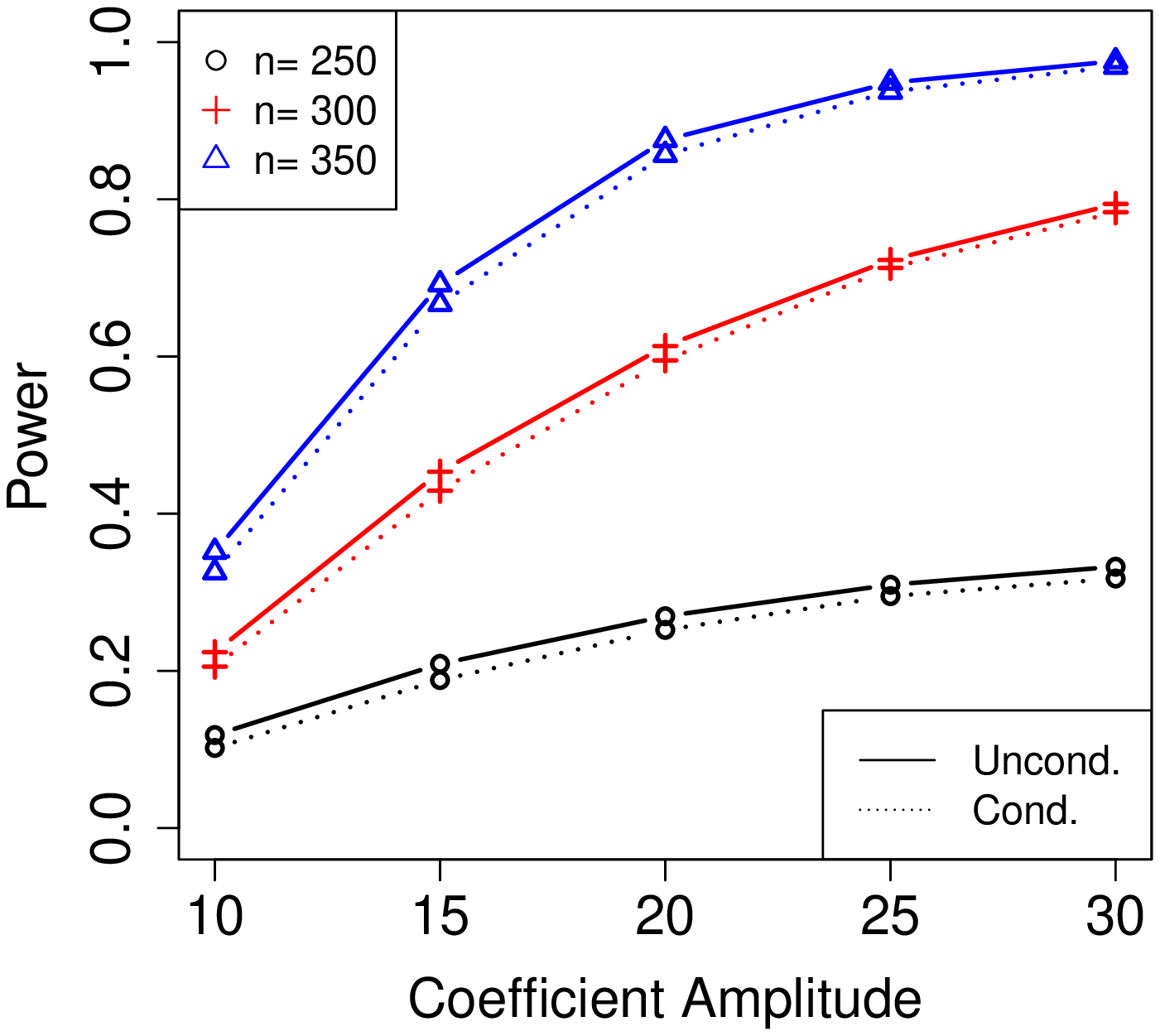} \label{fig:MC}
  }
  ~
    \subfloat[%$X_i$'s are i.i.d. samples from 10 independent 2D spatial-lattice Ising model of size $10\times 10$.
    ]{     \includegraphics[width=0.45\linewidth]{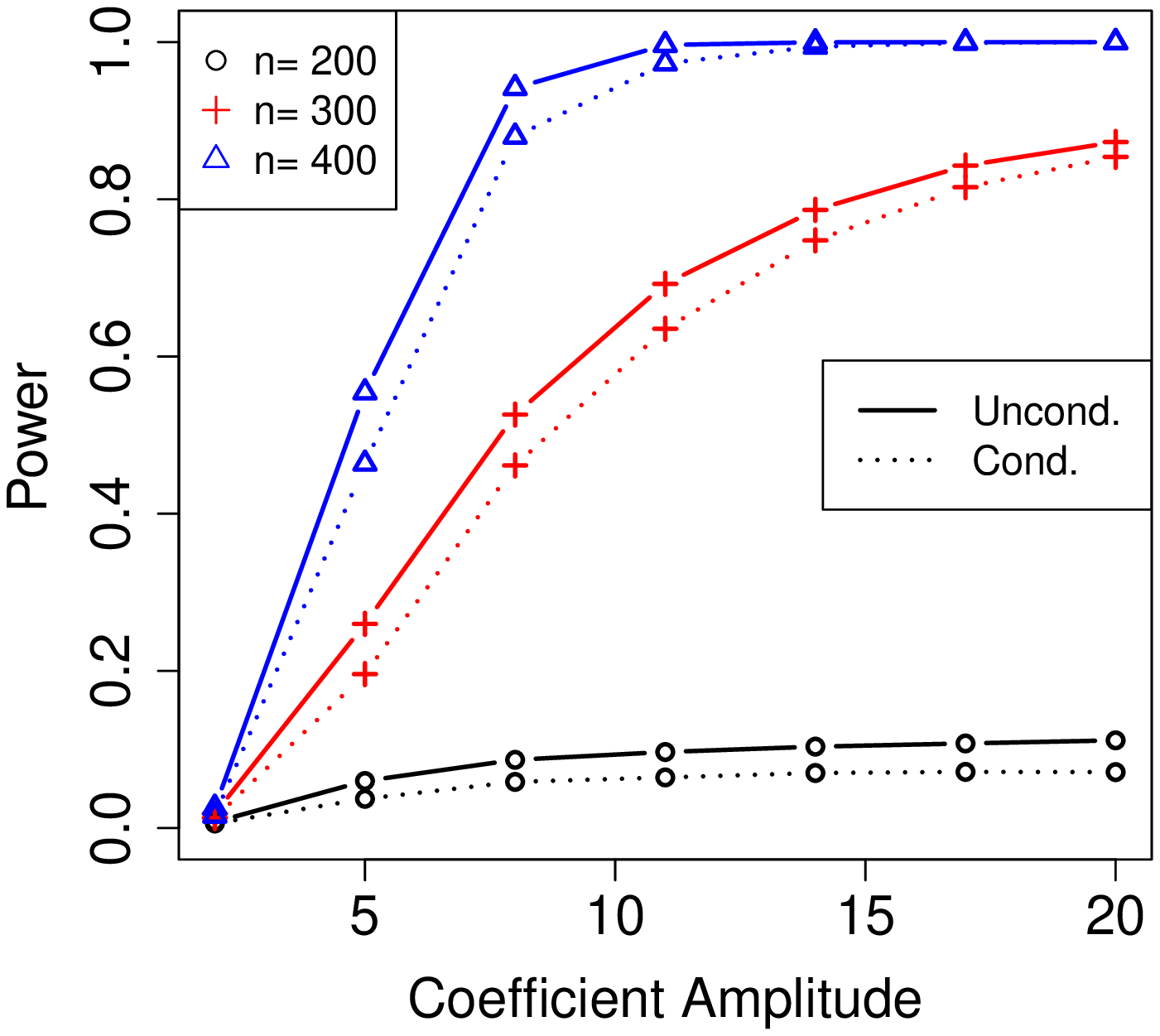} \label{fig:DGM}}
%  \caption{Power curves with $p=1000$ and a range of $n$ for (a) conditional and unconditional knockoffs for a Markov chain and (b) conditional knockoffs for an Ising model. Standard errors are below $0.008$.}
  \caption{Power curves of conditional and unconditional knockoffs with a range of $n$ for (a) a Markov chain of length $p=1000$ and (b) an Ising model of size $32\times 32$. Standard errors are all below $0.008$.} \label{fig:dgm-global}
  \end{figure}
% SE: mc=0.00516 , for Ising 0.007352931

\begin{comment}% [old ]
In Figure~\ref{fig:DGM}, the $\bx_i\in\R^{1000}$ are i.i.d. draws from 10 independent $10\times 10$ Ising models, with each Ising model given by:  %\ljmargin{}{For consistency, please make the curves in 4b dotted, since they represent conditional knockoffs.}
\begin{equation}\label{eq:model-ising}
	\pr{ Z =\bz} \propto \; \exp \left(  \sum_{ (s,t)\in E} \theta_{s,t} z_s z_t +\sum_{s\in V} h_s z_s \right),\qquad \bz\in \{-1,+1\}^{V},
\end{equation}
where the vertex set $V = [10]\times[10]$ and the edge set $E$ is all the pairs $(s,t)$ such that $\|s- t \|_{1}=1$. We take $\theta_{s,t}=0.2$ and the $h_s$ are i.i.d. Unif$([0,1])$ but held fixed across all replications. Sampling from Ising models is notoriously difficult, which is why we chose to sample from 10 smaller Ising models.\footnote{We used coupling from the past in R package \pkg{IsingSampler} to sample $Z_1,\dots , Z_{10}$ independently from \eqref{eq:model-ising} and set $X=(Z_1,\dots , Z_{10})$.} Model~\eqref{eq:model-ising} has $10\times (100 + 180)=2,800$ parameters, again far larger than any of the sample sizes simulated, yet conditional knockoffs is still able to be powerful with sufficient signal and $n$.
\end{comment}

In Figure~\ref{fig:DGM}, the $\bx_i\in\R^{32\times 32}$ are i.i.d. draws from an Ising model\footnote{We use the \emph{coupling from the past algorithm} \citep{JP-DW:1996} to sample exactly from this distribution.} given by:  
\begin{equation}\label{eq:model-ising}
	\pr{ X =\bx} \propto \; \exp \left(  \sum_{ (s,t)\in E} \theta_{s,t} x_s x_t +\sum_{s\in V} h_s x_s \right),\qquad \bx\in \{-1,+1\}^{V},
\end{equation}
where the vertex set $V = [32]\times[32]$ and the edge set $E$ is all the pairs $(s,t)$ such that $\|s- t \|_{1}=1$. We take $\theta_{s,t}=0.2$ and $h_s=0$.  Model~\eqref{eq:model-ising} has $2\times 32\times 31 + 32^2=3008$ parameters, again far larger than any of the sample sizes simulated, yet conditional knockoffs are still nearly as powerful as their unconditional counterparts.\footnote{We use the default subgraph width $w=5$ in \citet{SB-EC-LJ-WW:2019} for generating unconditional knockoffs.} The conditional knockoffs are generated by Algorithm~\ref{alg:datasplitting-discrete} with two-fold data-splitting ($m=2$\rev{, vertices are colored by the parity of the sum of their coordinates}) and no graph-expanding. Although it is possible to use graph-expanding, the power improvement is negligible because the sample size is quite small relative to the size of the neighborhoods in the expanded graph, resulting in the second round of knockoffs being nearly identical to their original counterparts.

%
%
%\begin{figure}\centering
%  \subfloat[FDR]{     \includegraphics[width=0.5\linewidth]{figures/DGM/FDR agains A} }
%    ~ %\hfill
%  \subfloat[Power]{     \includegraphics[width=0.5\linewidth]{figures/DGM/Power-agains-A} }
%  \caption{$X_i$'s are i.i.d. samples from 10 independent 2D spatial-lattice Ising model of size $10\times 10$. $n$ varies. The regression is the same as Figure~\ref{fig:MC}.}\label{fig:DGM}
% \end{figure}

%\section{Numerical Results}\label{sec:sims}
%compare conditional with unconditional in terms of power [maybe do this in the individual sections]

\section{Discussion}\label{sec:disc}
This paper introduced a way to use knockoffs to perform variable selection with exact FDR control under much weaker assumptions than made in \citet{EC-ea:2018}, while retaining nearly as high power in simulations. In fact, our method controls the FDR under arguably weaker assumptions than \emph{any} existing method (see Section~\ref{sec:contribution}).
%This paper introduced a way to use knockoffs to perform powerful variable selection with exact FDR control under much weaker assumptions than made in \citet{EC-ea:2018}, and arguably weaker assumptions than \emph{any} existing method (see Section~\ref{sec:contribution}). 
The key idea is simple, to generate knockoffs conditional on a sufficient statistic, but finding and proving valid algorithms for doing so required surprisingly sophisticated tools. One particularly appealing property of conditional knockoffs is how it directly leverages unlabeled data for improved power. We conclude with a number of open research questions raised by this paper:

\vspace{0.3cm}\noindent {\bf Algorithmic:} Perhaps the most obvious question is how to construct conditional knockoffs for models not addressed in this paper. Even for the models in this paper, what is the best way to choose the tuning parameters (e.g., $\bss$ in Algorithm~\ref{alg:ldg}, or the blocks $B_i$ in Algorithms~\ref{alg:datasplitting} and \ref{alg:datasplitting-discrete})?

\vspace{0.3cm}\noindent {\bf Robustness:} 
%Conditional knockoffs are perfectly robust to all distributions within the model for which the statistic they condition on is sufficient, in the sense that the FDR control is exact for all such distributions. 
Can techniques like those in \citet{BR-CE-SR:2018} be used to quantify the robustness of conditional knockoffs to model misspecification? 
\rev{Empirical evidence for such robustness is provided in Appendix~\ref{app:vary}.}
Also, it is worth pointing out that there are models for which no `small' sufficient statistic exists, i.e., every sufficient statistic $T(\bX)$ has the property that $\bX_j\mid\bX_{\noj},T(\bX)$ is a point mass at $\bX_j$, which forces the conditional knockoffs $\bXk_j$ to be trivial. In such models where the proposal of this paper can only produce trivial knockoffs, could postulating a distribution and generating knockoffs conditional on \emph{some} (not-sufficient) statistic still improve robustness to the parameter values in the model, relative to generating knockoffs for the same distribution but unconditionally? See \citet{BT-WY-BR-SR:2018} for a positive example for the related conditional randomization test.

\vspace{0.3cm}\noindent {\bf Power:} In this paper we always used unconditional knockoffs as a power benchmark for conditional knockoffs, as it seems intuitive that conditioning on less should result in higher power. Can this be formalized, and/or can the cost of conditioning in terms of power be quantified? Combining this with the previous paragraph, we expect there to be a \emph{power--robustness tradeoff} that can be navigated by conditioning on more or less when generating knockoffs.

\vspace{0.3cm}\noindent {\bf Conditioning:} There are reasons other than robustness that one might wish to generate knockoffs conditional on a statistic. For instance, if a model for $\bX$ needs to be checked by observing a statistic of $\bX$, generating knockoffs conditional on that statistic would guarantee a form of post-selection inference after model selection. Or when data contains variables that confound the variables of interest, it may be desirable to generate knockoffs conditional on those confounders (e.g., by Algorithm~\ref{alg:ldg-block}) in order to control for them. Also, can the conditioning tools and ideas in this paper be used to relax the assumptions of the conditional randomization test, generalizing \citet{RP:1984}?

\subsection*{Acknowledgments}
D. H. would like to thank Yu Zhao for advice on topological measure theory.
%}{please add a few more words, e.g., `for advice on geometric measure theory'}. 
L. J. would like to thank Emmanuel Cand\`es, Rina Barber, Natesh Pillai, Pierre Jacob, and Joe Blitzstein for helpful discussions regarding this project. The authors also thank the editors and the three referees for their constructive comments and suggestions.

%\dhmargin{-}{remember to fill in Wenshuo's citation}
\bibliography{references}
\bibliographystyle{apalike}

\appendix
% Is a description needed?
%\rev{Appendix A proves the theoretical results in the paper. Appendix B provides the algorithmic details. Appendix C discusses the hypotheses being tested and Appendix D provides additional simulations.	}

\section{Proofs for Main Text}\label{app:main-proof}
\subsection{Integration of Unlabled Data}
\begin{proof}[Proof of Proposition \ref{prop:unsuper}]
Denote by $\bX^{(u)}$ the last $n^{(u)} = n^{*}-n$ rows of $\bX^{*}$. Since the rows of $\bX^{*}$ are independent, $\bX^{(u)} \indp (\by,\bX)$. Then by the weak union property, $\bX^{(u)} \indp\by\,|\,\bX$. In addition, the condition that $\bXk^{*} \indp \by \,|\, (\bX,\bX^{(u)})$ and the fact that $\bXk$ is a function of $\bXk^{*}$ imply $\bXk \indp \by \,|\, (\bX,\bX^{(u)})$. By the contraction property, these two together show $\bXk \indp\by\,|\,\bX$. 

Let $\phi:\R^{n^*\times 2p} \mapsto \R^{n\times 2p}$ be the mapping that keeps the first $n$ rows of a matrix. We have $ [\bX,\,\bXk]=\phi([\bX^{*},\,\bXk^{*}] )$ and  $[\bX,\,\bXk]_{\text{swap}(A)}=\phi([\bX^{*},\,\bXk^{*}]_{\text{swap}(A)} )$ for any subset $A\subseteq [p]$. The given exchangeability condition implies that \[
\phi([\bX^{*},\,\bXk^{*}]_{\text{swap}(A)}) \eqd \phi([\bX^{*},\,\bXk^{*}] )\,\Big{|}\, T(\bX^{*}),
\]
which is simply $[\bX,\,\bXk]_{\text{swap}(A)} \eqd [\bX,\,\bXk] \,\Big{|}\, T(\bX^{*})$. It then follows that \[
[\bX,\,\bXk]_{\text{swap}(A)} \eqd [\bX,\,\bXk],
\]
%\ljmargin{}{whenever we write `swap' it should be in text mode (as it is in the main text), never in italics. I fixed the other instances in this section, and I think the main body of the paper is good, but please check this throughout the appendix}
and we conclude that $\bXk$ is a  model-X knockoff matrix for $\bX$. 
%The second part follows by taking $T(\bX^{*})$ to be a constant. 
\end{proof}

\subsection{Low-Dimensional Gaussian Models}
%\ljmargin{ we equip the sets of matrices defined in the proof with the Euclidean metric in the vectorized space, stacked column-wise.}{Do we really use this so many times that it's easier to state it once here at the beginning? I find it weird when later on we define a metric space as just a set, without mentioning a metric.} \ljmargin{-}{I'd still really like to explain how we use mathsf symbols vs. regular symbols.}

Throughout the appendix, bold-faced capital letters such as $\bs{A}$ are used for any matrix (random or not) \emph{except} when we need to distinguish between a random matrix and the values it may take, in which case we use bold sans serif letters for the values. For example, we will write $\pr{\bs{A} = \bs{\mathsf{A}}}$ to denote the probability that the random matrix $\bs{A}$ takes the (nonrandom) value $\bs{\mathsf{A}}$.

\rev{This section is planned as follows. Section~\ref{app:counterexample} clarifies a difficulty in the joint uniform distribution on a manifold. Section~\ref{app:ldgprf} contains the proof of Theorem~\ref*{thm:ldg_ko}, leaving the proofs of the lemmas in Section~\ref{app:ldg-lemma}. Section~\ref{app:ldg-intutive} discusses why a seemingly simpler proof for the theorem fails and thus justifies our technical contribution.}

\subsubsection{Counterexample for Conditional Uniformity}\label{app:counterexample}
The following statement is \uline{false}: `If a random variable $A$ is uniform on its support and another random variable $B$ is such that $B\mid A$ is conditionally uniform on its support for every $A$, with normalizing constant that does not depend on $A$, then $(A,B)$ is uniform on its support.' Although this statement seems intuitively true and holds for many simple examples (especially when $A$ and $B$ are both univariate), Figure~\ref{fig:nonunif} shows a counterexample. In it, although $X$ is uniform on $(0,1)$ and $(Y,Z)\mid X$ is uniform for every $X$ on a line whose length does not depend on $X$, the joint distribution of $(X,Y,Z)$ is not uniform on its 2-dimensional support.
\begin{figure}[h]\centering
 \includegraphics[width=0.4\linewidth]{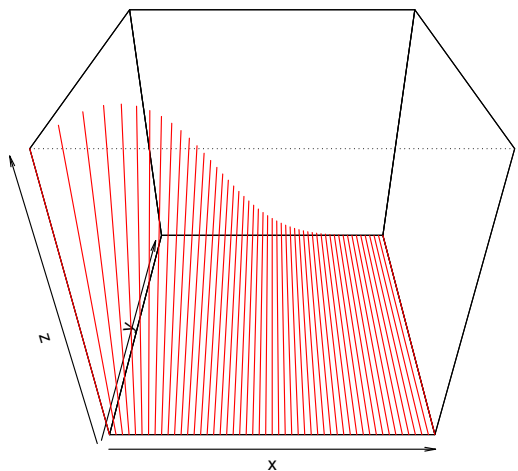}
\caption{A non-uniform distribution on a surface. $X\sim \text{Unif}(0,1)$ and $(Y,Z) \mid X \sim \text{Unif}(L_{X})$, where the line segment $L_{X}$ of length $1$ is orthogonal to the X-axis and has an angle with the Y-axis of $(1-X)^{10}\pi/2$. } \label{fig:nonunif}
 \end{figure}

%Since $\Pc{\bX}{T(\bX)}\propto\one{A_T}$ for some set $A_T$ depending on $T(\bX)$, and $\Pc{\bXk}{\bX}\propto\one{B_{\bX}}$ for some set $B_{\bX}$ depending on $\bX$, one might intuitively think that $\Pc{[\bX,\bXk]}{T(\bX)}\propto\one{C_T}$ for some set $C_T$ depending on $T(\bX)$. However this intuitive argument is fallacious, and although the conclusion happens to be correct in this section (which we prove using topological measure theory, using considerable extra information about $\bX$, $\bXk$, and $T$), so that this error would not effect the algorithm, for Gaussian graphical models this argument can lead to a subtly but absolutely incorrect algorithm, highlighting the importance of our topological measure theory approach here both for a rigorous proof and also proper algorithm development. 

\subsubsection{Proof of Theorem~\ref{thm:ldg_ko}}\label{app:ldgprf}
The proof of Theorem~\ref{thm:ldg_ko} follows three steps: Lemma~\ref{lem:ldg-invariant} states that the conditional distribution of $[\bX,\bXk]\mid T(\bX)$ is invariant on its support to multiplication by elements of the topological group of orthonormal matrices that have $\bs{1}_n$ as a fixed point, Lemma~\ref{lem:ldg-swap-invariant} states that the conditional distribution remains invariant (on the same support) after swapping $\bX_j$ and $\bXk_j$, and Lemma~\ref{lem:ldg-unique} states that the invariant measure on the support of $[\bX,\bXk]\mid T(\bX)$ is unique. These three steps combined show that the distributions before and after swapping are the same, and hence $\bXk$ is a valid conditional knockoff matrix for $\bX$.

To streamline notation, we redefine $\hbS := (\bX-\bs{1}_{n}\hbmu\tp)^{\top}(\bX-\bs{1}_n\hbmu\tp)$ as $n$ times the sample covariance matrix (it was defined as just the sample covariance matrix in the main text), and redefine $\bss$ such that $\bs{0}_{p\times p}\prec\diag{\bss}\prec 2 \hbS$ accordingly. With this new notation, $\bL$ is the Cholesky decomposition such that $\bL\tp\bL= 2\diag{\bss} - \diag{\bss}\hbS^{-1}\diag{\bss} $.
Let $\bC\in \R^{(n-1)\times n}$ be a matrix with orthonormal rows that are also orthogonal to $\bs{1}_{n}$. Then $\bC\tp \bC= \bs{I}_n-\bs{1}_n\bs{1}_n^{\top}/n$ is the centering matrix, $\bC\tp \bC\bX = \bX-\bs{1}_n\hbmu\tp$ and $ (\bC\bX)^{\top}\bC\bX=\hbS $; note $\bC$ is just a constant, nonrandom matrix. The statistic being conditioned on this this proof is $T(\bX) = ( \bX\tp \bs{1}_n / n , (\bC\bX)^{\top}\bC\bX) = (\hbmu,\hbS)$.
%{\color{red} Define $T(\bsfX)=( \bsfX\tp \bs{1}_n, (\bC\bsfX)^{\top}\bC\bsfX)$ for any $\bsfX\in \R^{n\times p}$. }
For any positive integers $s$ and $t$ such that $s\geq t$, 
denote by $\mathcal{O}_{s}$ the group of $s \times s$  orthogonal matrices (also known as the orthogonal group) and denote by $\mathcal{F}_{s,t}$ the set of $s \times t$ real matrices whose columns form an orthonormal set in $\R^s$ (also known as the Stiefel manifold). %, which is known as the Stiefel manifold (see, e.g., \citet[Chapter 1.5]{dubrovin2012modern}) and is compact. 

We will use techniques from topological measure theory to prove Theorem~\ref{thm:ldg_ko}, specifically on invariant measures (see e.g. \citet[Chapter 13]{schneider2008stochastic} and \citep[Chapter 44]{fremlin2000measure}). For readers unfamiliar with the field, the following is a short list of definitions we will use:
\begin{itemize}
    \item A group $\mathcal{G}$ is a \emph{topological group} if it has a topology such that the functions of multiplication and inversion, i.e., $(x,y)\mapsto xy$ and $x\mapsto x^{-1}$, are continuous.\footnote{A function between two topological spaces is \emph{continuous} if the inverse image of any open set is an open set. }
\item An \emph{operation} of a group $\mathcal{G}$ on a nonempty set $\setE$ is a function $\psi: \mathcal{G} \times \setE \mapsto \setE$ satisfying $\psi(g, \psi(g', x))=\psi(g g', x)$ and $\psi(e,x)=x$. The operation $\psi(g,x)$ is also written as $gx$ when there is no risk of confusion. For any subset $\mathcal{B}\subseteq \setE$ and $g\in \mathcal{G}$, denote by $g\mathcal{B}$ the image under the operation with $g$, i.e.,  $g\mathcal{B}=\{\psi(g,x): x\in \mathcal{B}\}$.
\item An operation $\psi$ is \emph{transitive} if for any $x,y \in \setE$ there exists $g\in \mathcal{G}$ such that $\psi(g,x)=y$. 
\item Suppose $\setE$ is a topological space and $\mathcal{G}$ is a topological group, the operation $\psi$ is  \emph{continuous} if $\psi$, as a function of two arguments, is continuous.
\item Suppose $\setE$ is a locally compact metric space. A Borel measure $\rho$ on $\setE$ is called \emph{$\mathcal{G}$-invariant} if for any $g\in \mathcal{G}$ and Borel subset $\mathcal{B}\subseteq \setE$, it holds that $\rho(\mathcal{B})=\rho( g\mathcal{B})$.
\end{itemize}

We can now define the elements of the proof. 
Suppose $\bs{S}\in\R^{p\times p}$ is a positive definite matrix and $\bs{m}\in \R^p$. 
Define a metric space
\begin{equation}
	\begin{aligned}
		\mathcal{M}_{} = &\left\{ [\bsfX,\bsfXk]\in \R^{n\times (2p)}:\bsfX^{\top}\bs{1}_n/n = \bs{m},\;(\bC\bsfX)^{\top}\bC\bsfX = \bs{S},\right.\\
&	\hspace{8em}	\left. \; \bsfXk^{\top}\bs{1}_n/n = \bs{m},\;(\bC\bsfXk)^{\top}\bC\bsfXk = \bs{S},\;
		(\bC\bsfXk)^{\top}\bC\bsfX = \bs{S}-\text{diag}\{\bss\}\right\},
	\end{aligned}
\end{equation}
equipped with the Euclidean metric in the vectorized space, stacked column-wise.
By Equation~\eqref{eq:ldg_ko_short}, it is straightforward to check that if $(\hbmu, \hbS)=(\bs{m}, \bs{S})$ then $[\bX,\bXk]\in \mathcal{M}$. 

Define $\mathcal{G}=\{\bG\in \mathcal{O}_{n} : \bG \bs{1}_{n}=\bs{1}_{n}\}$. It is easy to check that $\mathcal{G}$ is a group whose identity element is $\bs{I}_n$. $\mathcal{G}$ is also a topological group with the induced metric of $\R^{n\times n}$ because matrix inversion and multiplication are continuous. 
%$\bG \in \mathcal{G}$ then $\bG^{-1}=\bG\tp$, i.e., the inversion is a linear map. In addition, the multiplication is a bi-linear map. Hence they are both continuous.  

Define a mapping $\psi: \mathcal{G}\times \mathcal{M}_{} \mapsto \mathcal{M}_{}$ by $\psi(\bG, [\bsfX,\bsfXk])=[\bG\bsfX,\bG\bsfXk]$. Note that 
\begin{equation}\label{eq:ldg-operation-1}
\bG\tp \bC\tp \bC \bG=\bG\tp\bG - \bG\tp (\bs{1}_{n}\bs{1}_{n}\tp/n)\bG = \bs{I}_{n}-(\bs{1}_{n}\bs{1}_{n}\tp/n)=\bC\tp \bC,
\end{equation}
thus for any $[\bsfX,\bsfXk]\in \mathcal{M}$, we have $\psi(\bG, [\bsfX,\bsfXk])\in \mathcal{M}_{}$. It is also seen that $\psi(\bG_{1} \bG_{2}, [\bsfX,\bsfXk])=\psi(\bG_{1}, \psi( \bG_{2},[\bsfX,\bsfXk]))$ and $\psi(\bs{I}_n,[\bsfX,\bsfXk]) = [\bsfX,\bsfXk]$, so $\psi$ is an operation of $\mathcal{G}$ on $\mathcal{M}_{}$. By the continuity of matrix multiplication, $\psi$ is a continuous operation. %Because the operation is a bi-linear map, it is continuous. 

We can now state the three lemmas which comprise the proof.
\begin{lemma}[Invariance]\label{lem:ldg-invariant}
The probability measure of $[\bX,\bXk]$ conditional on $\hbmu=\bs{m}$ and $\hbS=\bs{S}$
%, i.e., $\nu(\mathcal{B})\, :=\, \Pcr{[\bX,\bXk]\in \mathcal{B}}{ \hbmu=\bs{m}, \hbS=\bs{S} }$ for any Borel subset $\mathcal{B}\subseteq \mathcal{M}$, 
is $\mathcal{G}$-invariant on $\mathcal{M}_{}$.
%joint distribution of $[\bX,\bXk]$ conditional on $\hbmu=\bs{m}$ and $\hbS=\bs{S}$ is $\mathcal{G}$-invariant on  $\mathcal{M}_{}$. 
\end{lemma}

\begin{lemma}[Invariance after swapping]\label{lem:ldg-swap-invariant}
The probability measure of $[\bX,\bXk]_{\text{\emph{swap}}(j)}$ conditional on $\hbmu=\bs{m}$ and $\hbS=\bs{S}$ is $\mathcal{G}$-invariant on $\mathcal{M}_{}$.
\end{lemma}

\begin{lemma}[Uniqueness]\label{lem:ldg-unique}
%Suppose $\bs{S}\in\R^{p\times p}$ is positive definite and $\bs{m}\in \R^p$. %and has $p$ distinct eigenvalues.
The $\mathcal{G}$-invariant probability measure on  $\mathcal{M}_{}$ is unique. 
\end{lemma}

Combining Lemmas~\ref{lem:ldg-invariant}, \ref{lem:ldg-swap-invariant} and \ref{lem:ldg-unique} together, we conclude that given $\hbmu=\bs{m}$ and $\hbS=\bs{S}$, swapping $\bX_{j}$ and $\bXk_{j}$ leaves the distribution of $[\bX, \bXk]$ unchanged.  
Since if swapping one column does not change the distribution, then by induction swapping any set of columns will not change the distribution and this completes the proof. 

\begin{remark}
 Although not shown here, one can define the uniform distribution on $\mathcal{M}$ via the Hausdorff measure and show that it is also $\mathcal{G}$-invariant. Therefore, by the uniqueness of the invariant measure, $[\bX,\bXk]$ is distributed uniformly on $\mathcal{M}$.
 \end{remark}

%Given Lemma~\ref{lem:ldg-invariant}, the conditional distribution of $[\bX,\bXk]$ is $\mathcal{G}$-invariant on $\mathcal{M}$, by which a measure $\nu$ is defined, i.e., $\nu(\mathcal{B})\, :=\, \Pcr{[\bX,\bXk]\in \mathcal{B}}{ \hbmu=\bs{m}, \hbS=\bs{S} }$ for any Borel subset $\mathcal{B}\subseteq \mathcal{M}$. 
%\ljmargin{}{what do you mean `defined'? This conditional distribution was already defined, no? Did it need to be $\mathcal{G}$-invariant? Why do we need all this formalism to define this measure?}

%It remains to check that swapping one column (w.l.o.g. let it be column 1) of $\bX$ with that of $\bXk$ leaves the conditional distribution $\mathcal{G}$-invariant and thus unchanged.  (

\subsubsection{Proofs of Lemmas}\label{app:ldg-lemma}
Before proving the lemmas, we introduce some notation and properties for Gaussian matrices. Let $r$, $s$, and $t$ be any positive integers. % such that $r\ge s\ge t$. %\ljmargin{}{why are these defined as ordered first? you go on to use these letters in the following sentences and I don't think this ordering is part of those definitions...}. 
For any matrix $\bs{A}\in \R^{s\times t}$, denote by $\ovec(\bs{A})$ the vector that concatenates its columns, i.e., $ ( \bs{A}_{1}\tp, \dots, \bs{A}_{t}\tp )\tp$.
Denote by $\otimes$ the Kronecker product. 
A $s\times t$ random matrix $\bs{A}$ is a Gaussian random matrix $\bs{A}\sim \N_{s,t}( \bs{M}, \bs{\Upsilon} \otimes \bs{\Sigma})$ if  $\ovec(\bs{A}\tp)\sim \N( \ovec(\bs{M}\tp) , \bs{\Upsilon} \otimes \bs{\Sigma})$ for some $\bs{M}\in \R^{s\times t}$ and matrices $\bs{\Upsilon}\succeq \bs{0}_{s\times s}$ and $\bs{\Sigma}\succeq \bs{0}_{t\times t}$.

If $\bs{A}\sim \N_{s,t}( \bs{M}, \bs{\Upsilon} \otimes \bs{\Sigma})$, then for any matrix $\bs{\Gamma}\in \R^{r\times s}$, $\ovec ( (\bs{\Gamma}\bs{A})\tp)=(\bs{\Gamma}\otimes \bs{I}_{t})\ovec(\bs{A}\tp)$ is still multivariate Gaussian and 
\begin{equation*}
\bs{\Gamma}\bs{A} \sim 	 \N_{r,t}( \bs{\Gamma}\bs{M}, (\bs{\Gamma}\bs{\Upsilon}\bs{\Gamma}\tp) \otimes \bs{\Sigma}) ,
\end{equation*}
because $(\bs{\Gamma}\otimes \bs{I}_{t})(\bs{\Upsilon} \otimes \bs{\Sigma})(\bs{\Gamma}\otimes \bs{I}_{t})\tp=(\bs{\Gamma}\bs{\Upsilon}\bs{\Gamma}\tp) \otimes (\bs{I}_{t}\bs{\Sigma}\bs{I}_{t})$ by the mixed-product property and transpose of Kronecker product. When the rows of $\bs{A}$ are i.i.d. samples from a multivariate Gaussian, $\bs{\Upsilon}=\bs{I}_{s}$ and $\bs{M}=\bs{1}_{s}\bmu\tp$ for some $\bmu\in \R^{t}$.
%\ljmargin{what is non-bold $\mu$?}. 
If further, $\bs{\Gamma}\bs{\Gamma}\tp=\bs{I}_{r}$, then 
\begin{equation}\label{eq:gaussian-matrix}
\bs{\Gamma}\bs{A} \sim 	 \N_{r,t}( \bs{\Gamma}\bs{1}_{s}\bmu\tp, \bs{I}_{r}\otimes \bs{\Sigma}).
\end{equation}
We write the Gram--Schmidt orthonormalization as a function $\Psi	(\cdot)$. We will make use of the property that for any $\bs{\Gamma}_{0}\in \mathcal{O}_{s}$ and any matrix $\bU_{0}\in \R^{s \times t}$ (for $s\ge t$), it holds that 
\begin{equation}\label{eq:ldg-gs-0}
\Psi(\bs{\Gamma}_{0}\bU_{0})=\bs{\Gamma}_{0} \Psi( \bU_{0}).
\end{equation}
See, e.g., \citet[Proposition 7.2]{eaton1983multivariate}.

\begin{proof}[Proof of Lemma~\ref{lem:ldg-invariant}]
Define $\nu(\mathcal{B})\, :=\, \Pcr{[\bX,\bXk]\in \mathcal{B}}{ \hbmu=\bs{m}, \hbS=\bs{S} }$ for any Borel subset $\mathcal{B}\subseteq \mathcal{M}$. 
For fixed $\bG\in \mathcal{G}$, we need to show the group operation given $\bG$, i.e., $g_{\bG}=\psi(\bG, \cdot)$, leaves $\nu$ unchanged. Define $\bX^{\prime}=\bG\bX$ and $\bXk^{\prime}=\bG\bXk$. We will show 
\[
[\bX,\bXk] \eqd [\bX^{\prime},\bXk^{\prime}] \mid T(\bX).
\]
%$[\bX,\bXk] \eqd [\bG\bX,\bG\bXk] \mid T(\bX)$. 

By Equation~\eqref{eq:ldg-operation-1}
%\ljmargin{}{you were using ref instead of eqref, and you do this many more times, please check} 
and $\bG\bs{1}_{n}=\bs{1}_{n}$, we have $T(\bG\bsfX)=T(\bsfX)$ 
%\ljmargin{}{we never defined $T$ in Section 3.1}
%\dhmargin{}{Let's define it at the beginning of the proof. } 
for any $\bsfX\in \R^{n\times p}$. 
%Next we show	unconditionally $[\bX,\bXk] \eqd [\bG\bX,\bG\bXk]$. 
Applying the property in Equation~\eqref{eq:gaussian-matrix}, we have 
\[
\bG\bX\sim \N_{n,p}( \bs{1}_{n}\bmu\tp, \bs{I}_{n}\otimes\bS), 
\]
where we have used $\bG\bs{1}_{n}=\bs{1}_{n}$ and $\bG\in \mathcal{O}_{n}$. Thus $\bX^{\prime}\eqd\bX$. By Equation~\eqref{eq:ldg-gs-0} and the definition of $[\bQ,\bU]$ in Algorithm~\ref{alg:ldg},  
\begin{align}\label{eq:ldg-gu}
\Psi([\bs{1}_{n},\bX^{\prime}, \bG \bW])=\bG\Psi(\bG\tp[\bs{1}_{n},\bX^{\prime}, \bG \bW])=\bG\Psi([\bs{1}_{n},\bX, \bW])=[\bG\bQ, \bG\bU]. 
\end{align}
Let $\bU^{\prime}=\bG\bU$. Since $\bW$ is independent of $\bX$ and $\bG\bW$ has the same distribution as $\bW$, we have $(\bX, \bW)\eqd (\bX^{\prime}, \bG\bW)$.  This together with Equation~\eqref{eq:ldg-gu} implies $(\bX, \bU)\eqd (\bX^{\prime}, \bU^{\prime})$. %\ljmargin{}{we concatenate matrices with $[]$, not $()$. Please check throughout carefully, as this occurs in a number of places}
Hence
\[
\Pcr{ \bX, \bU }{T(\bX)}=\Pcr{ \bX^{\prime}, \bU^{\prime}}{T(\bX^{\prime})}
\]
and since $T(\bX^{\prime})=T(\bX)$, we conclude
%\ljmargin{}{can you elaborate on this logic?} 
\[
(\bX, \bU) \eqd (\bX^{\prime}, \bU^{\prime}) \mid T(\bX). 
\]

Now recall we are conditioning on $T(\bX)=(\bs{m},\bs{S})$, and thus also $\bss$ and $\bL$. By Equation~\eqref{eq:ldg_ko_short} and the definition of $\bXk^{\prime}$,
\begin{align}
\bXk^{\prime} &=\; \bG\left(\bs{1}_n\hbmu\tp + (\bX-\bs{1}_n\hbmu\tp)(\bs{I}_p - \hbS^{-1}\diag{\bss}) + \bU\bL\right) \\
&=\; \bs{1}_n\hbmu\tp + (\bX^{\prime}-\bs{1}_n\hbmu\tp)(\bs{I}_p - \hbS^{-1}\diag{\bss}) + \bU^{\prime}\bL, 	
\end{align}
which would be the knockoff generated by Algorithm~\ref{alg:ldg} if $\bX^{\prime}$ was observed. As a consequence, 
\[
(\bX, \bXk) \eqd (\bX^{\prime}, \bXk^{\prime})\mid T(\bX).
\]
% Equation~\eqref{eq:ldg-invariant-eqd} then implies 
% 	\[
% 	\Pcr{[\bX,\bXk]\in \mathcal{B}}{ T(\bX) = (\bs{m}, \bs{S}) }=\Pcr{[\bX^{\prime},\bXk^{\prime}]\in \mathcal{B}}{ T(\bX) = (\bs{m}, \bs{S})  }.
% 	\]

This shows that for any Borel subset $\mathcal{B}\subseteq \mathcal{M}$, $\nu(\mathcal{B})=\;\nu(g_{\bG^{-1}}(\mathcal{B})) $. We conclude that for any $\bG\in \mathcal{G}$ and any Borel subset $\mathcal{B}\subseteq \mathcal{M}$
\[
\nu(g_{\bG}(\mathcal{B}))    =\; \nu(\mathcal{B}),
\]
that is, the conditional probability measure of $[\bX,\bXk]$ given $T(\bX)=(\bs{m},\bs{S})$ is $\mathcal{G}$-invariant. 
 %\ljmargin{}{we don't use $\mu$ in the lemma} 
\end{proof}

\begin{proof}[Proof of Lemma~\ref{lem:ldg-swap-invariant}]
Without loss of generality, we take $j=1$. Define a mapping $\phi: \R^{n\times(2p)}\mapsto \R^{n\times(2p)}$ by $\phi([\bsfX,\bsfXk])= [\bsfXsw,\bsfXksw]$, i.e., replacing $\bsfX$ and $\bsfXk$ with $\bsfXsw$ and $\bsfXksw$, respectively. It is easy to see that $\phi$ is isometric and $\phi^{-1}=\phi$. Furthermore, we will prove that $\phi$ is a bijective mapping of $\mathcal{M}$ to itself (Lemma~\ref{lem:M-bijective}). 
The conditional distribution of $\phi([\bX,\bXk])$ is the measure $\nu_{\phi}$ on $\mathcal{M}$ such that $\nu_{\phi}(\mathcal{B})=\nu(\phi^{-1}(\mathcal{B}))$, for any Borel subset $\mathcal{B}\subseteq \mathcal{M}$. We will show that $\nu_{\phi}$ is $\mathcal{G}$-invariant on $\mathcal{M}$ (Lemma~\ref{lem:M-invariant}). %Because $\hbS$ has different eigenvalues almost surely (see, e.g., \citet[Theorem 3.2.18]{muirhead2009aspects}), 
%By Lemma~\ref{lem:ldg-unique}, the $\mathcal{G}$-invariant probability measure on $\mathcal{M}$ is unique, and we conclude that  $\nu_{\phi}=\nu$, proving the theorem. 

%It remains to establish the following lemmas. 
% \ljmargin{\textbf{Claim 1}}{Why are these `claims' and not lemmas? I don't think `claim' is standard mathematical parlance, and when it is, I believe it usually refers to things that are NOT proved.}: 
\begin{lemma}\label{lem:M-bijective}
$\phi$ is a bijective mapping of $\mathcal{M}$ to itself. %\ljmargin{1:1}{do you mean bijective? The standard use of 1:1 just means injective. Same question about Prop A.6} 
\end{lemma}
\begin{proof} %Since $\phi$ is bijective and $\phi(\phi(x))=x$ for any $x\in \R^{n\times (2p)}$, it suffices to prove $\phi(\mathcal{M})\subseteq \mathcal{M}$. 
$\phi$ is easily seen to be injective, and to show surjectivity, we will first show $\phi(\mathcal{M})\subseteq\mathcal{M}$. Combining this with $\phi^{-1}=\phi$ gives $\mathcal{M}\subseteq\phi^{-1}(\mathcal{M})=\phi(\mathcal{M})$, and thus $\phi(\mathcal{M})=\mathcal{M}$ so $\phi$ is surjective from $\mathcal{M}$ to $\mathcal{M}$. We now complete the proof by showing something even stronger than $\phi(\mathcal{M})\subseteq\mathcal{M}$, namely the equivalence $\phi([\bsfX,\bsfXk])\in \mathcal{M} \iff [\bsfX,\bsfXk]\in\mathcal{M}$. 

Translating this equivalence to an equality of indicator functions, we need to show that 
\begin{align*}
&\;\one{\bsfX^{\top}\bs{1}_n/n = \bs{m}}\one{(\bC\bsfX)^{\top}\bC\bsfX = \bs{S}}\one{\bsfXk^{\top}\bs{1}_n/n = \bs{m}}\one{(\bC\bsfXk)^{\top}\bC\bsfXk = \bs{S}}\one{(\bC\bsfXk)^{\top}\bC\bsfX = \bs{S}-\text{diag}\{\bss\}} \\
= &\;\one{\bsfXsw^{\top}\bs{1}_n/n = \bs{m}}\one{(\bC\bsfXsw)^{\top}\bC\bsfXsw = \bs{S}} \\
&\;\cdot\one{\bsfXksw^{\top}\bs{1}_n/n = \bs{m}}\one{(\bC\bsfXksw)^{\top}\bC\bsfXksw = \bs{S}}\one{(\bC\bsfXksw)^{\top}\bC\bsfXsw = \bs{S}-\text{diag}\{\bss\}},
\end{align*}
where the righthand side is the same as the lefthand side but with $\bsfX$ and $\bsfXk$ replaced with $\bsfXsw$ and $\bsfXksw$, respectively. First note that for the first and third indicator functions on the lefthand side,
\[\one{\bsfX^{\top}\bs{1}_n/n = \bs{m}}\one{\bsfXk^{\top}\bs{1}_n/n = \bs{m}} = \left(\prod_{j=1}^p\one{\bsfX_j^{\top}\bs{1}_n/n=m_j}\right)\left(\prod_{j=1}^p\one{\bsfXk_j^{\top}\bs{1}_n/n=m_j}\right)\]
and exchanging the first term in each product and compressing the products each back into single indicator functions gives $\one{\bsfXsw^{\top}\bs{1}_n/n = \bs{m}}\one{\bsfXksw^{\top}\bs{1}_n/n = \bs{m}}$, so it just remains to show that 
\begin{align*}
&\;\one{(\bC\bsfX)^{\top}\bC\bsfX = \bs{S}}\one{(\bC\bsfXk)^{\top}\bC\bsfXk = \bs{S}}\one{(\bC\bsfXk)^{\top}\bC\bsfX = \bs{S}-\text{diag}\{\bss\}} \\
= &\;\one{(\bC\bsfXsw)^{\top}\bC\bsfXsw = \bs{S}}\one{(\bC\bsfXksw)^{\top}\bC\bsfXksw = \bs{S}}\one{(\bC\bsfXksw)^{\top}\bC\bsfXsw = \bs{S}-\text{diag}\{\bss\}}.
\end{align*}
Again it is useful to rewrite the three indicator functions as products:
\begin{align*}
&\;\one{(\bC\bsfX)^{\top}\bC\bsfX = \bs{S}}\one{(\bC\bsfXk)^{\top}\bC\bsfXk = \bs{S}}\one{(\bC\bsfXk)^{\top}\bC\bsfX = \bs{S}-\text{diag}\{\bss\}} \\
= &\;\left(\prod_{1\leq j\le k\leq p}^p\one{(\bC\bsfX_j)^{\top}\bC\bsfX_k = \bs{S}_{j,k}}\right)\left(\prod_{1\leq j\le k\leq p}^p\one{(\bC\bsfXk_j)^{\top}\bC\bsfXk_k = \bs{S}_{j,k}}\right)\left(\prod_{j,k=1}^p\one{(\bC\bsfXk_j)^{\top}\bC\bsfX_k = \bs{S}_{j,k}-1_{\{j=k\}}s_j}\right).
\end{align*}
Now if we exchange the terms in the first product with $k>j=1$ with the same terms in the third product, and exchange the terms in the second product with $k>j=1$ with the terms in the third product with $j>k=1$, we can compress the products each back into single indicator functions again to get $\one{(\bC\bsfXsw)^{\top}\bC\bsfXsw = \bs{S}}\one{(\bC\bsfXksw)^{\top}\bC\bsfXksw = \bs{S}}\one{(\bC\bsfXksw)^{\top}\bC\bsfXsw = \bs{S}-\text{diag}\{\bss\}}$.
We conclude that $ [\bsfX,\bsfXk]\in\mathcal{M} \iff \phi([\bsfX,\bsfXk])\in \mathcal{M}$. 
%\ljmargin{}{are you done with the proof? What about the 1:1 part?}
\end{proof}

\begin{lemma}\label{lem:M-invariant}
$\nu_{\phi}$ is $\mathcal{G}$-invariant on $\mathcal{M}$.
\end{lemma}
\begin{proof}
For any $\bG\in \mathcal{G}$, the group operation $g_{\bG}=\psi(\bG, \cdot)$ is exchangeable with $\phi$ because
\begin{align*}
\psi(\bG,\phi([\bsfX,\bsfXk])) &=\; [\bG\bsfXsw,\bG\bsfXksw] \\
&=\; [[\bG\bsfXk_1,\,\bG\bsfX_{\text{-}1}], [\bG\bsfX_1,\, \bG\bsfXk_{\text{-}1}] ] \\
&=\; \phi(  [\bG\bsfX,\bG\bsfXk])). 	\\
&=\; \phi( \psi(\bG, [\bsfX,\bsfXk])). 	
\end{align*}
Thus for any Borel subset $\mathcal{B}\subseteq \mathcal{M}$ ,  
\begin{align*}
    \nu_{\phi}(g_{\bG}\mathcal{B})&=\,\nu(\phi(g_{\bG}\mathcal{B}))\\
    &=\; \nu(g_{\bG}(\phi(\mathcal{B})))\\
    &=\; \nu(\phi(\mathcal{B}))\\
    &=\; \nu_{\phi}(\mathcal{B}),
\end{align*} 
where the third equality follows from Lemma~\ref{lem:ldg-invariant}. Thus we conclude that $\nu_{\phi}$ is $\mathcal{G}$-invariant. 
\end{proof}

\end{proof}
\begin{proof}[Proof of Lemma~\ref{lem:ldg-unique}]
Before the proof, we list a few results that will be used.
\begin{enumerate}[{Fact} 1.]
    \item For an operation $\psi$ of a group $\mathcal{G}$ on a space $\setE$, if there is some $z\in \setE$ such that for any $y\in \setE$ there exists $g_y\in \mathcal{G}$ such that $\psi(g_y,z)=y$, then $\psi$ is transitive. This is because for any $x,y\in \setE$, $\psi(g_{x}^{-1},x)=\psi(g_{x}^{-1}, \psi(g_{x}, z))=\psi(g_{x}^{-1}g_{x}, z)=z$  and $\psi(g_y g_{x}^{-1},x)=\psi(g_y,\psi(g_{x}^{-1},x))=\psi(g_y,z)=y $. \label{fact:transitive}
    \item For any compact Hausdorff\footnote{A topological space is \emph{Hausdorff} if every two different points can be separated by two disjoint open sets. } topological group $\mathcal{G}$, there exists a finite Borel measure $\nu$, called a \emph{Haar measure}, such that for any $g\in \mathcal{G}$ and Borel subset $\mathcal{B}\subseteq \mathcal{G}$, $\nu(\mathcal{B})=\nu( g\mathcal{B})=\nu( \mathcal{B}g)$ \citep[441E, 442I(c)]{fremlin2000measure}. As an example, the orthogonal group $\mathcal{O}_{n}$ has a Haar measure \citep[Chapter 6.2]{eaton1983multivariate}. \label{fact:exist-haar}
\end{enumerate}
The key theorem we use is the following.
\begin{lemma}[Theorem~13.1.5 in \citet{schneider2008stochastic}]\label{lem:unique-invariant}
Suppose that the compact group $\mathcal{G}$ operates continuously and transitively on the Hausdorff space $\setE$ and that $\mathcal{G}$ and $\setE$ have countable bases. Let $\nu$ be a Haar measure on $\mathcal{G}$ with $\nu(\mathcal{G}) = 1$. Then there exists a unique $\mathcal{G}$-invariant Borel measure $\rho$ on $\setE$ with $\rho(\setE) = 1$.
\end{lemma}

Now we are ready to prove Lemma~\ref{lem:ldg-unique}.  
Note $\mathcal{G}$ and $\mathcal{M}$ are compact subspaces of  the vectorized spaces, and Fact~\ref{fact:exist-haar} ensures the existence of a Haar measure on $\mathcal{G}$. Since $\psi$ is continuous, as long as $\psi$ is transitive we can apply Lemma~\ref{lem:unique-invariant} and conclude that the $\mathcal{G}$-invariant probability measure on $\mathcal{M}_{}$ is unique. 

To show $\psi$ is transitive by Fact~\ref{fact:transitive}, we first fix a point $[\bsfX_{0},\bsfXk_{0}]$ and then show for any $[\bsfX,\bsfXk]\in \mathcal{M}_{}$, we can find $\bG\in \mathcal{G}$ such that $\psi(\bG, [\bsfX_{0},\bsfXk_{0}])=[\bsfX,\bsfXk]$. 

%\dhmargin{-}{The proof goes backward and forward with $\bsfX$ and $\bsfX_{0}$ but we cannot adjust it. To avoid confusion, I split it into 3 parts. }

\textbf{Part 1.} We begin with representing $\bsfXk$ using the Stiefel Manifold.
Define $\mathcal{M}_{1}=\{\bsfX\in \R^{n\times p}:\, \bsfX^{\top}\bs{1}_n/n = \bs{m},\;(\bC\bsfX)^{\top}\bC\bsfX = \bs{S}\}=\{\bsfX\in \R^{n\times p}:\, T(\bsfX)=(\bs{m}, \bs{S})\}$. 
For any $\bsfX\in \mathcal{M}_{1}$, define 
\[
\mathcal{M}_{\bsfX} = \left\{\bsfXk\in \R^{n\times p}: \bsfXk^{\top}\bs{1}_n/n = \bs{m},\;(\bC\bsfXk)^{\top}\bC\bsfXk = \bs{S},\;(\bC\bsfXk)^{\top}\bC\bsfX = \bs{S}-\text{diag}\{\bss\}\right\}.
\]
Let $\bZ_{\bsfX}$ be a $n\times (n-1-p)$ matrix whose columns form an orthonormal basis for the orthogonal complement of $\text{span}(\left[ \bs{1}_n,\, \bsfX \right])$. Recall that $\mathcal{F}_{n-1-p,p}$ is the set of $(n-1-p) \times p$ real matrices whose columns form an orthonormal set in $\R^{n-1-p}$.  Define $\varphi_{\bsfX}: \mathcal{F}_{n-1-p,p} \mapsto  \R^{n\times p}$ by
\[
\varphi_{\bsfX}(\bsfV) =  \bs{1}_n\bs{m}^{\top} + (\bsfX-\bs{1}_n\bs{m}^{\top})(\bs{I}_p - \bs{S}^{-1}\diag{\bss}) +\bZ_{\bsfX}\bsfV \bL.
\]
The following result tells us that for any $[\bsfX,\bsfXk]\in\mathcal{M}$, there exists a $\bsfV\in \mathcal{F}_{n-1-p,p}$ such that $\bsfXk = \varphi_{\bsfX}(\bsfV)$, and thus we are implicitly decomposing $\bU$ from Algorithm~\ref{alg:ldg} into $\bZ_{\bX}\bV$ for some random $\bV$, and we think of $\bsfV$ as a realization of this $\bV$.
\begin{lemma}\label{lem:ldg-express-xk}
$\varphi_{\bsfX}$ is a bijective mapping from $\mathcal{F}_{n-1-p,p}$ to $\mathcal{M}_{\bsfX}$.
\end{lemma}
The proof of Lemma~\ref{lem:ldg-express-xk} involves mainly linear algebra and is deferred to the end of this section. 

%Because the eigenvalues of $\bs{S}$ are different, the eigenvalue decomposition of $\bs{S}$ is unique. 
\textbf{Part 2.} We now define $[\bsfX_{0},\bsfXk_{0}]$.  Let  the eigenvalue decomposition of $\bs{S}$ be $\bG_0 \bs{D}^{2} \bG_0\tp$, where $\bs{D}$ is a $p\times p$ diagonal matrix with positive non-increasing diagonal entries and $\bG_0\in \mathcal{O}_{p}$. Define a  $(n-1)\times p$ matrix $\bsfX_{*}$ and a $(n-1)\times (n-1-p)$ matrix $\bZ_{*}$ as
\[
\bsfX_{*}=\begin{bmatrix}
\bs{D}\bG_{0}\tp \\
\bs{0}_{(n-1-p)\times p}\\	
\end{bmatrix}, \; 
\bZ_{*}=\begin{bmatrix}
\bs{0}_{p\times (n-1-p)} \\
\bs{I}_{n-1-p}\\	
\end{bmatrix}.
\]
Then $\bZ_{*}\tp \bZ_{*}=\bs{I}_{n-1-p}$, $\bsfX_{*}\tp \bsfX_{*}=\bs{S}$ and $\bsfX_{*}\tp\bZ_{*}=\bs{0}$. Next define
\begin{align*}
\bsfX_{0}&=\,   \bs{1}_{n}\bs{m}\tp+\bC\tp \bsfX_{*}, \\
\bZ_{0}&=\, \bC\tp \bZ_{*}, \qquad\qquad\quad	\bsfV_{0}=  \begin{bmatrix}
 \bs{I}_{p}\\
 \bs{0}_{(n-1-2p)\times p}	
 \end{bmatrix}, \\
\bsfXk_{0}&=\, \bs{1}_{n}\bs{m}\tp+(\bsfX_{0}-\bs{1}_{n}\bs{m}\tp)(\bs{I}_{p}- \bs{S}^{-1}\diag{\bss})+\bZ_{0}\bsfV_{0} \bL. 
\end{align*}
One can check that $[\bsfX_{0},\bsfXk_{0}] \in \mathcal{M}_{}$. 

\textbf{Part 3.} Now for any $[\bsfX,\bsfXk]\in \mathcal{M}_{}$, we will find a $\bG\in \mathcal{G}$ such that $\psi(\bG, [\bsfX_{0},\bsfXk_{0}])=[\bsfX,\bsfXk]$. 

Let $\bs{Q}_{\bsfX}=\bC\bsfX\bG_{0} \bs{D}^{-1}$, which is a $(n-1)\times p$ matrix. 
Since $(\bC\bsfX)\tp \bC\bsfX=\bs{S}$, we have $\bs{Q}_{\bsfX}\tp \bs{Q}_{\bsfX}=\bs{I}_{p}$. Thus $\bs{Q}_{\bsfX}\in \mathcal{F}_{n-1, p}$. 
% and the eigenvalue decomposition of $\bs{S}$ is unique,  the singular value decomposition of $\bC\bsfX$ must be $\bs{Q}_{\bsfX} \bs{D}\bG_{0}\tp$ for some unique $\bs{Q}_{\bsfX}\in \mathcal{F}_{(n-1)\times p}$. 
By Lemma~\ref{lem:ldg-express-xk}, 
there is some $\bsfV\in \mathcal{F}_{n-1-p,p}$ such that $\bsfXk=\bs{1}_n\bs{m}^{\top} + (\bsfX-\bs{1}_n\bs{m}^{\top})(\bs{I}_p - \bs{S}^{-1}\diag{\bss}) +\bZ_{\bsfX}\bsfV \bL$. 
Let $\bs{Q}_{\bsfXk}$ be $\bC\bZ_{\bsfX}\bsfV$. We will show $\bs{Q}_{\bsfXk}\in \mathcal{F}_{n-1,  p}$ and $\bs{Q}_{\bsfXk}\tp \bs{Q}_{\bsfX}=\bs{0}$:\\
Because $\bZ_{\bsfX}\tp \bs{1}_{n}=\bs{0}$,  it holds $\bC\tp \bC\bZ_{\bsfX}=\bZ_{\bsfX}$. Thus  
\begin{align*}
\bs{Q}_{\bsfXk}\tp\bs{Q}_{\bsfXk}
&=\; (\bC \bZ_{\bsfX}\bsfV )\tp \bC \bZ_{\bsfX}\bsfV \\
&=\; \bsfV\tp \bZ_{\bsfX}\tp \bZ_{\bsfX}\bsfV\\
&=\; \bs{I}_{p}.
\end{align*}
In addition, because $\bZ_{\bsfX}\tp\bsfX=\bs{0}$, it holds that 
\begin{align*}
\bs{Q}_{\bsfXk}\tp\bs{Q}_{\bsfX}&=\; \bsfV\tp\bZ_{\bsfX} \tp \bC\tp \bs{Q}_{\bsfX} \\
&= \; \bsfV\tp\bZ_{\bsfX} \tp \bC\tp \bC\bsfX\bG_0 \bs{D}^{-1}\\
&=\; \bsfV\tp\bZ_{\bsfX} \tp \bsfX \bG_0 \bs{D}^{-1}\\
&=\; \bs{0}.
\end{align*}
Then we can find some $\bG_{*}\in \mathcal{O}_{n-1}$ such that 
\[
(\bG_{*})_{1:(2p)}=[\bs{Q}_{\bsfX} , \bs{Q}_{\bsfXk}]=[\bC\bsfX\bG_{0} \bs{D}^{-1} , \bC\bZ_{\bsfX}\bsfV]
\]
Define $\bG=\bC\tp \bG_{*} \bC+ \bs{1}_{n}\bs{1}_{n}\tp /n $. One can check that $\bG\tp\bG=\bs{I}_{n}$ and $\bG \bs{1}_{n} =\bs{1}_{n}$, and conclude that $\bG\in \mathcal{G}$. We next show $ [\bG\bsfX_{0}, \bG\bsfXk_{0}]=[\bsfX,\bsfXk]$. 

We first check 
\begin{align*} %\label{eq:x-manifold1}
\bG\bsfX_{0}&=\; (\bC\tp \bG_{*} \bC + \bs{1}_{n}\bs{1}_{n}\tp /n ) ( \bs{1}_{n}\bs{m}\tp+\bC\tp \bsfX_{*} ) \nonumber\\
&=\;  \bs{1}_{n}\bs{m}\tp+ \bC\tp \bG_{*} \bC \bC\tp \bsfX_{*} & (\because \bC\bs{1}_{n}=\bs{0}) \nonumber \\
&=\;  \bs{1}_{n}\bs{m}\tp+ \bC\tp \bG_{*}\bsfX_{*} & (\because \bC \bC\tp=\bs{I}_{n-1})\nonumber\\
&=\;  \bs{1}_{n}\bs{m}\tp+ \bC\tp \bs{Q}_{\bsfX} \bs{D}\bG_{0}\tp & (\because \text{definitions of } \bG_{*}, \bsfX_{*} ) \\
&=\;  \bs{1}_{n}\bs{m}\tp+ \bC\tp \bC \bsfX & (\because \text{definition of } \bs{Q}_{\bsfX} ) \nonumber\\
&=\;  \bsfX. \nonumber
\end{align*}
Next, note that
\begin{align*}
	\bG\bZ_0 \bsfV_0 &=\;  (\bC\tp \bG_{*} \bC + \bs{1}_{n}\bs{1}_{n}\tp /n )\bC\tp \bZ_{*}\bsfV_0\\
	&=\; \bC\tp \bG_{*} \bC \bC\tp \bZ_{*}\bsfV_0  & (\because \bC\bs{1}_{n}=\bs{0})  \\
	&=\; \bC\tp \bG_{*}\bZ_{*}\bsfV_0 & (\because \bC \bC\tp=\bs{I}_{n-1})\\
	&=\; \bC\tp \bG_{*}\begin{bmatrix}
\bs{0}_{p\times p}	\\
\bs{I}_{p} \\
\bs{0}_{(n-1-2p)\times p}
\end{bmatrix}  & (\because \text{definitions of } \bZ_{*}, \bsfV_0 ) \\
&=\; \bC\tp \bC \bZ_{\bsfX}\bsfV & (\because \text{definition of } \bG_{*} ) \\
&=\; \bZ_{\bsfX}\bsfV, 
\end{align*}
and hence it holds that
\begin{align}\label{eq:xk-manifold2}
\bG\bsfXk_{0}&=\; \bG\left(  \bs{1}_{n}\bs{m}\tp+(\bsfX_{0}-\bs{1}_{n}\bs{m}\tp)(\bs{I}_{p}- \bs{S}^{-1}\diag{\bss})+\bZ_{0}\bsfV_{0} \bL \right) \nonumber\\
&=\;  \bs{1}_{n}\bs{m}\tp+ (\bsfX-\bs{1}_{n}\bs{m}\tp)(\bs{I}_{p}- \bs{S}^{-1}\diag{\bss})  +  \bZ_{\bsfX}\bsfV \bL \nonumber \\
&=\; \bsfXk. \nonumber
\end{align}
Hence the operation $\psi$ is transitive, and the proof is complete. 
\end{proof}

%%%%%%
\begin{proof}[Proof of Lemma~\ref{lem:ldg-express-xk}]
The proof takes four steps. 

\textbf{Step 1}: $\bL$ is invertible.

Let
%\ljmargin{}{we definitely shouldn't use $\bG$ here, after having just used it so much in the preceding text to mean something totally different!}
\[ \bs{S}_{*} = \left[\begin{array}{cc} 2\diag{\bss} & \diag{\bss} \\
				   \diag{\bss} & \bs{S} \end{array}\right]. \]

By construction of $\bss$, $2\diag{\bss}\succ\bs{0}_{p\times p}$ and
\begin{align*}
&&2\bs{S}&\succ\diag{\bss} \\
\Rightarrow&& \bs{S}-\frac{1}{2}\diag{\bss}&\succ\bs{0}_{p\times p} \\
\Rightarrow&& \bs{S}-\diag{\bss}\left(2\diag{\bss}\right)^{-1}\diag{\bss}&\succ\bs{0}_{p\times p},
\end{align*}
where the lefthand side of the last line is exactly the Schur complement of $2\diag{\bss}$ in $\bs{S}_{*}$, and therefore $\bs{S}_{*}\succ\bs{0}_{p\times p}$. But since $\bs{S}\succ\bs{0}_{p\times p}$, the fact that $\bs{S}_{*}\succ\bs{0}_{p\times p}$ implies that the Schur complement of $\bs{S}$ in $\bs{S}_{*}$ is also positive definite: 
\[2\diag{\bss}-\diag{\bss}\bs{S}^{-1}\diag{\bss} = \bL^{\top}\bL\succ\bs{0}_{p\times p}, \]
and therefore $\bL$ is invertible. %, which we will use repeatedly in the remainder of this proof.

\textbf{Step 2}: $\varphi_{\bsfX}(\mathcal{F}_{n-1-p,p})\subseteq \mathcal{M}_{\bsfX}$.

Let $\bsfXk=\varphi_{\bsfX}(\bsfV)$ for some $\bsfV\in \mathcal{F}_{n-1-p,p}$. First we show $\bsfXk^{\top}\bs{1}_n/n = \bs{m}$:
\begin{flalign*}
\bsfXk^{\top}\bs{1}_n/n &= \left(\bs{1}_n\bs{m}^{\top} + (\bsfX-\bs{1}_n\bs{m}^{\top})(\bs{I}_p - \bs{S}^{-1}\diag{\bss}) + \bZ_{\bsfX}\bsfV\bL\right)^{\top}\bs{1}_n/n && \\
&= \bs{m} + \bL^{\top}\bsfV^{\top}\bZ_{\bsfX}\tp \bs{1}_n/n &&\because (\bsfX-\bs{1}_n\bs{m}^{\top})^{\top}\bs{1}_n = \bs{0}_p \\
&= \bs{m} &&\because \bZ_{\bsfX}^{\top}\bs{1}_n = \bs{0}_{n-1-p}.
\end{flalign*}
Next we show $(\bC\bsfXk)^{\top}\bC\bsfXk = \bs{S}$:
\begin{align}
(\bC\bsfXk)^{\top}\bC\bsfXk 
&=\left((\bsfX-\bs{1}_n\bs{m}^{\top})(\bs{I}_p - \bs{S}^{-1}\diag{\bss}) + \bZ_{\bsfX}\bsfV\bL\right)^{\top} \left((\bsfX-\bs{1}_n\bs{m}^{\top})(\bs{I}_p - \bs{S}^{-1}\diag{\bss}) + \bZ_{\bsfX}\bsfV\bL\right)  \nonumber\\
%&\hspace{8 cm}\text{since }
 &\hspace{200pt minus 1fil}  \because \bC\tp \bC[ \bsfX-\bs{1}_n\bs{m}^{\top} , \bZ_{\bsfX}] =  [ \bsfX-\bs{1}_n\bs{m}^{\top} , \bZ_{\bsfX}] \nonumber\\
&= (\bs{I}_p - \bs{S}^{-1}\diag{\bss})^{\top}\bs{S}(\bs{I}_p - \bs{S}^{-1}\diag{\bss}) + \bL^{\top}\bsfV^{\top}\bZ_{\bsfX}\tp \bZ_{\bsfX}\bsfV\bL \nonumber  \\%\label{eq:a3}
&\hspace{200pt minus 1fil} \because (\bsfX-\bs{1}_n\bs{m}^{\top})^{\top}\bZ_{\bsfX} = \bs{0}_{p\times p}  \hfilneg \nonumber\\
&= \bs{S}-2\diag{\bss}+\diag{\bss}\bs{S}^{-1}\diag{\bss} + \bL^{\top}\bL \nonumber \\
%\hspace{1.35cm}\text{since }
 &\hspace{200pt minus 1fil}  \because  \bZ_{\bsfX}^{\top}\bZ_{\bsfX}=\bs{I}_{n-1-p}, \bsfV\in \mathcal{F}_{n-1-p,p} \nonumber\\
&= \bs{S}%\hspace{6.7cm}\text{since }
\hspace{180pt minus 1fil}  \because \bL^{\top}\bL = 2\diag{\bss}-\diag{\bss}\bs{S}^{-1}\diag{\bss}.\nonumber
\end{align}
And finally we show $(\bC\bsfXk)^{\top}\bC\bsfX = \bs{S}-\text{diag}\{\bss\}$:
\begin{align*}
(\bC\bsfXk)^{\top}\bC\bsfX 
&= \left((\bsfX-\bs{1}_n\bs{m}^{\top})(\bs{I}_p - \bs{S}^{-1}\diag{\bss}) +  \bZ_{\bsfX}\bsfV\bL\right)^{\top} \left(\bsfX-\bs{1}_n\bs{m}^{\top}\right) \\
&= (\bs{I}_p - \bs{S}^{-1}\diag{\bss})^{\top}\bs{S} + \bL^{\top}\bsfV^{\top}\bZ_{\bsfX}\tp\left(\bsfX-\bs{1}_n\bs{m}^{\top}\right) \\
&= \bs{S}-\diag{\bss}\hspace{7cm}\because (\bsfX-\bs{1}_n\bs{m}^{\top})^{\top}\bZ_{\bsfX} = \bs{0}_{p\times (n-1-p)}.
\end{align*}
We conclude that $\bsfXk\in \mathcal{M}_{\bsfX}$ and therefore $\varphi_{\bsfX}(\mathcal{F}_{n-1-p,p})\subseteq \mathcal{M}_{\bsfX}$.\medskip

\textbf{Step 3}: $\varphi_{\bsfX}$ is injective.

Since $\bZ_{\bsfX}\tp \left[ \bs{1}_n,\, \bX \right]=\bs{0}$ and $\bL$ is invertible, $\bZ_{\bsfX}\tp\varphi_{\bsfX}(\bsfV)\bL^{-1}=\bsfV$. Thus $\varphi_{\bsfX}$ is injective. 

\textbf{Step 4}: $\varphi_{\bsfX}$ is surjective.

Let $\bsfXk\in \mathcal{M}_{\bsfX}$. By the definition of $\bZ_{\bsfX}$, the columns of $\left[ \bs{1}_n,\, (\bsfX-\bs{1}_n\bs{m}^{\top}), \bZ_{\bsfX} \right]$ form a basis of $\R^n$. Hence we can uniquely define $\bs{\alpha}\tp \in \R^{1\times p}$, $\bs{\Lambda} \in \R^{p\times p}$ and $\bs{\Theta}\in \R^{(n-1-p)\times p}$ such that 
\begin{equation}\label{eq:pf-ldg-bsfxk}
\bsfXk= \bs{1}_n\bs{\alpha}\tp + (\bsfX-\bs{1}_n\bs{m}^{\top})\bs{\Lambda}+\bZ_{\bsfX}\bs{\Theta}
\end{equation}

First, $\bs{m}=\bsfXk^{\top}\bs{1}_n/n=\bs{\alpha}$ because $\left[(\bsfX-\bs{1}_n\bs{m}^{\top}), \bZ_{\bsfX} \right]\tp \bs{1}_n=\bs{0}_{(n-1)\times 1}$.

Next we show $\bs{\Lambda}=\bs{I}_p - \bs{S}^{-1}\diag{\bss}$:
\begin{align*}
&&\bs{S}-\text{diag}\{\bss\} &= (\bC\bsfXk)^{\top}\bC\bsfX \\
\Rightarrow&&\bs{S}-\text{diag}\{\bss\} &= \bs{\Lambda}^{\top}(\bsfX-\bs{1}_n\bs{m}^{\top})^{\top}\left(\bsfX-\bs{1}_n\bs{m}^{\top}\right) + \bs{\Theta}^{\top}\bZ_{\bsfX}\tp  \left(\bsfX-\bs{1}_n\bs{m}^{\top}\right) \\
\Rightarrow&&\bs{S}-\text{diag}\{\bss\} &= \bs{\Lambda}^{\top}(\bsfX-\bs{1}_n\bs{m}^{\top})^{\top}\left(\bsfX-\bs{1}_n\bs{m}^{\top}\right)   \\
\Rightarrow&&\bs{S}-\text{diag}\{\bss\} &= \bs{\Lambda}^{\top}\bs{S}  \\
\Rightarrow&&\bs{I}_p - \bs{S}^{-1}\diag{\bss} &= \bs{\Lambda} .
\end{align*}

And finally, we show $\bs{\Theta}=\bsfV\bL$ for some %\ljmargin{unique}{is this needed or even proved?}\dhmargin{}{it is unique but it seems we don't need the uniqueness. }
$\bsfV\in \mathcal{F}_{n-1-p,p}$. Using Equation~\eqref{eq:pf-ldg-bsfxk},
\begin{align*}
&&(\bC\bsfXk)^{\top}\bC\bsfXk &= \bs{S} \\
\Rightarrow&& \bs{\Lambda}\tp \bsfX\tp \bC\tp \bC \bsfX \bs{\Lambda} + \bs{\Theta}\tp \bZ_{\bsfX}\tp\bZ_{\bsfX} \bs{\Theta} &=\bs{S}\qquad\qquad  \\
\Rightarrow&& (\bs{I}_p - \bs{S}^{-1}\diag{\bss})^{\top}\bs{S}(\bs{I}_p - \bs{S}^{-1}\diag{\bss}) + \bs{\Theta}\tp \bs{\Theta} &=\bs{S}\qquad\qquad   \\
\Rightarrow&& \bs{S}-2\diag{\bss}+\diag{\bss}\bs{S}^{-1}\diag{\bss} + \bs{\Theta}\tp \bs{\Theta}  &=\bs{S} \\
\Rightarrow&& \bs{\Theta}\tp \bs{\Theta} &=2\diag{\bss}-\diag{\bss}\bs{S}^{-1}\diag{\bss} \\
\Rightarrow&& \left(\bL^{-1}\right)^\top\bs{\Theta}\tp \bs{\Theta}\bL^{-1}  &=\bs{I}_p,
\end{align*}
%\ljmargin{-}{I don't understand the 2nd line, particularly why it follows from (A.5). Note (A.5) assumed something different, since it was part of proving the opposite inclusion.}
where again the second equality uses $\bC\bs{1}_{n}=\bs{0}$ and $\bC\tp\bC\bZ_{\bsfX}=\bZ_{\bsfX}$, the third equality uses $\bs{\Lambda}=\bs{I}_p - \bs{S}^{-1}\diag{\bss}$ and $\bZ_{\bsfX}\tp\bZ_{\bsfX}=\bs{I}_{n-1-p}$, and the last equality follows from the invertibility of $\bL$. Define $\bsfV \, :=\, \bs{\Theta}\bL^{-1}$, then the last equality implies $\bsfV \in \mathcal{F}_{n-1-p,p}$.
We conclude that $\bsfXk =  \varphi_{\bsfX}(\bsfV)$. 
\end{proof}

\rev{
 \subsubsection{An Intuitive Proof That Does Not Quite Work}\label{app:ldg-intutive}
The astute reader may think there is a more straightforward way than the previous subsection to prove Theorem~\ref*{thm:ldg_ko} using the fact that all the randomness in the conditional knockoffs construction of Algorithm~\ref{alg:ldg} comes from $[\bU,\tilde{\bU}]$ which follows the Haar measure on $\mathcal{F}_{n,2p}$, and this Haar measure has many known properties including swap-invariance. We show here why we were not able to follow this route, and resorted instead to a more technical proof using topological measure theory.
 
 % this is the referee's question
 For simplicity, consider the special case where the mean vector is known to be zero, i.e. $\bx_ i\sim N(0,\bs{I}\otimes \bS)$. Let $\bX=\bU \bs{D}\bV\tp$  be the singular value decomposition of $\bX$ where  $\bU \in \mathbb{R}^{n \times p}, \bs{D} \in \mathbb{R}^{p \times p}, \bV \in \mathbb{R}^{p \times p}$. It is not hard to see  that $\bU$ is uniformly distributed on $\mathcal{F}_{n,p}$ and is independent of $\bs{D} \bV^{\top}$. This claim implicitly uses the existence of a Haar measure on  $\mathcal{F}_{n,p}$, but this is well-known (we denote this measure by $\operatorname{Unif}\left(\mathcal{F}_{n, p}\right)$). Conditioning on $\bX^{T} \bX=\bV^{T} \bs{D}^{2} \bV=\hbS$, 
\[
\bX \stackrel{d}{=} \operatorname{Unif}\left(\mathcal{F}_{n, p}\right) \bs{D} \bV \stackrel{d}{=} \operatorname{Unif}\left(\mathcal{F}_{n, p}\right) \hbS^{1 / 2}
\]
Thus in principle, it would be sufficient to construct $\bXk$ such that
\begin{equation}\label{eq:r2}
[\bX,  \bXk] \sim \operatorname{Unif}\left(\mathcal{F}_{n, 2 p}\right) \left[ \begin{array}{cc}{\hbS} & {\hbS-\diag{\bss}} \\ {\hbS-\diag{\bss}} & {\hbS}\end{array}\right]^{1 / 2}
\end{equation}
which simply requires generating the left singular vectors of $\operatorname{Unif}\left(\mathcal{F}_{n, 2 p}\right) $ conditioned on $\bU$ being the first $p$ columns. This can be easily achieved by stacking $\bW$ on the right of $\bX$ and calculating the left singular values of $[\bX, \bW]$, which is exactly what is done in Algorithm 3.1.

To prove the validity of this construction, we just need to check that the right hand side of Equation~\eqref{eq:r2} is swap-invariant. Indeed, 
%\dhmargin{-}{here is our response}
%However, there is a key step missing from this argument. 
%Note that the unknown mean is not crucial, we will ignore it in the following.
$\operatorname{Unif}\left(\mathcal{F}_{n, 2 p}\right)$ is easily shown to be swap-invariant, and the matrix multiplying it appears to be swap-invariant as well. However, the matrix square root complicates things. Denote
$$\bG = \left[ \begin{array}{cc}{\hbS} & {\hbS-\diag{\bss}} \\ {\hbS-\diag{\bss}} & {\hbS}\end{array}\right]^{1 / 2}.$$ 

To make the argument more precise, suppose that we want to show that swapping $\bX_1$ with $\bXk_1$ does not change the joint distribution of $[\bX,\tilde{\bX}]$. 
Let $\bs{P}\in \R^{2p\times 2p}$ be the permutation matrix that swaps columns $1$ and $1+p$ of a matrix when multiplied on the right. By Equation~\eqref{eq:r2}, what we need to show is
\begin{equation}\label{eq:r2-2}
\operatorname{Unif}\left(\mathcal{F}_{n, 2 p}\right) \,\bG \bs{P}\eqd \operatorname{Unif}\left(\mathcal{F}_{n, 2 p}\right)\, \bG 
\end{equation}
The left hand side equals to $\left( \operatorname{Unif}\left(\mathcal{F}_{n, 2 p}\right)\, \bs{P}  \right) \bs{P} \bG \bs{P}$. By known properties of the Haar measure, we have that $\operatorname{Unif}\left(\mathcal{F}_{n, 2 p}\right)\, \bs{P}\eqd \operatorname{Unif}\left(\mathcal{F}_{n, 2 p}\right)$, and hence Equation~\eqref{eq:r2-2} is equivalent to
\begin{equation}\label{eq:r2-3}
\operatorname{Unif}\left(\mathcal{F}_{n, 2 p}\right)\, \bs{P} \bG \bs{P}\eqd \operatorname{Unif}\left(\mathcal{F}_{n, 2 p}\right)\, \bG
\end{equation}
The only way we can see how one might prove this more simply than the proof in our paper is to show that $\bs{P} \bG \bs{P} = \bG$, i.e., that $\bG$ is swap-invariant.

	%However, it is not the case. 
 % 	The matrix $\left[ \begin{array}{cc}{\hbS} & {\hbS-\diag{\bss}} \\ {\hbS-\diag{\bss}} & {\hbS}\end{array}\right]^{1 / 2}$ written in Eq. \eqref{eq:r2}  is ambiguous. 
	
	%\item 
	$\bG$ visually appears to be swap-invariant, and indeed is the square root of a swap-invariant matrix, but the fact that a matrix is swap-invariant does not directly imply that its square root is swap-invariant. The square root of a matrix in general is not unique, so we may hope that there exists (and we can identify) a swap-invariant square root in this case, but in the representation of Equation~\eqref{eq:r2}, we can actually only use the square root that has  $\bs{D}\bV^{\top}$ on its upper left block and has $\bs{0}$ on its bottom left block, in order to match $\bX=\bU\bs{D}\bV\tp$ on the left hand side. Therefore, we can actually say with certainty that
$$
\bG=\left[ \begin{array}{cc}{\bs{D}\bV\tp} & { \bs{D}^{-1}\bV\tp \diag{\bss}-\hbS^{-1}\diag{\bss}} \\ {0} & {\bL}\end{array}\right],
$$ 
where $\bL\tp \bL=2 \diag{\bss}-\diag{\bss} \hbS^{-1}\diag{\bss}$  is a Cholesky decomposition. However, this matrix \emph{cannot} be swap-invariant. This is why we were unable to prove swap-exchangeability of $[\bX,\tilde{\bX}]$ directly from swap-invariance of $\operatorname{Unif}\left(\mathcal{F}_{n, 2 p}\right)$, and were instead forced to prove it directly using topological measure theory. Note that our proof uses similar machinery to the first-principles proof of the known result that $\operatorname{Unif}\left(\mathcal{F}_{n, 2 p}\right)$ is swap-invariant.
%\end{enumerate}

%Above, we illustrate why we cannot just use the swap-invariance of $[\bU, \tilde{\bU}]$ and have to work on probability measure of $[\bX, \tilde{\bX}]$ directly. 
}

%%%%%%%%%%%
%\input{duplicateProof.tex}

\subsection{Gaussian Graphical Models}\label{app:ggmor}

%The following definition characterize the condition that makes the knockoff construction possible. 

% Algorithm~\ref{alg:sparse-gaussian-or} provides an alternative knockoff construction for Gaussian Graphical Models. Suppose we run Algorithm~\ref{alg:sparse-gaussian} with $n'\leq n-1$ and some $\pi$ but without the sampling step, it returns the blocking set $B$. Proposition~\ref{prop:ggm-block-comp} shows $\mathcal{G}$ is $n'$-separated by $B$. After identifying the connected components $W_{\icc}$'s of the subgraph that deletes $B$, Algorithm 3 can be used to generate knockoff. 

%We now provide the proof of Blocking-SDP construction. 
\begin{proof}[Proof of Theorem~\ref{prop:ggm-or}]
By classical results for the multivariate Gaussian distribution, we have

\begin{equation}\label{eq:ggm-cond-XBc}
	 X_{B^c}\mid X_B \sim \N(\bmu^{\ast}+\bs{\Xi}X_B ,\; \bS^{\ast}),
\end{equation}
 where $\bs{\Xi}=\bS_{B^c,B} (\bS_{B,B})^{-1}$, $\bmu^{\ast}=\bmu_{B^c} -\bs{\Xi}\bmu_B$ and $(\bS^{\ast})^{-1}= ( \bS^{-1} )_{B^c,B^c} $. By the condition that $G$ is $n$-separated by $B$, $(\bS^{\ast})^{-1}$ is block diagonal with blocks defined by the $V_{\icc}$'s. Thus $X_{V_1},\dots, X_{V_{\ncc}}$ are conditionally independent given $X_{B}$. 
 
To show $[\bX,\bXk]$ is invariant to swapping $A$ for any $A\subseteq[p]$,  by conditional independence of the $X_{V_\icc}$'s, it suffices to show that for any $\icc \in [\ncc]$ and $A_{\icc} \,:=A\,\cap V_{\icc}$, 
 \begin{equation}\label{eq:ggm-swap-separate}
 [  \bX_{V_{\icc}} ,  \bXk_{V_{\icc}}] _{\swap{A_{\icc}} } \eqd   [   \bX_{V_{\icc}} ,  \bXk_{V_{\icc}} ]  \mid \bX_{B}.
 \end{equation}
  
Before proving Equation~\eqref{eq:ggm-swap-separate}, we set up some notation.  Let $\bO=\bS^{-1}$, then by block matrix inversion (see, e.g., \citet[Proposition 1.3.3]{kollo2006advanced}), 
\[
\bS_{B^c,B^c}-\bS_{B^c,B}(\bS_{B,B})^{-1}\bS_{B,B^c}=\bO_{B^c,B^c}^{-1},
\]
and
\[
-\left(\bS_{B^c,B^c}-\bS_{B^c,B}(\bS_{B,B})^{-1}\bS_{B,B^c}\right)^{-1}\bS_{B^c,B}(\bS_{B,B})^{-1}=\bO_{B^c,B}.
\]
Thus $\bs{\Xi}$ can be written as $-\bO_{B^c,B^c}^{-1}\bO_{B^c,B}$. %\ljmargin{}{This needs a bit more explanation}.  
 
 Now fix $\icc \in [\ncc]$.  Let $B_{\icc}=I_{V_{\icc}} \, \cap \, B$. Since $V_{\icc}$ and  $B\setminus B_{\icc}$ are not adjacent,  $\bO_{V_{\icc}, B\setminus B_{\icc}}$, and thus $\bs{\Xi}_{V_{\icc}, B\setminus B_{\icc}}$, equals $\bs{0}$.    Equation~\eqref{eq:ggm-cond-XBc} implies 
\[
 X_{V_{\icc}} \mid X_{B} \sim \N (\bmu^{\ast}_{V_{\icc}}+\bs{\Xi}_{V_{\icc}, B_{\icc}}X_{B_{\icc}} ,\; \bS^{\ast}_{V_{\icc},V_{\icc}}). 
\]
This also implies that \begin{equation}\label{eq:ggm-cond-gaussian}
X_{V_{\icc}} \indp X_{B\setminus B_{\icc}} \mid X_{B_{\icc}}.
\end{equation}

Since the rows of $\bX_{V_{\icc} \uplus B_{\icc} }$ are i.i.d. Gaussian, the validity of Algorithm~\ref{alg:ldg-block} (see Theorem~\ref{thm:ldg-block} in Appendix~\ref{app:detail-ldg-block}) says that $\bXk_{V_{\icc}}$ generated in Line~\ref{line:ggm-alg1} of Algorithm~\ref{alg:sparse-gaussian-or} satisfies
 \[
 [  \bX_{V_{\icc}} ,  \bXk_{V_{\icc}}] _{\swap{A_{\icc}} } \eqd   [   \bX_{V_{\icc}} ,  \bXk_{V_{\icc}} ]  \mid X_{B_{\icc}},  
 \]
This together with Equation~\eqref{eq:ggm-cond-gaussian} shows Equation~\eqref{eq:ggm-swap-separate}. This completes the proof. 
\end{proof}

% \subsubsection{Results on Graphs}\label{sec:graph-str} 
\begin{proof}[Proof of Proposition~\ref{prop:ggm-block-comp}]
This proof will be about Algorithm~\ref{alg:ggm-greedy-search} in Appendix~\ref{app:detail-ggm}, which is shown there to be equivalent to Algorithm~\ref{alg:ggm-greedy-search-graph}. Without loss of generality, we assume $\pi=(1,\dots,p)$. Denote by $N_j^{(h)}$ the set $N_j$ in the Algorithm~\ref{alg:ggm-greedy-search} after the $h$th step. The updating steps of the algorithm ensure $j\notin N_j^{(h)}$ for any $j$ and $h$.  Note that $N_j$ does not change after the $(j-1)$th step, i.e., $N_j^{(j-1)}=N_j^{(j)}=\dots=N_j^{(p)}$.

%, and we write this final value as $M_j$ for short. 

It suffices to show the following inequality for each connected component $W$, whose vertex set is denoted by $V$, of the subgraph induced by deleting $B$:
	\[
	1+ 2|V|+| I_{V} \cap \, B| \leq n', \text{ where } I_{V}\, :=\,\bigcup\limits_{j\in V} I_j.
	\]
	
\textbf{Part 1. }% : $W$ is blocked by $I_{V} \cap \, B$.}
First note that by definition of $V$, every element of $I_V$ is either in $V$ or $B$. Now define $F=[p]\setminus ( V \uplus (I_{V} \cap \, B) )$.  
%\dhmargin{\color{red}For any $k\in F$, every path from $k$ to $V$ must pass through $I_{V} \cap \, B$. }{seems useless?} 
We will show that $k\in F$ will never appear in $N_j$ for any $j\in V$.

Initially, for any $j\in V$, $N_j^{(0)}=I_j$ does not intersect $F$. Suppose $h$ is the smallest integer such that there exists some $j\in V$ such that $N_j^{(h)}$ contains some $k\in F$. By the construction of the algorithm, $h\notin B$, $j>h$ and $j\in N_h^{(h-1)}$ (otherwise $N_j^{(h)}$ would not have been altered in the $h$th step), $k\in N_h^{(h-1)}$ (otherwise $k$ could not have entered $N_j^{(h)}$ at the $h$th step), and $h\in N_j^{(h-1)}$ (by symmetry of $N_j^{(i)}$ and $N_h^{(i)}$ for $i<\min(h,j)$).

%Since the $h$th step is the first time after which $F$ intersect with some $N_{j'}^{(h')}$ and 
 Since $h\in N_j^{(h-1)}$, the definition of $h$ guarantees $h\notin F$ (otherwise $h-1$ would be smaller and satisfy the condition defining $h$), and thus $h$ is in either $V$ or $I_V\cap B$. But since $h\notin B$, we must have $h\in V$. Now we have shown $k\in N_h^{(h-1)}$, i.e., $F$ intersects $N_h$ before the $h$th step, and $h\in V$, but this contradicts the definition of $h$. We conclude that for any $j\in V$ and any $h\in [p]$,  $F\cap\, N_j^{(h)}=\emptyset$ and thus $N_{j}^{(p)}\subseteq ( I_{V} \cap \, B) \uplus\, (V\setminus\{j\} )$. 
 
\textbf{Part 2:} We now characterize $N_{j}^{(p)}$.
	For any $j\in V$, define
\begin{equation*}
\begin{aligned}
L_j&:=\left\{ v\in (I_{V} \cap \, B) \uplus (V\setminus \{j\})  :\, \exists \text{ a path }(j,j_1,\dots,j_m,v)\text{ in $G$} , \right.\\
&\hspace{14 em} \left.  \text{ where } j_i <j \text{ and } j_i\in V ,\; \forall i\in [m]\right\}.\\
\end{aligned}
\end{equation*}
We will show $L_{j}\subseteq N_{j}^{(p)}$ by induction.
This is true for the smallest $j\in V$ because $L_j=I_j\subseteq N_{j}^{(p)}$.
Now assume $L_{j}\subseteq N_{j}^{(p)}$  for any $j<j_0$ (both in $V$), we will show $L_{j_0}\subseteq N_{j_{0}}^{(p)}$. For any $v\in L_{j_0}$, if $v\in I_{j_0}$ it is trivial that $v\in N_{j_{0}}^{(p)}$. If $v\in L_{j_0}\setminus I_j$, there is a path $(j_{0},j_1,\dots,j_m,v)$ in $G$ where $\{j_i\}_{i=1}^{m}\subseteq V$ are all smaller than $j_0$. Let $j_{i^*}$ be the largest among $\{j_i\}_{i=1}^{m}$. With the two paths $(j_0,j_1,\dots,j_{i^*} )$ and $( j_{i^*},\dots, j_m,v)$, we have $j_0, v \in L_{j_{i^*}}\subseteq N_{j_{i^*}}^{(p)}$ by the inductive hypothesis. 
%Since $j_0\in N_{j_{i^*}}^{(p)}$ and $j_0>j_{i^*}$, we conclude $j_{0}\in N_{j_{i^*}}^{+}$ as defined in Line~\ref{line:alg-search-pm} of  Algorithm~\ref{alg:ggm-greedy-search}. 
Since $j_0\in N_{j_{i^*}}^{(p)}$ and $j_0>j_{i^*}$, 
in the $j_{i^*}$th step on Line~\ref{alg:alg-search-absorb},  $N_{j_0}$ absorbs $N_{j_{i^*}}\setminus \{j_0\}$, and it follows that $v\in N_{j_0}^{(j_{i^*})}$ and thus $v\in N_{j_0}^{(p)}$. We finally conclude that $L_{j_0}\subseteq N_{j_0}^{(p)}$, and by induction, $L_j\subseteq N_{j}^{(p)} $ for all $j\in V$. 
	
	%Define $N_j^{-}=N_j\cap \{\pi_1,\dots, \pi_{j-1} \} \setminus B$, and $N_j^{+}=N_j\cap \{\pi_{j+1},\dots, \pi_{p}\} $. \label{line:alg-search-pm}  %$N_j^{+}=N_j\setminus N_j^{-} \setminus B $.  %{\color{red} this step is not greedy enough. We can pick a $j$ that minimizes $|N_j|+|N_j^-|$}% better to be clear than to be comprehensive if not both
	
\textbf{Part 3.}	
	Let $j^*$ be the largest number in $V$. Since $W$ is connected and $j^*$ is the largest, the definition of $L_{j^*}$ implies $( I_{V} \cap \, B) \uplus (V\setminus\{j^*\} ) = L_{j^*}$. Part 1 showed that $N_{j^*}^{(p)}\subseteq ( I_{V} \cap \, B) \uplus\, (V\setminus\{j^*\} )$ and Part 2 showed that $L_{j^*}\subseteq N_{j^*}^{(p)}$. Thus $N_{j^*}^{(p)}= ( I_{V} \cap \, B) \uplus\, (V\setminus\{j^*\} )$. 
	
	Since $B$ keeps growing, at the $j$th step of Algorithm~\ref{alg:ggm-greedy-search},  the set $\{ 1,\dots,j-1\}\setminus B$ with the current $B$ is the same as that with the final $B$.  
At the $j^*$th step of the algorithm,
$N_{j^*}\cap (\{ 1,\dots,j^{*}-1\}\setminus B) $ 
	%in Line~\ref{line:alg-search-pm} 
	equals $V\setminus\{j^*\}$ (since $j^*$ is the largest in $V$). 
	Hence 
	\[
|N_{j^*}|+|N_{j^*}\cap (\{ 1,\dots,j^{*}-1\}\setminus B)| = \left( | (I_{V} \cap \, B) \, \uplus V|-1 \right)+ \left(  |V|-1  \right) =2|V|+| I_{V} \cap \, B|-2.
	\] 
	Since $j^*\notin B$, the requirement in Line~\ref{line:alg-search} and the equality above implies
	\[
	1+2|V|+| I_{V} \cap \, B| \leq n',
	\]
	and this completes the proof. 
\end{proof}

\subsection{Discrete Graphical Models}\label{app:graph-expanding}
%For the ease of notation, we write the conditional probabililty $\pr{\cdot}$ without  \textquoteleft $|T_{B}(\bX)$\textquoteright\ since we always conditional on $T_{B}(\bX)$

%\subsubsection{General Blocking}
\begin{proof}[Proof of Theorem~\ref{thm:dgm-block}]
We first show
\begin{equation}\label{eq:dgm-decompostion}
\Pcr{X_{B^c}}{X_{B}} \, =\, \prod_{j \in B^c} \Pcr{X_{j}}{X_{B}}\, =\, \prod_{j \in B^c} \Pcr{X_{j}}{X_{ I_j}}. 
\end{equation}
Suppose $j\in B^c$,  then $I_j\subseteq B$. By the local Markov property, 
\[
X_j \indp (X_{B^c\setminus  \{j\} }, X_{B\setminus  I_j} ) \mid X_{I_j}.
\]
By the weak union property, we have 
\[
X_j \indp X_{B^c\setminus  \{j\}  } \mid ( X_{I_j}, X_{B\setminus  I_j} ), 
\]
which implies $\Pcr{X_{B^c}}{X_{B}} \, = \, \Pcr{X_j}{X_B} \Pcr {X_{B^c\setminus  \{j\} } }{X_B}$. 
Following this logic for the remaining elements of $B^c\setminus \{j\} $, we have $\Pcr{X_{B^c}}{X_{B}} \, =\,\prod_{j \in B^c} \Pcr{X_{j}}{X_{B}}$, which is then equal to $\prod_{j \in B^c} \Pcr{X_{j}}{X_{ I_j}}$ because  $X_j \indp  X_{B\setminus  I_j}  \mid X_{I_j}$. 

Secondly, as justified in Section~\ref{sec:eg-dgm-1}, the construction of $\bXk_{j}$ in Algorithm~\ref{alg:block-dg} implies that conditional on  $T_{B}(\bX)$, $ \bXk_j$ and  $\bX_{j}$ are independent and identically distributed, and thus 
\[
(\bX_j,\bXk_j)\eqd (\bXk_j,\bX_j) \mid T_{B}(\bX).\] 
By the law of total probability, it follows that
\begin{equation}\label{eq:pf-dgm-exchange}
   (\bX_j,\bXk_j)\eqd (\bXk_j,\bX_j) \mid \bX_{B}
\end{equation}
Since $ \bXk_j$ is generated without looking at $\bX_{B^{c}\setminus  \{j\}  }$, it holds that 
\begin{equation}\label{eq:pf-dgm-cindp}
    \bXk_j \indp (\bX_{B^{c}\setminus  \{j\}  }, \bXk_{B^{c}\setminus  \{j\}  }) \mid (\bX_B, \bX_j).
\end{equation}

Next we show $(\bX_{B^{c}},\bXk_{B^c})_{\text{swap}(A)} \eqd (\bX_{B^c},\bXk_{B^c})$ for any $A\subseteq B^c$. 
For any pair of column vectors $(\bsfX_{j}, \bsfXk_{j})\in [\St_j]^{n}\times  [\St_j]^{n}$, define 
\[
(\bsfX_{j}^{A}, \bsfXk_{j}^{A})=\left\{
\begin{array}{c l}	
	(\bsfX_{j}, \bsfXk_{j}) & j\notin A \\
	(\bsfXk_{j}, \bsfX_{j}) & j \in A
\end{array}\right.
\]
By Equations~\eqref{eq:dgm-decompostion} and \eqref{eq:pf-dgm-cindp},
\begin{align*}
&\; \Pc{ ( \bX_{B^{c}},\bXk_{B^c} )= (\bsfX_{B^{c}},\bsfXk_{B^c} )}{ \bX_{B}}\\
=&\; \prod_{j\in B^c}\Pc{ ( \bX_{j},\bXk_{j} )= ( \bsfX_{j},\bsfXk_{j} )}{ \bX_{B}}\\
=&\; \prod_{j\in B^c\setminus  A}\Pc{ ( \bX_{j},\bXk_{j} )=( \bsfX_{j},\bsfXk_{j} )}{ \bX_{B}}\times \prod_{j\in A}\Pc{ ( \bX_{j},\bXk_{j} )=( \bsfXk_{j}, \bsfX_{j} ) }{ \bX_{B}}\\
=& \; \prod_{j\in B^c\setminus  A}\Pc{ ( \bX_{j},\bXk_{j} )=( \bsfX_{j}^{A},\bsfXk_{j}^{A} )}{ \bX_{B}}\times \prod_{j\in A}\Pc{ ( \bX_{j},\bXk_{j} )=( \bsfX_{j}^{A}, \bsfXk_{j}^{A} ) }{ \bX_{B}}\\
=&\; \Pc{ ( \bX_{B^{c}},\bXk_{B^c} )= (\bsfX_{B^{c}}^{A},\bsfXk_{B^c}^{A} )}{ \bX_{B}},\\
=&\; \Pc{ ( \bX_{B^{c}},\bXk_{B^c} )_{\text{swap}(A)}= (\bsfX_{B^{c}},\bsfXk_{B^c} )}{ \bX_{B}},
\end{align*}
% not prove in this way because it has the same complexity as above
% For $A$ being a singleton, say $\{k\}$ for some $k$
where the third equality (which swaps the order of $\bsfX_j$ and $\bsfXk_j$ and adds superscript $A$'s in the second product) follows from Equation~\eqref{eq:pf-dgm-exchange}. 

Together with $\bXk_{B}=\bX_{B}$, we conclude $\bXk$ is a valid knockoff for $\bX$. 
\end{proof}

\section{Algorithmic Details}
% low d gaussian
\subsection{Low Dimensional Gaussian}\label{app:detail-ldg}
\subsubsection{Additional Details on Algorithm~\ref{alg:ldg}} % fail to compile with "Computing Parameter $\bss$"
We begin with the construction of a suitable $\bss$ by extending existing algorithms for computing $\bss$ to our situation. 
%\dhmargin{Set $\hbS := (\bX-\hbmu)^{\top}(\bX-\hbmu)$ as in Appendix~\ref{app:ldgprf}. }{this line is not needed now.}
Without loss of generality we assume $\hbS_{j,j}=1$ for $j=1,\dots,p$ here; otherwise denote by $\hat{\bs{D}}$ the diagonal matrix with $\hat{\bs{D}}_{j,j}=\hbS_{j,j}$, set $\hbS^{0}$ to be $\hat{\bs{D}}^{-1/2}\hbS\hat{\bs{D}}^{-1/2}$, define $\bss^{0}=\hat{\bs{D}}^{-1}\bss$, and proceed with $\hbS$ and $\bss$ replaced by $\hbS^{0}$ and $\bss^{0}$ respectively. For any $\epsilon, \delta \in (0,1)$, we can compute $\bss$ in any of the following ways: %In order to ,  can easily be extended:
\begin{itemize}
\item \emph{Equicorrelated} \citep{RB-EC:2015}: Take $s_j^{\text{EQ}} = (1-\epsilon)\min\left( 2\lambda_{\text{min}}( \hbS ),  1 \right) $ for all $j=1,\dots,p$.
\item \emph{Semidefinite program (SDP)} \citep{RB-EC:2015}: Take $\bss^{\text{SDP}}$ to be the solution to the following convex optimization:
\begin{align}\label{opt:sdp}
\min\sum_{j=1}^p\left(1-s_j\right)\quad\text{subject to: }\qquad \begin{array}{l}\delta\le s_j\le 1,\; j=1,\dots,p\\ \diag{\bss}\preceq(1-\epsilon)2\hbS.\end{array}
\end{align}
\item \emph{Approximate SDP} \citep{EC-ea:2018}: Choose an approximation $\hbS_{\text{approx}}$ of $\hbS$ and compute $\bss^{\text{approx}}$ by the SDP method as if $\hbS=\hbS_{\text{approx}}$. Then set $\bss = \gamma\bss_{\text{approx}}$ where $\gamma$ solves
\[\max\gamma\qquad\text{subject to: }\quad \diag{\gamma\bss_{\text{approx}}}\preceq(1-\epsilon)2\hbS.\]
\end{itemize}
Noting that $\bXk_j^{\top}\bX_j/n = \hat{\Sigma}_{j,j}-s_j$, it will always be preferable to take $\epsilon$ as small as possible (for all methods), so that $\bss$ is as large as possible and $\bX_j$ and $\bXk_j$ are as different as possible. For the SDP method, the lower bound $\delta$ can be set to be $s_j^{\text{EQ}}$ multiplied by a small number, e.g., $\delta = 0.1 \cdot 2\lambda_{\text{min}}( \hbS ) $, to guarantee feasibility; this choice is used in the simulations in Sections~\ref{sec:ldg} and \ref{sec:eg-ggm}.

% In practice, the SDP method has a disadvantage that some of $s_j$ may be close to 0. A hybrid method can lower bound $\bss$ uniformly by taking $\bss$ to be a convex combination of $\bss^{\text{EQ}}$ and $\bss^{\text{SDP}}$ as
% \begin{equation}\label{eq:ldg-s-hybrid}
% \bss=(1-\alpha) \bss^{\text{EQ}} + \alpha\bss^{\text{SDP}},
% \end{equation}
% with weight $\alpha$, say $\alpha=0.9$. 

We now prove the computational complexity of Algorithm~\ref{alg:ldg}. The Cholesky decomposition takes $O(p^3)$ operations and the Gram--Schmidt orthonormalization takes $O(n p^2)$ operations. If $\bss$ is computed by the \emph{Equicorrelated} method whose complexity is no larger than $O(p^3)$, the overall complexity of Algorithm~\ref{alg:ldg} is $O(np^2)$.

\subsubsection{Gaussian Knockoffs with Known Mean}
Algorithm~\ref{alg:ldg-known-mean} is a slight modification of Algorithm~\ref{alg:ldg} for mulitvariate Gaussian models with mean parameter $\bmu$ known. The proof of its validity requires only minor modification of the proof of Theorem~\ref{thm:ldg_ko}, and is thus omitted.
%; this will be used in Algorithm~\ref{alg:sparse-gaussian-or}.
% Suppose $n\geq 2p$, let $\hbS=\bX\tp \bX$, find $\bss$ and $\bL$ as in Algorithm~\ref{alg:ldg}, %let $\bZ\in \R^{n\times (n-p)}$ be orthogonal to $\bX$ and let $\bV$ be a uniform $(n-p)\times p$ orthonormal matrix,
% generate $\bW$ a $n\times p$ matrix whose entries are i.i.d. $\N(0,1)$ and independent of $\bX$ and compute the Gram--Schmidt orthonormalization $[\bs{Q},\bU ]$ of $\left[\bX ,\, \bW \right]$, 
%  then a valid knockoff is given by
% \begin{equation}\label{eq:ldg_ko_nomean}
% \bXk =  \bX(\bs{I}_p - \hbS^{-1}\diag{\bss}) + \bU\bL
% \end{equation}

\begin{algorithm}[]
\caption{Conditional Knockoffs for Low-Dimensional Gaussian Models with Known Mean} \label{alg:ldg-known-mean}
\begin{algorithmic}[1]
\ENSURE$\bX\in \R^{n\times p}$, $\bmu\in \R^{p}$.
\REQUIRE $n\geq 2p$.
\STATE Define $\hbS=(\bX-\bs{1}_{n}\bmu\tp)\tp (\bX-\bs{1}_{n}\bmu\tp)$.
\STATE Find $\bss\in\R^p$ such that $\bs{0}_{p\times p}\prec\diag{\bss}\prec 2 \hbS$.
\STATE Compute the Cholesky decomposition of $2\diag{\bss} - \diag{\bss}\hbS^{-1}\diag{\bss}$ as $\bL^{\top}\bL$.
\STATE Generate $\bW$ a $n\times p$ matrix whose entries are i.i.d. $\N(0,1)$ and independent of $\bX$ and compute the Gram--Schmidt orthonormalization $[\underbrace{\bs{Q}}_{n\times p} ,\underbrace{\bs{U}}_{n\times p}]$ of $\left[\bX,\, \bW \right]$.
\STATE Set 
\begin{equation}\label{eq:ldg-known-mean_short}
\bXk = \bs{1}_{n}\bmu\tp + (\bX-\bs{1}_{n}\bmu\tp)  (\bs{I}_p - \hbS^{-1}\diag{\bss}) + \bU\bL.
\end{equation}\vspace{-.5cm}
\RETURN $\bXk$.
\end{algorithmic}
\end{algorithm}

\subsubsection{Partial Gaussian Knockoffs with Fixed Columns}\label{app:detail-ldg-block}
Consider the case where some of the variables are known to be relevant and thus do not need to have knockoffs generated for them. Let $B\subseteq [p]$ be the set of variables that no knockoffs are needed for, so we only want to construct knockoffs for variables in $V=B^c$, i.e., to generate $\bXk_V$ such that for any subset $A\subseteq V$, 
\[
 [  \bX_{V} ,  \bXk_{V}] _{\swap{A} } \eqd   [   \bX_{V} ,  \bXk_{V} ]  \mid \bX_{B}.
\]
Algorithm~\ref{alg:ldg-block} provides a way to generate such knockoffs. We can find its computational complexity as follows. Fitting the least squares in Line~\ref{line:ls} takes $O(n|B|^{2}|V|)$, computing $\hbS$ takes $O(n|V|^{2})$, both the most efficient construction of $\bss$ and inverting $\hbS$ take $O(|V|^{3})$, and the Gram--Schmidt orthonormalization takes $O(n (1+|B|+2|V|)^{2})$. Hence the overall computational complexity is $O\left( n |B|^{2} |V| +n |V|^{2} \right)$.
\begin{algorithm}[h]
\begin{algorithmic}[1]
\caption{Partial Conditional Knockoffs for Low-Dimensional Gaussian Models} \label{alg:ldg-block}
\ENSURE $\bX_V \in \R^{n\times |V|}$, columns to condition on: $\bX_B\in\R^{n\times |B|}$.
\REQUIRE $n> 2|V|+|B|$.
\STATE Compute the least squares fitted value $\hat{\bX}_j$ and residual $\bs{R}_{j}$ from regressing $\bX_j$ on $[\bs{1}_{n},\bX_{B}]$ for each $j\in V$. Let $\bs{R}=[\dots,\bs{R}_{j},\dots]_{j\in V}$ and compute $\hbS=\bs{R}\tp \bs{R}$.\label{line:ls} 
\STATE Find $\bss\in\R^{|V|}$ such that $\bs{0}_{|V|\times |V|}\prec\diag{\bss}\prec 2 \hbS$.
\STATE Compute the Cholesky decomposition of $2\diag{\bss} - \diag{\bss}\hbS^{-1}\diag{\bss}$ as $\bL^{\top}\bL$.
\STATE Generate $\bW$ a $n\times |V|$ matrix whose entries are i.i.d. $\N(0,1)$ and independent of $\bX$ and compute the Gram--Schmidt orthonormalization 
$[\!\!\underbrace{\bU_{0}}_{n\times(1+|B|+|V|)}\!\!, \;\;\;  \underbrace{ \bU}_{n\times |V|}\, ]$ of $\left[\, \bs{1}_{n},\bX_{B}, \bs{R}, \,\bW\right]$.  \label{line:ldg-block-sub-1}
\STATE Set $\bXk_{V}=\hat{\bX}_{V} + \bs{R}\left( \bs{I}_{|V|}- (\hbS)^{-1}\diag{\bss} \right)+\bU \bL$. \label{line:ldg-block-1}
\RETURN $\bXk_{V}$.
\end{algorithmic}
\end{algorithm}

The validity of Algorithm~\ref{alg:ldg-block} relies on its equivalence to a straightforward but slow algorithm, Algorithm~\ref{alg:ldg-block-2}. We first show the validity of Algorithm~\ref{alg:ldg-block-2} and then show the equivalence. 

\begin{algorithm}[h]
\begin{algorithmic}[1]
\caption{Alternative Form of Algorithm~\ref{alg:ldg-block}} \label{alg:ldg-block-2}
\ENSURE $\bX \in \R^{n\times p}$, $B\subseteq [p]$ and $V=B^c$.
\REQUIRE $n> 2|V|+|B|$.
\STATE Generate a $n\times(n-1- | B|)$ orthonormal matrix $\bs{Q}$ that is orthogonal to $[\bs{1}_{n},\bX_{  B}]$, and compute $\bs{Q}_{\perp}$ an orthonormal basis for the column space of $[\bs{1}_{n},\bX_{  B}]$.\label{line:ldg-block-Q}
\STATE  Construct low-dimensional knockoffs $\bs{J}$ for $\bs{Q}\tp \bX_{V}$ via Algorithm~\ref{alg:ldg-known-mean} with $\bmu = \bs{0}$. \label{line:ldg-block-alg1}
\STATE Set $\bXk_{V}=\bs{Q}\bs{J} + \bs{Q}_{\perp}\bs{Q}_{\perp}\tp \bX_{V}$.\label{line:genxk}
\RETURN $\bXk_{V}$.
\end{algorithmic}
\end{algorithm}

\begin{proposition}
	Algorithm~\ref{alg:ldg-block-2} generates valid knockoff for $\bX_{V}$ conditional on $\bX_{B}$. 
\end{proposition}
\begin{proof}
By classical results for the multivariate Gaussian distribution, we have
\begin{equation}\label{eq:ldg-block-cond-XBc}
	 X_{V}\mid X_B \sim \N(\bmu^{\ast}+\bs{\Xi}X_B ,\; \bS^{\ast}),
\end{equation}
 where $\bs{\Xi}=\bS_{V,B} (\bS_{B,B})^{-1}$, $\bmu^{\ast}=\bmu_{V} -\bs{\Xi}\bmu_B$ and $(\bS^{\ast})^{-1}= ( \bS^{-1} )_{V,V} $. 
 
We want to show that for any $A\subseteq V$,  
 \begin{equation}\label{eq:ldg-block-swap-separate}
 [  \bX_{V} ,  \bXk_{V}] _{\swap{A} } \eqd   [   \bX_{V} ,  \bXk_{V} ]  \mid \bX_{B}.
 \end{equation}
 
 For $n$ i.i.d. samples, 
 \[
 \bX_{V}\mid \bX_{B} \sim \N_{n,|V|}(\bs{1}_{n}(\bmu^{\ast})\tp + \bX_{B} \bs{\Xi}\tp ,\; \bs{I}_{n}\otimes \bS^{\ast}  ).
 \]
By the definition of $\bs{Q}$ in Algorithm~\ref{alg:sparse-gaussian-or}, $\bs{Q}\tp [\bs{1}_{n},\bX_{B}]=\bs{0}_{ (n-1-|B|)\times (1+|B|)}$ and $\bs{Q}\tp\bs{Q}=\bs{I}_{n-1-|B| }$. %(note: not used) and $\bs{Q}_\perp$, $\bs{Q}\tp \bX_{V}$ is independent of $\bs{Q}_{\perp}\tp \bX_{V}$ and 
This together with the property in Equation~\eqref{eq:gaussian-matrix} implies
\[
\bs{Q}\tp \bX_{V} \mid \bX_B \sim \N_{n-1-|B|,|V|}(\bs{0},\; \bs{I}_{n-1-|B|} \otimes \bS^{\ast} ).
\]

% This conditional distribution does not depend on $\bX_B$, thus $\bs{Q} \tp \bX_{V}$ is independent of $\bX_B$. 

Since $n-1-|B| \geq 2|V|$, Algorithm~\ref{alg:ldg-known-mean} can be used to generate knockoffs $\bs{J}$ for $\bs{Q}\tp \bX_{V}$, which satisfies that
\[
 [ \bs{Q}\tp \bX_{V} ,\bs{J}] _{\swap{A} } \eqd   [ \bs{Q}\tp \bX_{V} ,\bs{J}] \mid \bX_B,
 \]
%Since both $ \bs{Q}\tp \bX_{V}$ and $\bs{J}$ are independent of $\bX_B$, it holds that
%\[
% [ \bs{Q}\tp \bX_{V} ,\bs{J}] _{\swap{A} } \eqd   [ \bs{Q} \tp\bX_{V} ,\bs{J}]  \mid \bX_{B},
%\]	
%and thus
and thus
\[
 [  \bs{Q}\bs{Q}\tp \bX_{V} , \bs{Q}\bs{J}] _{\swap{A} } \eqd   [  \bs{Q}\bs{Q}\tp \bX_{V} , \bs{Q}\bs{J}]  \mid \bX_{B}.
\]	
Adding $[\bs{Q}_{\perp}\bs{Q}_{\perp}\tp \bX_{V},\bs{Q}_{\perp}\bs{Q}_{\perp}\tp \bX_{V}]$, which is trivially invariant to swapping, to both sides and using $\bs{I}_{n}=\bs{Q}\bs{Q}\tp+ \bs{Q}_{\perp}\bs{Q}_{\perp}\tp$ and the definition of $\bXk_{V}$ in Line~\ref{line:genxk} of Algorithm~\ref{alg:sparse-gaussian-or}, we have
\[
 [  \bX_{V} ,  \bXk_{V}] _{\swap{A} } \eqd   [   \bX_{V} ,  \bXk_{V} ]  \mid \bX_{B}.
\]
Since this holds for any $A\subseteq V$, $\bXk_{V}$ is a valid knockoff matrix for $\bX_{V}$ conditional on $\bX_{B}$.
	 %$\bS^{\ast}=\bS_{B,B}- \bS_{B^c,B} (\bS_{B,B})^{-1} \bS_{B,B^{c}}= ( ( \bS^{-1} )_{B^c,B^c} )^{-1}$
\end{proof}

\begin{theorem}\label{thm:ldg-block}
	Algorithm~\ref{alg:ldg-block} generates valid knockoffs for $\bX_V$ conditional on $\bX_B$. 
\end{theorem}
\begin{proof}
		It suffices to show that if the same $\bss$ and $\bL$ in Algorithm~\ref{alg:ldg-block}  are used to generate $\bs{J}$ in Line~\ref{line:ldg-block-alg1} of Algorithm~\ref{alg:ldg-block-2}, then the output $\bXk_{V}$ in Algorithm~\ref{alg:ldg-block} and the output in Algorithm~\ref{alg:ldg-block-2}, which is denoted by $\bXk_{V}^{\prime}$ to avoid confusion,  have the same conditional distribution given $\bX_{B}$ and $\bX_{V}$. 

	We write the Gram--Schmidt orthonormalization as a function $\Psi	(\cdot)$. 
Let $b=1+|B|$ and $d=|V|$. By assumption, $b+2d\leq n$. 
	
	By the definition of $\bs{Q}$ and $\bs{Q}_{\perp}$ in Line~\ref{line:ldg-block-Q} of Algorithm~\ref{alg:ldg-block-2}, we have 
	\[
	\hat{\bX}_{V}=\bs{Q}_{\perp}\bs{Q}_{\perp}\tp \bX_{V},\; \bs{R}=\bs{Q}\bs{Q}\tp \bX_{V}. 
	\]

First, we express $\bXk_{V}^{\prime}$ in a similar form as  $\bXk_{V}$ in Line~\ref{line:ldg-block-1} of Algorithm~\ref{alg:ldg-block}. The conditional knockoff matrix for $\bs{Q}\tp \bX_{V}$ generated by Algorithm~\ref{alg:ldg-known-mean} (with $\bmu=\bs{0}_{(n-b)\times 1}$) is given by 
\[
\bs{J}=\bs{Q}\tp\bX_{V}(\bs{I}_d - (\hbS^{\prime})^{-1}\diag{\bss}) + \bU^{\prime}\bL,\]
where $\hbS^{\prime}=\bX_{V}\tp\bs{Q}\bs{Q}\tp\bX_{V}=\bs{R}\tp \bs{R}=\hbS$ and $\bU^{\prime}$ is the last $d$ columns of the Gram--Schmidt orthonormalization of $[\bs{Q}\tp\bX_{V}, \bW^{\prime}]$ with $\bW^{\prime}\sim \N_{n-b,d}(\bs{0}, \bs{I}_{n-b}\otimes \bs{I}_{d})$ independent of $\bX_{V}$ and $\bX_{B}$. 
Hence we have %\ljmargin{}{I think the dimension of the identity matrix should be $|V|$, not $p$}
	 	\begin{align}
\bXk_{V}^{\prime}&=\, \bs{Q}_{\perp}\bs{Q}_{\perp}\tp \bX_{V}+\bs{Q}\bs{J}\nonumber\\
&=\,  \hat{\bX}_{V} + \bs{R}\left( \bs{I}_{d}- \hbS^{-1}\diag{\bss} \right)+\bs{Q}\bU^{\prime}\bL. \label{eq:ldg-block-Q}
	 	\end{align}
It suffices to show $\bU$ in Line~\ref{line:ldg-block-sub-1} of Algorithm~\ref{alg:ldg-block} is distributed the same as $\bs{Q}\bU^{\prime}$ conditional on $\bX$. 

Without loss of generality (by choosing $\bs{Q}_{\perp}$ in Line~\ref{line:ldg-block-Q} of Algorithm~\ref{alg:ldg-block-2}), assume the Gram--Schmidt orthonormalization of $[\bs{1}_{n},\bX_{B},\bX_{V}]$ is $[\bs{Q}_{\perp}, \bs{M}]$, where $\bs{M}$ is a $n\times d$ matrix.
%\ljmargin{}{what is $\bs{M}$? Are you defining it here?}.\dhmargin{}{Yes, it is defined here. }
Hence $\text{span}(\bs{M})=\text{span}(\bs{Q}\bs{Q}\tp\bX_{V})$. 
Let $\bZ$ be a $(n-b)\times (n-b-d)$ matrix whose columns form an orthonormal basis for the orthogonal complement of $\text{span}( \bs{Q}\tp\bX_{V})$. 

\textbf{Characterizing $\bU$}:
Let $\bs{\Gamma}=[\bs{Q}_{\perp}, \bs{M}, \bs{Q}\bs{Z}]$. 
Since $\bZ\tp \bs{Q}\tp\bX_{V}=\bs{0}$, we have $\bZ\tp\bs{Q}\tp\bs{Q}\bs{Q}\tp\bX_{V}=\bs{0}$ and thus $\bZ\tp\bs{Q}\tp \bs{M}=\bs{0}$. Together with $\bs{Q}_{\perp}\tp \bs{Q}\bZ=\bs{0}$ and $(\bs{Q}\bZ)\tp\bs{Q}\bZ=\bs{I}_{n-b-d}$, we have $\bs{\Gamma}\in \bs{O}_{n}$. 

Using Equation~\eqref{eq:ldg-gs-0}, 
\begin{align}\label{eq:ldg-block-sub-gs1}
  \Psi(\left[\bs{1}_{n},\bX_{B}, \bX_{V}, \bW\right])
=\; &\Psi( \left[\Psi(\left[\bs{1}_{n},\bX_{B}, \bX_{V}\right]), \bW \right]) \nonumber\\
=\; &\Psi( \left[\bs{Q}_{\perp}, \bs{M}, \bW \right]) \nonumber\\
=\; & \bs{\Gamma}\Psi(\bs{\Gamma}\tp \left[\bs{Q}_{\perp}, \bs{M}, \bW  \right]),
\end{align}
where the first equality is due to the fact that Gram--Schmidt orthonormalization treats the columns of its inputs sequentially. 
An elementary calculation shows
\begin{align*}
	\bs{\Gamma}\tp \left[\bs{Q}_{\perp}, \bs{M}, \bW  \right]=
		\begin{bmatrix}
		\bs{I}_{b} & \bs{0}_{b\times d} &  \bs{Q}_{\perp}\tp \bW \\
		\bs{0}_{d\times b}&  \bs{I}_{d}  &   \bs{M}\tp\bW \\
		\bs{0}_{(n-b-d)\times b}&  \bs{0}_{(n-b-d)\times d} & \bZ\tp\bs{Q}\tp\bW
\end{bmatrix}, 
\end{align*}%\ljmargin{}{typo in the $(2,3)$ entry of the matrix}
thus 
%\begin{align}
%	\Psi\left(  \bs{\Gamma}\tp \left[\bs{Q}_{\perp}, \bs{M}, \bW  \right] \right)=
%		\begin{bmatrix}
%		\bs{I}_{b} & \bs{0}_{b\times (n-b)} \\
%		\bs{0}_{(n-b)\times b}& \Psi \left( 
%		\begin{bmatrix}
%		\bs{I}_{d}  &  \bU\tp\bW \\
%		\bs{0}_{(n-b-d)\times d} & \bZ\tp\bs{Q}\tp
%			\end{bmatrix}
%					 \right)
%		\end{bmatrix}, 
%\end{align}
\begin{align}\label{eq:ldg-block-sub-gs2}
	\Psi\left(  \bs{\Gamma}\tp \left[\bs{Q}_{\perp}, \bs{M}, \bW  \right] \right)=
			\begin{bmatrix}
		\bs{I}_{b} & \bs{0}_{b\times d} &  \bs{0}_{b\times(n-b-d)}\\
		\bs{0}_{d\times b}&  \bs{I}_{d}  & \bs{0}_{d\times(n-b-d)} \\
		\bs{0}_{(n-b-d)\times b}&  \bs{0}_{(n-b-d)\times d} & \Psi(\bZ\tp\bs{Q}\tp\bW)
\end{bmatrix}. 
\end{align}
Using the definition of $\bs{\Gamma}$ and Equations~\eqref{eq:ldg-block-sub-gs1} and \eqref{eq:ldg-block-sub-gs2}, we conclude \[\Psi(\left[\bs{1}_{n},\bX_{B}, \bX_{V}, \bW\right])=[\bs{Q}_{\perp}, \bs{M}, \bs{Q}\bs{Z} \Psi(\bZ\tp\bs{Q}\tp\bW) ],\]which implies
\[
\bU=\bs{Q}\bs{Z} \Psi(\bZ\tp\bs{Q}\tp\bW). 
\]
%It remains to show $\bs{Z} \Psi(\bZ\tp\bs{Q}\tp\bW)\eqd \bU^{\prime} \mid ( \bX_B, \bX_{V})$. 

Noting that $\bZ\tp\bs{Q}\tp\bs{Q}\bZ=\bs{I}_{n-b-d}$ and $\bW\sim \N_{n,d}(\bs{0}, \bs{I}_{n}\otimes \bs{I}_{d})$, Equation~\eqref{eq:gaussian-matrix} implies $\bZ\tp\bs{Q}\tp\bW\sim \N_{n-b-d,d}(\bs{0}, \bs{I}_{n-b-d}\otimes \bs{I}_{d})$. 
By the classic result in \citet[Proposition 7.2]{eaton1983multivariate}, the conditional distribution of $ \Psi(\bZ\tp\bs{Q}\tp\bW)$ given $(\bX_B, \bX_{V})$ is the unique $\mathcal{O}_{n-b-d}$-invariant probability measure on $\mathcal{F}_{n-b-d,d}$. 

\textbf{Characterizing $\bU^{\prime}$}: Let $\bZ_{\perp}\in \R^{(n-b)\times d}$ be the Gram--Schmidt orthonormalization of $\bQ\tp\bX_{V}$, and thus $\bZ\tp\bZ_{\perp}=\bs{0}$. Let $\bs{\Gamma}_{z}=[\bZ_{\perp},\bZ]$, then $\bs{\Gamma}_{z}\in \mathcal{O}_{n-b}$. Again using the properties of Gram--Schmidt orthonormalization, 
\begin{equation}\label{eq:ldg-block-sub-gs3}
\Psi( [\bs{Q}\tp\bX_{V}, \bW^{\prime}] )=\Psi( [\bZ_{\perp}, \bW^{\prime}] )=\bs{\Gamma}_{z}\Psi(\bs{\Gamma}_{z}\tp[\bZ_{\perp}, \bW^{\prime}] ).
\end{equation}
Since 
\begin{align*}
	\bs{\Gamma}_{z}\tp[\bZ_{\perp}, \bW^{\prime}] =
		\begin{bmatrix}
		\bs{I}_{d} &                       \bZ_{\perp}\tp \bW^{\prime} \\
		\bs{0}_{(n-b-d)\times d}&  \bZ\tp\bW^{\prime}
\end{bmatrix}, 
\end{align*}
it holds that
\begin{align}\label{eq:ldg-block-sub-gs4}
	\Psi\left( \bs{\Gamma}_{z}\tp[\bZ_{\perp}, \bW^{\prime}]  \right)=
				\begin{bmatrix}
		\bs{I}_{d} &                      \bs{0}_{d\times d} \\
		\bs{0}_{(n-b-d)\times d}&  \Psi(\bZ\tp\bW^{\prime})
\end{bmatrix}. 
\end{align}
Hence $\bU^{\prime}=\bZ \Psi(\bZ\tp\bW^{\prime})$ by combining Equations~\eqref{eq:ldg-block-sub-gs3} and \eqref{eq:ldg-block-sub-gs4}. As before, we can conclude that the conditional distribution of $\Psi(\bZ\tp\bW^{\prime})$ given $( \bX_B, \bX_{V})$ is the unique $\mathcal{O}_{n-b-d}$-invariant probability measure on $\mathcal{F}_{n-b-d,d}$. 

Combining the two parts above and using the uniqueness of the invariant measure, we conclude that
\[
\bU\eqd \bs{Q}\bU^{\prime} \mid ( \bX_B, \bX_{V}).
\]
Using the definition of $\bXk_{V}$ in Line~\ref{line:ldg-block-sub-1} and Equation~\eqref{eq:ldg-block-Q}, it follows that $\bXk_{V} \eqd \bXk_{V}^{\prime} \mid (\bX_B, \bX_{V})$. 
\end{proof}

% graphical gaussian
\subsection{Gaussian Graphical Models}\label{app:detail-ggm}
The computational complexity of Algorithm~\ref{alg:sparse-gaussian-or} can be shown by summing up the computational complexity of Algorithm~\ref{alg:ldg-block} in Line~\ref {line:ggm-alg1} for individual connected components, which is $O\left(n  \left| I_{V_{\icc}} \, \cap \, B \right|^{2} \left| V_{\icc} \right|+  | V_{\icc} |^{2} \right)$, 
as shown in Appendix~\ref{app:detail-ldg-block}. Its upper bound is due to the facts that $\sum_{\icc=1}^{\ncc}| V_{\icc}|\leq p$ and $\max_{1\leq \icc\leq \ncc} |V_{\icc}|\leq n'$. 

\subsubsection{Greedy Search for a Blocking Set}
Algorithm~\ref{alg:ggm-greedy-search} is the virtual implementation of Algorithm~\ref{alg:ggm-greedy-search-graph}. In Line~\ref{alg:alg-search-absorb} of Algorithm~\ref{alg:ggm-greedy-search}, we only need to keep track of $N_j$ the neighborhood of each unvisited $j$ in $\bar{G}$ among the vertices in $[p]$. This is because  if $k\in N_j$ and $\tilde{k}$ exists in $\bar{G}$ then it is guaranteed by Algorithm~\ref{alg:ggm-greedy-search-graph} that $\tilde{k}$ is a neighbor of $j$ in $\bar{G}$, and $j$ is a neighbor of both $k$ and $\tilde{k}$. This also implies that $|N_{j}\cap \{\pi_1,\dots, \pi_{t-1} \} \setminus B|$ equals the size of the neighborhood of $j$ in $\bar{G}$ among the knockoff vertices. 
Also note that the neighborhood of a visited vertex is no longer used in Line~\ref{line:alg-search-graph} of Algorithm~\ref{alg:ggm-greedy-search-graph}, therefore the update step in Line~\ref{alg:alg-search-absorb} of Algorithm~\ref{alg:ggm-greedy-search} can be restricted to the unvisited $k$'s. 
In the following, we use the equivalence between Algorithm~\ref{alg:ggm-greedy-search-graph} and Algorithm~\ref{alg:ggm-greedy-search} to prove the properties of Algorithm~\ref{alg:ggm-greedy-search-graph}. 

\begin{algorithm}[h]
\caption{Greedy Search for a Blocking Set} \label{alg:ggm-greedy-search}
\begin{algorithmic}[1]
\ENSURE  $\pi$ a permutation of $[p]$, $G=([p],E)$, $n'$.
\STATE Initialize $N_j=I_j$ for all $j\in [p]$, $B=\emptyset$.\label{line:alg-init}
\FOR{$t = 1,\dots,p$}
\STATE Let $j$ be $\pi_{t}$.
% \STATE Define $N_j^{-}=N_j\cap \{\pi_1,\dots, \pi_{j-1} \} \setminus B$, and $N_j^{+}=N_j\cap \{\pi_{j+1},\dots, \pi_{p}\} $. \label{line:alg-search-pm}  %$N_j^{+}=N_j\setminus N_j^{-} \setminus B $.  %{\color{red} this step is not greedy enough. We can pick a $j$ that minimizes $|N_j|+|N_j^-|$}% better to be clear than to be comprehensive if not both
\IF{$n' \geq 3+ |N_{j}|+|N_{j}\cap \{\pi_1,\dots, \pi_{t-1} \} \setminus B|  $} \label{line:alg-search}
\STATE Update $N_k\leftarrow N_k \cup ( N_{j}\setminus  \{k\})$ for all $k\in N_{j}\cap \{\pi_{t+1},\dots, \pi_{p}\}$.\label{alg:alg-search-absorb}
\ELSE
\STATE $B\leftarrow B\cup \{j\}$.
\ENDIF
\ENDFOR
\RETURN $B$.
\end{algorithmic}
\end{algorithm}% Avoid immediate line after the environment.

The following proposition shows that if the tail of the input permutation to Algorithm~\ref{alg:ggm-greedy-search-graph} is already a blocking set of the graph, then the output from the algorithm is a subset of this blocking set. This property allows one to refine a known but large blocking set (e.g., one could apply Algorithm~\ref{alg:ggm-greedy-search-graph} to the blocking set from Example~\ref{eg:square-lattice} in Appendix~\ref{app:ggm-eg}).

% Proposition~\ref{prop:ggm-perm-block}. We present some properties of the blocking set used in Algorithm~\ref{alg:ggm-greedy-search}. We end this subsection with a converse result. 

%\begin{proposition}\label{prop:ggm-perm-block}
%	If $G$ is $n'$-separated by $D\subseteq [p]$, then the blocking set determined by  Algorithm~\ref{alg:ggm-greedy-search} with $\pi=(\tau(D^c),\tau(D))$ is a subset of $D$. 
%\end{proposition}
%\begin{proof}
%	Let $m=|D^c|$. Without loss of generality, reindex the variables so that $\pi = 1:p$. Let $W$ be a connected component of the subgraph $\tilde{G}$ induced by deleting $D$, and $V$ be the vertex set of $W$. 
%	
%	We modify Part 1 of the proof of Proposition~\ref{prop:ggm-block-comp} to show $N_j\subset V\uplus\, (I_{V}\cap D)$, for all $j\in V$. Define $F=[p]\setminus V \setminus (I_{V} \cap \, D)$. Initially, $N_j$ doesn't include any $k\in F $. Suppose the $h$th step is the first time after which some $k\in F$ appears in $N_j$ for some $j\in V$. This implies  $h\in N_j$ and $k\in N_h$ before the $h$th step. Since this is the first time and $h\in N_j$,  $h\notin F$. Since $N_j$ gets expanded, $h < j\leq m $. Together shows $h\in V$. This contradict the definition of $h$, with $k\in N_h$. 
%
%	For any $j\in V$, $N_j\subset V\uplus\, (I_{V}\cap D)$ gives the inequality
%	\[
%3+	 |N_j^-|+|N_j^+| \leq 3+ |V|-1+|V\uplus\,  (I_{V}\cap D) | -1 \leq n',
%	 \]
%	which implies $j$ is not in the blocking set. 
%	
%The same	argument applies to any connected component of $\tilde{G}$, thus the blocking set is a subset of $D$. 
%\end{proof}
\begin{proposition}\label{prop:ggm-perm-block}
	Suppose $n'$ and $\pi$ are the inputs of Algorithm~\ref{alg:ggm-greedy-search-graph}, which returns a blocking set $B$. If $G$ is $n'$-separated by $\{\pi_{m+1}, \dots,\pi_{p}\}$ for some $m\in [p]$, then $\pi_1,\dots,\pi_{m}$ will not be in $B$.
\end{proposition}
\begin{proof}
In this proof, we use the equivalence between Algorithm~\ref{alg:ggm-greedy-search-graph} and Algorithm~\ref{alg:ggm-greedy-search}. 
	Let $D=\{\pi_{m+1}, \dots,\pi_{p}\}$. Without loss of generality, re-index the variables so that $\pi_{j}= j$ for every $j\in [p]$, and thus $D=\{m+1,\dots,p\}$. 
	Denote by $N_j^{(h)}$ the set $N_j$ in Algorithm~\ref{alg:ggm-greedy-search} after the $h$th step, as in the proof of Proposition~\ref{prop:ggm-block-comp}. 
	Let $W$ be any of the connected components of the subgraph induced by deleting $D$,  and $V$ be the vertex set of $W$. Then $V\subseteq \{1,\dots,m\}$.	
	%We modify Part 1 of the proof of Proposition~\ref{prop:ggm-block-comp} to show $N_j\subseteq V\uplus\, (I_{V}\cap D^{\pi})$, for all $j\in V$. 

% 	Define $F=[p]\setminus V \setminus (I_{V} \cap \, D)$. Initially, $N_j$ doesn't include any $k\in F $ and $j\in V$. Suppose the $h$th step is the first time after which some $k\in F$ appears in $N_j$ for some $j\in V$. This implies  $h\in N_j$ and $k\in N_h$ before the $h$th step. Since this is the first time and $h\in N_j$,  $h\notin F$. Since $N_j$ gets expanded, $h < j\leq m $. Together shows $h\in V$. This contradict the definition of $h$, with $k\in N_h$. 
	
% 	For any $j\in V$, $N_j\subseteq V\uplus\, (I_{V}\cap D)$ gives the inequality
% 	\[
% 3+	 |N_j^-|+|N_j^+| \leq 3+ |V|-1+|V\uplus\,  (I_{V}\cap D) | -1 \leq n',
% 	 \]
% 	which implies $j$ is not blocked in the algorithm. This shows $V\subseteq D^c$.
	
% Since $W$ is arbitrary, we conclude that any vertex in $\{1,\dots,m\}$ is not blocked. 

\textbf{Part 1.} 
We first show that $N_j^{(h)}\subseteq V\uplus\, (I_{V}\cap D)$ for any $j\in V$ and $h\in [p]$. The proof is similar to Part 1  in the proof of Proposition~\ref{prop:ggm-block-comp}.  
By definition of $V$, every element of $I_V$ is either in $V$ or $D$. 
Define $F=[p]\setminus ( V \uplus (I_{V} \cap \, D) )$. It suffices to show   that $k\in F$ will never appear in $N_j$ for any $j\in V$.
%It suffices to show that $k\in F$ does not intersect $N_j^{(h)}$ for any $j\in V$ and $h\in [p]$.

Initially, for any $j\in V$, $N_j^{(0)}=I_j$ does not intersect $F$. 
Suppose $h$ is the smallest integer such that there exists some $j\in V$ such that $N_j^{(h)}$ contains some $k\in F$. 
By the construction of the algorithm, $j>h$ and $j\in N_h^{(h-1)}$ (otherwise $N_j^{(h)}$ would not have been altered in the $h$th step), $k\in N_h^{(h-1)}$ (otherwise $k$ could not have entered $N_j^{(h)}$ at the $h$th step), and $h\in N_j^{(h-1)}$ (by symmetry of $N_j^{(i)}$ and $N_h^{(i)}$ for $i<\min(h,j)$). The fact that $h < j\leq m$ implies $h\notin D$. 
Since $h\in N_j^{(h-1)}$, the definition of $h$ guarantees $h\notin F$ (otherwise $h-1$ would be smaller and satisfy the condition defining $h$), and thus $h$ is in either $V$ or $I_V\cap D$. But since $h\notin D$, we must have $h\in V$. Now we have shown $k\in N_h^{(h-1)}$, i.e., $F$ intersects $N_h$ before the $h$th step, and $h\in V$, but this contradicts the definition of $h$. %[duplicate]Thus $F\cap N_j^{(h)}=\emptyset$ for any $j\in V$ and any $h\in [p]$. 
We conclude that for any $j\in V$ and any $h\in [p]$,  $F\cap\, N_j^{(h)}=\emptyset$ and thus $N_{j}^{(j-1)}\subseteq ( I_{V} \cap \, D) \uplus\, (V\setminus\{j\} )$.  

\textbf{Part 2.}
For any $j\in V$, at the $j$th step of Algorithm~\ref{alg:ggm-greedy-search}, 
%$N_{j}^{-}$ defined in Line~\ref{line:alg-search-pm} is a subset of $N_{j}^{(j-1)} \cap \{ 1,\dots,j-1\}$, thus 
%$N_{j}^{-}\subseteq V\setminus \{j\}$. 
$N_{j}\cap \{1,\dots,j-1\}\subseteq  V\setminus \{j\}$ by the definition of $D$. 
Hence we have %the inequality
\begin{align*}
3+	 | N_{j}\cap \{1,\dots,j-1\} \setminus B |+|N_j| &\leq\; 3+ |V|-1+|V\uplus\,  (I_{V}\cap D) | -1 \\
&\leq \; 1+2|V|+| (I_{V}\cap D) | \\
&\leq \; n', 
\end{align*}
where the last inequality is because of the condition that $G$ is $n'$-separated by 	$D$. Thus the requirement in Line~\ref{line:alg-search} of Algorithm~\ref{alg:ggm-greedy-search} is satisfied and $j$ is not in the blocking set.  %. This shows $V$ does not intersect the blocking set.
	
Finally, since $j$ and $W$ are arbitrary, we conclude that any vertex in $\{1,\dots,m\}$ is not blocked. 
\end{proof}

%\end{comment}

% greedy search
\subsubsection{Searching for Blocking Sets}\label{app:ggm-greedy-B}
Given any $\ncg$, Algorithm~\ref{alg:greedy-blocking} performs a randomized greedy search for the blocking sets $B_{i}$. Although there is no guarantee that the $B_i$'s found by Algorithm~\ref{alg:greedy-blocking} satisfy $\bigcap\limits_{i=1}^{\ncg}B_i=\emptyset$, one can subsequently check whether $\eta_{j}=\ncg$ for any $j\in [p]$, in which case the algorithm can be run again. Inspecting the vertices with $\eta_{j}=\ncg$ may reveal the difficulties of blocking for this graph. Changing the inputs $\ncg$ and $n'$ may also help.

\begin{algorithm}[h]
\caption{Randomized Greedy Search for Blocking Sets}\label{alg:greedy-blocking}
\begin{algorithmic}[1]
\ENSURE $G$, $\ncg$, $n'$ (by default $n'=\floor{n/\ncg}$).
\REQUIRE $n'\leq n/\ncg$.
\STATE Let $\{\eta_{j}=0\}_{j=1}^{p}$ be a sequence counting how often a variable is in a blocking set.
\FOR{ $i= 1,\dots,\ncg$  }
%\STATE Set $\pi^{(i)}$ to be the decreasing order of $\eta_{j}$'s, i.e. $\eta_{\pi^{(i)}_{1}}\geq \eta_{\pi^{(i)}_{2}}\geq\dots\geq\eta_{\pi^{(i)}_{p}}$, and randomly permute at ties.
\STATE Set $\pi^{}$ to be the decreasing order (with ties broken randomly) of $\eta_{j}$'s, so  $\eta_{\pi^{}_{1}}\geq \eta_{\pi^{}_{2}}\geq\dots\geq\eta_{\pi^{}_{p}}$.
\STATE Run Algorithm~\ref{alg:ggm-greedy-search-graph} with $\pi^{}$ and $n'$ and let $B^{(i)}$ be the returned blocking set.
\STATE Update $\eta_{j}\leftarrow \eta_{j}+\one{j\in B^{(i)}}$ for each $j=1,\dots,p$.
\ENDFOR
%\STATE ; otherwise turn to step 1. 
\RETURN $B^{(1)},\dots,B^{(m)}$.
\end{algorithmic}
\end{algorithm}

% graphs and algorithm
%\begin{comment}

\subsubsection{Examples of $(m,n)$-Coverable  Graphs} \label{app:ggm-eg}
% $(\ncg,n)$

%[ AR(1) is fully a special case of AR(r), so delete this one ]
%\begin{example}
%Consider a non-stationary Gaussian AR($1$) model, so that the sparsity pattern $E = \{(i,j): |i-j|\le 1\}$. Let $d=\floor{(n-6)/4}$ and suppose $n\geq 10$ and $p$ is arbitrarily large. Define $B_{1}=[p]\cap \{ \icc d -1 : \icc\in\mathbb{N} \}$ and $B_{2}=[p]\cap \{ \icc d  : \icc\in\mathbb{N} \}$. 
%
%Any connected component $W$ of the subgraph that deletes $B_{1}$ is no larger than $d$ and $W$'s vertices $V$ satisfy $|I_V\cap B_1|\leq 2$, so $G$ is $(2d+3)$-separated by $B_1$ and $2d+3\leq n/2$. The same holds for $B_2$.
%Running Algorithm \ref{alg:sparse-gaussian-or} on each half of $\bX$ with $B_1$ and $B_2$ will create non-trivial knockoff columns for $B_{1}^{c}$ and $B_2^{c}$ respectively. As $[p]=B_{1}^c\cup B_{2}^{c}$, we conclude that such a graph is $(2,n)$-coverable for $n\geq 10$. 
%\end{example}

% AR : time-inhomogeneous Gaussian AR models
\begin{example}[Time-inhomogeneous Autoregressive Models  ]
Consider a time-inhomogeneous Gaussian AR($r$) model (assuming\footnote{When $r=0$, the graph is isolated and is $(1,n)$-coverable for any $n\geq 3$.} 
 $r\geq 1$), so that the sparsity pattern $E = \{(i,j): 1\le |i-j|\le r\}$. Suppose $n\geq 2+8r$. A simple choice of blocking sets is given as follows. Let $d=\floor{(n-2)/8}$, then $d\geq r$. Let $B_1= [p]\cap \{k d + i: k \text{ odd, and } i=1,\dots,d\} $ and $B_2= [p]\cap \{kd + i: k \text{ even, and } i=1,\dots,d\} $. Any connected component $W$ of the subgraph that deletes $B_{1}$ is no larger than $d$ and $W$'s vertices $V$ satisfy $|I_V\cap B_1|\leq 2r$, so $G$ is $(2d+2r+1)$-separated by $B_1$ and $2d+2r+1\leq n/2$. The same holds for $B_2$. Note that $[p]=B_{1}^c\cup B_{2}^{c}$, thus the graph is $(2,n)$-coverable. 
\end{example}

\begin{example}[$d$-dimensional Square-lattice Models ]\label{eg:square-lattice}
Consider a finite subset of the $d$-dimensional lattice $\Z^{d}$ where pairs of vertices with distance 1 are adjacent. Suppose $n\geq 6+4d$, one could take $B_1$ as the grid points whose coordinates sum up to an odd number, and $B_2$ as the complement of $B_1$. The subgraph that deletes $B_1$ (or $B_2$) is isolated and each vertex has a neighborhood of size $2d$, so the graph is $(3+2d)$-separated by $B_1$ (or $B_2$). Since $3+2d\leq n/2$, the graph is $(2,n)$-coverable. 
\end{example}
\begin{example}
Consider a $m$-colorable graph $G$. Let each of $V_{1}, \dots, V_{m}$ be the vertex set of the same color. For any $i\in [m]$, the subgraph that deletes $B_{i}\, :=\, \cup_{\ell \neq i} V_{\ell}$ is the subgraph that restricts on $V_{i}$, of which each vertex is isolated. Thus $G$ is $(1+2+\max_{v\in V_{i}} |I_v|)$-separated by $B_{i}$. If $n\geq \sum_{i\in [m]}(3+\max_{v\in V_{i}} |I_v|)$, the graph is $(m,n)$-coverable. Note this subsumes Example~\ref{eg:square-lattice} which has $m=2$, but also applies to many other graphs such as forests, stars, and circles.
\end{example}

\subsection{Discrete Graphical Models}\label{app:detail-dgm}

\subsubsection{Details about the Algorithms}

We begin by proving the computational complexity of Algorithm~\ref{alg:block-dg}. For each $j\in B^{c}$, enumerating all nonempty configurations of $\bst_{I_j}$ takes no more than $\prod\limits_{\ell\in I_j}\St_{\ell}$ operations by checking each $\bst_{I_j}$ or $n|I_j|$ operations by checking each observed $X_{i,I_j}$. The random permutation takes no more than $n$ steps in total, so the overall complexity is $O\left(  \sum\limits_{j\in B^{c}}(n+ \min (\prod\limits_{\ell\in I_j}K_{\ell} , n |I_j| ) ) \right)$. 

As mentioned at the beginning of Section~\ref{sec:eg-dgm}, we can generate knockoffs without assuming the covariate categories being finite. First of all, with infinite $\St_{\ell}$'s, Algorithm~\ref{alg:block-dg} can still be used since in Line~\ref{alg:dgm-enum} it is only needed to enumerate  those $\bst_{I_j}$ actually appearing in the observed data, which is at most $n$. 
Furthermore, the proof of Theorem~\ref{thm:dgm-block} does not require the $\St_{\ell}$'s to be finite.

\subsubsection{Graph Expanding}
As mentioned in Section~\ref{sec:eg-dgm-1}, in Algorithm~\ref{alg:block-dg} variables in $B$ are blocked and their knockoffs are trivial. One way to mitigate this drawback is to run multiple times of Algorithm~\ref{alg:block-dg} with expanded graphs that include the generated knockoff variables. 

Specifically, denote by $\bar{G}$ a graph being augmented from $G$. For each $j\in B^c$, we add an edge between every pair of $j$'s neighbors and add to $\bar{G}$ the `knockoff vertex'  $\tilde{j}$ which has the same neighborhood as $j$. One can show that $[\bX,\bXk_{B^c}]$ is locally Markov w.r.t. the new graph. Applying Algorithm~\ref{alg:block-dg} to $[\bX,\bXk_{B^c}]$ with graph $\bar{G}$ but with a different global cut set $\bar{B}$ which pre-includes $B^c$ and also the knockoff vertices, we can generate knockoffs for some of the variables that have been blocked in the first run. One can continue to expand the graph to include the new knockoff variables, although the neighborhoods may become so large that the knockoff variables generated are constrained (through conditioning on these large neighborhoods) to be identical to their corresponding original variables. Algorithm~\ref{alg:dgm-expanding} formally describes this process, whose validity is guaranteed by Theorem~\ref{thm:dgm-expand}.

\begin{algorithm}[h]
\caption{Conditional Knockoffs for Discrete Graphical Models with Graph-Expanding}\label{alg:dgm-expanding}
\begin{algorithmic}[1]
\ENSURE $\bX\in\mathbb{N}^{n\times p}$, $G=([p],E)$, $Q$ the maximum number of steps to expand the graph. 
\STATE \textbf{Initialization}: the augmented graph $\bar{G}\leftarrow G$, whose vertex set is denoted by $\bar{V}$; $D\leftarrow \emptyset$ contains the variables that have knockoffs ; $q\leftarrow 1$ is the step of the graph expansion.
%\WHILE{ $|D|<p$ or some other criterion}
%\FOR{$q$ from $1$ to $Q$}
\WHILE{$q\leq Q$ and $|D|<p$}
\STATE Find a global cut set $B$ for $\bar{G}$ such that $B \supseteq D$.
\STATE Construct knockoffs $\bs{J}$ for $[\bX,\bXk_{\bar{V}\setminus [p]}]$ w.r.t. the graph $\bar{G}$ and the cut set $B$ by Algorithm~\ref{alg:block-dg}.
\STATE Set $\bXk_{[p]\setminus B}=\bs{J}_{[p]\setminus B}$.
\FOR{ $ j\in B^c$}
\STATE Add $\tilde{j}$ to $\bar{G}$  with the same neighborhood as $j$.
\STATE Add an edge between each pair of $j$'s neighbors.
\STATE Update $D\leftarrow D\uplus \{j\}$.
\ENDFOR
\STATE $q\leftarrow q+1$.
\ENDWHILE
\STATE Set $\bXk_{[p]\setminus D}=\bX_{[p]\setminus D}$.
\RETURN $\bXk$.
%\ENDWHILE
\end{algorithmic}
\end{algorithm}

\begin{theorem}\label{thm:dgm-expand}
Algorithm~\ref{alg:dgm-expanding} generates valid knockoff for model~\eqref{model:dgm}.
\end{theorem}

\begin{proof}[Proof of Theorem~\ref*{thm:dgm-expand}]
We will first define some notation to describe the process of graph-expanding, and then write down the joint probability mass function. Finally, we show the p.m.f. remains unchanged when swapping one variable, which suffices to prove the theorem by induction. 

\textbf{Part 1.} 
To streamline the notation, we redefine $Q$ as the number of steps that have actually been taken to expand the graph in Algorithm~\ref*{alg:block-dg} (rather than the input value).
%\ljmargin{Assume Algorithm~\ref*{alg:block-dg} has been run for $Q$ times.}{isn't this the definition of the algorithm?} 
For each $q\in [Q]$, denote by $G^{(q)}$ the augmented graph and by $B^{(q)}$ the blocking set used in the $q$th run of Algorithm~\ref*{alg:block-dg}. Let $V^{(q)}=[p]\setminus B^{(q)}$. Also denote  by $D^{(q)}$ the variables for which knockoffs have already been generated before the $q$th run of Algorithm~\ref*{alg:block-dg}. Then $D^{(r)}=\bigcup_{q=1}^{r-1} V^{(q)}$ for any $r\geq 2$ and $D^{(1)}=\emptyset$. 
Let $I_j^{(q)}$ be the neighborhood of $j$ in $G^{(q)}$. 
For ease of notation, we neglect to write the ranges of $\st_j$ and $\bst_A$ when enumerating them in equations (e.g., taking a product over all their possible values).
For a $n\times c$ matrix $\bZ$ and integers $k_1,\dots, k_c$, let $\mrc{\bZ,\st_1,\dots, \st_c}\; :=\; \sum_{i=1}^{n} \one{Z_{i,1}=\st_1,\dots,Z_{i,c}=\st_c}$, i.e., the number of rows of $\bZ$ that equal the vector $(\st_1,\dots,\st_c)$. 
%For simplicity, we assume $\bigcup\limits_{q=1}^{Q} V^{(q)} =[p]$. Other cases can be proved in a similar way. 

\textbf{Part 2.} 
For any $q\in [Q]$ and $j\in V^{(q)}$, the neighborhood $I_j^{(q)}$ consists of three parts:
\begin{enumerate}
    \item $([p]\cap I_{j}^{(q)})\setminus D^{(q)}$ : the neighbors in $[p]$ for which no knockoffs have been generated,
    \item $I_{j}^{(q)}\cap D^{(q)}$ : the neighbors in $[p]$ for which knockoffs have been generated, and
    \item $\{\tilde{\ell}: \ell \in I_{j}^{(q)}\cap D^{(q)} \}$ : the neighbors that are knockoffs. 
\end{enumerate} 
The generation of $\bXk_{j}$  by Algorithm~\ref*{alg:block-dg} is to sample uniformly from all vectors in $ [\St_j]^{n}$ such that the contingency table for variable $j$ and its neighbors in $I_{j}^{(q)}$ remains the same if $\bX_j$ is replaced by any of these vectors. Define  %\ljmargin{Define  }{don't want to use $\Omega$, as it looks like a precision matrix} 
\begin{align*}
&\mathcal{M}_{j} \left( \bsfX_{j}, \bsfX_{ ([p]\cap I_{j}^{(q)})\setminus D^{(q)}}, \bsfX_{I_{j}^{(q)}\cap  D^{(q)}}, \bsfXk_{I_{j}^{(q)}\cap D^{(q)}} \right)\\
=\; &
\left\{ \bsfW_{j} \in [\St_j]^{n}: 
	\mrc{ [\bsfW_{j}, \bsfX_{ ([p]\cap I_{j}^{(q)})\setminus D^{(q)}}, \bsfX_{I_{j}^{(q)}\cap  D^{(q)}}, \bsfXk_{I_{j}^{(q)}\cap D^{(q)}}], \st_{j}, \bst_{I_{j}^{(q)}}}= \right.\\
&\qquad\qquad\qquad\;\, 	\mrc{ [\bsfX_{j}, \bsfX_{ ([p]\cap I_{j}^{(q)})\setminus D^{(q)}}, \bsfX_{I_{j}^{(q)}\cap  D^{(q)}}, \bsfXk_{I_{j}^{(q)}\cap D^{(q)}}], \st_{j}, \bst_{I_{j}^{(q)}}} , \\
&\qquad\qquad\qquad\;\, 	\left.\forall \st_{j}, \bst_{I_{j}^{(q)}} \right\}, 
\end{align*}
and then
\begin{align*}
&\Pcr{ \bXk_{j}=\bsfXk_{j} }{ \bX=\bsfX,  \bXk_{V^{(1)}}=\bsfXk_{V^{(1)}},\dots,\bXk_{V^{(q-1)}}=\bsfXk_{V^{(q-1)}} } \\
=\;  &\frac
	{ 
	\one{\bsfXk_{j}\in \mathcal{M}_{j} \left( \bsfX_{j}, \bsfX_{ ([p]\cap I_{j}^{(q)})\setminus D^{(q)}}, \bsfX_{I_{j}^{(q)}\cap  D^{(q)}}, \bsfXk_{I_{j}^{(q)}\cap D^{(q)}} \right)}
	}
	{ \left|\mathcal{M}_{j} \left( \bsfX_{j}, \bsfX_{ ([p]\cap I_{j}^{(q)})\setminus D^{(q)}}, \bsfX_{I_{j}^{(q)}\cap  D^{(q)}}, \bsfXk_{I_{j}^{(q)}\cap D^{(q)}} \right)\right| }. 
\end{align*}

Denote the probability mass function of $X$ by $f(\bsfx)$.
The joint probability mass of the distribution of $[\bX, \bXk]$ is %By Algorithm~\ref{alg:dgm-expanding}, 
\begin{align}\label{eq:graph-expand-pmf}
	&\pr{[\bX, \bXk]=[\bsfX, \bsfXk]}  \\
=&\prod_{i=1}^{n}f(\bsfx_i)\nonumber\\
		& \times \Pcr{ \bXk_{V^{(1)}}=\bsfXk_{V^{(1)}} }{ \bX=\bsfX }\nonumber \\
		& \times \prod_{q=2}^{Q} \Pcr{ \bXk_{V^{(q)}}=\bsfXk_{V^{(q)}} }{ \bX=\bsfX,  \bXk_{V^{(1)}}=\bsfXk_{V^{(1)}},\dots,\bXk_{V^{(q-1)}}=\bsfXk_{V^{(q-1)}} }    \nonumber\\
	& \times \prod_{j\in [p]\setminus (\bigcup\limits_{q=1}^{Q} V^{(q)})} \one{\bsfXk_j=\bsfX_j}\nonumber\\
=& \prod_{i=1}^{n}f(\bsfx_i)  \nonumber\\
	& \times \prod_{q=1}^{Q}\prod_{j\in V^{(q)}} 
	\frac
	{ 
	\one{\bsfXk_{j}\in \mathcal{M}_{j} \left( \bsfX_{j}, \bsfX_{ ([p]\cap I_{j}^{(q)})\setminus D^{(q)}}, \bsfX_{I_{j}^{(q)}\cap  D^{(q)}}, \bsfXk_{I_{j}^{(q)}\cap D^{(q)}} \right)}
	}
	{ |\mathcal{M}_{j} \left( \bsfX_{j}, \bsfX_{ ([p]\cap I_{j}^{(q)})\setminus D^{(q)}}, \bsfX_{I_{j}^{(q)}\cap  D^{(q)}}, \bsfXk_{I_{j}^{(q)}\cap D^{(q)}} \right)| } \nonumber\\
	& \times \prod_{j\in [p]\setminus (\bigcup\limits_{q=1}^{Q} V^{(q)})} \one{\bsfXk_j=\bsfX_j},\nonumber
\end{align}
where the product is partitioned into three parts: the distribution of $\bX$, the distributions of the knockoff columns generated in each step and the indicator functions for the variables that have no knockoffs generated within the $Q$ steps. 

\textbf{Part 3}. 
It suffices to show that for any $\ell \in [p]$,  
\begin{equation}\label{eq:graph-expand-goal}
\pr{[\bX, \bXk]=[\bsfX, \bsfXk]}=\pr{[\bX, \bXk]=[\bsfX, \bsfXk]_{\swap{\ell}}}.
\end{equation}
If both sides of Equation~\eqref{eq:graph-expand-goal} equal zero, it holds trivially. Without loss of generality, we will prove this equation under the assumption that the left hand side is non-zero. One can redefine $[\bsfX',\bsfXk']=[\bsfX,\bsfXk]_{\swap{\ell}}$ and apply the same proof when assuming the right hand side is non-zero. 

First, suppose $\ell \in  [p]\setminus (\bigcup\limits_{q=1}^{Q} V^{(q)})$. 
%if the left hand side is non-zero, then 
Since the left hand side is non-zero, 
by Equation~\eqref{eq:graph-expand-pmf}, $\bsfXk_{\ell}=\bsfX_{\ell}$ and the p.m.f. does not change when swapping $\bsfX_{\ell}$ with $\bsfXk_{\ell}$. 
%\ljmargin{The same reasoning applies when the right hand side is non-zero.}{needed?}

Second, suppose $\ell \in V^{(q_{\ell})}$ for some $q_{\ell}\in [Q]$. Since the left hand side of Equation~\eqref{eq:graph-expand-goal} is non-zero, then in Equation~\eqref{eq:graph-expand-pmf}, the indicator function in the second part with $q=q_{\ell}$ and $j=\ell$ being non-zero indicates
\begin{align}
&\mrc{ [\bsfXk_{\ell}, \bsfX_{ ([p]\cap I_{\ell}^{(q_{\ell})})\setminus D^{(q_{\ell})}}, \bsfX_{I_{\ell}^{(q_{\ell})}\cap  D^{(q_{\ell})}}, \bsfXk_{I_{\ell}^{(q_{\ell})}\cap D^{(q_{\ell})}}], \st_{\ell}, \bst_{I_{\ell}^{(q_{\ell})}}} \nonumber\\
	=\;&\mrc{ [\bsfX_{\ell}, \bsfX_{ ([p]\cap I_{\ell}^{(q_{\ell})})\setminus D^{(q_{\ell})}}, \bsfX_{I_{\ell}^{(q_{\ell})}\cap  D^{(q_{\ell})}}, \bsfXk_{I_{\ell}^{(q_{\ell})}\cap D^{(q_{\ell})}}], \st_{\ell}, \bst_{I_{\ell}^{(q_{\ell})}}} , \forall \st_{\ell}, \bst_{I_{\ell}^{(q_{\ell})}}. \label{graph-expand-cond-1}
\end{align}
The only difference between the two sides of Equation~\ref{graph-expand-cond-1} is that  $\bsfXk_{\ell}$ is replaced by $\bsfX_{\ell}$ in the first columns of the matrix. Such a  difference will keep appearing in the equations that will be showed below. %Equation~\ref{graph-expand-cond-1} 

In the following, we show that everywhere $\bsfX_{\ell}$ or $\bsfXk_{\ell}$ appears in the first or second part of the product in Equation~\eqref{eq:graph-expand-pmf} remains unchanged when swapping $\bsfX_{\ell}$ and $\bsfXk_{\ell}$. 

\begin{enumerate}
    \item 
As shown in Section~\ref*{sec:eg-dgm}, there exist some functions $\psi_{\ell}$ and $\theta_\ell(\st_\ell, \bst_{I_\ell})$'s such that 
\begin{equation*}
\prod_{i=1}^{n}f(\bsfx_i)=\prod_{i=1}^{n}\psi_{\ell}(\bsfx_{i, [p]\setminus  \{\ell\}})\prod_{\st_\ell, \bst_{I_\ell}} \theta_\ell(\st_\ell, \bst_{I_\ell})^{\mrc{[\bsfX_\ell, \bsfX_{I_\ell}], \st_\ell, \bst_{I_\ell}} }
\end{equation*}
Note that the initial neighborhood $I_{\ell}\subseteq I_{\ell}^{(q_{\ell})}$, by summing over Equation~\eqref{graph-expand-cond-1}, one can conclude that \[
\mrc{[\bsfX_\ell, \bsfX_{I_\ell}], \st_\ell, \bst_{I_\ell}} =\mrc{[\bsfXk_\ell, \bsfX_{I_\ell}], \st_\ell, \bst_{I_\ell}}\]
for all $\st_{\ell}, \bst_{I_{\ell}}$. 
Thus $\prod_{i=1}^{n}f(\bsfx_i)$ remains unchanged when swapping $\bsfX_\ell$ and $\bsfXk_\ell$. 
\item 
For any $j$ such that $j\in V^{(q_{j})}$ and $\ell\in I_j^{(q_{j})}$, the second part of the product in Equation~\ref{eq:graph-expand-pmf} involves $\bsfX_{\ell}$ or $\bsfXk_{\ell}$ with the indices $q_{j}$ and $j$. Since $B^{(q_{j})}$ is a blocking set, we have $q_{j}\neq q_{\ell}$. 
\begin{enumerate}
    \item 
If $q_{j}>q_{\ell}$, then $\ell \in I_{j}^{(q_{j})}\cap  D^{(q_{j})}$. Note that swapping $\bsfX_{\ell}$ with $\bsfXk_{\ell}$ only changes the order of the dimensions of the contingency table formed by 
$\bsfX_{j}$ and $\left[\bsfX_{[p]\cap I_{j}^{(q_{j})}}, \, \bsfXk_{ I_{j}^{(q_{j})} \cap D^{(q_{j})} }\right]$,
and thus does not change whether or not the following system of equations (where only the first column of the first argument of $\varrho$ differs between the two lines) holds
\begin{align*}
&\mrc{ [\bsfW_{j}, \bsfX_{ ([p]\cap I_{j}^{(q_{j})})\setminus D^{(q_{j})}}, \bsfX_{I_{j}^{(q_{j})}\cap  D^{(q_{j})}}, \bsfXk_{I_{j}^{(q_{j})}\cap D^{(q_{j})}}], \st_{j}, \bst_{I_{j}^{(q_{j})}}} \\
= &\mrc{ [\bsfX_{j},\; \bsfX_{ ([p]\cap I_{j}^{(q_{j})})\setminus D^{(q_{j})}}, \bsfX_{I_{j}^{(q_{j})}\cap  D^{(q_{j})}}, \bsfXk_{I_{j}^{(q_{j})}\cap D^{(q_{j})}}], \st_{j}, \bst_{I_{j}^{(q_{j})}}},\\
&\forall \st_{j},\bst_{I_{j}^{(q_{j})}}.
\end{align*}
Thus the indicator function 
\[\one{\bsfXk_{j}\in \mathcal{M}_{j} \left( \bsfX_{j}, \bsfX_{ ([p]\cap I_{j}^{(q_{j})})\setminus D^{(q_{j})}}, \bsfX_{I_{j}^{(q_{j})}\cap  D^{(q_{j})}}, \bsfXk_{I_{j}^{(q_{j})}\cap D^{(q_{j})}} \right)} \]
and the cardinal number 
\[\left|\mathcal{M}_{j} \left( \bsfX_{j}, \bsfX_{ ([p]\cap I_{j}^{(q_{j})})\setminus D^{(q_{j})}}, \bsfX_{I_{j}^{(q_{j})}\cap  D^{(q_{j})}}, \bsfXk_{I_{j}^{(q_{j})}\cap D^{(q_{j})}} \right)\right|\] from Equation~\eqref{eq:graph-expand-pmf} remain unchanged when swapping $\bsfX_\ell$ and $\bsfXk_\ell$. 
\item

Now we show the same conclusion for the remaining $j$ values, which will require a few intermediate steps. If $q_{j}<q_{\ell}$, then $\ell \in I_j^{(q_{j})}\setminus D^{(q_{j})}$ and $j\in I_{\ell}^{(q_{\ell})}\cap D^{(q_{\ell})}$. By the graph expanding algorithm, we have $I_j^{(q_{j})}\setminus \{\ell\} \subseteq I_{\ell}^{(q_{\ell})}$, and $D^{(q_{j})}\subseteq D^{(q_{\ell})}$. This shows 
\[
([p]\cap I_{j}^{(q_{j})})\setminus (\{\ell\}\cup D^{(q_{j})}) \subseteq ([p]\cap I_{\ell}^{(q_{\ell})}\setminus D^{(q_{\ell})}) \uplus ( I_{\ell}^{(q_{\ell})}\cap D^{(q_{\ell})}),
\]
and 
\[
I_{j}^{(q_{j})}\cap  D^{(q_{j})} \subseteq I_{\ell}^{(q_{\ell})}\cap  D^{(q_{\ell})}.
\]

Summing over Equation~\eqref{graph-expand-cond-1} and rearranging the columns of the first argument of $\varrho$, one can conclude that
\begin{align}
&\; \mrc{[ \bsfX_{j},\bsfXk_{j}, \bsfXk_{\ell}, \bsfX_{ ([p]\cap I_{j}^{(q_{j})})\setminus (\{\ell\}\cup D^{(q_{j})}) }, \bsfX_{I_{j}^{(q_{j})}\cap  D^{(q_{j})}}, \bsfXk_{I_{j}^{(q_{j})}\cap D^{(q_{j})}}], \st_{j},\st_{\tilde{j}}, \st_{\ell}, \bst_{I_{j}^{(q_{j})}\setminus\{\ell\}}} \label{eq:graph-expand-smallq-0}  \\
=&\; \mrc{[ \bsfX_{j},\bsfXk_{j}, \bsfX_{\ell}, \bsfX_{ ([p]\cap I_{j}^{(q_{j})})\setminus (\{\ell\}\cup D^{(q_{j})}) }, \bsfX_{I_{j}^{(q_{j})}\cap  D^{(q_{j})}}, \bsfXk_{I_{j}^{(q_{j})}\cap D^{(q_{j})}}], \st_{j},\st_{\tilde{j}}, \st_{\ell}, \bst_{I_{j}^{(q_{j})}\setminus\{\ell\}}},\nonumber \\
&\qquad\qquad\qquad\qquad\qquad\qquad\qquad\qquad\qquad\qquad\qquad\qquad\qquad\qquad\quad\forall  \st_{j},\st_{\tilde{j}},\st_{\ell}, \bst_{I_{j}^{(q_{j})}\setminus\{\ell\}},  \nonumber
\end{align}
where the two lines only differ in the third column of the first argument of $\varrho$. 
Summing over Equation~\eqref{eq:graph-expand-smallq-0} 
%in two different ways, w.r.t. {\sout{$\bsfX_{j}$ and w.r.t. $\bsfXk_{j}$} $\st_{\tilde{j}}$ and w.r.t. $\st_{j}$} respectively, 
w.r.t. $\st_{\tilde{j}}$, 
one can further conclude that 
\begin{align}
&\; \mrc{[ \bsfX_{j},\bsfXk_{\ell}, \bsfX_{ ([p]\cap I_{j}^{(q_{j})})\setminus (\{\ell\}\cup D^{(q_{j})}) }, \bsfX_{I_{j}^{(q_{j})}\cap  D^{(q_{j})}}, \bsfXk_{I_{j}^{(q_{j})}\cap D^{(q_{j})}}], \st_{j}, \st_{\ell}, \bst_{I_{j}^{(q_{j})}\setminus\{\ell\}}}   \label{eq:graph-expand-smallq-1}\\
=&\; \mrc{[ \bsfX_{j}, \bsfX_{\ell}, \bsfX_{ ([p]\cap I_{j}^{(q_{j})})\setminus (\{\ell\}\cup D^{(q_{j})}) }, \bsfX_{I_{j}^{(q_{j})}\cap  D^{(q_{j})}}, \bsfXk_{I_{j}^{(q_{j})}\cap D^{(q_{j})}}], \st_{j},\st_{\ell}, \bst_{I_{j}^{(q_{j})}\setminus\{\ell\}}}, \nonumber \\
& \forall  \st_{j},\st_{\ell}, \bst_{I_{j}^{(q_{j})}\setminus\{\ell\}}, \nonumber 
\end{align}
where the two lines only differ in the second column of the first argument of $\varrho$, and the first column is $\bsfX_j$. 
%Further,
Similarly, summing over Equation~\eqref{eq:graph-expand-smallq-0}  w.r.t. $\st_{j}$, we have
\begin{align}
&\; \mrc{[ \bsfXk_{j}, \bsfXk_{\ell}, \bsfX_{ ([p]\cap I_{j}^{(q_{j})})\setminus (\{\ell\}\cup D^{(q_{j})}) }, \bsfX_{I_{j}^{(q_{j})}\cap  D^{(q_{j})}}, \bsfXk_{I_{j}^{(q_{j})}\cap D^{(q_{j})}}], \st_{\tilde{j}}, \st_{\ell}, \bst_{I_{j}^{(q_{j})}\setminus\{\ell\}}}   \label{eq:graph-expand-smallq-2}  \\
=&\; \mrc{[\bsfXk_{j}, \bsfX_{\ell}, \bsfX_{ ([p]\cap I_{j}^{(q_{j})})\setminus (\{\ell\}\cup D^{(q_{j})}) }, \bsfX_{I_{j}^{(q_{j})}\cap  D^{(q_{j})}}, \bsfXk_{I_{j}^{(q_{j})}\cap D^{(q_{j})}}], \st_{\tilde{j}}, \st_{\ell}, \bst_{I_{j}^{(q_{j})}\setminus\{\ell\}}},\nonumber \\ & \forall  \st_{\tilde{j}},\st_{\ell}, \bst_{I_{j}^{(q_{j})}\setminus\{\ell\}},\nonumber
%&\; \Phi([ \bsfXk_{j}, \bsfXk_{\ell}, \bsfX_{ ([p]\cap I_{j}^{(q_{j})})\setminus (\{\ell\}\cup D^{(q_{j})}) }, \bsfX_{I_{j}^{(q_{j})}\cap  D^{(q_{j})}}, \bsfXk_{I_{j}^{(q_{j})}\cap D^{(q_{j})}}], \st_{\tilde{j}}, \bst_{I_{j}^{(q_{j})}})  \nonumber  \\
%=&\; \Phi([ \bsfXk_{j}, \bsfX_{\ell}, \bsfX_{ ([p]\cap I_{j}^{(q_{j})})\setminus (\{\ell\}\cup D^{(q_{j})}) }, \bsfX_{I_{j}^{(q_{j})}\cap  D^{(q_{j})}}, \bsfXk_{I_{j}^{(q_{j})}\cap D^{(q_{j})}}], \st_{\tilde{j}}, \bst_{I_{j}^{(q_{j})}}), \forall  \st_{\tilde{j}}, \bst_{I_{j}^{(q_{j})}},
\end{align}
where again the two lines only differ in the second column of the first argument of $\varrho$, but now the first column is $\bsfXk_j$. 

Note that $\left|\mathcal{M}_{j} \left( \bsfX_{j}, \bsfX_{ ([p]\cap I_{j}^{(q_{j})})\setminus D^{(q_{j})}}, \bsfX_{I_{j}^{(q_{j})}\cap  D^{(q_{j})}}, \bsfXk_{I_{j}^{(q_{j})}\cap D^{(q_{j})}} \right)\right|$ is a product of some multinomial coefficients, and each  multinomial coefficient depends on a unique combination of $(\st_{\ell}, \bst_{I_{j}^{(q_{j})}\setminus\{\ell\}})$ and  the values of 
\[
\mrc{[ \bsfX_{j},\bsfX_{\ell}, \bsfX_{ ([p]\cap I_{j}^{(q_{j})})\setminus (\{\ell\}\cup D^{(q_{j})}) }, \bsfX_{I_{j}^{(q_{j})}\cap  D^{(q_{j})}}, \bsfXk_{I_{j}^{(q_{j})}\cap D^{(q_{j})}}], \st_{j}, \st_{\ell}, \bst_{I_{j}^{(q_{j})}\setminus\{\ell\}}}, \forall \st_{j}. 
\] 
These quantities are the ones on the right hand side of Equation~\eqref{eq:graph-expand-smallq-1}.  Thus we conclude that  $\left|\mathcal{M}_{j} \left( \bsfX_{j}, \bsfX_{ ([p]\cap I_{j}^{(q_{j})})\setminus D^{(q_{j})}}, \bsfX_{I_{j}^{(q_{j})}\cap  D^{(q_{j})}}, \bsfXk_{I_{j}^{(q_{j})}\cap D^{(q_{j})}} \right)\right|$ remains unchanged when swapping $\bsfX_\ell$ and $\bsfXk_\ell$ by checking the terms in Equation~\eqref{eq:graph-expand-smallq-1} that appear in the multinomial coefficients. 
%Equation~\eqref{eq:graph-expand-smallq-1} implies $\left|\mathcal{M}_{j} \left( \bsfX_{j}, \bsfX_{ ([p]\cap I_{j}^{(q_{j})})\setminus D^{(q_{j})}}, \bsfX_{I_{j}^{(q_{j})}\cap  D^{(q_{j})}}, \bsfXk_{I_{j}^{(q_{j})}\cap D^{(q_{j})}} \right)\right|$ remains unchanged when swapping $\bsfX_\ell$ and $\bsfXk_\ell$ by \ljmargin{checking the terms appear in the multinomial coefficients.}{?} 

Note that $\bsfXk_{j}\in \mathcal{M}_{j} \left( \bsfX_{j}, \bsfX_{ ([p]\cap I_{j}^{(q_{j})})\setminus D^{(q_{j})}}, \bsfX_{I_{j}^{(q_{j})}\cap  D^{(q_{j})}}, \bsfXk_{I_{j}^{(q_{j})}\cap D^{(q_{j})}} \right)$ if and only if the right hand sides of Equations~\eqref{eq:graph-expand-smallq-1} and \eqref{eq:graph-expand-smallq-2} are equal for all $\st_{\tilde{j}},\st_{\ell}, \bst_{I_{j}^{(q_{j})}\setminus\{\ell\}}$, which is equivalent to the left hand sides of the equations are equal, which holds if and only if $\bsfXk_{j}\in \mathcal{M}_{j} \left( \bsfX_{j},\bsfXk_{\ell},  \bsfX_{ ([p]\cap I_{j}^{(q_{j})})\setminus (\{\ell\}\cup D^{(q_{j})})}, \bsfX_{I_{j}^{(q_{j})}\cap  D^{(q_{j})}}, \bsfXk_{I_{j}^{(q_{j})}\cap D^{(q_{j})}} \right)$. Therefore the indicator function $\one{\bsfXk_{j}\in \mathcal{M}_{j} \left( \bsfX_{j}, \bsfX_{ ([p]\cap I_{j}^{(q_{j})})\setminus D^{(q_{j})}}, \bsfX_{I_{j}^{(q_{j})}\cap  D^{(q_{j})}}, \bsfXk_{I_{j}^{(q_{j})}\cap D^{(q_{j})}} \right)} $  remains unchanged when swapping $\bsfX_\ell$ and $\bsfXk_\ell$.

%\[
%\one{\bsfXk_{j}\in \mathcal{M}_{j} \left( \bsfX_{j}, \bsfX_{ ([p]\cap I_{j}^{(q_{j})})\setminus D^{(q_{j})}}, \bsfX_{I_{j}^{(q_{j})}\cap  D^{(q_{j})}}, \bsfXk_{I_{j}^{(q_{j})}\cap D^{(q_{j})}} \right)}= \one{\bsfXk_{j}\in \mathcal{M}_{j} \left( \bsfX_{j},\bsfXk_{\ell},  \bsfX_{ ([p]\cap I_{j}^{(q_{j})})\setminus (\{\ell\}\cup D^{(q_{j})})}, \bsfX_{I_{j}^{(q_{j})}\cap  D^{(q_{j})}}, \bsfXk_{I_{j}^{(q_{j})}\cap D^{(q_{j})}} \right)}
%\]
%Equations~\eqref{eq:graph-expand-smallq-1} and \eqref{eq:graph-expand-smallq-2} together imply that $\one{\bsfXk_{j}\in \mathcal{M}_{j} \left( \bsfX_{j}, \bsfX_{ ([p]\cap I_{j}^{(q_{j})})\setminus D^{(q_{j})}}, \bsfX_{I_{j}^{(q_{j})}\cap  D^{(q_{j})}}, \bsfXk_{I_{j}^{(q_{j})}\cap D^{(q_{j})}} \right)} $  remains unchanged when swapping $\bsfX_\ell$ and $\bsfXk_\ell$.

\end{enumerate}
\end{enumerate}
%To sum up, when the left hand side of Equation~\eqref{eq:graph-expand-goal} is non-zero, Equation~\eqref{eq:graph-expand-pmf} remains unchanged with swapping $\bsfX_{\ell}$ and $\bsfXk_{\ell}$.  When the right hand side is non-zero, one can redefine $[\bsfX',\bsfXk']=[\bsfX,\bsfXk]_{\swap{\ell}}$ and apply the same reasoning above. 

To sum up, Equation~\eqref{eq:graph-expand-goal} holds for any $\ell \in [p]$, and the proof is complete. 
\end{proof}

\subsection{Alternative Knockoff Generations for Discrete Markov Chains}\label{app:mc}
We provide alternative constructions of conditional knockoffs for Markov chains that make use of the simple chain structure. Proofs in this section are deferred to Appendix~\ref{app:proof-mc}.

We introduce a notation that makes the display clear.  % or recall at Appendix~\ref{app:graph-expanding}
For any $n\times c$ matrix $\bZ$  and any integers $k_1,\dots, k_c$, let $\mrc{\bZ,\st_1,\dots, \st_c}\; :=\; \sum_{i=1}^{n} \one{Z_{i,1}=\st_1,\dots,Z_{i,c}=\st_c}$, i.e., the number of rows of $\bZ$ that equal the vector $(\st_1,\dots,\st_c)$. 

Suppose the components of $X$ follow a general discrete Markov chain, then the joint distribution for $n$ i.i.d. samples is
\begin{align*}
\pr{\bX} 
&= \; \prod_{\st=1}^{ \St_{1}}(\pi_{  \st}^{(1)})^{\sum_{ \st' =1}^{ \St_{2}}N_{\st,\st' }^{(2)}} \prod_{j=2}^p \prod_{\st =1}^{\St_{j-1}}\prod_{ \st' =1}^{\St_{j}}(\pi_{\st ,\st' }^{(j)})^{N_{\st ,\st' }^{(j)}},
\end{align*}
where $\pi^{(1)}_{\st} = \pr{X_{i,1}=\st}$ and $\pi^{(j)}_{\st_{},\st'} = \Pc{X_{i,j}=\st'}{X_{i,j-1}=\st_{}}$ are model parameters and $N^{(j)}_{\st_{},\st'} 
= \mrc{ [\bX_{j-1},\bX_j], \st_{},\st' } $ 
is the number of samples such that the $(j-1)$th and $j$th components are $\st_{}$ and $\st'$, respectively. So all the $N^{(j)}_{\st_{},\st'}$'s together form a sufficient statistic which we denote by $T(\bX)$, and although it has some redundant entries (for example, $\sum_{\st_{}=1}^{\St_{j-1}}\sum_{ \st'=1}^{\St_{j}}N^{(j)}_{\st_{},\st'} = n$), it is nevertheless minimal, and we prefer to keep the redundant entries for notational convenience.  % minimal but not complete!

Conditional on $T(\bX)$, $\bX$ is uniformly distributed on $\mathcal{Q}\, :=\, \{\bsfW\in \prod_{j=1}^{p} [ \St_j ]^{n}: T(\bsfW)=T(\bX)\}$. Hereafter, we distinguish notationally between $\bsfX$'s and $\bsfW$'s (and $\bsfWk$'s), with the former denoting realized values of the data in $\bX$ and the latter denoting hypothetical such values not necessarily observed in the data.
The conditional distribution of $\bX$ can be decomposed as
\begin{align}\label{eq:mc-condp} \Pcr{\bX=\bsfW}{T(\bX)} &= \; C_0 \prod_{j=2}^{p}\left( \prod_{\st_{}=1}^{ \St_{j-1}}\prod_{\st'=1}^{\St_{j} }\one{ \mrc{ [\bsfW_{j-1},\bsfW_j], \st,\st' }=N_{\st_{},\st'}^{(j)}}  \right) \nonumber \\
& =\; C_0 \prod_{j=2}^{p}\phi_{j}( \bsfW_{j-1},\bsfW_{j} \mid T(\bX) ),
\end{align}
where $C_0$ only depends on $T(\bX)$. This decomposition implies that conditional on $T(\bX)$, the columns of $\bX$ still comprise a vector-valued Markov chain (see Appendix~\ref{app:proof-mc}). For ease of notation, in what follows we will write $\pr{\cdot}$ without \textquoteleft $\mid T(\bX)$\textquoteright\ since we always condition on $T(\bX)$ in this section.

\subsubsection{SCIP}
The sequential conditional independent pairs (SCIP) algorithm from \cite{EC-ea:2018} was introduced in completely general form for any distribution for X with the substantial caveat that actually carrying it out for any given distribution can be quite challenging. \cite{MS-CS-EC:2017} show how to run SCIP for Markov chains \emph{unconditionally}. When applied to vectors instead of scalars, SCIP can also be adapted to generate \emph{conditional} knockoffs for Markov chains because the conditional distribution of $\bX$ is uniform on $\mathcal{Q}$, making it a Markov chain, and conditional knockoffs are simply knockoffs for this conditional distribution.

%Therefore, the same proof that SCIP works unconditionally also establishes its conditional validity, and the only question that remains is computation.

%{still keep this part because it is slightly different from Sesia's formula} 
SCIP sequentially samples $\bXk_j \sim \mathcal{L}(\bX_{j} |\bX_{\noj}, \bXk_{1:(j-1)})$\footnote{$\mathcal{L}(\bX_{j} |\bX_{\noj}, \bXk_{1:(j-1)})$ is the conditional distribution of $\bX_{j}$ given $(\bX_{\noj}, \bXk_{1:(j-1)})$.}  independently of $\bX_j$, for $j=1,\dots,p$. For a Markov chain, this sampling can be reduced to $\bXk_j \sim \mathcal{L}(\bX_{j} |\bX_{j-1}, \bX_{j+1}, \bXk_{j-1})$. 
The main computational challenge is to keep track of the following conditional probabilities:
\begin{align*}
	f_1( \bsfW_1, \bsfW_2) :=\; &\Pcr{ \bX_1=\bsfW_1 }{ \bX_2=\bsfW_{2} },\\
	f_j( \bsfW_j, \bsfWk_{j-1},\bsfW_{j+1}) :=\; &\Pcr{ \bX_j=\bsfW_j }{ \bX_{j-1}=\bsfX_{j-1},\bXk_{j-1}=\bsfWk_{j-1},\bX_{j+1}=\bsfW_{j+1} }, \\
&  \hspace{200pt minus 1fil} \forall \; j\in \{2,\dots,p-1\},   \hfilneg  \\
	f_p( \bsfW_p, \bsfWk_{p-1}) :=\; &\Pcr{ \bX_{p}=\bsfW_p }{\bX_{p-1}=\bsfX_{p-1}, \bXk_{p-1}=\bsfWk_{p-1} }.
\end{align*}
Algorithm~\ref{alg:scip-mc} describes how to generate conditional knockoffs for a discrete Markov Chain with finite states by SCIP, where the functions $f_j(\cdot)$ are computed recursively by the formulas in Proposition~\ref{prop:mc-scip-f}. These formulas are different from the ones in \citet[Proposition 1]{MS-CS-EC:2017}, in which the authors assume transition probabilities can be evaluated directly.  
% find out why different

\begin{algorithm}
\caption{Conditional Knockoffs for Discrete Markov Chains by SCIP} \label{alg:scip-mc}
\begin{algorithmic}[1]
\ENSURE $\bX\in \prod_{j=1}^{p} [\St_j]^{n}$.
\STATE Sample $\bXk_1$ uniformly from $\{ \bsfW_1 \in [\St_1]^n: (\bsfW_1, \bX_2,\dots, \bX_p)\in \mathcal{Q}\}$.    %for all $\bsfW_{2}$ such that $(\bX_1, \bsfW_{2}, \bX_3,\dots, \bX_p)\in \mathcal{Q}\} $. 
\STATE Compute $f_1( \bXk_1, \bsfW_2)$ for all $\bsfW_2\in [\St_2]^n$.
\FOR{$j$ from $2$ to $p-1$}
\STATE  Compute $f_j( \bsfW_j, \bXk_{j-1},\bsfW_{j+1})$ for all $\bsfW_j\in [\St_j]^n$ and $\bsfW_{j+1}\in [\St_{j+1}]^n$.
\STATE  Sample $\bXk_j$ from $f_j(\cdot, \bXk_{j-1}, \bX_{j+1})$.
\ENDFOR
\STATE  Compute $f_p( \bsfW_j, \bXk_{p-1})$ for all $\bsfW_p\in [\St_p]^n$.
\STATE Sample $\bXk_p$ from $f_p(\cdot, \bXk_{p-1})$.
\RETURN $\bXk = [\bXk_1,\dots,\bXk_p]$.
\end{algorithmic}
\end{algorithm}

\begin{proposition}\label{prop:mc-scip-f}
% [said before ]Denote by $\bsfX$ the realization of $\bX$. 
Define $\frac{0}{0}=0$.
We formally write $f_1( \bsfW_1, \bsfXk_{0}, \bsfW_2)$ for $f_1( \bsfW_1, \bsfW_2)$. Suppose $\bsfXk$ is a realization of $\bXk$ generated by Algorithm 2. Then the following equations hold
%Define \[\phi_{j}( \bsfW_{j-1},\bsfW_{j})=\prod_{\st_j\in [\St_j]}\prod_{ \st_{j+1}\in [\St_{j+1}] } \one{ \mrc{ [\bsfW_{j-1},\bsfW_j], \st,\st' }}\]
\[f_1( \bsfW_1, \bsfW_2)=\frac{\one{(\bsfW_{1}, \bsfW_{2}, \bsfX_3,\dots, \bsfX_p)\in \mathcal{Q}}}{ \# \{ \bsfW_1^{\prime}\in [\St_1]^{n}: (\bsfW_1^{\prime}, \bsfW_{2}, \bsfX_3,\dots, \bsfX_p)\in \mathcal{Q}\} },\]
\[
\forall \; 1<j<p, \; f_j( \bsfW_j, \bsfXk_{j-1},\bsfW_{j+1})=\frac{\one{(\bsfX_{1},\dots, \bsfX_{j-1}, \bsfW_j,\bsfW_{j+1},\dots, \bsfX_p)\in \mathcal{Q}}  f_{j-1}(\bsfXk_{j-1}, \bsfXk_{j-2},\bsfW_{j}) }
{ \sum_{ \bsfW_j^{\prime}\in [\St_j]^{n} }\one{(\bsfX_{1},\dots, \bsfX_{j-1}, \bsfW_j^{\prime},\bsfW_{j+1},\dots, \bsfX_p)\in \mathcal{Q}} f_{j-1}(\bsfXk_{j-1}, \bsfXk_{j-2},\bsfW_{j}^{\prime})},
\]
\[
f_p( \bsfW_p, \bsfXk_{p-1})=\frac{\one{(\bsfX_{1},\dots, \bsfX_{p-1}, \bsfW_p)\in \mathcal{Q}}  f_{p-1}(\bsfXk_{p-1}, \bsfXk_{p-2},\bsfW_{p}) }
{ \sum_{ \bsfW_{p}^{\prime}\in [\St_p]^{n}  } \one{(\bsfX_{1},\dots, \bsfX_{p-1}, \bsfW_{p}^{\prime})\in \mathcal{Q}}  f_{p-1}(\bsfXk_{p-1}, \bsfXk_{p-2},\bsfW_{p}^{\prime}) }.
\]
\end{proposition}

Computing $f_j( \bsfW_j, \bsfXk_{j-1},\bsfW_{j+1})$ in Proposition~\ref{prop:mc-scip-f} requires enumerating all possible configurations of $\bsfW_j\in [\St_j]^{n}$ and $\bsfW_{j+1}\in [\St_{j+1}]^{n}$, making the total computational complexity of SCIP $O(\sum_{j\leq p-1} (\St_{j}\St_{j+1})^n)$. Due to the $n$ in the exponent, SCIP quickly becomes intractable as the sample size grows, even for binary states and $n\gtrsim 10$. A simple remedy is to first randomly divide the rows of $\bX$ into disjoint folds of small size around $n_0$, say $n_0=10$, and then run SCIP for each fold separately. %The knockoff construction in each fold is conditional on a separate sufficient statistic, thus the whole construction is conditional on more statistics than the one-time construction. 
This construction conditions on a statistic which is $n/n_0$ times as large as that conditioned on before dividing into folds, but the former's computation time scales linearly with $n$, instead of exponentially. 
Conditioning on more should tend to degrade the quality of the knockoffs, but is necessary to enable computation. Still, compared to Algorithm~\ref{alg:datasplitting-discrete}, SCIP does not block any variables and thus has the potential to generate better knockoffs. 
%In fact when unlabeled data is available, this lack of blocking is particularly beneficial.
%However, empirically the improvement by using unlabeled data is not significant because the loss of efficiency is mainly caused by $n_0$ being small. LJ: especially since we don't show these simulations, I don't think it's worth including this.

\subsubsection{Refined Blocking}

One can apply Algorithm~\ref{alg:datasplitting-discrete} to Markov Chains, as a 2-colorable graph,  with two blocking sets, one with all even numbers and the other with all odd numbers. Instead of running Algorithm~\ref{alg:block-dg} in Line~\ref{line:dgm-sub},  a refined blocking algorithm, Algorithm~\ref{alg:mh-mc-simple}, can be used for $1<j<p$. 
%The cases of $j=1$ and $p$ are the same as in Algorithm~\ref{alg:block-dg}. For 
It introduces more variability in the knockoff generation because it first draws a new contingency table $\tilde{\ct}$ that is conditionally exchangeable with the observed contingency table $\ct$ of $(\bX_{j-1},\bX_{j},\bX_{j+1})$, and then samples $\bXk_j$ given $\tilde{\ct}$. This algorithm constructs a reversible Markov Chain by proposing random walks on the space of contingency tables, moved by $\Delta \ct$ and corrected by acceptance ratio $\alpha$. In the following, we discuss how to sample $\Delta \ct$ and compute $\alpha$, and provide a detailed version of Algorithm~\ref{alg:mh-mc-simple} at the end of this section.

\begin{algorithm}[]
\caption{Improved Conditional Knockoffs for $\bX_j$ in a Discrete Markov Chain (Pseudocode)} \label{alg:mh-mc-simple}
\begin{algorithmic}[1]
\ENSURE $\bX_j$, columns to condition on: $\bX_{j-1}$ and $\bX_{j+1}$.
\STATE Initialize a chain of contingency table $\tilde{\ct}= \ct(\bX_{j})$.
\FOR{$t$ from $1$ to $t_{\max}$ }
\STATE Draw $\Delta \ct$ independent of $\tilde{\ct}$, and propose $\tilde{\ct}^{*}=\tilde{\ct}+\Delta \ct$.
\IF{all elements of $\tilde{\ct}^{*}$ are nonnegative}
\STATE  Calculate the Metropolis--Hastings acceptance ratio $\alpha$.
\STATE  With probability $\alpha$, update $\tilde{\ct} \leftarrow \tilde{\ct}^{*}$.
\ENDIF
\ENDFOR
\STATE Sample $\bXk_j$ uniformly on all possible vectors that match $\tilde{\ct}$.
\RETURN $\bXk_{j}$.
\end{algorithmic}
%\hline 
\end{algorithm}

Conditional on $T(\bX)$ and $\bX_B$, $\bX_j$ is uniformly distributed on all $\bsfW_{j} \in [ \St_j]^{n}$ such that 
\[
\mrc{[\bX_{j-1},\bsfW_{j}], \st_{j-1}, \st_j}=N^{(j)}_{\st_{j-1}, \st_{j}} \mbox{ and }\mrc{[\bsfW_{j},\bX_{j+1}], \st_j, \st_{j+1} }=N^{(j+1)}_{\st_{j}, \st_{j+1}}
\]
for all $(\st_{j-1}, \st_j, \st_{j+1})$. 
In the following, we view $T(\bX)$ and $\bX_B$ as fixed and only $\bX_j$ as being random. 

We begin with some notation. To avoid burdensome subscripts, we write $(\st_{-},\st_{}, \st_{+})$  for $(\st_{j-1},\st_{j}, \st_{j+1})$.
% (resp. $(\St_{-},\St_{}, \St_{+})$ ) (resp. $(\St_{j-1},\St_{j}, \St_{j+1})$). 
%Denote by  $\cell \in  \Z^{ \St_{j-1} \times \St_{j} \times \St_{j+1 } }$ the three dimensional array, whose entries are $\cell_{\st_{-},\st_{}, \st_{+}}$ for $\st_{-},\st_{}, \st_{+}$ 
%$\cell[\st_{-},\st_{}, \st_{+}]$ 
%the element at the $(\st_{-},\st_{}, \st_{+})$ position of a table $\cell$ in  $\Z^{ \St_{j-1} \times \St_{j} \times \St_{j+1 } }$. 
Let $\ct=\ct(\bX_{j})$ be the three-dimensional array with elements $\ct_{\st_{-}, \st_{},\st_{+}}:=\mrc{ \bX_{(j-1):(j+1)}, \st_{-},\st_{},\st_{+} }$ for all $(\st_{-},  \st_{},\st_{+})$. The statistic $\ct$ is essentially a three-way contingency table and its three-dimensional marginals satisfy 
\begin{align*}
&\sum_{\st_{-} \in [\St_{-}] } \ct_{\st_{-}, \st_{},\st_{+}}&=&\, N^{(j+1)}_{ \st_{}, \st_{+}}, ~~ \forall \; \st_{}, \st_{+}; 
&& \sum_{\st_{+} \in [\St_{+}] } \ct_{\st_{-}, \st_{},\st_{+}} =\, N^{(j)}_{\st_{-}, \st_{}} , ~~ \forall \; \st_{-}, \st_{};\\
&\sum_{ \st_{}\in [\St_{}] } \ct_{\st_{-}, \st_{},\st_{+}} &=&\, M_{\st_{-}, \st_{+}},\text{ \hspace{2em} where } && M_{\st_{-}, \st_{+}}\; :=  \,\mrc{ [\bX_{j-1},\bX_{j+1}], \st_{-},\st_{+}},  ~~ \forall \; \st_{-}, \st_{+}.  &&  &
\end{align*}
Here $M_{\st_{-}, \st_{+}}$ is a function of $\bX_{B}$ and thus fixed. Conditional on $\ct$, $\bX_j$ is uniform on all vectors in $[\St_j]^{n}$ that match the three-way contingency table. 
The probability function for $\ct$ is 
\begin{align}
\label{eq:prob-table}
\Pcr{ \forall \; \st_{-}, \st_{},\st_{+},  \ct_{\st_{-}, \st_{},\st_{+}} = \smallcell_{\st_{-}, \st_{},\st_{+}}  }{\bX_{B} }\propto \;\prod_{\st_{-}\in [ \St_{j-1}],\st_{+}\in [ \St_{j+1}]} {M_{\st_{-}, \st_{+}}  \choose \smallcell_{\st_{-},1,\st_{+}}, \dots, \smallcell_{\st_{-}, \St_{j},\st_{+}}},  
\end{align}
where the counts $\smallcell_{\st_{-}, \st_{},\st_{+}}$ satisfy $\sum\limits_{ \st_{}\in [\St_j]} \smallcell_{\st_{-}, \st_{},\st_{+}} = M_{\st_{-}, \st_{+}}$ for each pair of $(\st_{-}, \st_{+}) \in  [\St_{j-1}]\times[\St_{j+1}]$. 

The construction of the Markov chain on contingency tables begins with defining the basic moves: suppose there are $L$ different three-way tables $\{\bs{\Delta}^{(\ell)}\}_{\ell=1}^{L}\subseteq \Z^{ \St_{j-1} \times \St_{j} \times \St_{j+1 } }$ such that the marginals of each table $\bs{\Delta}^{(\ell)}$ are $0$'s:\footnote{The set $\{\bs{\Delta}^{(\ell)}\}_{\ell=1}^{L}$ is similar to the Markov bases used in algebraic statistics (see \cite{PD-BS:1998}), but %is less restricted because it 
it does not need to connect every two possible contingency tables. 
%[old sentence]it does not require generating all possible contingency tables. 
}  
\begin{align*}
\forall \; \st_{},\st_{+},	\sum_{\st_{-}} \bs{\Delta}^{(\ell)}_{\st_{-}, \st_{},\st_{+}}=0, &&
\forall \; \st_{-}, \st_{},	\sum_{\st_{+}} \bs{\Delta}^{(\ell)}_{\st_{-}, \st_{},\st_{+}}=0, \\
\forall \; \st_{-},\st_{+},	\sum_{ \st_{}} \bs{\Delta}^{(\ell)}_{\st_{-}, \st_{},\st_{+}}=0. &&
\end{align*}

A simple set of basic moves, indexed by $\bs{\ell}$, can be constructed as follows: 
for each $\bs{\ell}=(r_1,r_2,c_1,c_2,d_1,d_2)$ where $r_1, r_2 \in [\St_{j-1}]$, $c_1,c_2\in [\St_{j+1}]$, $d_1,d_2\in [\St_{j}]$ and $r_1\neq r_2$, $c_1\neq c_2$, $d_1\neq d_2$,  define 
\[
 \bs{\Delta}^{(\bs{\ell})}_{\st_{-}, \st_{},\st_{+}}=\left\{
\begin{array}{c l}	
     (-1)^{\one{\st_{-}=r_1}+ \one{\st_{+}=c_1}+\one{\st_{}=d_1}} , &\text{ if }\; \st_{-}\in \{r_1,r_2\},
     \st_{+}\in \{c_1,c_2\}, z_j\in \{d_1,d_2\}\\
     0, & \text{otherwise}
\end{array}\right.
 \]
Algorithm~\ref{alg:mh-mc} is a detailed sampling procedure, whose validity is guaranteed by Proposition \ref{prop:mc-blocking-exchange}.

\begin{proposition}\label{prop:mc-blocking-exchange}
For $j\in B^c$, if $\bXk_j$ is drawn from Algorithm \ref{alg:mh-mc}, then \[(\bX_j,\bXk_j)\eqd (\bXk_j,\bX_j) \mid \bX_B.\]
\end{proposition}

A final remark is that one can generalize the refined blocking algorithm to Ising models. By Equation~\eqref{eq:model-ising}, the sufficient statistic is the vector that includes all the counts of configurations of  adjacent pairs, i.e., $\sum_{i=1}^{n} \one{X_{i,s}=\st, X_{i,t}=\st'}$ for all $\st,\st'\in \{-1,1\}$ and $(s,t)\in E$. Instead of sampling a three-way contingency table in Algorithm~\ref{alg:mh-mc}, now one has to construct a reversible Markov Chain for the $(2d+1)$-way contingency table for each vertex and its neighborhood. The basic moves can be constructed similarly as the $\Delta^{(\bs{\ell})}$ given before. 

\begin{comment}
%\dhmargin{}{Calculate the complexity of the super model. }
\dhmargin{[Note that the sufficient statistics of Ising model is also the count statistics of the neighboring pairs, so that the refine blocking algorithm for Markov Chain can also be applied. ]}{Unfortunately, Ising model is the only one Markov Random Fields on the lattice. By Hammersley-Clifford, the PMF can be written as the product of potential functions. The log of a potential function is uniquely determined by a linear combinaton of $(x,y,xy)$ so Ising model is enough. There is no freedom left. }
\end{comment}

\begin{algorithm}[]
\caption{Improved Conditional Knockoffs for $\bX_j$ in a Discrete Markov Chain} \label{alg:mh-mc}
\begin{algorithmic}[1]
\ENSURE $\bX_j$, columns to condition on: $\bX_{j-1}$ and $\bX_{j+1}$.%$\bX_{(j-1):(j+1)}$ 
\STATE Initialize $\tilde{\ct}^{0}$ with $\tilde{\ct}^{0}_{\st_{j-1}, \st_{j},\st_{j+1}}:=\sum_{i=1}^n\one{X_{i,j-1}=\st_{j-1}, X_{i,j}= \st_{j}, X_{i,j+1}=\st_{j+1}}$ for all $\st_{j-1}$, $ \st_{j}$ and $\st_{j+1}$.
\FOR{$t$ from $1$ to $t_{\max}$ }
\STATE Sample uniformly without replacement the pair $(R_{1}, R_{2})$ from $\{1,\dots,\St_{j-1}\}$, the pair $(C_{1},C_{2})$ from $\{1,\dots, \St_{j+1}\}$, and the pair $(D_{1},D_{2})$ from $\{1,\dots, \St_{j}\}$. 
\STATE Define 
$\tilde{\ct}^{*}=\tilde{\ct}^{t-1}+ \bs{\Delta}^{(R_{1},R_{2},C_{1},C_{2},D_{1},D_{2})}$.
\IF{all elements of $\tilde{\ct}^{*}$ are nonnegative}
\STATE
Calculate the Metropolis-Hastings acceptance ratio 
\[
\alpha = \min \left( 
\frac{\prod_{\st_{j-1},\st_{j+1}} {M_{\st_{j-1}, \st_{j+1}}  \choose \tilde{\ct}^{*}_{\st_{j-1},1,\st_{j+1}}, \dots, \tilde{\ct}^{*}_{\st_{j-1}, \St_{j},\st_{j+1} } } }
{\prod_{\st_{j-1},\st_{j+1}} {M_{\st_{j-1}, \st_{j+1}}  \choose \tilde{\ct}^{t-1}_{\st_{j-1},1,\st_{j+1}}, \dots, \tilde{\ct}^{t-1}_{\st_{j-1}, \St_{j},\st_{j+1} } } }, 
1 \right).
\]
\STATE Sample $U\sim \text{Unif}(0,1)$. 
\IF{$U\leq \alpha$}
\STATE Set $\tilde{\ct}^{t}=\tilde{\ct}^{*}$.
\ELSE
\STATE Set $\tilde{\ct}^{t}=\tilde{\ct}^{t-1}$.
\ENDIF
\ELSE
\STATE Set $\tilde{\ct}^{t}=\tilde{\ct}^{t-1}$.
\ENDIF
\ENDFOR
\STATE Set $\tilde{\ct}=\tilde{\ct}^{t_{\max}}$. 
\STATE Sample $\bXk_j$ uniformly on all possible vectors that match $\tilde{\ct}$. In other words, for each $(\st_{j-1},\st_{j+1})\in[\St_{j-1}]\times [\St_{j+1}]$, set
the subvector of $\bXk_j$ with indices $\{i\in [n]: X_{i,j-1}= \st_{j-1}, X_{i,j+1}(i)=\st_{j+1} \}$ as a random uniform permutation of
\[
\big( \underbrace{1,\dots,1}_{ \tilde{\ct} \left(  \st_{j-1},1,\st_{j+1}  \right) }, \dots,  \underbrace{ \St_{j} ,\dots, \St_{j}}_{ \tilde{\ct}\left( \st_{j-1},\St_{j}, \st_{j+1}   \right)   } \big) .
\]
\RETURN $\bXk_{j} $.
\end{algorithmic}
%\hline 
\begin{remark}
\footnotesize
The calculation of  $\alpha$ is simply
\begin{align*}
&\prod_{\st_{j-1}\in \{R_1,R_2\},\st_{j+1}\in \{C_1,C_2\} } \frac{ \tilde{\ct}^{t-1}_{\st_{j-1},D_1,\st_{j+1}} !  }{\tilde{\ct}^{*}_{ \st_{j-1},D_1,\st_{j+1} } !  }\frac{\tilde{\ct}^{t-1}_{\st_{j-1},D_2,\st_{j+1} } !}{  \tilde{\ct}^{*}_{ \st_{j-1},D_2,\st_{j+1} } ! }  \\
=\; &\frac{ \tilde{\ct}^{t-1}_{R_1,D_1,C_1}   }{\tilde{\ct}^{t-1}_{R_1,D_2,C_1} +1 }\times
\frac{ \tilde{\ct}^{t-1}_{R_1,D_2,C_2} }{ \tilde{\ct}^{t-1}_{R_1,D_1,C_2}  +1 }\times
\frac{ \tilde{\ct}^{t-1}_{R_2,D_2,C_1} }{ \tilde{\ct}^{t-1}_{R_2,D_1,C_1}  +1 }\times
\frac{ \tilde{\ct}^{t-1}_{R_2,D_1,C_2}   }{\tilde{\ct}^{t-1}_{R_2,D_2,C_2} +1 },
\end{align*}
where all quantities can be read off directly from $\tilde{\ct}^{t-1}$. 
\end{remark}
\end{algorithm}

\subsubsection{Proofs}\label{app:proof-mc}
\paragraph{Conditional Markov Chains}
We first show that conditional  on $T(\bX)$, the sequence of $\bX_j$'s forms a Markov chain, i.e., Equation~\eqref{eq:mc-condp} describes a Markov chain. We still write \textquoteleft $\mid T(\bX)$\textquoteright\ in the probability here to emphasize the dependence on $T(\bX)$.

Summing Equation~\eqref{eq:mc-condp} over $\bsfW_p$, we have
%\ljmargin{}{we've got some major notation problems here, e.g., there is no $\bsfW_p$ in (B.10), and where did the $\bsfXk$'s come from?}
\begin{equation}\label{eq:mc-cond-mc-p}
	\Pcr{\bX_{1:(p-1)}=\bsfW_{1:(p-1)} }{T(\bX)} = C_0 \prod_{j=2}^{p-1}\phi_{j}( \bsfW_{j-1},\bsfW_{j} \mid T(\bX) ) \sum_{\bsfW_p} \phi_{p}( \bsfW_{p-1},\bsfW_{p} \mid T(\bX) ),
\end{equation}
thus 
\[
\Pcr{\bX_p=\bsfW_p}{\bX_{1:(p-1)}=\bsfW_{1:(p-1)},T(\bX)} =\frac{  \phi_{p}( \bsfW_{p-1},\bsfW_{p} \mid T(\bX) ) }{ \sum_{\bsfW_p'} \phi_{p}( \bsfW_{p-1},\bsfW_{p}' \mid T(\bX) ) }.
\]
Since the right hand side of the last equation does not involve $\bsfW_{1:(p-2)}$, we conclude 
\[
\Pcr{\bX_p }{\bX_{1:(p-1)},T(\bX)}=\Pcr{\bX_p}{\bX_{p-1},T(\bX)}.
\] 
In addition, let  ${\phi}_{p-1}^{\prime}( \bsfW_{p-2},\bsfW_{p-1} \mid T(\bX) )=\phi_{p-1}( \bsfW_{p-2},\bsfW_{p-1} \mid T(\bX) ) \sum_{\bsfW_{p}'} \phi_{p}( \bsfW_{p-1},\bsfW_{p}' \mid T(\bX) )$  % \ljmargin{}{don't want to use tilde for non-knockoff stuff} 
and for $j<p-1$, let $\phi_{j}^{\prime}=\phi_{j}$, \eqref{eq:mc-cond-mc-p} can be rewritten as
\begin{equation}\label{eq:mc-cond-mc-p-2}
	\Pcr{\bX_{1:(p-1)}=\bsfW_{1:(p-1)} }{T(\bX)} = C_0 \prod_{j=2}^{p-1}\phi_{j}^{\prime}( \bsfW_{j-1},\bsfW_{j} \mid T(\bX) ), 
\end{equation}
which has the same form as Equation~\eqref{eq:mc-condp}. 
Continuing the same reasoning for $\bX_{p-1},\bX_{p-2},\dots, \bX_2$, we conclude 
\[
\Pcr{\bX_j }{\bX_{1:(j-1)},T(\bX)}=\Pcr{\bX_j}{\bX_{j-1},T(\bX)}, \; 2\le j\le p, 
\] 
that is, the sequence of $\bX_j$'s is a Markov chain conditional on $T(\bX)$.

%scip
\paragraph{SCIP}
\begin{proof}[Proof of Proposition~\ref{prop:mc-scip-f}]
The first equation follows from the uniform distribution and the Markovian property 
	\[\Pcr{ \bX_1=\bsfW_1 }{\bX_2=\bsfW_{2} }=\; \Pcr{ \bX_1=\bsfW_1}{ \bX_2=\bsfW_{2} , \bX_{3:p}=\bsfX_{3:p}}.\]
%Proving the second and fourth equation is similar to proving the third equation if we allow the notation $1:0$, $(p+1):p$ and $(p+2):p$ for empty set. \\
Next we prove the second equation, except for the case of $j=2$. However, the second equation with $j=2$ and also the third equation both follow the same proof, by allowing $k_1:k_2$ for  $k_1>k_2$ to denote the empty set.

Before the proof, we show an implication of Bayes' rule. For any $\bsfW_{j}$ and $\bsfW_{j+1}$, %For the realized $\bsfX_{1:(j-2)}$,
%\ljmargin{}{I don't follow your notation in these equations, why do you use $\bsfW$ sometimes and $\bsfX$ at other times?}\dhmargin{}{$\bsfX$ is reserved for the observed $\bX$ }, 
\begin{align}
	&\Pcr{ \bX_j=\bsfW_j }{ \bX_{1:(j-1)}=\bsfX_{1:(j-1)}, \bXk_{1:(j-2)}=\bsfXk_{1:(j-2)},\bX_{j+1}=\bsfW_{j+1}, \bX_{(j+2):p}=\bsfX_{(j+2):p} }\nonumber \\
	\propto \;&\Pcr{ \bX_j=\bsfW_j }{ \bX_{1:(j-1)}=\bsfX_{1:(j-1)}, \bX_{j+1}=\bsfW_{j+1}, \bX_{(j+2):p}=\bsfX_{(j+2):p} }\nonumber\\
	&\times  \Pcr{\bXk_{1:(j-2)}=\bsfXk_{1:(j-2)}}{ \bX_j=\bsfW_j , \bX_{1:(j-1)}=\bsfX_{1:(j-1)}, \bX_{j+1}=\bsfW_{j+1}, \bX_{(j+2):p}=\bsfX_{(j+2):p} }\label{eq:mc-scip-f1} \\
	\propto \;&\Pcr{ \bX_j=\bsfW_j }{\bX_{1:(j-1)}=\bsfX_{1:(j-1)}, \bX_{j+1}=\bsfW_{j+1}, \bX_{(j+2):p}=\bsfX_{(j+2):p} }\nonumber\\
	&\times  \Pcr{\bXk_{1:(j-2)}=\bsfXk_{1:(j-2)}}{  \bX_{1:(j-1)}=\bsfX_{1:(j-1)} }\label{eq:mc-scip-f2}\\
	\propto \;&\Pcr{ \bX_j=\bsfW_j }{\bX_{1:(j-1)}=\bsfX_{1:(j-1)}, \bX_{j+1}=\bsfW_{j+1}, \bX_{(j+2):p}=\bsfX_{(j+2):p} } \label{eq:mc-scip-f3}\\
	\propto \;& \one{ (\bsfX_{1:(j-1)}, \bsfW_{j},\bsfW_{j+1}, \bsfX_{(j+2):p} )\in \mathcal{Q} }\label{eq:mc-scip-f4},
\end{align}
%\ljmargin{}{in subscripts, should be $1:(j-1)$, not $1:(j-1)$. Please check this everywhere}
where Equation~\eqref{eq:mc-scip-f1} is due to Bayes' rule, Equation~\eqref{eq:mc-scip-f2} is due to the Markovian property and the fact that SCIP sampling of $\bXk_{1:(j-2)}$ only depends on $\bX_{1:(j-1)}$, and Equation~\eqref{eq:mc-scip-f3} is because the conditional probability of $\bXk_{1:(j-2)}$ does not depend on $\bsfW_j$. Note that the normalizing constant in Equation~\eqref{eq:mc-scip-f4} depends on $\bsfW_{j+1}$ but not on $\bsfW_j$.

Now the second equation of Proposition~\ref{prop:mc-scip-f} can be shown as follows
	\begin{align*}
		&\Pcr{ \bX_j=\bsfW_j }{ \bX_{j-1}=\bsfX_{j-1},\bXk_{j-1}=\bsfXk_{j-1},\bX_{j+1}=\bsfW_{j+1} } \nonumber \\
		=\;&\Pcr{ \bX_j=\bsfW_j }{ \bX_{1:(j-1)}=\bsfX_{1:(j-1)}, \bXk_{1:(j-2)}=\bsfXk_{1:(j-2)},\bXk_{j-1}=\bsfXk_{j-1},\bX_{j+1}=\bsfW_{j+1}, \bX_{(j+2):p}=\bsfX_{(j+2):p}  } \nonumber  \\
&\rightline{ \text{ (Since $\bX_j, \bXk_{1:(j-2)}, \bX_{(j+2):p}$ are conditionally independent given $\bX_{j-1}, \bX_{j+1}$)}	}	\nonumber \\
		\propto \;& \; \Pcr{\bXk_{j-1}=\bsfXk_{j-1} }{ \bX_j=\bsfW_j , \bX_{1:(j-1)}=\bsfX_{1:(j-1)}, \bXk_{1:(j-2)}=\bsfXk_{1:(j-2)},\bX_{j+1}=\bsfW_{j+1}, \bX_{(j+2):p}=\bsfX_{(j+2):p} }\nonumber \\
		& \times \Pcr{ \bX_j=\bsfW_j}{\bX_{1:(j-1)}=\bsfX_{1:(j-1)}, \bXk_{1:(j-2)}=\bsfXk_{1:(j-2)},\bX_{j+1}=\bsfW_{j+1}, \bX_{(j+2):p}=\bsfX_{(j+2):p} }\nonumber \\
		& \rightline{ \text{ (By Bayes' rule)}	}	\nonumber \\
		\propto \;& \; f_{j-1}(	\bsfXk_{j-1}, \bsfXk_{j-2}, \bsfW_{j})  \one{ (\bsfX_{1:(j-1)}, \bsfW_{j},\bsfW_{j+1}, \bsfX_{(j+2):p} )\in \mathcal{Q} }, \\
		&\rightline{\text{ (By the Markovian property and Equation~\eqref{eq:mc-scip-f3})}	}	\nonumber 
	\end{align*}
where the normalizing constant does not depend on $\bsfW_j$. Hence we have 
\[
f_{j}(\bsfW_{j},\bsfXk_{j-1},\bsfW_{j+1})=\frac{ f_{j-1}(	\bsfXk_{j-1}, \bsfXk_{j-2}, \bsfW_{j})  \one{ (\bsfX_{1:(j-1)}, \bsfW_{j},\bsfW_{j+1}, \bsfX_{(j+2):p} )\in \mathcal{Q} } }{
\sum_{\bsfW_{j}^{\prime}\in [\St_{j}]^{n}} f_{j-1}(	\bsfXk_{j-1}, \bsfXk_{j-2}, \bsfW_{j}^{\prime})  \one{ (\bsfX_{1:(j-1)},\bsfW_{j}^{\prime},\bsfW_{j+1}, \bsfX_{(j+2):p} )\in \mathcal{Q} }
}
\]
\end{proof}

\paragraph{Refined Blocking}

\begin{proof}[Proof of Proposition \ref{prop:mc-blocking-exchange}]
In the following, we view $T(\bX)$ and $\bX_B$ as fixed and only $\bX_j$ being random, and denote this conditional probability %the conditional probability for $\bX_j$ given $\bX_{B}$ 
by $\P_{j}(\cdot)$. 

We first show $( \ct, \tilde{\ct}^{t_{\max}}) \eqd (\tilde{\ct}^{t_{\max}},\ct)$. Denote the probability mass function in \eqref{eq:prob-table} as $g( \smallcell )$. 
%where $\cell=<\smallcell_{\st_{-},\st,\st_{+}}| \st_{-}\in \St_{-}, \st_{}\in \St_{}, \st_{+}\in \St_{+} >$. 
Since $\tilde{\ct}^{0}=\ct\sim g$ and the transition kernel constructed in Algorithm~\ref{alg:mh-mc} is in detailed balance with density $g(\smallcell)$, $(\tilde{\ct}^{t})_{t=0}^{t_{\max}}$ is reversible. Thus
\begin{equation}\label{eq:mc-table-swap}
	( \tilde{\ct}^0, \tilde{\ct}^{t_{\max}}) \eqd (\tilde{\ct}^{t_{\max}},\tilde{\ct}^0).
\end{equation}
By the sampling of $\bXk_j$ in the algorithm, we have
\begin{equation}\label{eq:mc-x-xk-indp}
\bXk_j \indp \bX_j \mid  \tilde{\ct}^{t_{\max}},  \ct.
\end{equation}
and 
\begin{equation}\label{eq:mc-table-to-x}
\Pcrs{j}{\bXk_j=\bsfX_{j} }{ \tilde{\ct}^{t_{\max}}= \smallcell }=\Pcrs{j}{\bX_j=\bsfX_{j} }{ \ct (\bX_j)= \smallcell }.
\end{equation}
Hence
\begin{align*}
&\prs{j}	{ \bX_j=\bsfW_{j}, \ct(\bX_j) =\smallcell ,\bXk_j =\bsfWk_{j}, \tilde{\ct}^{t_{\max}}=\tilde{\smallcell}}\\
=\;&\prs{j}	{\ct(\bX_j) =\smallcell ,  \tilde{\ct}^{t_{\max}}=\tilde{\smallcell}}\Pcrs{j}{ \bX_j=\bsfW_{j}, \bXk_j =\bsfWk_{j}}{ \ct(\bX_j) =\smallcell , \tilde{\ct}^{t_{\max}}=\tilde{\smallcell}}\\
=\;&\prs{j}	{\ct(\bX_j) =\smallcell ,  \tilde{\ct}^{t_{\max}}=\tilde{\smallcell}}\Pcrs{j}{ \bX_j=\bsfW_{j} }{ \ct(\bX_j) =\smallcell}\Pcrs{j}{\bXk_j =\bsfWk_{j} }{ \tilde{\ct}^{t_{\max}}=\tilde{\smallcell}}\\
=\;&\prs{j}	{ \tilde{\ct}^{t_{\max}}=\smallcell, \ct(\bX_j) =\tilde{\smallcell} }\Pcrs{j}{\bXk_j =\bsfW_{j} }{  \tilde{\ct}^{t_{\max}}=\smallcell}\Pcrs{j}{ \bX_j=\bsfWk_{j} }{  \ct(\bX_j) =\tilde{\smallcell}}\\
=\;&\prs{j}	{ \bX_j=\bsfWk_{j}, \ct(\bX_j) =\tilde{\smallcell} ,\bXk_j =\bsfW_{j}, \tilde{\ct}^{t_{\max}}=\smallcell},
\end{align*}
where the second equality is due to \eqref{eq:mc-x-xk-indp} and the third equality is due to Equations~\eqref{eq:mc-table-swap} and  \eqref{eq:mc-table-to-x}. 
Summing over all $\smallcell, \tilde{\smallcell}$, we conclude that $(\bX_j,\bXk_j)\eqd (\bXk_j,\bX_j) \mid \bX_{B}$.
%\ljmargin{}{Seems like $\smallcell$ and $\tilde{\smallcell}$ should be $\smallcell$ and $\tilde\smallcell$, right?}
\end{proof}

\section{Conditional Hypothesis}\label{app:cond-hypo}
%By treating all stochasticity conditional on $T(\bX)$, the conditional knockoffs are valid model-X knockoffs as defined in \citet{EC-ea:2018}, and the same properties (Lemma 3.2 and Theorem 3.4 ) hold conditionally on $T(\bX)$. 
This section concerns the hypotheses actually tested by conditional knockoffs.
Suppose  $(\bx_i,Y_i)\iid (X,Y)$ and $T(\bX)$ is a statistic of $\bX$. 
The knockoff procedure using conditional knockoffs treats the variables in $\mathcal{H}_{0,T}=\{j\; :\;  \by\indp \bX_j\mid\bX_\noj,T(\bX)\}$  as null. It is of interest to compare $\mathcal{H}_{0,T}$ with $\nullset$, the original set of null variables defining the variable selection problem we actually care about.
\begin{proposition}\label{prop:hypo-sub}
%	If $\by\indp \bX_j\mid\bX_\noj$, then $\by\indp \bX_j\mid\bX_\noj,T(\bX)$. Therefore for i.i.d. data, 
%If samples are i.i.d., then 
$\nullset \subseteq \mathcal{H}_{0,T}$. 
\end{proposition}

\begin{proof}%[Proof of Proposition~\ref{prop:hypo-sub}]
	Suppose $j\in \nullset$. For i.i.d. data,  $j\in \nullset$ implies $Y_{i}\indp X_{i,j}\mid X_{i,\noj}$, which together with the independence among $\{(Y_{i}, \bx_{i})\}_{i=1}^{n}$ implies $\by\indp \bX_j\mid \bX_\noj$. 
	Note that $\by \indp T(\bX)\mid \bX$ (since $T(\bX)$ is deterministic given $\bX$), which together with $\by\indp \bX_j\mid \bX_\noj$ implies $\by\indp (\bX_j,T(\bX))\mid \bX_\noj$ by the contraction property of conditional independence. And by the weak union property of conditional independence, $\by\indp (\bX_j,T(\bX))\mid \bX_\noj$ implies $\by\indp \bX_j\mid \bX_\noj, T(\bX)$.
	Thus  $j\in \mathcal{H}_{0,T}$. This holds for arbitrary $j\in \nullset$ and thus $\nullset \subseteq \mathcal{H}_{0,T}$.
\end{proof}

The converse is not true in general, for instance if $T(\bX)=\bX$ and $\nullset=\emptyset$, then all variables are automatically null conditional on $T(\bX)$ and thus $\nullset \subsetneq  \mathcal{H}_{0,T}$. In general, when $T(\bX)$ does not allow full reconstruction of $\bX_j$ it should be rare for a non-null variable $\bX_j$ to be null conditional on $T(\bX)$, as this can only happen if there is a perfect synergy of $\Fyx$ and $\Fx$ so that $\Fyx$ is only a function of $X_j$ through a transformation computable from the sufficient statistic $T(\bX)$ of $\Fx$. For most problems of interest, %we provide the following sufficient condition for $\nullset \supset \mathcal{H}_{0,T}$.
 Theorem~\ref{thm:suff-nn} provides a sufficient condition for $\by\nindp \bX_{j} \mid \bX_\noj , T(\bX)$, i.e., $\nullset =  \mathcal{H}_{0,T}$: the conditional mean of $Y_{i}$ (or some transformation of $Y_{i}$) given $\bx_{i}$, say $\phi(\bx_i)$, should not be deterministic after conditioning on $\bX_{\noj}$ and $T(\bX)$. 

\begin{theorem}\label{thm:suff-nn}
Suppose for a bounded function $g(y)$ and $\phi(\bx) :=\,  \Ec{g(Y) }{X=\bx} $,  there exist two disjoint Borel sets $B_1,B_2\subset \R^{p}$ such that $\inf_{\bx\in B_1} \phi(\bx) > \sup_{\bx\in B_2}\phi(\bx)$. 
If for each $j\in [p]$, it holds with positive probability that 
\[
\Pcr{ \bx_1\in B_i}{\bX_{\noj } , T(\bX) }>0, i=1,2,
\]
then $\nullset = \mathcal{H}_{0,T}$. 
\end{theorem}

This theorem is based on Proposition~\ref{prop:hypo-sub} and the following Proposition~\ref{prop:nn-nn}. By Proposition~\ref{prop:nn-nn}, for each $j\notin \nullset$, it holds that $j\notin \mathcal{H}_{0,T}$; hence $\mathcal{H}_{0,T}\subseteq \nullset$. In addition, Proposition~\ref{prop:hypo-sub} shows $\nullset \subseteq \mathcal{H}_{0,T}$, thus $\nullset = \mathcal{H}_{0,T}$.

\begin{proposition}\label{prop:nn-nn} % nonnull becomes nonnull
Suppose $Y\nindp X_{j} \mid X_{\noj}$, and $g(y)$ is a bounded function. Define $K :=\,  \Ec{g(Y_1) }{ \bx_{1}} $, and 
%$M\; :=\;  \Ec{K}{ X_{1,\noj } , T(\bX) }$  (??difficult)
%$M\; :=\;  \Ec{K}{ X_{1,\noj }, X_{2:n,\noj} , T(\bX) }$ 
$M :=\,  \Ec{K}{ \bX_{\noj } , T(\bX) }$. 
\begin{enumerate}[(a)]
\item If $K$ is different from $M$, then 
\[
\by\nindp \bX_{j} \mid \bX_\noj , T(\bX).
\] 

\item %Define $K,M$ as in part (a). 
If $K$ can be written as $\phi(\bx_1)$ and $\phi$ is not constant on the support of the conditional distribution of $\bx_1$ given $\bX_\noj$ and $T(\bX)$, i.e.,  
 there exist two disjoint Borel sets $B_1,B_2$ such that $\inf_{\bx\in B_1} \phi(\bx) > \sup_{\bx\in B_2}\phi(\bx)$, and
 \[
0< \pr{   \Pcr{ \bx_1\in B_i}{\bX_{\noj } , T(\bX) }>0, i=1,2},
 \]
 
 then $K$ is different from $M$.

\end{enumerate}

%% will this separation be better to understand?
%Let  $(\bx_i,Y_i)\iid (X,Y)$ and $T(\bX)$ is a statistic of $\bX$. 
%\begin{enumerate}[(a).]
%	\item If  $Y\nindp X_{j} |X_{\noj}$ in the following sense: 
%	
%	$K\; :=\;  \Ec{Y_1 }{ \bx_{1}} $ is different from 
%%$M\; :=\;  \Ec{K}{ X_{1,\noj } , T(\bX) }$  (??difficult)
%%$M\; :=\;  \Ec{K}{ X_{1,\noj }, X_{2:n,\noj} , T(\bX) }$ 
%$M\; :=\;  \Ec{K}{ \bX_{\noj } , T(\bX) }$ 
%
%then $\by\nindp \bX_{j} | \bX_\noj , T(\bX)$.
%\item Furthermore, if $K$ can be written as $\phi(X_1)$ and $\phi$ is not constant on the support of conditional probability of $X_1$ given $\bX_\noj$ and $T(\bX)$:
%
% there exist two disjoint Borel sets $B_1,B_2$ such that $\inf_{x\in B_1} \phi(x) > \sup_{x\in B_2}\phi(x)$, 
% \[
% \pr{0<\Pcr{X_1\in B_i}{\bX_{\noj } , T(\bX) }>0, i=1,2}>0,
% \]
% 
% then $K$ is different from $M$. 
%\item If $K$ is defined as $\Ec{g(Y_1) }{ \bx_{1}} $ for a bounded function $g$, the above statements still holds. 
%\end{enumerate}

 %In particular, if $K$ can be written as $\phi(X_1)$ and one of the following holds, 
%\begin{enumerate}[(1).]
%\item there exists $\omega_1,\omega_2$ such that $\Pcr{X_1=\omega_{i}}{X_{1,\noj }, \bx_{2:n} , T(\bX) }>0, i=1,2$ and $\phi(\omega_1)\neq \phi(\omega_2)$ then $K$ is different from $M$. 
%\item there exists 
%\end{enumerate}

\end{proposition}

To prove this proposition, we need the following lemma. 

\begin{lemma}\label{lem:nindp-nindp}
Suppose $Y\nindp X$ and $T$ is a function of $X$. Furthermore, if there exists a bounded function $g$ such that % $\E{|g(Y)|}<\infty$ and 
$K\; :=\; \Ec{ g(Y)}{X}$ is not conditionally deterministic in the following sense:
\[
%(not easy to see)0<\P \Big{(} \Pcr{K-\Ec{K}{T} >0 }{T}>0 \Big{)},
0<\pr{K\neq \Ec{K}{T} },
\]
then $Y\nindp X \; | \; T$. 
\end{lemma}
\begin{proof}
Let $M  :=\;  \Ec{K}{T}$. Then $M$ is $\sigma(T)$-measurable. Since $\pr{K\neq M }=\pr{K>M }+\pr{K< M }$, without loss of generality, we assume $0<\pr{K>M }$. 

% and the event $ B  :=\; \{ K>M \} \in \sigma(X)$. 

We compute $\Ec{g(Y)\one{K>M} }{T}$ in two different ways.	On the one hand,
%\ljmargin{}{without the $-M$?}
%\ljmargin{}{move line-by-line explanations to righthand side, as in other proofs} 
%\ljmargin{}{Aren't the last two lines identical?} 
	\begin{align*}
&		\Ec{g(Y)\one{K>M} }{T}  & \\
\; \overset{a.s.}{=}\; & \Ec{ \Ec{g(Y)\one{K>M} }{X} }{T}  & (\because \sigma(T) \subseteq \sigma(X) )\\
\;\overset{a.s.}{=}\; & \Ec{ \Ec{g(Y) }{X}\one{K>M} }{T} & (\because \{K>M\} \in \sigma(X) ) \\
\; =\; & \Ec{ K\one{K>M} }{T} & (\because \mbox{definition of }K) 
	\end{align*}
On the other hand, if $Y\indp X \; | \; T$ then 
\begin{align*}
&	\Ec{g(Y)\one{K>M} }{T} & \\
 \overset{a.s.}{=}\; & \Ec{g(Y)}{T}\Ec{ \one{K>M} }{T} & (\because Y\indp X \; | \; T ) \; \\
 \overset{a.s.}{=}\; & \Ec{ \Ec{g(Y)}{X}}{T}\Ec{ \one{K>M} }{T} & (\because\mbox{law of total expectation}) \; \\
 \overset{a.s.}{=}\; & M\Ec{ \one{K>M} }{T} & (\because\mbox{definition of }M) \; \\
\overset{a.s.}{=} \; & \Ec{M\one{K>M} }{T}. & (\because M\in \sigma(T)) \; 
\end{align*}

\begin{comment}
	On the one hand, 
	\begin{align*}
&		\Ec{(g(Y)-M)\one{K>M} }{T}\\
=\; & \Ec{g(Y)\one{K>M} }{T} -\Ec{M\one{K>M} }{T}\\
(\because \sigma(T) \subseteq \sigma(X) )\; \overset{a.s.}{=}\; & \Ec{ \Ec{g(Y)\one{K>M} }{X} }{T} -\Ec{M\one{K>M} }{T}\\
(\because \{K>M\} \in \sigma(X) ) \;  \overset{a.s.}{=}\; & \Ec{ \Ec{g(Y) }{X}\one{K>M} }{T} -\Ec{M\one{K>M} }{T}\\
(\because \mbox{definition of }K) \; =\; & \Ec{ K\one{K>M} }{T} -\Ec{M\one{K>M} }{T}\\
=\; & \Ec{ (K-M)\one{K>M} }{T}
	\end{align*}
On the other hand, if $Y\indp X \; | \; T$ then 
\begin{align*}
&	\Ec{(g(Y)-M)\one{K>M} }{T}\\
=\; & \Ec{g(Y)\one{K>M} }{T} -\Ec{M\one{K>M} }{T}\\
(\because Y\indp X \; | \; T ) \; \overset{a.s.}{=}\; & \Ec{g(Y)}{T}\Ec{ \one{K>M} }{T} -\Ec{M\one{K>M} }{T}\\
(\because\mbox{law of total expectation}) \; \overset{a.s.}{=}\; & \Ec{ \Ec{g(Y)}{X}}{T}\Ec{ \one{K>M} }{T} -\Ec{M\one{K>M} }{T}\\
(\because\mbox{definition of }M) \; \overset{a.s.}{=}\; & M\Ec{ \one{K>M} }{T} -\Ec{M\one{K>M} }{T}\\
(\because M\in \sigma(T)) \; \overset{a.s.}{=} \;& 0. 
\end{align*}
\end{comment}
Combining these two expressions shows that if $Y\indp X \; | \; T$ then ${\Ec{ (K-M)\one{K>M} }{T}\overset{a.s.}{=}  0}$, and thus $\E{ (K-M)\one{K>M} }{=}  0$. However, this implies $\pr{K>M}=0$ and contradicts the condition; hence  $Y\nindp X \; | \; T$.
\end{proof}

%[the condition should be : K is not constant on the support of X_1 given X_-j, T ]

\begin{proof}[Proof of Proposition \ref{prop:nn-nn}]
$ $ % to make a new line in the proof environment
%\ljmargin{}{please get this to start on a new line}
\begin{enumerate}[(a)]
	\item The condition that $Y\nindp X_{j}\mid X_{\noj}$ implies $Y_1\nindp \bX_{j}\mid\bX_{\noj}$ (see Lemma~\ref{lem:iid-nindp} below). 
Because $\bx_1 \indp \bx_{2:n}$, we have $K=\Ec{g(Y_1) }{ \bx_{1}, \bx_{2:n} }=\Ec{g(Y_1) }{ \bX_{j}, \bX_{\noj} }$.
The condition $\P(K\neq M)>0$ implies that  $\Pc{K\neq M}{\bX_\noj}>0$ holds with positive probability. 

To apply Lemma \ref{lem:nindp-nindp},  $\bX_{\noj}$ is treated as fixed, and $\bX_j$ (resp. $Y_1$) is treated as $X$ (resp. $Y$). Then we have $Y_1\nindp \bX_{j} \mid \bX_{\noj}, T(\bX)$, which immediately implies $\by\nindp \bX_{j} \mid \bX_\noj , T(\bX)$. 
%Suppose $\by\indp \bX_{j} | \bX_\noj , T(\bX)$, then $Y_1\indp \bX_{j} | \bX_\noj , T(\bX)$. 

\item The existence of $B_1$ and $B_2$ implies that there exists a real number $s$ such that 
\[
\sup_{x\in B_2}\phi(x)< s< \inf_{x\in B_1} \phi(x) .
\]
We will prove by contradiction that $K$ is different from $M$. Suppose $\pr{K\neq M}=0$, then $\Pcr{ K\neq M }{\bX_{\noj } , T(\bX)}\overset{a.s.}{=}0$. Thus a.s. we have
\begin{align*}
&\Pcr{X_1\in B_{1}}{\bX_{\noj } , T(\bX) }\\
=\; &\Pcr{X_1\in B_{1}, K=M }{\bX_{\noj }, T(\bX) }\\
\leq \; & \Pcr{X_1\in B_{1}, M>s}{\bX_{\noj } , T(\bX) } & \hfill{(\because s<\inf_{x\in B_1} \phi(x) ) } \\
 \leq \; & \one{M>s}. & \hfill{ (\because M \in \sigma(\bX_{\noj } , T(\bX) ))} \;
\end{align*}
Similarly $\Pcr{X_1\in B_{2}}{\bX_{\noj } , T(\bX) }\overset{ a.s.}{\leq} \one{M<s}$. Since $\one{M>s} \cdot  \one{M<s}=0$, it follows that
\[\Pcr{X_1\in B_{1}}{\bX_{\noj } , T(\bX) } \cdot \Pcr{X_1\in B_{2}}{\bX_{\noj } , T(\bX) }\overset{ a.s.}{\leq} 0,\] which contradicts the condition that $0< \pr{   \Pcr{X_1\in B_i}{\bX_{\noj } , T(\bX) }>0, i=1,2}$. Hence we conclude $\pr{K\neq M}>0$. 
\end{enumerate}
\end{proof}

\begin{lemma}\label{lem:iid-nindp}
If $Y\nindp X\mid U$ and $(X,Y,U) \indp (V,W)$, then $Y\nindp (X, V)\mid (U,W)$. 
\end{lemma}
\begin{proof}
%It suffices to show that given 	$(X,Y,U) \indp (V,W)$, $Y\indp (X, V)\mid (U,W)$ implies $Y\indp X\mid U$. 

Suppose $Y\indp (X, V)\mid (U,W)$. Then 
\begin{equation}\label{eq:hypo-lem-5-1}
Y\indp X\mid (U,W),
\end{equation}

The condition that $(X,Y,U) \indp (V,W)$ implies $(X,Y,U) \indp W$, and thus by weak union property of conditional independence, we have
\begin{equation}\label{eq:hypo-lem-5-2}
 (X,Y) \indp W \mid U. 	
\end{equation}

\begin{comment}%just by contraction
For any bounded continuous functions $f(x),h(y)$, 
\begin{align*}
&	\Ec{f(X)h(Y) }{U}\\
	=\;	 & \Ec{ \Ec{f(X)h(Y)}{U,W} }{U}\\
	(\because \text{Equation~}\eqref{eq:hypo-lem-5-1})=\; & \Ec{ \Ec{f(X)}{U,W}\Ec{h(Y)}{U,W} }{U}\\
(\because \text{Equation~}\eqref{eq:hypo-lem-5-2})	=\; & \Ec{ \Ec{f(X)}{U}\Ec{h(Y)}{U} }{U}\\
	=\; &\Ec{f(X)}{U}\Ec{h(Y)}{U}. \\
\end{align*}
\end{comment}
Equations~\eqref{eq:hypo-lem-5-1} and \eqref{eq:hypo-lem-5-2} together with the contraction property of conditional independence imply $Y \indp X \mid U$. This contradicts the condition, so we conclude that $Y\nindp (X, V)\mid (U,W)$. 
\end{proof}

\section{Supplementary Simulations}\label{app:Simulation}
\subsection{Nonlinear Response Models}\label{app:nonliner}
We re-conduct the same the simulations in Section~\ref{sec:ex}  on logistic regression, confirming that the variable selection by using conditional knockoff allows for general response dependence. The experiments follow the same designs as in Sections~\ref{sec:ldg-sims}, \ref{sec:ggm-sims} and \ref{sec:dgm-sims}, but with binary responses sampled as $Y_i\mid \bx_{i}\sim \text{Bernoulli}( \varsigma(\bx_i\tp \bb/\sqrt{n}))$, where $\varsigma(t)=e^{t}/(1+e^t)$ is the  logistic function, and slightly larger sample sizes $n$ for the re-conducted simulations of Sections \ref{sec:ggm-sims} and \ref{sec:dgm-sims}.

%\subsection{Low-Dimensional Gaussian}
\begin{figure}[H]\centering
  \subfloat [
  ]{     \includegraphics[width=0.45\linewidth]{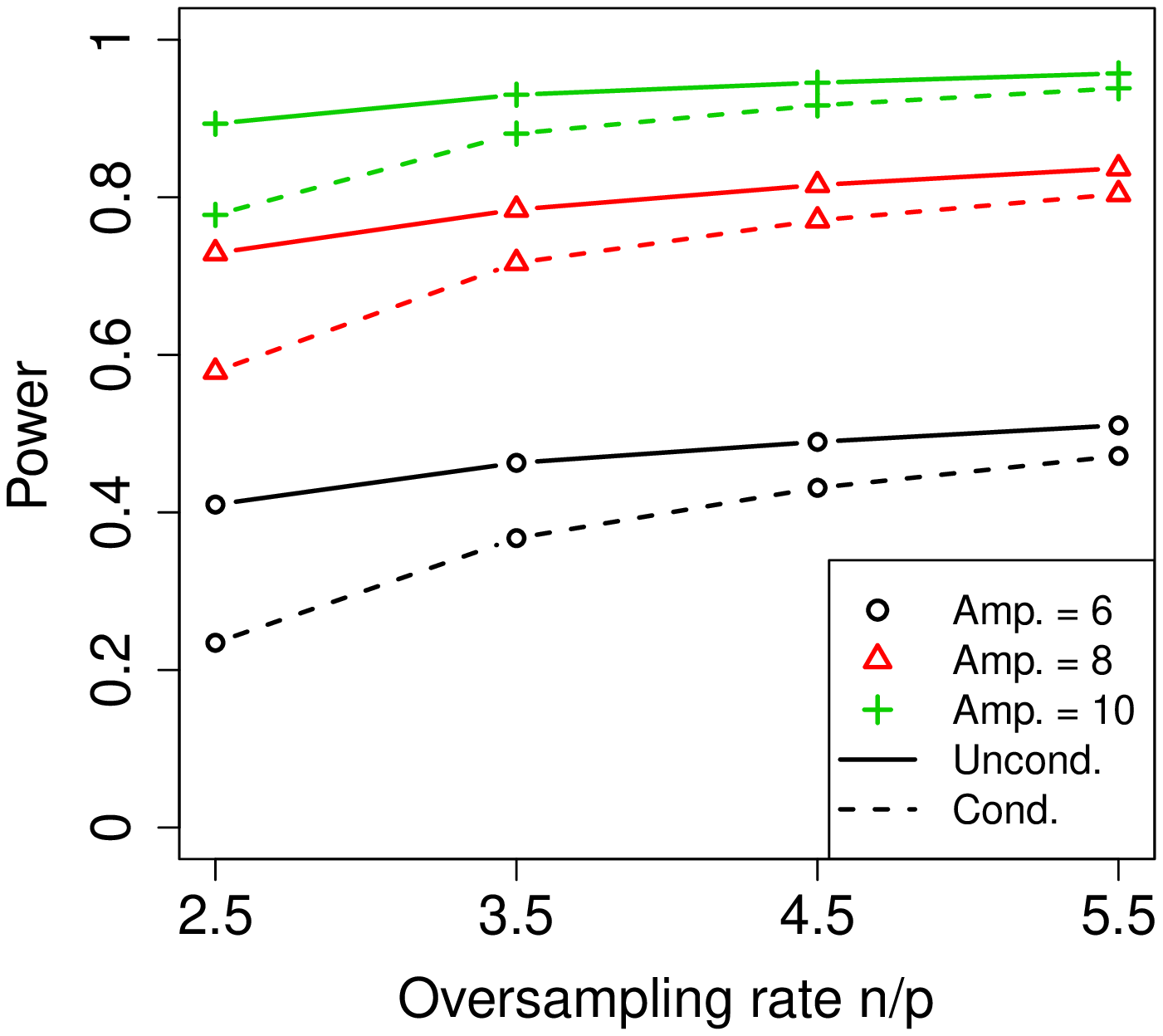} 
  \label{fig:Bin_LDG1}}
  \subfloat[
  ]{     \includegraphics[width=0.45\linewidth]{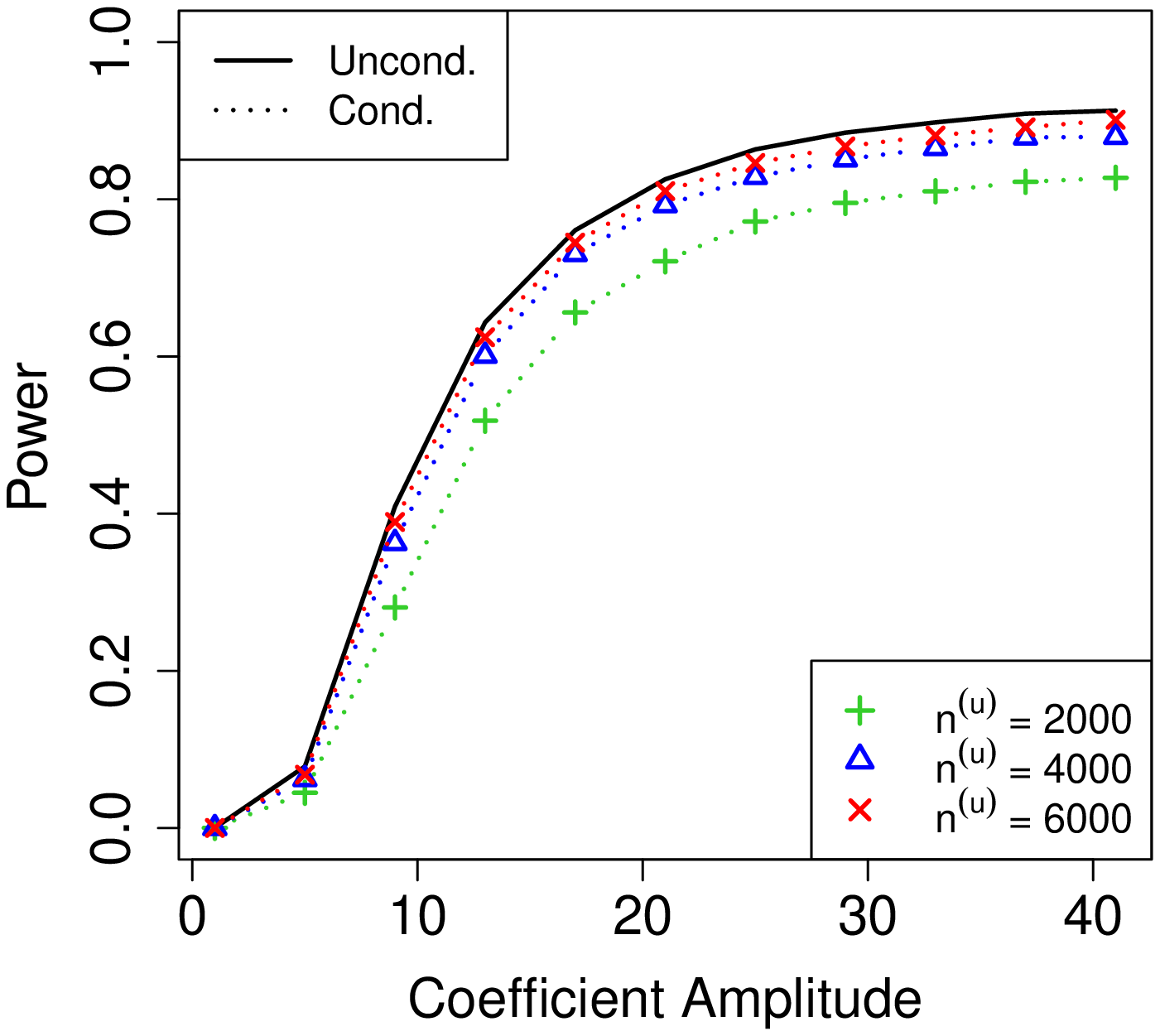} 
  \label{fig:Bin_LDG2}
  }
    \caption{(Logistic regression version of Figure~\ref{fig:ldg}) Power curves of conditional and unconditional knockoffs for an AR(1) model with $p=1000$ (a) as $n/p$ varies for various coefficient amplitudes and (b) as the coefficient amplitude varies for various values of $n^{(u)}$, with $n=800$ fixed. Standard errors are all below $0.006$.}%The nominal FDR level is $0.2$.
 \end{figure}
 % se : unlabel:0.005878904, ldg:0.005696278

%\subsection{Gaussian Graphical Models}
\begin{figure}[H]
\centering
  \subfloat[
  ]{ \includegraphics[width=0.45\linewidth]{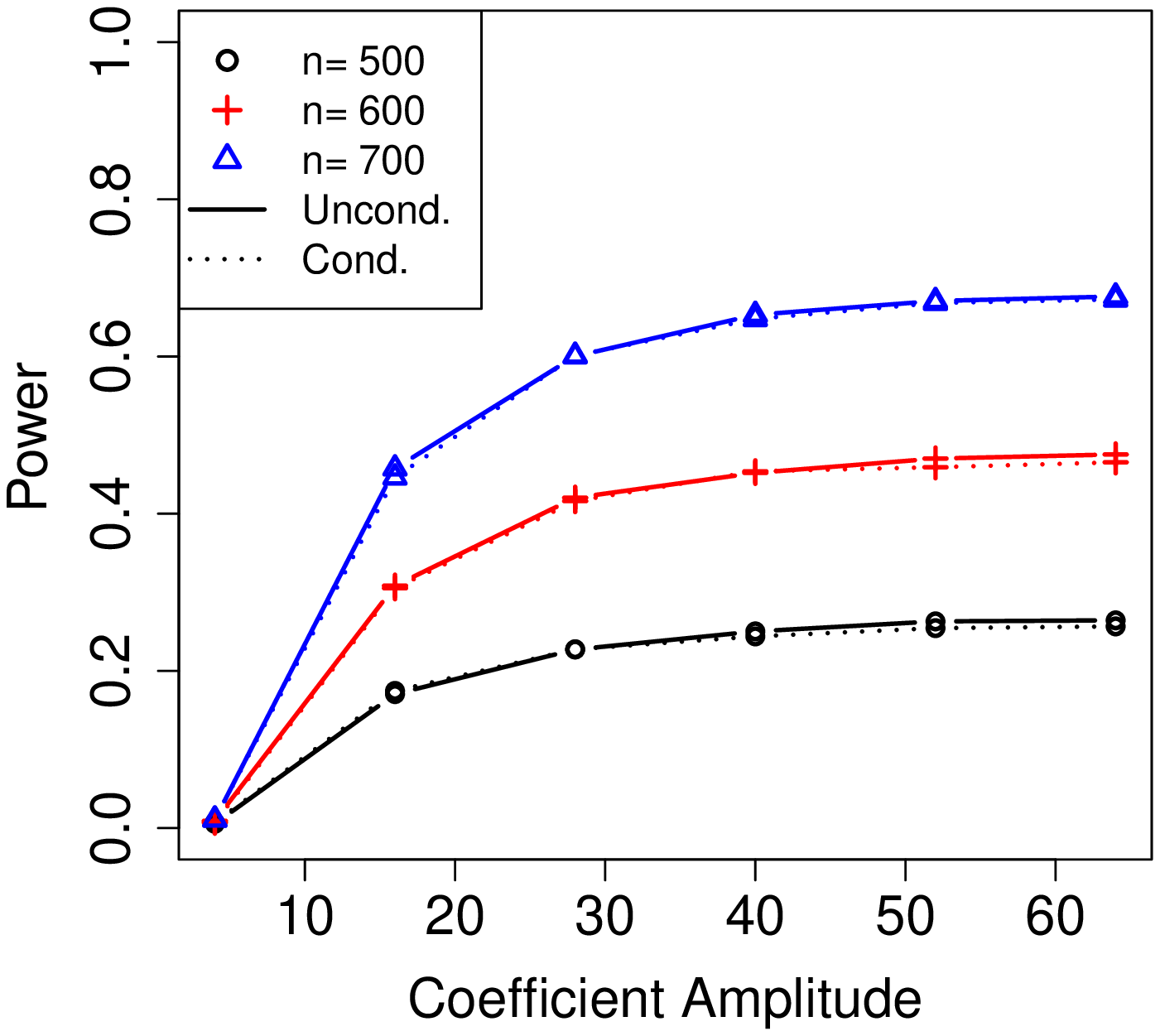}
   \label{fig:Bin_ggm1}
    }
  \subfloat[
  ]{     \includegraphics[width=0.45\linewidth]{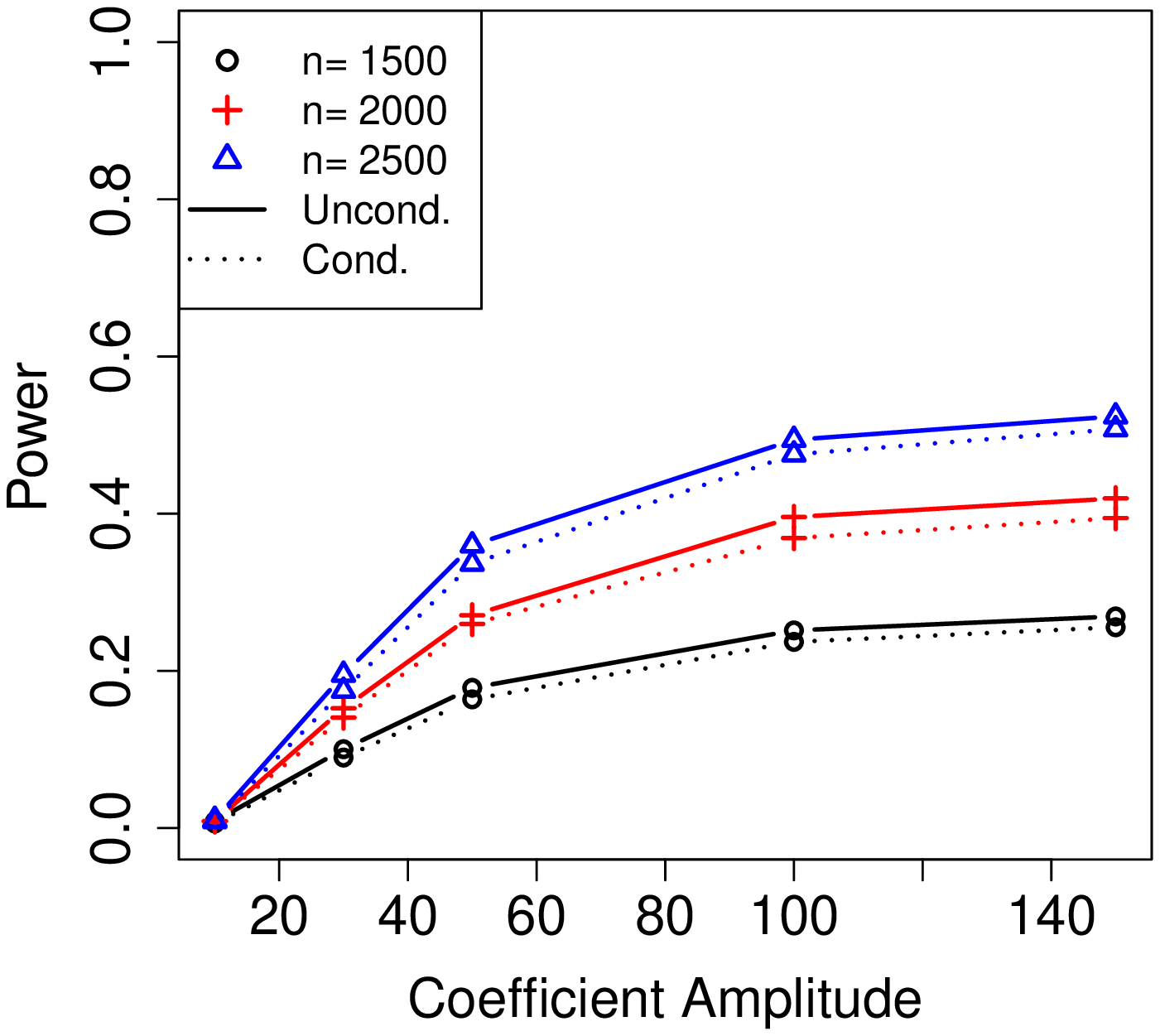} 
  \label{fig:Bin_ggm10}
  }
 \caption{(Logistic regression version of Figure~\ref{fig:ggm}) Power curves of conditional and unconditional knockoffs for $p=2000$ and a range of $n$ for (a) an AR$(1)$ model and (b) an AR$(10)$ model. Standard errors are all below $0.008$.} 
 \end{figure}
 
 %se: ar1: 0.0049, ar10:0.007881473

%\subsection{Discrete Graphical Models}
\begin{figure}[H]\centering
  \subfloat[  ]{     \includegraphics[width=0.45\linewidth]{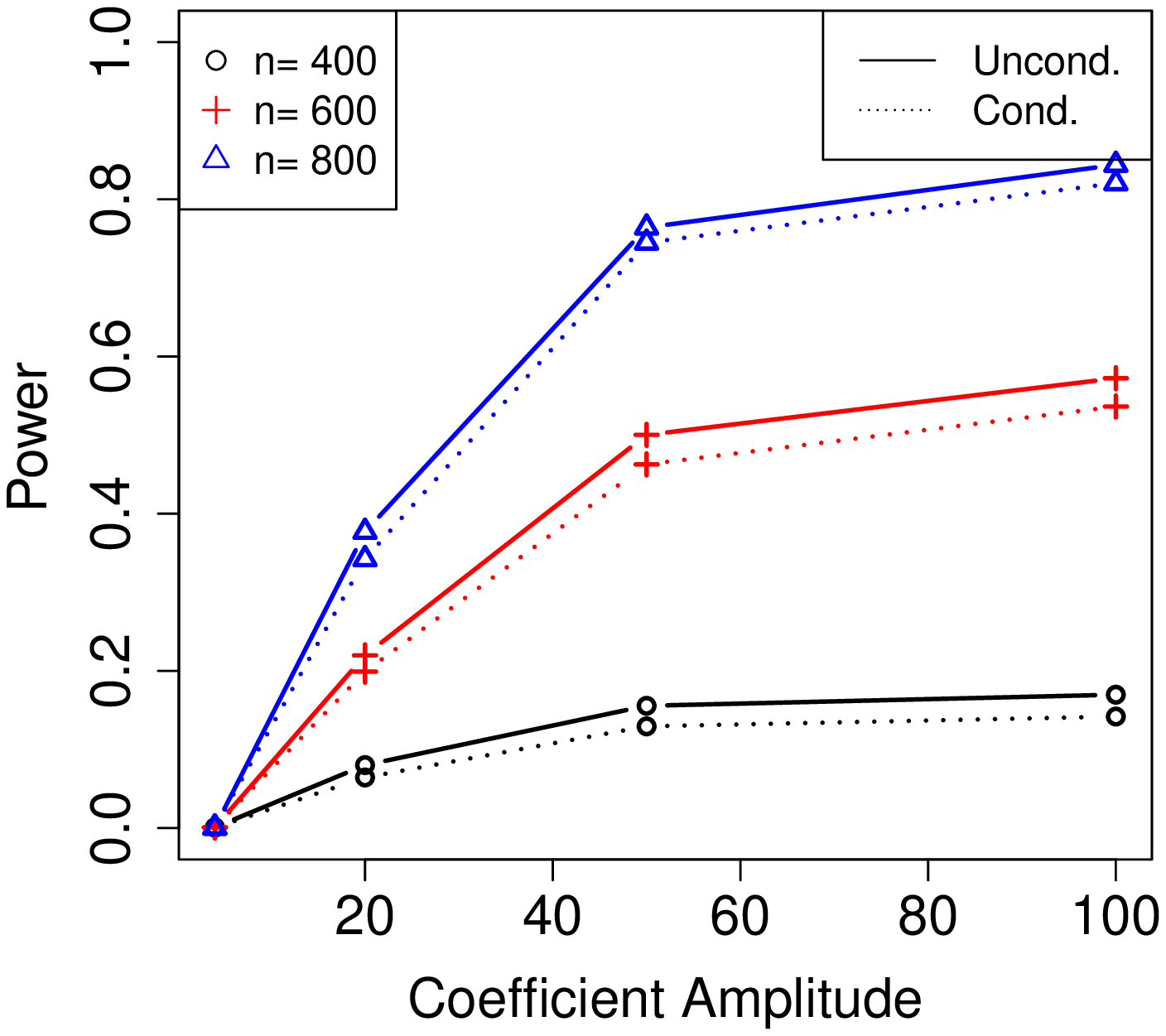} \label{fig:Bin_MC}
  }
  ~
    \subfloat[
    ]{     \includegraphics[width=0.45\linewidth]{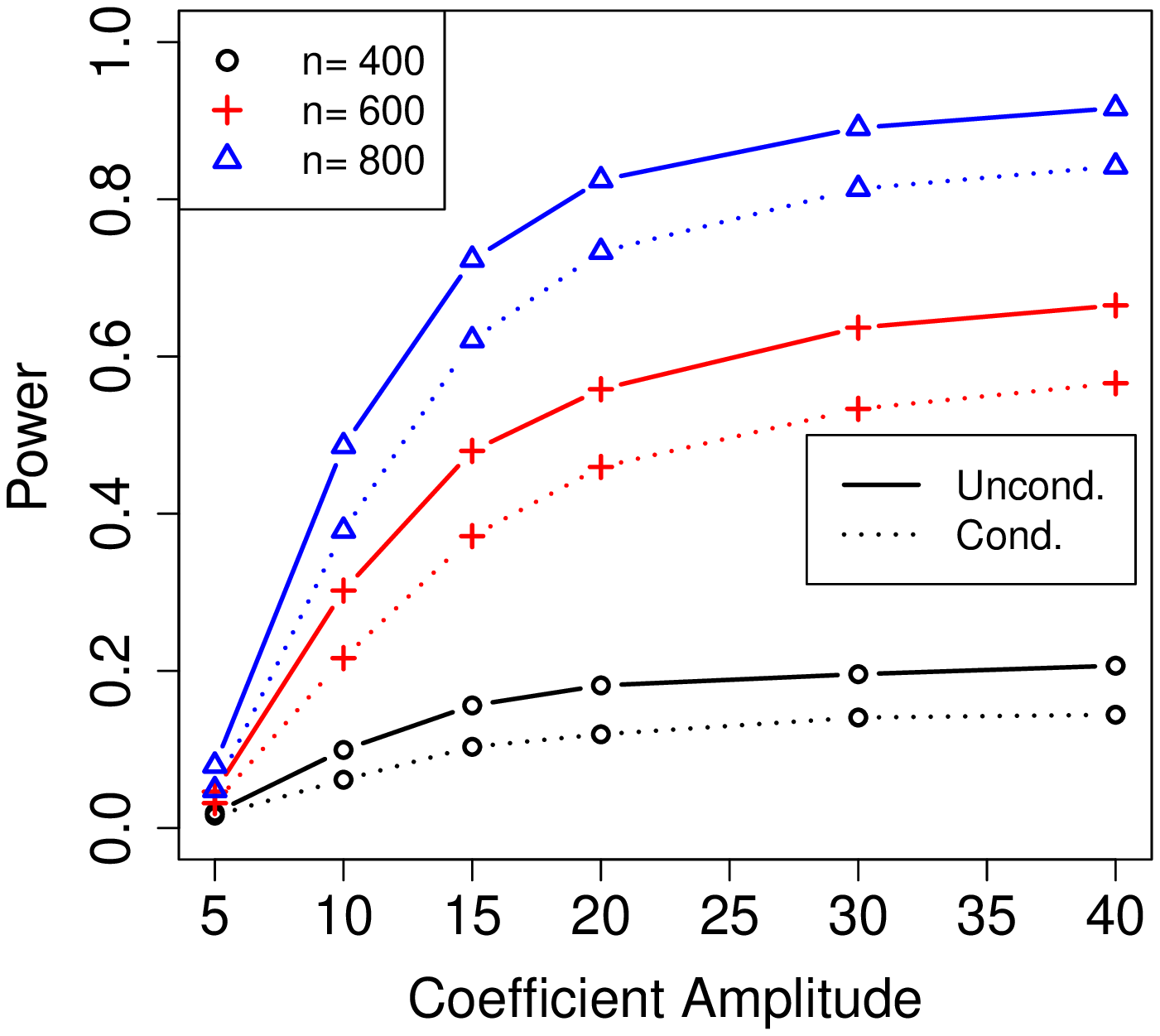} \label{fig:Bin_DGM}}
  \caption{(Logistic regression version of Figure~\ref{fig:dgm-global}) Power curves of conditional and unconditional knockoffs with a range of $n$ for (a) a Markov chain of length $p=1000$ and (b) an Ising model of size $32\times 32$. Standard errors are all below $0.006$.}.\end{figure}
  
  % se: mc: 0.005574358, dgm: 0.005605907

\subsection{Varying the Sparsity and Magnitude of the Regression Coefficients}\label{app:vary}
The following simulations reproduce Figure~\ref*{fig:ggm1} but with varying sparsities and magnitudes. Specifically, the sparsity level $k$ varies between $30$, $60$, and $90$, and the nonzero entries are randomly sampled from Unif$(1,2)$. The message from these experiments is the same as those in the main paper, that is, the power of conditional knockoffs is almost the same as that of unconditional knockoffs even though it does not know the exact distribution of $\bX$.
\begin{figure}[H]\centering
\subfloat {     \includegraphics[width=0.45\linewidth]{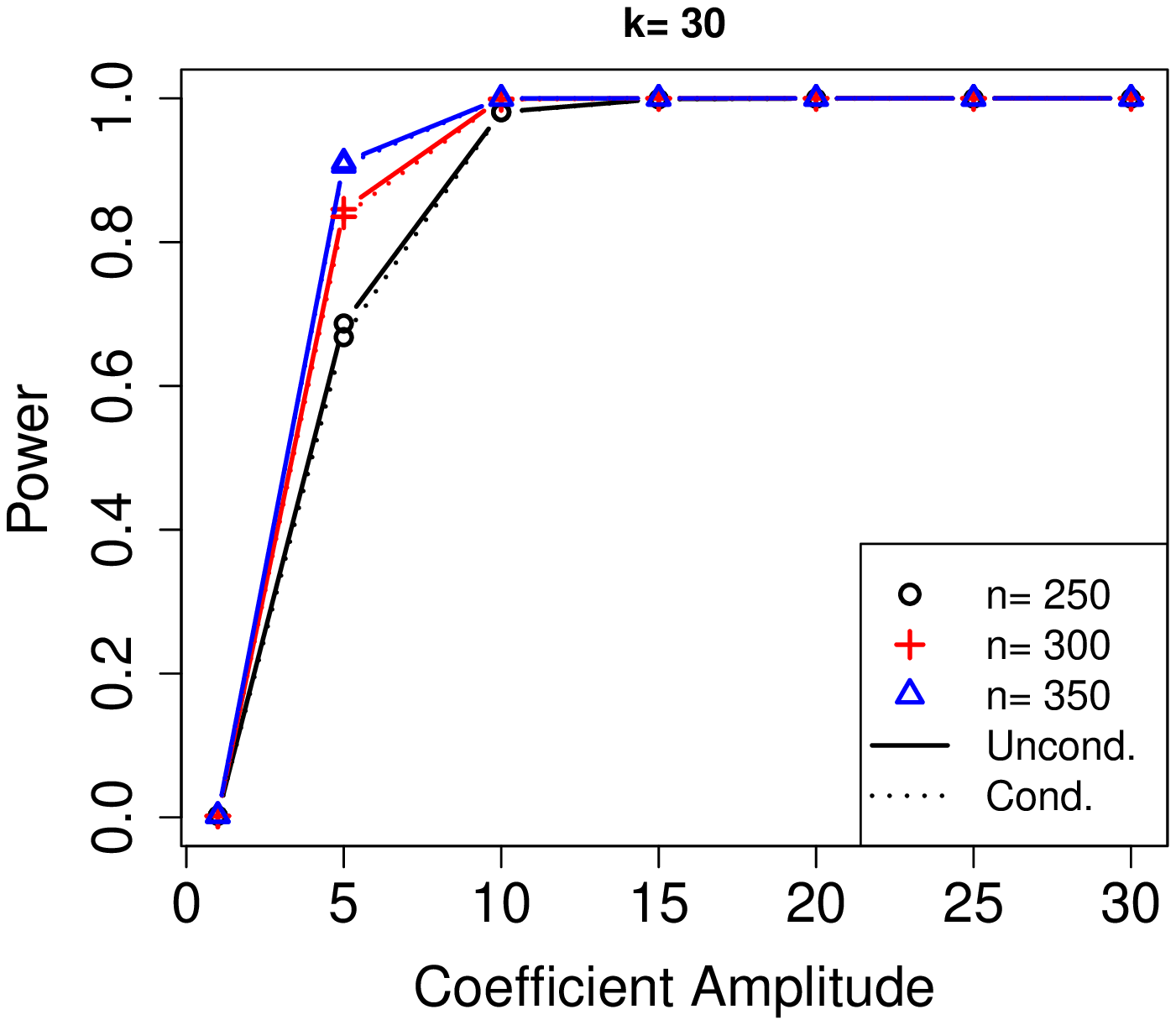} } 
\subfloat {     \includegraphics[width=0.45\linewidth]{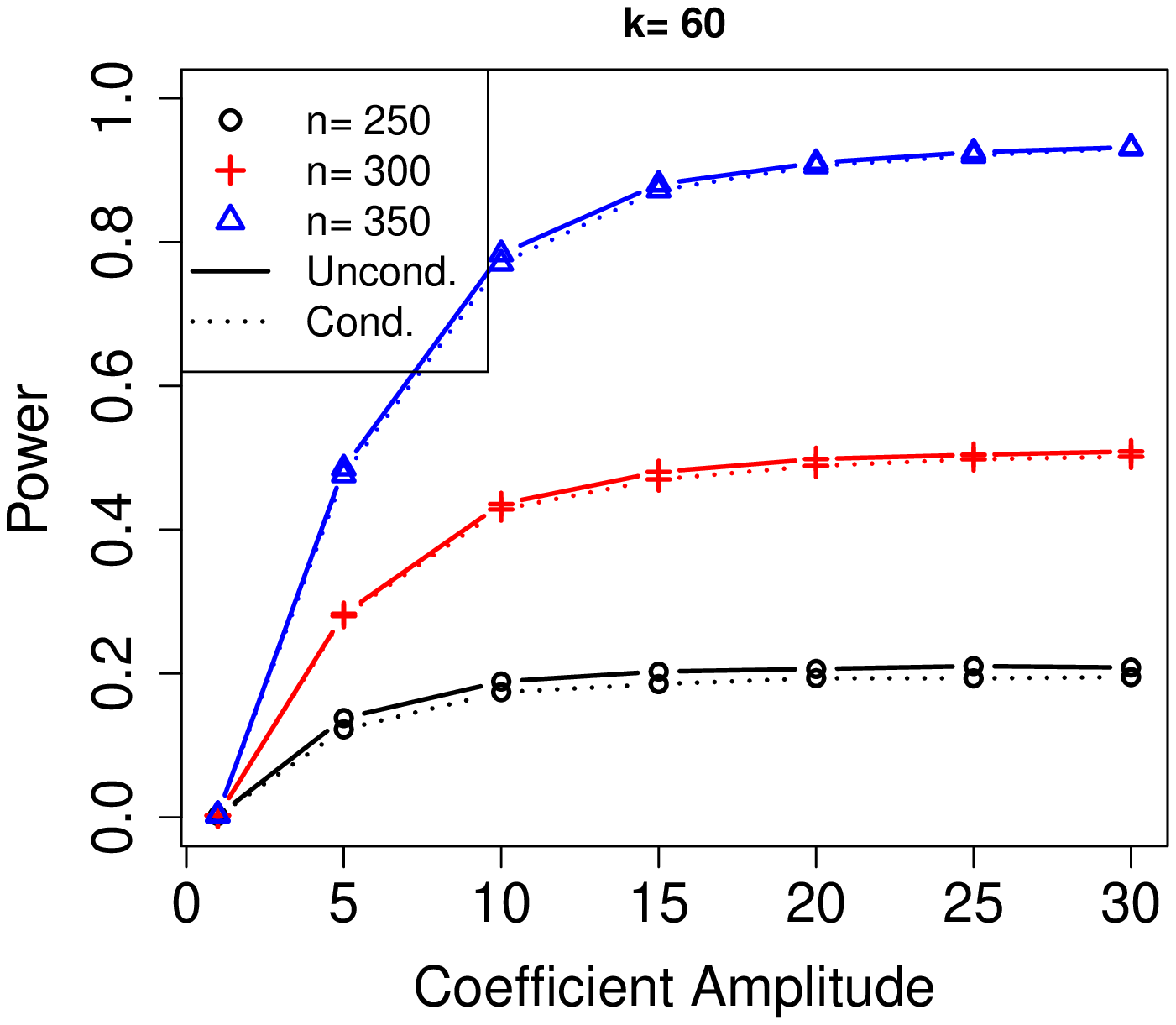} }\\
\subfloat {     \includegraphics[width=0.45\linewidth]{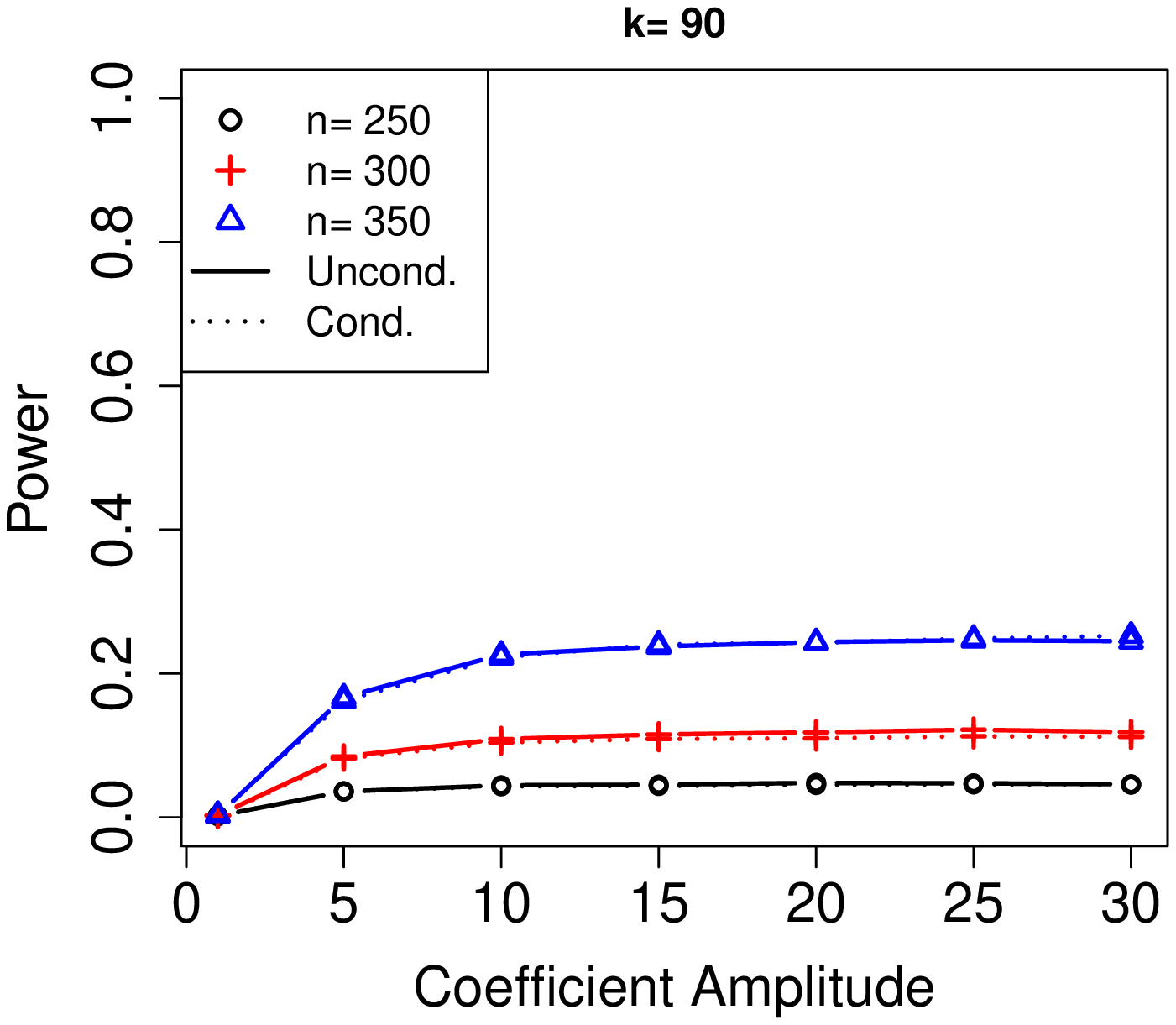} }
\caption{Reproducing Figure~2(a) of the paper with different $\bb$. Power curves of conditional and unconditional knockoffs for $p=2000$ and a range of $n$ for an AR$(1)$ model.  }\label{fig:rep-ar1-various-beta}
 \end{figure}

\rev{

 \subsection{Power of Different Sufficient Statistics}\label{app:simu-different-ss}
 We provide the following experiment to examine the power performance of conditional knockoffs that are generated using different sufficient statistics. It confirms our intuition that conditioning on more generally leads to lower power. 

Specifically, for a Gaussian graphical model, we have run Algorithm~\ref{alg:datasplitting} for a sequence of nested sufficient statistics to see how this choice affects the power. In the following simulation, $X$ is sampled from an AR$(1)$ distribution with autocorrelation coefficient 0.3, and models of (nonstationary) AR($q$) with various $q\geq 1$ are used to model $X$, i.e., each model assume a banded precision matrix with bandwidth $q$, and we increase $q$ beyond 1 to study the effect of more conditioning. Since the models are nested, all of them lead to valid conditional knockoffs. As $q$ grows, the graphical model gets denser and the sufficient statistic conditioned on in Algorithm~\ref{alg:datasplitting} contains more elements (and always contains all the elements of the sufficient statistic conditioned on for all smaller $q$), which can be done by choosing two increasing sequences of blocking sets for the two split data folds and making sure that these two sequences never intersect with each other. Thus, we expect to see some loss of power when $q$ increases. We chose $n=400$, $p=2000$, and the algorithmic parameter $n'$ to be set to $160$, and produced results for a range of $Y\mid X$'s linear model coefficient amplitudes and for $q$ ranging from 1--30; see Figure~\ref{fig:rep-ars}. Although a larger value of $q$ indeed lowers the power, the loss is relatively small in this example despite conditional knockoffs with $q=30$ conditioning on far more than with $q=1$.
%[not justified], especially considering the largest model / largest sufficient statistic is roughly 30 times as large as the smallest. 

\begin{figure}[h]\centering
\includegraphics[width=0.45\linewidth]{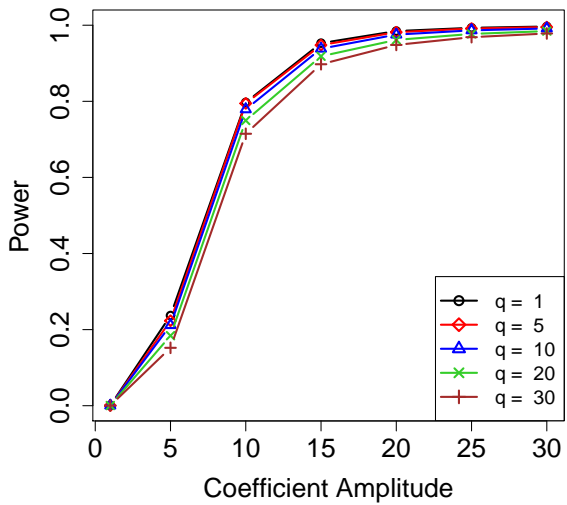}
    \caption{Power curves of conditional knockoffs for AR($q$) models with $p=2000$  as coefficient amplitude varies for various $q$. The nominal FDR level is $0.2$. Standard errors are all below 0.007.}  \label{fig:rep-ars}
 \end{figure}  

\subsection{Gaussian Graphical Models with Unknown Edge Supersets}\label{app:simu-unknown-graph}
The conditional knockoff generation in  Section~\ref{sec:eg-ggm} requires knowing a superset of the true edge set of the Gaussian graphical model. 
We present the following experiment to preliminarily examine the idea mentioned in  Remark~\ref{rem:unknown-graph} for when such a superset is not known a priori and one has to estimate the edge set (or its superset) using the data. 

%?? Note that we cannot provide formal guarantees for this procedure or any other when one does not know a superset of the edge locations, we did not think to try and explore this further. But the reviewers suggestion to at least try some simulations is an excellent one.

Suppose the true covariance matrix $\bS$ is a rescaled (to have diagonal entries equal to 1) version of $\bS^{(0)}$, where $(\bS^{(0)})_{j,k}^{-1} = \bs{1}_{j=k} - \frac{1}{7}\bs{1}_{1\leq |j-k|\leq 3}$. In other words, every node in the true graph is connected to its 6 nearest neighbors. In the following simulations, we set $p=400$ and $n=200$. 

We can estimate the edge set $E$ by the nonzero entries $\hat{E}$ of the estimated precision matrix using the \textit{graphical Lasso} \citep{JF-TH-RT:2008}, which is implemented via the \texttt{R} package \texttt{huge}. The tuning parameter of the graphical Lasso is selected by the standard method \textit{StARS} \citep{HL-ea:2010}. 
%The authors of that paper show their method outperforms cross-validation, which dramatically overselects edges---overselection is generally good for conditional knockoffs, since we only need to have a \emph{superset} of the true edges, but we found even the less-conservative StARS procedure worked well in our simulation. 
Once $\hat{E}$ is computed, we can then construct conditional knockoffs as if $\hat{E}$ were given. The blocking set used in our algorithms is obtained by Algorithm~\ref{alg:sparse-gaussian-or} with input $n'=80$. We additionally consider the case where a set of $n_u=1,600$ unlabeled data points is available and is used together with the labeled covariates to estimate the edge set. 

The FDR and power curves are shown in Figure~\ref{fig:rep-FDR-power-estimateE}. ``Label Cond.'' refers to conditional knockoffs generated with $\hat{E}$ estimated using only labeled data, and ``Unlabel. Cond.'' refers to conditional knockoffs with $\hat{E}$ estimated additionally with the unlabeled data. Both methods control the FDR and in fact are conservative. One might attribute the FDR control to some over-conservative choice of graphical Lasso tuning parameter, but in fact $\hat{E}$ estimated with just the labeled data, although it tended to find a larger graph than the truth (its maximal degree was often above 20), also missed around 40 true edges on average. Unsurprisingly, the use of unlabeled data improves the power by improving the estimate of the edge set, with much fewer false negative and false positive edges.
%, based on the histograms in Fig~\ref{fig:rep-hist-estimateE}, 
%The estimate using only labeled covariates often missed more than $40$ true edges and its maximal degree was often above $20$. The estimate using additional unlabeled dat never missed any true edge in all simulations, and its maximal degree was at most $13$, much smaller than the previous one. 
\begin{figure}[h]\centering
\subfloat{\includegraphics[width=0.45\linewidth]{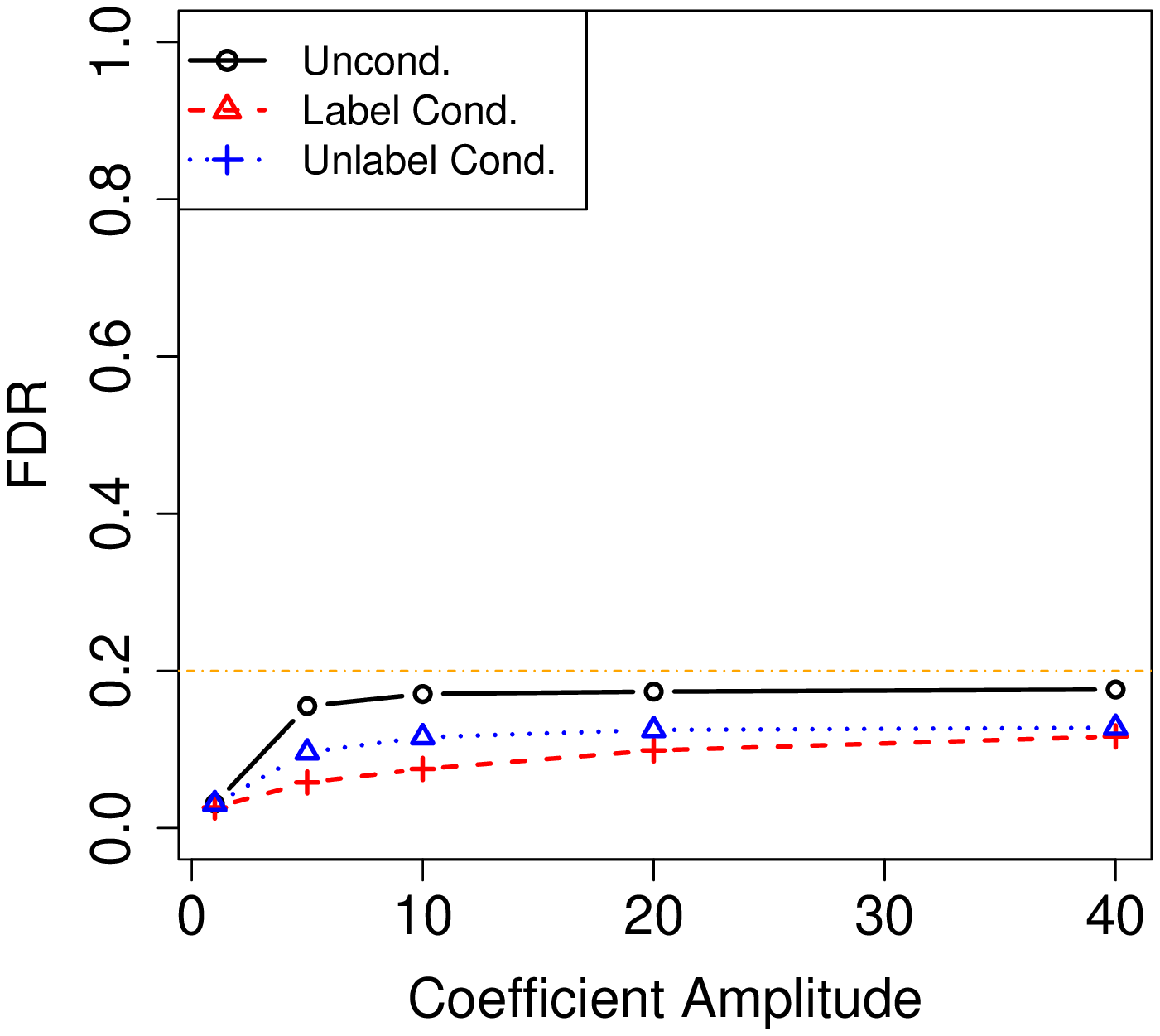} }
\subfloat{\includegraphics[width=0.45\linewidth]{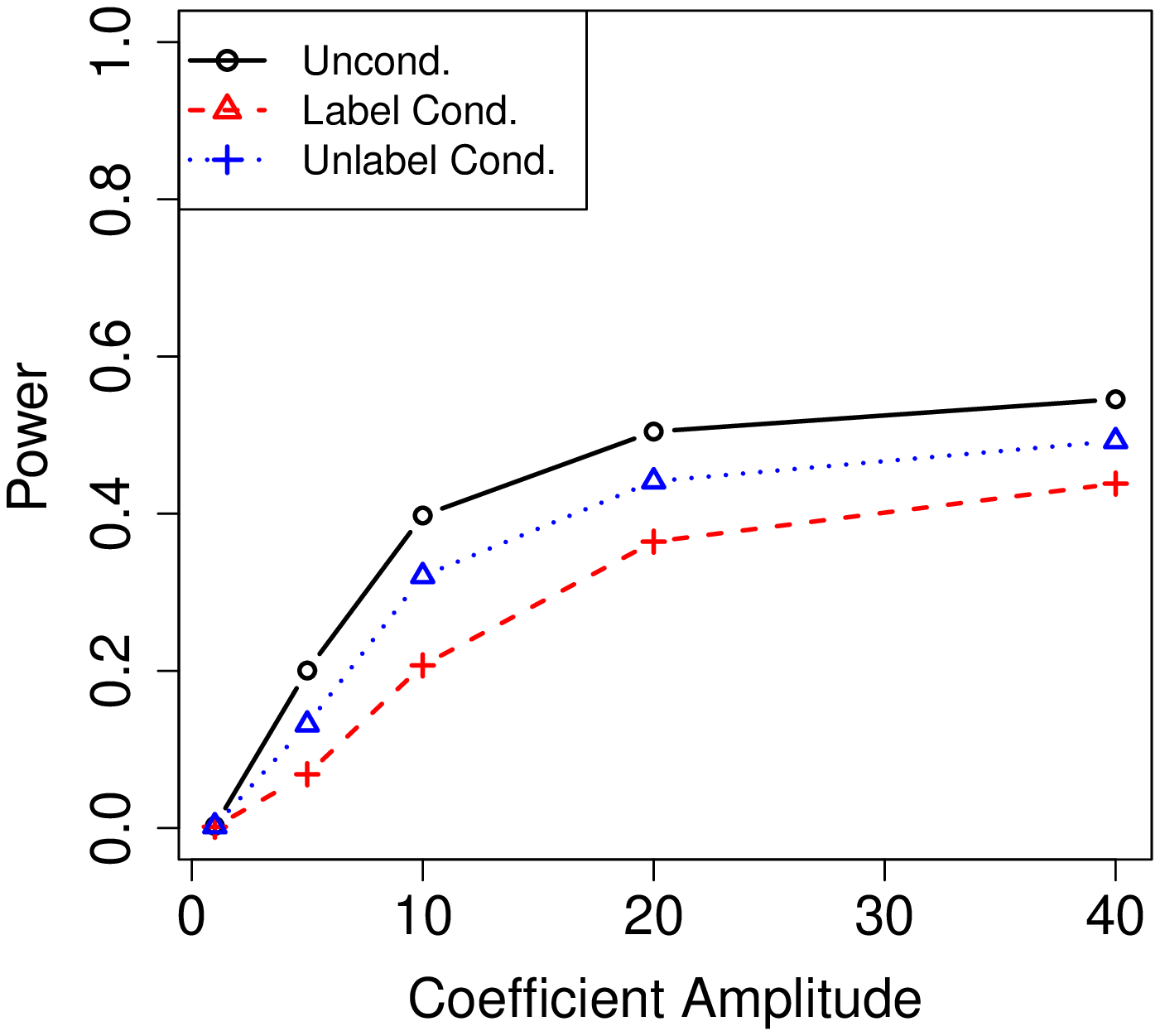} }
    \caption{FDR and Power curves of unconditional knockoffs and (approximate) conditional knockoffs using an estimated edge set either from purely labeled or from both labeled and unlabeled data.  $p=400$ and $n=200$. Standard errors are all below 0.007.}  \label{fig:rep-FDR-power-estimateE}
 \end{figure}

 \subsection{Robustness to Model Misspecification}\label{app:simu-discrete}
The current paper focuses on the cases where the models for the covariates are known and well-specified. In practice, practitioners may not know what the true model is. Here we provide an experiment to examine the robustness of  Gaussian conditional knockoffs ($n>2p$). 
The following robustness experiment constructs a set of distributions that approximate a multivariate Gaussian by discretizing it at different resolutions by varying a parameter $K$. 

We first generate $X^{(0)}\sim \N(0, \bS)$ where $\bS_{i,j}=0.3^{|i-j|}$, and then discretize each coordinate as follows
\[
X_{j} = \frac{ \lfloor X^{(0)}_{j} \times K + 1/2 \rfloor }{K}, \quad j=1,\dots, p. 
\]
In other words, $X^{(0)}_{j}$ is rounded to the nearest $\frac{1}{K}$-grid value. Since $|X_{j}-X^{(0)}_{j}| \leq \frac{1}{K}$, the larger $K$, the closer $X$ is to a multivariate Gaussian vector, and indeed as $K\rightarrow \infty$,  $X\rightarrow X^{(0)}$ and becomes multivariate Gaussian. However, for small $K$, the distribution is very far from Gaussian.
For conditional knockoffs, we pretend that $X$ is drawn from a multivariate Gaussian distribution and directly apply Algorithm 3.1. To get a baseline for power (since changing $K$ not only affects the model misspecification, but also changes the nature of the data-generating distribution and thus the power of any procedure), we also generate exactly-valid unconditional knockoffs for $X$ by discretizing an unconditional knockoff of $X^{(0)}$ with the same $K$ (of course this procedure would be impossible in practice, since $X^{(0)}$ is unobserved).

We fix $p=1,000$, linear model coefficient amplitude at $4$, vary $n\in\{2001,3000,4000\}$ and vary $K\in\{1/2, 1, 2, 3, \infty\}$. The other details of the experiment are the same as in Figure 1a of the paper, where the response is drawn from $Y_{i}\mid X_i\sim N(X_{i}\tp \bb/ \sqrt{n},1)$.  The result is shown in Figure~\ref{fig:rep-discrete}. 

Note that $K=1/2$ produces a distribution that is almost entirely concentrated on just three values $\{-2,0,2\}$, making it extremely non-Gaussian, yet the FDR is controlled quite well for all values of $n$ at this $K$ value and all others. The power difference between conditional knockoffs and unconditional knockoffs is also quite insensitive to $K$ and, as seen in all other simulations, quite small for all $n$ except when $n\approx 2p$ (in the $n\approx 2p$ setting the power gap is substantial, although conditional knockoffs still has quite a bit of power). 
%only varies slightly across different $K$. This indicates the conditional knockoff still performs well under model-misspecification. 
%(wrong!) Note that different $K$'s correspond to different underlying distributions of $Y\mid X$ and are not comparable in general. 

%The power with a small $K$ is higher because the conditional knockoff, unlike the observed covariate, takes continuous values, and is thus by no means looks like the observed covariates. As $K$ increases, the power drops and converges to the curve
%the power drops and looks closer to the curve in Fig~\ref{fig:rep-LDG1}, a plot we showed in the last letter. 

\begin{figure}[h]\centering
\subfloat{ \includegraphics[width=0.45\linewidth]{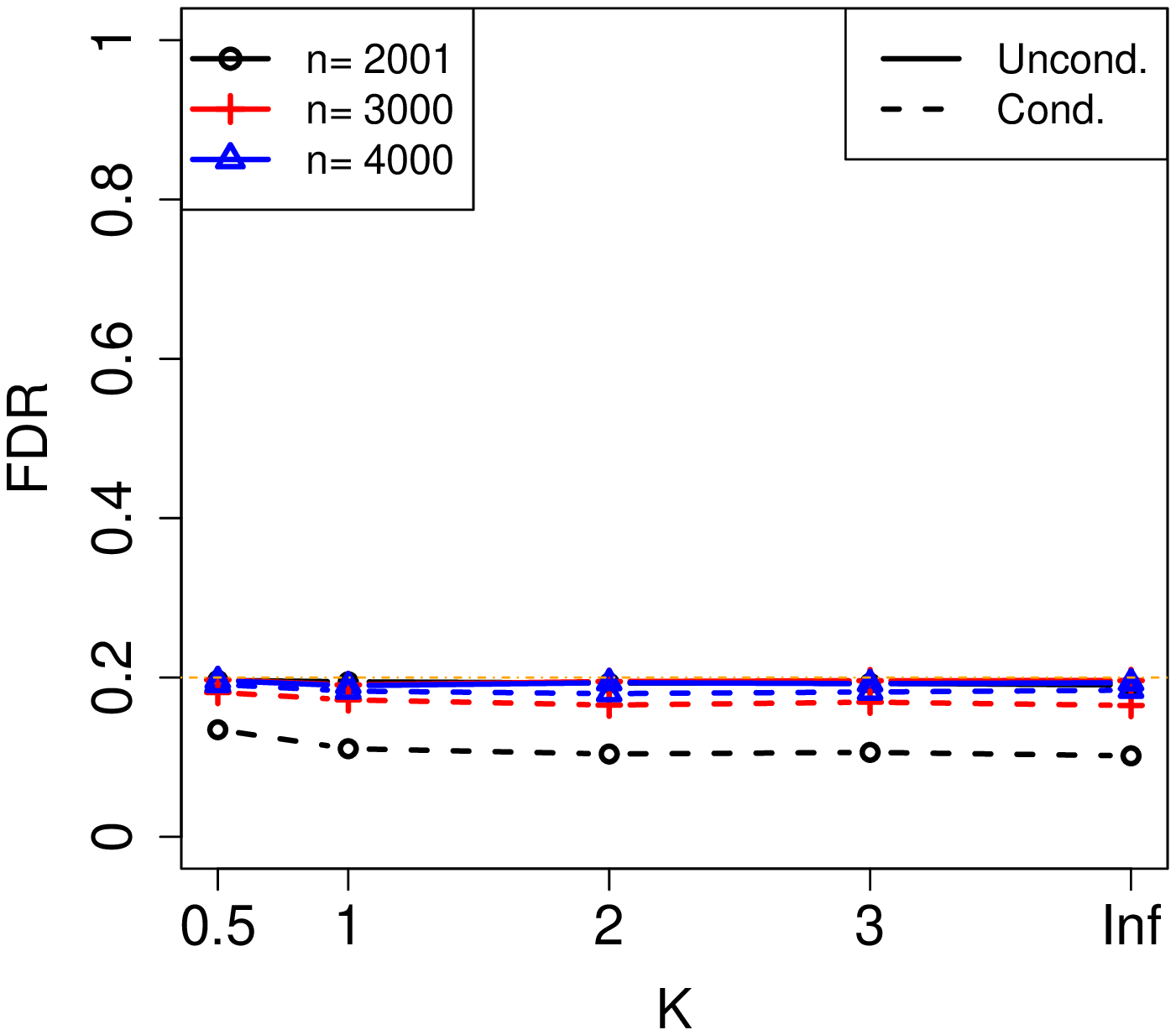}} 
\subfloat { \includegraphics[width=0.45\linewidth]{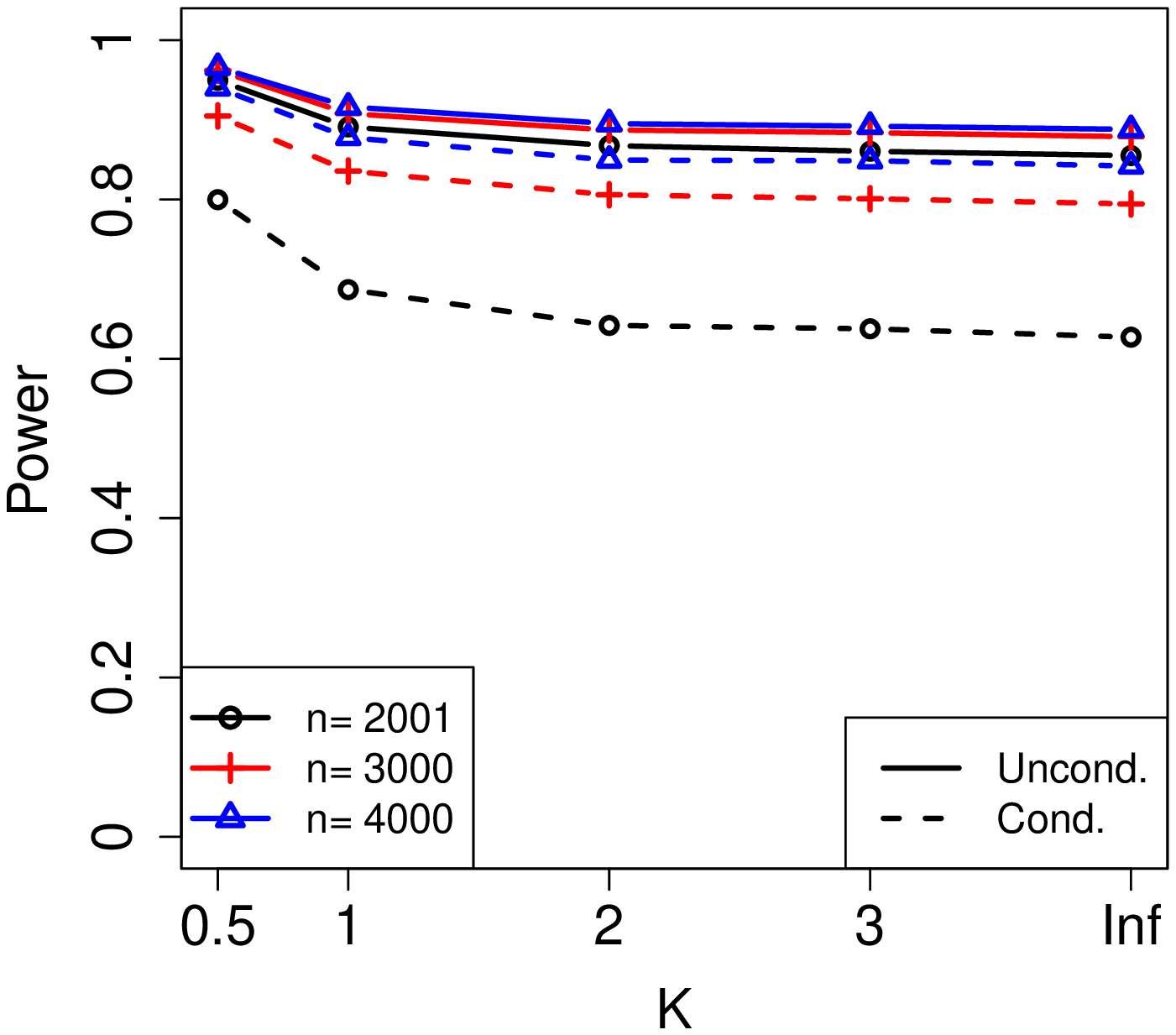}}\\
    \caption{FDR and Power curves of unconditional and conditional knockoffs for a discretized AR($1$) model with $p=1000$ and linear model coefficient amplitude $4$. The nominal FDR level is $0.2$. Standard errors are all below 0.004.}  \label{fig:rep-discrete}
 \end{figure}
 
 }

\begin{comment}
\section{A Sequential Knockoff Construction for Gaussian Graphical Models}
\input{appSeqGGM}
\end{comment}

\end{document}